\documentclass[12pt]{article}
\usepackage[a4paper,margin=3cm]{geometry}

\usepackage{graphicx}
\usepackage{bbm}
\usepackage{color}
\usepackage{amsmath}
\usepackage{amsfonts}
\usepackage{graphicx, float}
\usepackage{amssymb, amsthm}
\usepackage{natbib}
\usepackage{hyperref}
\usepackage{algorithm, algpseudocode}
\usepackage{caption, subcaption}
\usepackage{authblk}
\usepackage[shortlabels]{enumitem}
\usepackage{nicematrix}

\newcommand{\vect}[1]{\boldsymbol{\mathbf{#1}}}
\newcommand{\x}{\vect{x}}
\newcommand{\X}{\vect{X}}

\newcommand{\z}{\vect{z}}
\newcommand{\y}{\vect{y}}
\newcommand{\vc}{\vect{c}}
\newcommand{\vbeta}{\vect{\beta}}
\newcommand{\veta}{\vect{\eta}}
\newcommand{\vtheta}{\vect{\theta}}

\newcommand{\R}{\mathbb{R}}

\newcommand{\sumj}{\sum_{j=1}^p}
\newcommand{\sumi}{\sum_{i=1}^n}
\newcommand{\prodj}{\prod_{j=1}^p}
\newcommand{\prodi}{\prod_{i=1}^n}

\newcommand{\xij}{x_{ij}}
\newcommand{\betaj}{\beta_{j}}
\newcommand{\zj}{z_{j}}
\newcommand{\cj}{c_{j}}
\newcommand{\etai}{\eta_{i}}
\newcommand{\thetai}{\theta_{i}}
\newcommand{\yi}{y_{i}}
\newcommand{\vmu}{\vect{\mu}}
\newcommand{\DparamScalar}{a}
\newcommand{\DparamVector}{\boldsymbol{\xi}}

\newcommand{\iid}{i.i.d. }
\newcommand{\ie}{i.e., }
\newcommand{\eg}{e.g., }

\theoremstyle{thmstyleone}
\newtheorem{theorem}{Theorem}
\newtheorem{lemma}[theorem]{Lemma}
\newtheorem{corollary}[theorem]{Corollary}

\theoremstyle{thmstyletwo}

\newtheorem{remark}{Remark}
\theoremstyle{thmstylethree}
\newtheorem{definition}{Definition}
\usepackage[doublespacing]{setspace}
\usepackage[fontsize=12pt]{fontsize}

\begin{document}

\title{Conjugate priors for Bayesian Variable Selection in exponential family Generalized Linear Models}

\author[1,2]{Lucia Filippozzi (\texttt{lucia.filippozzi@unitn.it})}
\author[3,4]{I\~{n}igo Urteaga (\texttt{iurteaga@bcamath.org})}
\author[1]{Claudio Agostinelli (\texttt{claudio.agostinelli@unitn.it})}

\affil[1]{\small Department of Mathematics, University of Trento, Trento, Italy.}
\affil[2]{\small Fondazione Bruno Kessler (FBK), Trento, Italy.}
\affil[3]{\small BCAM --- Basque Center for Applied Mathematics, Bilbao, Spain.}
\affil[4]{\small Ikerbasque --- Basque Foundation for Science, Bilbao, Spain.}
\date{}

\maketitle

\begin{abstract}
Covariate selection in Generalized Linear Models (GLMs) is a fundamental problem in statistics, as including irrelevant predictors might lead to overfitting and poor interpretability, 
while omitting relevant ones might result in biased estimates.

Most Bayesian approaches to variable selection 
---including spike-and-slab priors and continuous shrinkage priors--- have key limitations, e.g., (i) are based on non fully conjugate formulations, (ii) are restricted to a linear model formulation, or (iii) lack posterior consistency guarantees for the variable selection procedure and model parameters.

In this work, we propose a fully Bayesian hierarchical and conjugate GLM framework for covariate selection, applicable to any distribution in the exponential family, based on modeling a binary inclusion indicator that directly encodes covariate inclusion in the linear predictor.
In our approach, variable selection and parameter estimation are performed simultaneously, incorporating both sources of uncertainty in posterior inference. Consequently, our methodology provides a valid post-model Bayesian selection procedure.

We present theoretical guarantees of the proposed fully conjugate Bayesian variable selection for GLMs, establishing posterior consistency of both the inclusion indicators
and the active regression coefficients. We derive an efficient Gibbs Sampling algorithm with a corresponding \textsf{R} package implementation. We validate the proposed method on synthetic and real-world datasets, demonstrating competitive predictive and inferential performance.

\end{abstract}

\maketitle

\section{Introduction}
\label{sec:intro}
Covariate or variable selection in Generalized Linear Models (GLMs) ---identifying which covariates are relevant to explain a response--- is a fundamental problem in statistics: including irrelevant covariates leads to overfitting and poor interpretability, while omitting relevant ones biases estimates and degrades prediction.

Frequentist approaches typically address it through penalized likelihood, \eg the LASSO \citep{TibshiraniLASSO}, which return a single selected model ---but do not by themselves quantify the uncertainty of the selection. Bayesian approaches instead place priors on the regression coefficients, \eg spike-and-slab or continuous shrinkage priors,
and perform selection and estimation jointly through the posterior (see Section~\ref{sec:background}).

Existing Bayesian variable selection methods, however, are limited in at least one of three ways: ($i$) they rely on non-conjugate prior formulations, which require approximations beyond the Gaussian linear model; ($ii$) their theoretical guarantees are restricted to linear Gaussian regression; or ($iii$) their posterior consistency results for GLMs rely on sub-Gaussianity of the score function, which excludes common models such as Poisson regression \citep{narisetty2014spike, narisetty2019Skinny, lee2021bayesian, cao2024_NLP_in_glm}.

In this work, we propose a fully conjugate Bayesian hierarchical GLM for covariate selection, applicable to any distribution in the exponential family. We augment the linear predictor $\etai$ with a binary inclusion vector $\z \in \{0,1\}^p$, so that $\etai = \x_i^\top(\vbeta \circ \z)$; place the conjugate GLM prior of \citet{chen2003conjugate_prior_glm} on $\vbeta$ conditionally on $\z$; and induce sparsity through a Beta--Bernoulli prior on $\z$ rather than by shrinking $\vbeta$ toward zero.
Conjugacy is what makes the approach applicable across the exponential family: all conditionals are available in conjugate form,
yielding a Gibbs sampler that performs selection and estimation jointly.

Our model's posterior decomposes into active and non-active components naturally, as guided by observations, which leads to a Skinny Gibbs Sampler \citep{narisetty2019Skinny}. Conditionally, the posterior of the active coefficients is driven by the GLM's data likelihood,
while the posterior of each excluded coefficient reverts to its conjugate, non-sparse prior. Marginally, the posterior of $\vbeta$ is a mixture of two components with the mixture probabilities determined by the posterior inclusion probabilities of $z_j$ after observing the data,
not fixed a-priori.

Our contributions are:
\begin{enumerate}
\item \textbf{A fully conjugate, indicator-based model for exponential-family GLMs.} Our use of inclusion indicators in the linear predictor follows the seminal work of \citet{kuo1998variable}, and our hierarchical structure resembles that of \citet{dellaportas2002bayesian}.
    Unlike these, we place conjugate GLM priors on $\vbeta$ given $\z$, which yields conjugate full conditionals for any exponential-family likelihood. \citet{chen2008bayesian_glm} also use the conjugate priors of \citet{chen2003conjugate_prior_glm} for variable selection in GLMs, but within a model-comparison framework: a prior is specified for each candidate model and selection relies on posterior model probabilities computed over subsets of covariates, which becomes infeasible beyond low-dimensional settings. We instead embed the inclusion indicators $\z$ in a hierarchical model and sample them simultaneously with $\vbeta$ in a Gibbs sampler.
\item \textbf{Posterior consistency for GLMs, beyond sub-Gaussian conditions.} We prove posterior consistency of the inclusion vector $\z$ (Theorem~\ref{thm:consistency_z}), and of the active regression coefficients, together with their asymptotic normality (Theorem~\ref{thm:consistency_beta}). Our results hold under a $\varpi$-subexponential distribution function model with $\varpi > 1/2$, which covers also Poisson regression, excluded by sub-Gaussian analyses \citep{lee2021bayesian, cao2024_NLP_in_glm}, and extends the recent results of \citet{lee2025advances} that require $\varpi \geq 1$.
\item \textbf{Empirical validation.} On synthetic and real data for Linear, Logistic and Poisson models, the method (an implementation is provided as an \textsf{R} package) achieves accurate sparsity recovery and parameter estimation, with predictive performance competitive with spike-and-slab and horseshoe baselines. We also observe good mixing of the Gibbs sampler in our experiments (Remark~\ref{rem:mixing_trace}).
\end{enumerate}

The rest of the work is structured as follows. Section~\ref{sec:background} reviews related Bayesian variable selection methods, Section~\ref{sec:glm_vs} presents the model and Gibbs sampler, Section~\ref{sec:posterior_consistency} the posterior consistency results, and Section~\ref{sec:numerical_experiments} the empirical evaluation.
Posterior derivations, proofs, sampling details and additional experiments are provided in the Supplementary Material.

\section{Background}
\label{sec:background}

Covariate selection has a long history in frequentist statistics, where it is typically addressed through penalized likelihood that trades off model fit against complexity.
The LASSO \citep{TibshiraniLASSO, donohoL1penalty} uses an $L^1$ penalty that sets some coefficients exactly to zero;
Bridge estimators \citep{frank1993BRIDGEpenalty, fu1998BRIDGEpenalty} generalize it to $L^q$ penalties;
and non-convex penalties such as SCAD \citep{fan2001_penalties} reduce the bias on large coefficients.

Bayesian methods instead place priors on the regression coefficients whose posteriors express sparsity;
see \citet{ohara2009review} and \citet{Lu2022bayesianapproaches} for extensive reviews.
Two main classes of priors are typically used:
the \textit{``spike-and-slab'' prior} and \textit{shrinkage priors}.

Spike-and-slab \citep{george1993spike, mitchell1988spike, ishwaran2005spike, narisetty2014spike, malsiner2016_spike} model each regression coefficient as a mixture of a point mass (or narrow distribution) at zero that enforces exclusion (the ``spike'') and a diffuse, often Gaussian, distribution (the ``slab'') allowing for nonzero values:
\ie $\betaj \sim (1-\gamma_j)\mathcal{N}(0, \delta) + \gamma_j \mathcal{N}(0, \sigma^2)$ with $\delta \rightarrow 0$.
Examples are the Stochastic Search Variable Selection (SSVS) \citep{george1993spike}, the Gibbs variable selection \citep{dellaportas2002bayesian}, and the unconditional prior of \citet{kuo1998variable}, where the indicator variables enter the regression equation directly.
Zellner's $g$-priors \citep{liang2008_gpriors} are typically used in a related way, as priors on the coefficients of each candidate model combined with a prior over the model space.

In the continuous \textit{shrinkage prior} literature, we find the Bayesian Lasso \citep{park2008bayesian, cai2011_bayesianLasso},
and Horseshoe priors \citep{bhadra2019lasso, bhadra2021horseshoe}.
These priors pose covariate selection as a shrinkage mechanism, pushing irrelevant coefficients toward zero without explicit inclusion indicators, often through global-local structures that adapt the amount of shrinkage to each covariate.
Nonlocal Priors (NLPs), defined by \citet{johnson2010_NLP} as priors that vanish when a parameter equals its null value,
were extended to model selection by \citet{johnson2012_NLP} through product moment (pMoM) and product inverse moment (piMoM) priors,
and have been adapted to high-dimensional settings and GLMs \citep{rossell2013_NLP_highdim, shin2018_NLP_highdim, cao2024_NLP_in_glm}.

The hierarchical GLM we propose in the next section is based on 
conjugate priors under conditional independence between active and non-active components and a GLM likelihood, explicitly dependent on indicator variable $\z$, unifying many of the above approaches within a coherent framework.
The corresponding posterior captures information from observations only for the truly active covariates, reverting to the original ---not necessarily sparse--- prior for the inactive covariates.

\section{Covariate Selection in Generalized Linear Models}
\label{sec:glm_vs}

Let $\vect X$ be a $n\times p$ matrix of covariates, and target vector $Y$ of dimension $n$. Consider a Generalized Linear Model (GLM) with likelihood function $f(\cdot)$ in the exponential family, canonical link function $g$, canonical parameter $\vtheta$ and linear predictor $\veta$ ($i=1, \ldots, n$): 
\begin{equation}
  \yi \mid \thetai, \tau \sim f_Y(\yi\mid  \thetai, \tau) = \exp \left\{ \frac{\yi \thetai - b(\thetai)}{A_i(\tau)} + c(\yi, \tau)\right\} \label{eq:classic_glm}
\end{equation}  
with $\etai \mid \x_i, \vbeta = \x_i^\top \vbeta$ (linear predictor) and $\vmu = \mathbb{E}(Y \mid  \X ) = g^{-1}(\veta)$ (inverse link function). We consider the link function $g$ to be canonical, $g(\vmu) = \vtheta = \veta$,
and functions $A(\cdot)$, $b(\cdot)$, $c(\cdot, \cdot)$ determine each specific member of the exponential family.
Alternative but equivalent parameterizations of the likelihood exist; see \citep{bedbur2021multivariate} for a detailed discussion. We recall that the scale parameter $\tau$ appears only in certain cases (\eg the Linear model).

To perform Bayesian GLM variable selection, we introduce a binary random vector  $\z \in \{0,1\}^p$, 
where each component $\zj$ indicates whether the $j$-th covariate is included ($\zj = 1$) or excluded ($\zj = 0$) from the model.
We pose a Bernoulli distribution over the indicators,
and complete the Bayesian hierarchical formulation by assigning conjugate priors to $\z$.
We also use priors on $\vbeta$ and $\tau$ that are conjugate with respect to the likelihood model.

Specifically, the linear predictor is modified to depend only on the subset of covariates selected by $\z$, with the corresponding regression coefficients $\vbeta$ modeled conditionally on $\z$, \ie
\begin{align}
\label{eq:glm_variable_selection_model_begin}
    \yi \mid \thetai, \tau \ 
        & \sim \ f_Y(\yi\mid  \thetai, \tau) =  \exp \left\{ \frac{\yi \thetai - b(\thetai)}{A_i(\tau)} + c(\yi, \tau)\right\} \\
    \thetai & = \etai \mid  \x_i, \z, \vbeta \ = \ \x_i^\top (\vbeta \circ \z) = \sumj \xij\betaj \zj  \nonumber \\  
    \z  \mid \vc \             & \sim \ \prod_{j=1}^p Bern (\cj)  \nonumber \\
    \vc \mid \alpha \          & \sim \ \prod_{j=1}^p Beta\left( \frac{\alpha}{p}, 1 \right)  \nonumber \\
    \vbeta, \ \tau \mid \z           & \sim \ p(\vbeta \mid \z, \tau) p(\tau) 
    \nonumber
\end{align}
where the $\circ$ operator denotes the element-wise product. In a standard GLM, identifiability of $\vbeta$ is typically ensured by assuming that $\X$ has full column rank,
\ie $\operatorname{rank}(\X) = p$,
and that the link function $g$ is strictly monotone.
Under these conditions, the mapping from $\vbeta$ to the mean $\mu = g^{-1}(\X\vbeta)$ is injective and $\X \vbeta = \X \vbeta^*$ implies $\vbeta = \vbeta^*$.
In our model,
the linear predictor is $\eta = \X(\vbeta \circ \z)$,
and the identifiability of the parameters $\vbeta$ depends on both the design matrix $\X \in \mathbb{R}^{n \times p}$ and the selection vector $\z \in \{0,1\}^p$.
In particular, it holds that $\X(\vbeta \circ \z) = \X(\vbeta^* \circ \z)$ implies $\vbeta \circ \z = \vbeta^* \circ \z$.
Therefore, parameter identifiability holds only for the subvector $\vbeta^{(1)}$ of active coefficients, or equivalently, the product $\vbeta \circ \z$, when the submatrix $\X^{(1)}$ has full column rank.
Because components of $\vbeta$ associated with zero entries in $\z$, $\vbeta^{(0)}$, do not influence the likelihood, these are not identifiable from the data.

From a Bayesian perspective, identifiability may also be assessed via the posterior distribution. As discussed by \citet{gelfand1999identifiability}, even in a nonidentifiable likelihood setting, the posterior of $\vbeta$ is proper (and hence meaningful for inference) only if the prior on $\vbeta$ is proper, condition that is always satisfied in our model. We allow the total number of covariates $p$ to grow with $n$, potentially exceeding it. Precises growth rates are presented in Supplementary Material Section \ref{apd:proofs}. We choose proper conjugate priors for $\vbeta$ and $\tau$  consistent with the exponential family structure of the likelihood,
in line with the work by \citet{chen2003conjugate_prior_glm}.

In a classical GLM with canonical link a conjugate prior for the regression coefficients $\vbeta$ takes the form
\begin{equation}
    \label{eq:prior_beta_exp_family}
    p(\vbeta \mid  \tau, \DparamScalar_0, \DparamVector_0) = \exp \left\{ 
    \DparamScalar_0 \tau \left(  \DparamVector_0^\top \X \vbeta - \vect{J}_n^\top b(\X \vbeta)  + \vect{J}_n^\top c(\DparamVector_0, \tau) \right) 
    \right\}
\end{equation}
where $\DparamScalar_0 \in \R^+$ and $\DparamVector_0 \in \R^{n \times 1}$ are  hyperparameters, and $\vect{J}_n = \mathbf{1}\in \R^{n \times 1}$. When $\tau$ is unknown, a conjugate joint prior on $(\vbeta, \tau)$ can be formulated as $p(\vbeta, \tau \mid \DparamScalar_0, \DparamVector_0) = p(\vbeta \mid \tau, \DparamScalar_0, \DparamVector_0) p(\tau)$ where $p(\tau)$ must be chosen to ensure conjugacy, depending on the particular form of the exponential family. In our model, to ensure conjugacy of the distribution $p(\vbeta \mid \z, \tau)$, we adapt the prior in Eq.~\eqref{eq:prior_beta_exp_family} to be dependent on $\z$.
To that end, we use $\mathbf{M}^{(s)} \ (\vect{v}^{(s)})$ for the submatrix (subvector) containing only the columns (components) $j$ such that $z_j=s$, for $s=\{0, 1\}$,
and write:
\begin{align}
    p(\vbeta \mid \z, \tau ) & = p\left(\vbeta^{(0)}, \vbeta^{(1)} \mid \z, \tau\right) =  p \left( \vbeta^{(1)} \mid \z, \tau \right) p\left(\vbeta^{(0)} \mid \z, \tau \right) \label{eq:prior_beta_cond_z_exp_family_factorized}  \\
    & \propto \exp\left\{ \DparamScalar_0 \tau \left(\DparamVector_0^\top \X \vbeta - \vect{J}_n^\top b\left(\X^{(1)} \vbeta^{(1)} \right) - \vect{J}_n^\top b\left(\X^{(0)} \vbeta^{(0)} \right) \right) \right\}   \nonumber
\end{align}
where $\X \vbeta \ =  \ \X \left( \vbeta \circ \left(\z + (\vect{J}_p - \z) \right) \right) =  \X^{(1)} \vbeta^{(1)} + \X^{(0)} \vbeta^{(0)} $.
In what follows we denote with $\mathcal{D}(\vbeta^{(s)}; \DparamScalar_0, \DparamVector_0)$ the distribution of Eq.~\eqref{eq:prior_beta_exp_family}, highlighting the influence of variable $\vbeta^{(s)}$, $s=\{0, 1\}$.
The form of this conjugate prior distribution $\mathcal{D}(\vbeta^{(s)}; \DparamScalar_0, \DparamVector_0)$, depends on the choice of the hyperparameters $\DparamScalar_0$ and $\DparamVector_0$.
Details on the properties and the choice of these hyperparameters can be found in \citep[Section 3]{chen2003conjugate_prior_glm}.
In general, the parameter $\DparamVector_0$ controls the location of $\vbeta$ and the symmetry of the prior. One possible strategy for choosing this hyperparameter is to use domain expertise or case-specific information. Alternatively, it can be chosen in order to have a specific desired prior mode of $\vbeta$, $\vmu_0$, as $\DparamVector_0 = b^\prime(\X \vmu_0)$ \cite[Theorem 3.1]{chen2003conjugate_prior_glm}. 
The parameter $\DparamScalar_0 > 0$ can be viewed as a dispersion parameter that controls the heaviness of the tails of the prior. When $\DparamScalar_0$ gets larger, the prior becomes more informative in $\vbeta$, and when $\DparamScalar_0 \rightarrow \infty$ it reduces to a point mass at the mode.
Sensitivity to the elicited hyperparameters varies across the exponential family, with specific distributions (\eg Bernoulli) exhibiting greater inference susceptibility.

Given the model in Eq.~\eqref{eq:glm_variable_selection_model_begin} and the conjugate prior structure in Eq.~\eqref{eq:prior_beta_exp_family}, we now characterize the posterior distributions over all model parameters, which serve as the basis for inference under the proposed covariate selection model:
\begin{align}
\label{eq:post_glm_beta0}
    \vbeta^{(0)} \mid  \X, \y, \z, \tau 
    &\sim  \mathcal{D} \left( \vbeta^{(0)}; \DparamScalar_0, \DparamVector_0 \right) \\
\label{eq:post_glm_beta1} 
    \vbeta^{(1)} \mid  \X, \y, \z,  \tau 
    &\sim  \mathcal{D}\left(  \vbeta^{(1)}; \frac{1 + A(\tau) \DparamScalar_0 \tau}{A(\tau)}, \frac{\y + A(\tau) \DparamScalar_0 \tau \DparamVector_0}{1 + A(\tau) \DparamScalar_0 \tau} \right) \\ 
\nonumber
    \vc \mid \z \ 
     &\sim \ \prodj Beta\left( \zj + \frac{\alpha}{p}, \ 2 - \zj \right)  \\
\label{eq:post_glm_zh}
    z_h \mid  \X, \y, \z_{-h}, c_h, \vbeta, \tau   &\sim  
    Bern\left( 
    \frac{c_h\cdot P^1 \cdot p(\vbeta \mid z_h=1, \z_{-h}, \tau)}{ \displaystyle \sum_{s=0,1}c_h^s (1-c_h)^{1-s} P^s \cdot  p(\vbeta \mid z_h=s, \z_{-h}, \tau)}
    \right) \\
    \label{eq:post_glm_tau}
    \tau \mid \X, \y, \z, \vbeta 
    & \sim \ p(\tau \mid  \z, \y, \vbeta, \x) 
\end{align}
with $P^s := \exp \left\{ \sum_{i=1}^n (\etai \yi - b(\etai))/A_i(\tau) \}\right\rvert_{z_h  = s}$ ($s = 0,1$), $\etai = \sumj \xij \betaj \zj = \x_i^\top (\vbeta \circ \z)$, and $p(\vbeta \mid z_{h}=s, \z_{-h}, \tau)$ denotes the density of the regression coefficients $\vbeta$ evaluated when $z_h=s$. Full details on the posterior computations are provided in Supplementary Material Section \ref{apd:posterior_general}. In Eq.~\eqref{eq:post_glm_tau}, if prior conjugacy is ensured, the posterior of $\tau$ is of the same form of the prior.
\begin{remark}
The prior over $\vbeta$ induces identical and independent priors for $\vbeta^{(0)}$ and $\vbeta^{(1)}$, as shown in Eq.~\eqref{eq:prior_beta_cond_z_exp_family_factorized}.
However, their posteriors differ. When $z=0$, the conditional posterior coincides with the prior, as data do not provide any information (see Eq.~\eqref{eq:post_glm_beta0}).
Conversely,  when $z=1$, 
the posterior belongs to the same family of the prior with updated parameters (see Eq.~\eqref{eq:post_glm_beta1}).
\end{remark}

\begin{remark}
It is possible to integrate out variable $\vc$ from the posterior distribution of $z_h$. 
Marginalization of $\vc$ leads to equivalent posterior inference, but improves the Gibbs sampler. 
In this case, the posterior over indicator variables results in
\begin{equation*}
    z_h \mid  \X, \y, \z_{-h}, \vbeta, \tau
    \sim  Bern\left( \frac{ P^1 \cdot p(\vbeta \mid z_h =1, \z_{-h}, \tau )  }{ P^1 \cdot p(\vbeta \mid z_h =1, \z_{-h}, \tau )  +  \displaystyle \frac{p}{\alpha} P^0 \cdot p(\vbeta \mid z_h =0, \z_{-h}, \tau ) }  \right) \ . 
\end{equation*}
\end{remark}
All conjugate posterior conditionals presented in equations~\eqref{eq:post_glm_beta0}--\eqref{eq:post_glm_tau} have closed-form expressions, and they define a Gibbs sampling scheme for posterior inference.
As shown in Lemma~\ref{lemma:jointgibbs} in Supplementary Material Section \ref{apd:proofs}, sampling from this Gibbs Scheme is equivalent to sampling from the joint posterior distribution of the parameters.  
In practice, efficient posterior inference via Gibbs sampling relies on the possibility of drawing samples from each conditional posterior.
The conjugate distribution of the model parameters' distribution $\mathcal{D}(\vbeta; \DparamScalar, \DparamVector)$ is, in general, known up to a normalization constant, which cannot be computed in closed form for most likelihoods within the exponential-family
---except for specific GLM settings, such as the Gaussian likelihood, for which the conjugate prior and posterior over $\vbeta$ are Normal.
Nonetheless,
we can sample \iid from $\mathcal{D}(\vbeta ; \cdot)$
via Sampling-Importance Resampling (SIR),
with a SIR-proposal of choice.
We describe a sampling procedure tailored to the Bayesian GLM of Eq.~\eqref{eq:glm_variable_selection_model_begin} in Supplementary Material Section \ref{sec:sampling},
based on a Gaussian proposal centered at the posterior mode of $\mathcal{D}(\vbeta; \DparamScalar, \DparamVector)$
and covariance given by the inverse of the observed negative Hessian evaluated at that same mode.

\begin{remark}[On the mixing of the proposed Gibbs sampler]
\label{rem:mixing_trace}
The theoretical and algorithmic mixing properties of MCMC methods for Bayesian variable selection have received considerable attention in the literature, \eg by \citet{zhou2022dimension, yang2015computationalcomplexity, lee2025advances}. A fundamental bottleneck of Gibbs sampling in spike-and-slab models is the propensity of the chain to become trapped in local selection states \citep{george1993spike, scheipl2012spikeforfunction, yang2015computationalcomplexity}. This phenomenon arises when the chain enters a state in which a covariate with truly nonzero effect is effectively excluded from the model.
Driven by a prior sharply peaked at zero, the sampler tends to draw $\beta_j$ values too small to trigger the reactivation during subsequent updates. As a result, adequate exploration of the model space demands a higher number of MCMC iterations and a slower convergence.

This phenomenon of becoming trapped in local selection states
is less likely to arise in our framework.
As shown in Eq.~\eqref{eq:prior_beta_cond_z_exp_family_factorized},
both subvectors $\vbeta^{(0)}$ and $\vbeta^{(1)}$ are distributed as $\mathcal{D}(\cdot\, ; a_0, \xi_0)$; and, as shown in Eq.~\eqref{eq:post_glm_beta0},
when $z_j=0$ the conditional posterior of $\beta_j$ coincides with this non-degenerate distribution, rather than collapsing towards zero, with enough mass to explore the $\beta_j\neq0$.
We illustrate this non-trapped behavior empirically:
we show accurate recovery of the active and non-active components for a Poisson GLM
in Figures~\ref{fig:poi_acc_cd_n100}, \ref{fig:poi_acc_cd_n1000}, and \ref{fig:settingC_z_acc_np_curve} of Section~\ref{subsubsec:syntheticdata_poisson_settingC};
and Figure~\ref{fig:traceplot_z_n1000} in Supplementary Material Section \ref{apd:trace_poisson},
detailed Gibbs sampling trace plots of $z_j$ for the same setting,
show how the chain transitions without restrictions between active and inactive.
Note that, for those trace-plots, we avoid initializing the chain at any of the extremes,
\ie $\z = \mathbf{0}_p$ or $\z = \mathbf{1}_p$,
and instead draw initial $\z \sim \text{Binom}(p, 1/2)$.
All posterior estimates are stable across independent runs from random starting points, indicating that the recovered posterior does not depend on the initial configuration of the chain.
\end{remark}

\section{Posterior Consistency}
\label{sec:posterior_consistency}
We here analyze the asymptotic behavior of the posterior distributions within our covariate selection model. The Bayesian approach treats parameters as random variables, in practice however, the parameter of interest is unknown but fixed. For this reason, we adopt a frequentist perspective: we analyze how the posterior behaves as the sample size increases, assuming data are generated from a fixed, true parameter $\vbeta^*$. This viewpoint allows us to assess whether the Bayesian procedure concentrates around the true parameter value ---a property known as posterior consistency. As discussed by \citet[Chapter 4]{ghosh2007introduction}, such asymptotic analyses offer a form of frequentist validation for Bayesian methods.

Consider the model presented in Eq.~\eqref{eq:glm_variable_selection_model_begin}, in terms of true regression coefficients $\vbeta^*$ and covariate selection indicators $\z^*$, \ie $\thetai = \etai | \x_i, \z^*, \vbeta^* \ = \ \x_i^\top (\vbeta^* \circ \z^*)$ such that $z^*_h=1$ if and only if $\beta^*_h\neq 0$, and $0$ otherwise. In what follows,
we establish posterior consistency of the covariate selection indicator $\z$ as the sample size increases, under regularity conditions on the exponential family, the design matrix $\X$, the true model, and the growth rate of the model dimensions.
The conditions of the Theorem are stated precisely in Supplementary Material Section \ref{apd:proofs},
and summarized below for clarity.

\begin{theorem}[Posterior consistency of variable selection]
\label{thm:consistency_z}
Consider the model in Eq.~\eqref{eq:glm_variable_selection_model_begin} with true parameters $\vbeta^*$ and $\z^*$.
Assume conditions \ref{cond:1}-\ref{cond:4} specified in Supplementary Material Section \ref{apd:proofs} and discussed below hold,
and that the model is identifiable.
Then, the posterior probability concentrates on the true model, almost surely, as the sample size grows $P(\z = \z^* \mid \X, \y, \tau) \longrightarrow 1$ as $n \to \infty$.
\end{theorem}
The detailed proof is reported in Supplementary Material Section \ref{apd:proofs}.
\begin{remark}[Conditions of the Theorem]
\label{rem:consistency_z_conditions}
Condition~\ref{cond:1:a} requires the data generating exponential-family distribution to be log-concave, which is satisfied by many of its members, such as those with canonical links \citep{diaconis1979conjugate} ---exactly the case we consider in this work--- see \citep{Wedderburn1976} for noncanonical link-based cases. 
Condition~\ref{cond:1:b} requires the exponential family to be $\varpi$-subExponential with $\varpi \in (1/2, 2]$:
we verify this for several members of the exponential family, such as Gaussian, Bernoulli, Poisson, Gamma, and Inverse-Gaussian
---see Lemma \ref{lemma:sub-gauss}, Corollary \ref{corollary:sub-gauss} and Remark \ref{rem:locally-lipschitz} in Supplementary Material Section \ref{apd:proofs}.
Condition~\ref{cond:1:c} requires that the maximum likelihood estimator exists, is unique, and depends continuously on the data. Conditions under which it holds are discussed, \eg in \citep{Wedderburn1976} for some members of the exponential family regarding existence and uniqueness, while continuity can be achieved under Berge's maximum theorem, see for instance \citet[Theorem 17.31]{AliprantisBorder2006}.   \\
Condition~\ref{cond:2} requires the design matrix to have bounded entries,
and that the relevant eigenvalues of the log-likelihood are bounded away from zero and infinity, uniformly over all models up to a given size. 
Similarly, Condition~\ref{cond:3} requires the nonzero coefficients of the true model to be bounded away from zero.\\
Finally, condition~\ref{cond:4} describes the rates at which the number of covariates and the maximal model dimensions (i.e. the maximum number of active covariates) can grow with the sample size $n$, calibrated by the tail behavior of the exponential family through parameter $\varpi$. In particular, the number of covariates can grow faster than $O(n)$ (\ref{cond:4:d}), while the maximal model dimension can grow at a rate $O(n^\delta)$ with $0 < \delta < 1$ (\ref{cond:4:e}).  
\end{remark}

In what follows,
we denote with $\pi_n(\vbeta \mid \X, \y, \z, \tau)$
the conditional posterior distribution of the regression coefficients given the inclusion vector $\z$ and $n$ data samples, \ie $\X \in \mathbb{R}^{n \times p}$ and $\y \in \mathbb{R}^{n \times 1}$;
and with $\pi_n(\vbeta \mid \X, \y, \tau)$, the (marginal) posterior distribution of the regression coefficients,
where the latter can be written in terms of the former as 
$\pi_n(\vbeta \mid \X, \y, \tau) = \sum_{\z} \pi_n(\vbeta \mid \X, \y, \z, \tau) p(\z | \X, \y, \tau)$. In a similar same way $\pi_n(\vbeta^{(1)}\mid \X^{(1)}, \y, \tau)$ denotes the (marginal) posterior distribution of the \textit{active} regression coefficients. 

\begin{theorem}[Posterior consistency of the active GLM coefficients]
\label{thm:consistency_beta}
Consider the model in Eq.~\eqref{eq:glm_variable_selection_model_begin} with true parameters $\vbeta^*$ and $\z^*$. Let $\vbeta^{*^{(1)}} \in \mathcal{B}$ be the active coefficients subvector with non-zero components, with parameter space $\mathcal{B} \subset \mathbb{R}^{p_1}$ convex and open with $p_1 \le p$. Assume that the prior over parameters $p(\vbeta^{(1)} \mid \z^{(1)})$ is continuous and strictly positive in a neighborhood of $\vbeta^{*^{(1)}}$.
Let $\ell ( \vbeta^{(1)} )$ denote the negative expected log-likelihood and $\ell_n(\vbeta^{(1)})$ its empirical version.
Assume the following conditions hold: (i) there exists $\hat{\vbeta}^{(1)} \in \mathcal{B}$ such that $\ell^{\prime} \left( \hat{\vbeta}^{(1)}\right) = 0$; (ii) $\mathbb{E}\left[ {\X^{(1)}}_i \yi \right] < + \infty$ and  $\mathbb{E}\left[ b\left({{\X^{(1)}}_i}^\top  \vbeta^{(1)} \right) \right] < + \infty$, for all $\vbeta^{(1)} \in \mathcal{B}$, $\forall i= 1, \ldots n $; (iii) for any $\mathbf{a} \in \mathbb{R}^p$, if $X_i^\top \mathbf{a} \overset{a.s.}{=} \mathbf{0}$ then $\mathbf{a} = \mathbf{0}$, $\forall i= 1, \ldots n $; (iv) there exists an open ball $E \subseteq \mathbb{R}^p$, such that $\hat \vbeta \in E, \ \Bar{E} \subseteq \mathcal{B}$, and $\ \forall j,k,l \in \{ 1, \ldots p\}$: $\mathbb{E} \left[ \displaystyle \sup_{\vbeta \in \Bar{E}} \mid b^{\prime \prime \prime}({\X^{(1)}}_i^\top \vbeta) {\X^{(1)}}_{ij} {\X^{(1)}}_{ik} {\X^{(1)}}_{il} \mid \right] < \infty $; with $f^{\prime}(\cdot)$ and $f^{\prime \prime \prime}(\cdot)$ denoting respectively the first and third derivative of the function $f$, with respect to its argument.
Then, (a) there is a sequence $\widehat{\vbeta}^{(1)}_n$ that converges to $ \vbeta^{*^{(1)}}$ with respect to the Euclidean norm, such that $\ell^{\prime}(\widehat{\vbeta}_n^{(1)}) = 0$ and, as $n\rightarrow + \infty$, it satisfies $\ell_n\left(\widehat{\vbeta}^{(1)}_n \right) \rightarrow \ell \left(\vbeta^{*^{(1)}} \right)$;  (b) for any $\varepsilon > 0$, the posterior distribution of $\vbeta^{(1)}$  concentrates around the true active coefficient vector $\vbeta^{*^{(1)}}$: $\pi_n \left( \vbeta^{(1)} \in B_{\varepsilon}\left(\vbeta^{*^{(1)}} \right) \mid \X^{(1)}, \y, \tau \right) \longrightarrow 1$ as $n \rightarrow + \infty$; (c) when $\vbeta^{(1)} \sim  \pi_n(\cdot \mid \X^{(1)}, \y, \tau)$ then $\sqrt{n} \left( \vbeta^{(1)} - \widehat{\vbeta}^{(1)}_n \right) \overset{t.v.}{\longrightarrow} \mathcal{N} \left( \mathbf{0}, H_0^{-1}\right)$ with $H_0 = \ell^{\prime \prime}\left(\vbeta^{*^{(1)}}\right)$.
\end{theorem}
\begin{proof}
    For every selection vector $\z \in \{0,1\}^p$, there exists a conditional posterior distribution of $\pi_n(\vbeta \mid \X, \y, \z, \tau)$. The same reasoning can be applied to the subvectors related to the active covariates. 
    The marginal posterior distribution of $\vbeta^{(1)}$ can be therefore written as a mixture:
    \[ 
    \pi_n(\vbeta^{(1)}\mid \X^{(1)}, \y, \tau) = \sum_{\z^{(1)}} \pi_n(\vbeta^{(1)} \mid \X^{(1)}, \y, \z^{(1)}, \tau) p(\z^{(1)} | \X^{(1)}, \y, \tau) \ .
    \]
    By Theorem~\ref{thm:consistency_z}, the posterior distribution of $\z$ will concentrate on the true inclusion vector $\z^*$. Hence, among all the  $2^p-1$ possible configurations, the only one with non-vanishing posterior probability is the true model $\z^*$. As a direct consequence, this convergence also holds for the subvector related to the active covariates, \ie the posterior distribution of $\z^{(1)}$ will concentrate on the true inclusion vector $\z^{*^{(1)}}$.
    Therefore, asymptotically, the marginal posterior $\pi_n(\vbeta^{(1)}\mid \X^{(1)}, \y, \tau)$ is dominated by the component corresponding to $\z^{*^{(1)}}$: $\pi_n(\vbeta^{(1)}\mid \X^{(1)}, \y, \tau) \approx \pi_n(\vbeta^{(1)}\mid \X^{(1)}, \y, \z^{*^{(1)}}, \tau)$ as $n \to \infty$. The rest of the proof follows directly from Theorems 5 and 13 of \citet{miller2021asymptotic}, which establish posterior consistency and concentration for general GLMs under equivalent conditions. 
\end{proof}

\section{Empirical evaluation}
\label{sec:numerical_experiments}
We assess the performance of the proposed, conjugate posterior-based inference of our model's regression coefficient and indicators through a series of numerical experiments, considering both
synthetic data (where we have access to the ground truth)
and real-world datasets, illustrating the practical applicability of the proposed Bayesian variable selection procedure.

\subsection{Synthetic data experiments}
\label{subsec:results_synthetic}
We begin by evaluating the performance of our covariate selection and regression parameter estimation accuracy on synthetic datasets.
In this controlled setting, we generate data from different GLMs,
where the true parameter values
---in particular, the true sparsity pattern encoded by $\z^*$ and the  coefficients $\vbeta^{*(1)}$---
are known.
This set-up enables a rigorous quantitative assessment of both the parameter posteriors' convergence and the covariate indicators' selection accuracy.

We replicate the simulated data-generation procedure by \citet{khan2007robust}.
We first introduce a set of $k$ latent variables, 
independently and identically distributed with a standard Gaussian distribution:
$L_j \overset{\iid}{\sim} \mathcal{N}(0,1)$, for $j= 1, \ \ldots, \ k$.
We concatenate $n$ draws of these variables to form the matrix $\mathbf{L} \in \mathbb{R}^{n \times k}$. 
These latent variables are used to define the true linear predictor $\veta = \mathbf{L} \vbeta^{*(1)} \in \mathbb{R}^n$, where the coefficient vector $\vbeta^{*(1)} \in \mathbb{R}^k$ is sampled from $Unif([a,b]^k)$, with range $[a,b]$ according to the corresponding GLM family.
The dependent variable, \ie observed vector $\y \in \mathbb{R}^n$, is sampled from a distribution in the exponential family, with corresponding natural parameter $\veta$; that is $\yi \sim f_Y(\yi \mid \etai)$ ($i = 1, \ldots , n$). 
While the response variable $\y$ is generated according to the latent, informative variables $\mathbf{L}$,
the model performs inference based only on observed, noisy covariates $\X$.
Specifically, we construct the design matrix $\X \in \mathbb{R}^{n \times p}$ via noisy and redundant transformations of $\mathbf{L}$, $X = \left[ X^{(1)} \, \middle| \, X^{(2)} \, \middle| \, X^{(3)} \right]$, with a total number of covariates of $ p = k + 2c + d $,
where: 
$ X^{(1)} \in \mathbb{R}^{n \times k} $ contains \textbf{informative} covariates, directly associated with the outcome; 
$ X^{(2)} \in \mathbb{R}^{n \times 2c}$ contains \textbf{correlated} covariates, each being a noisy linear combination of the informative ones; and 
$ X^{(3)} \in \mathbb{R}^{n \times d} $ contains \textbf{noisy} covariates, drawn independently from a standard normal distribution.
Each of this covariate generation follows:
\begin{align*}
    X_j &= L_j + \phi e_j & j= 1, \ \ldots, \ k & \quad \text{(informative)} \\
    X_{k + 2j -1} &= L_j + \delta e_{k + 2j -1}  & j= 1, \ \ldots, \ c & \quad \text{(correlated)}\\
    X_{k + 2j} &= L_j + \delta e_{k + 2j}  & j= 1, \ \ldots, \ c & \quad \text{(correlated)} \\
    X_j &= e_j  & i= k+2c+1, \ \ldots, \ k+2c+d & \quad \text{(noisy)} 
\end{align*}
with $\phi = 0.001, \ \delta = 5$ and 
$e_j \overset{\iid}{\sim} \mathcal{N}(0,1)$, for all $j= 1, \ \ldots, \ k+2c+d$.

For this data-generating mechanism,
the true GLM coefficient vector $\vbeta^*$ and the corresponding inclusion vector $\z^*$ are defined as 
\[
\vbeta^* = ( \vbeta^{*(1)}, \underbrace{0, \ldots, 0}_{2c + d} ) \ \in \mathbb{R}^{p}, \qquad 
\z^* = ( \underbrace{1, \ldots, 1}_{k}, \underbrace{0, \ldots, 0}_{2c + d} ) \ \in \{0,1\}^p.
\]

We run different experimental configurations defined by the values of $n,\ p,\ k,\ c,\ d$ and GLM family, with corresponding hyperparameters.
For each GLM family, the simulations can be organized into three progressively complex scenarios. \textbf{Setting A:} informative and noisy covariates ($c=0,d>0$); \textbf{Setting B:} informative and redundant correlated covariates ($c>0,d=0$); \textbf{Setting C:} informative, noisy and correlated covariates ($c>0,d>0$). We report results for the most challenging configuration, i.e., \textbf{Setting C}, where both correlated and noisy covariates are present ($c>0, \ d>0$).
Results for \textbf{Setting A} and \textbf{B} gave very similar conclusions and are therefore omitted.
For a dedicated experiment examining the effect of increasing correlation on variable selection for different models, see Supplementary Material Section \ref{apd:results_correlation}.

We fix the total number of covariates $p$, and increase the sample size $n$, ensuring $n \gg p$,
to investigate how posterior recovery improves as the ratio $n/p$ increases.
For each experiment setting,
we performed 10 independent runs with different random seeds
to characterize the impact of randomness in $\mathbf{L}$, $\X$, and $\y$.

We assess performance based on the following metrics, which are aggregated across seeds and visualized through boxplots: (i) Sparsity pattern $\z^*$ recovery: we report the component-wise accuracy of the inclusion variable $\z$ across multiple simulations (for each fixed experimental configuration), as well as the total selection accuracy (defined as the proportion of correctly recovered entries in $\z$) as $n/p$ increase; (ii) Overall posterior estimation accuracy for coefficients $\vbeta^{*}$: we inspect the posterior distribution of $\vbeta \circ \z$ in each experimental configuration, and report the Mean Squared Error (MSE) between the estimated coefficients and the ground truth $\vbeta^{*}$, computed as $MSE(\vbeta, \vbeta^{*}) = \sum_{j=1}^{p} \left( \mathbb{E}\left[ \beta_j z_j\right] -  \beta_j^{*} \right)^2/p$; (iii) Posterior estimation accuracy for active coefficients $\vbeta^{*(1)}$: we report the Relative Mean Squared Error (RelMSE) between the estimated active coefficients and the ground truth $\vbeta^{*(1)}$, computed as $RelMSE(\vbeta^{(1)}, \vbeta^{*(1)}) = \sum_{j=1}^{k} \left( (\mathbb{E}\left[ \beta_j^{(1)} z_j^{(1)}\right] -  \beta_j^{*(1)})/\beta_j^{*(1)} \right)^2/k$; (iv) Posterior of inactive coefficients $\vbeta^{*(0)}$: we also inspect the posterior of coefficients deemed inactive, \ie $\z^{(0)}=0$, assessing that they remain consistent with the prior in our model. We compare the results of our proposed model, which we will refer to as \texttt{BayesVS-GLM}, to the following frequentist and Bayesian alternatives: (i) the frequentist estimate of regression coefficients computed using the built-in \texttt{glm} function of \textsf{R} (\texttt{MLE}); (ii) the Bayesian version of GLM, where the inclusion vector is fixed to all ones, \ie $\z = \mathbf{1}\in \R^{p \times 1}$ (\texttt{Bayes-GLM}); (iii) an oracle Bayesian-GLM, where the inclusion vector $\z$ is fixed to the ground truth $\z^*$  (\texttt{Bayes-GLM(Oracle)}), representing the best achievable performance, since the true active set is known by the oracle; (iv) a spike-and-slab prior-based Bayesian variable selection method (\texttt{SpikeSlab}), as implemented in the \textsf{R} package \texttt{BoomSpikeSlab} \citep{boomspikeslab}; (v) a continuous shrinkage prior-based Bayesian variable selection baseline, as implemented in \textsf{R} package \texttt{bayesreg}~\citep{bayesreg}; specifically, adopting the popular horseshoe prior (\texttt{Horseshoe}).
For all Bayesian algorithms,
we impose same priors over all parameters,
and posterior inference of unknowns is carried out using Gibbs Sampling with $5000$ iterations and a burn-in of $10\%$, except for the \texttt{Horseshoe} baseline, which needed a larger burn-in (fixed to $10^4$) to converge.

We present here results for the Poisson model
---experiments for the Linear and Logistic model are provided in Supplementary Material Sections \ref{apd_subsec:linear} and \ref{apd_subsec:binomial}.

\subsubsection{Bayesian Variable Selection for Poisson regression}
\label{subsubsec:syntheticdata_poisson_settingC}
The Poisson model reads $\yi \mid \lambda_i \sim \text{Poisson}(\lambda_i)$ with canonical log link $\log(\lambda_i) = \etai = \sum_{j=1}^p x_{ji} \betaj \zj$, 
and priors: $\z \mid \vc \sim \prod_{j=1}^p \text{Bern}(\cj),
    \vc \mid \alpha \sim \prod_{j=1}^p \text{Beta}\left( \frac{\alpha}{p}, 1 \right),  \
    \vbeta \sim \mathcal{D}(\vbeta; \DparamScalar_0, \DparamVector_0)$, where we fix the prior hyperparameters for this model to $\DparamScalar_0 = 0.001$, $\DparamVector_0 \overset{\iid}{\sim} \mathcal{N}(0,1)$, $\alpha=1$. Derivations of the posterior distributions are reported in Supplementary Material Section \ref{apd:posterior_poisson}.

In the Poisson regression with informative, noisy and correlated covariates (Setting C), the design matrix $\X$ contains $k$ informative, $d$ purely noisy, and $2c$ correlated covariates, resulting in a total of $p = k + d + 2c$.
For the experiments, we set $p = 20$, $d = 10$, and $c = 2$; while $\vbeta^{*(1)} = (0.633, 0.686, 0.786, 0.954, 0.601, 0.949)$, sampled from a $Unif([0.5,1]^k)$. We first evaluate the model’s ability to correctly recover the sparsity pattern $\z^*$, 
that is, to distinguish informative covariates from irrelevant ones.
To this end, we consider two types of analyses. First, we assess the accuracy of each individual component  $\zj$ over repeated simulations 
with fixed configuration of $n$, $p$, $c$ and $d$. 
This component-wise evaluation allows us to understand whether some coordinates are consistently easier or harder to recover. 
Second, we summarize the overall ability to retrieve $\z^*$ by computing the total selection accuracy 
(\ie the proportion of correctly recovered entries in $\z$) for each run. 

Figure \ref{fig:settingC_Poisson_z_accuracy} reports the component-wise accuracy of the inclusion vector $\z$ over multiple simulation runs for two different sample sizes ($n = 100$ and $n = 1000$).
Even with a limited number of observations ($n = 100$), the model achieves remarkably high accuracy, which further improves as the sample size increases.
The trend is summarized in Figure~\ref{fig:settingC_z_acc_np_curve}, which displays the mean and standard deviation of the global selection accuracy as a function of the ratio $n/p$.
\begin{figure}
   \centering
     \subfloat[$n=100$]{
        \includegraphics[width=0.48\textwidth]{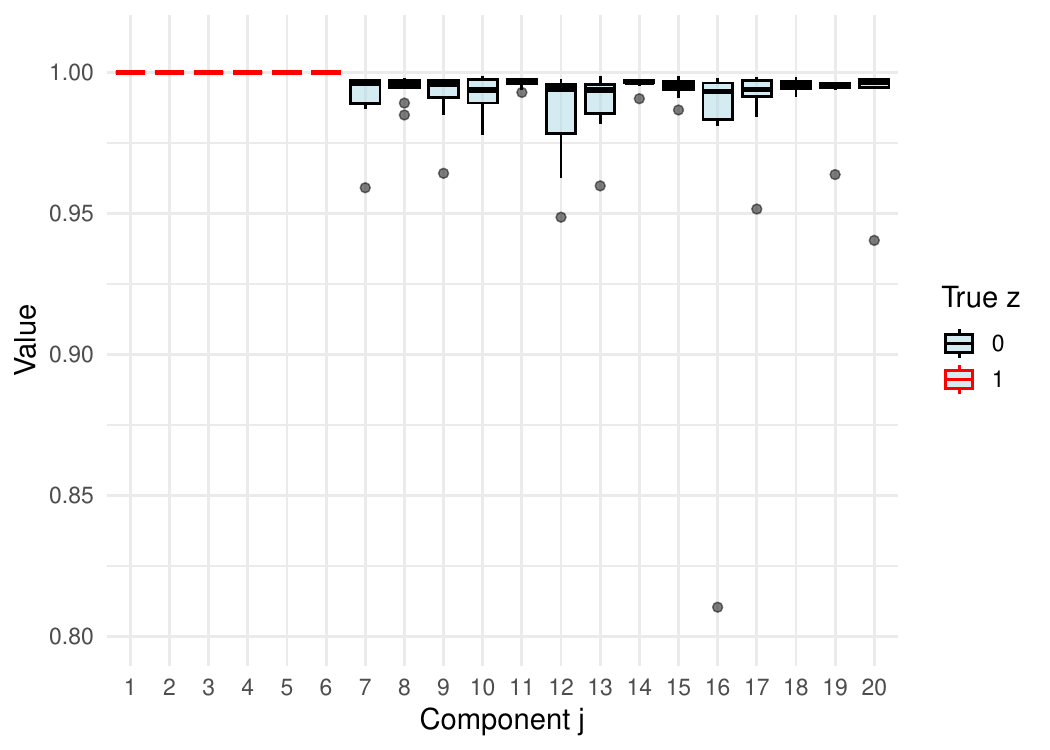}
        \label{fig:poi_acc_cd_n100}}
     \hfill
     \subfloat[$n=1000$]{
        \includegraphics[width=0.48\textwidth]{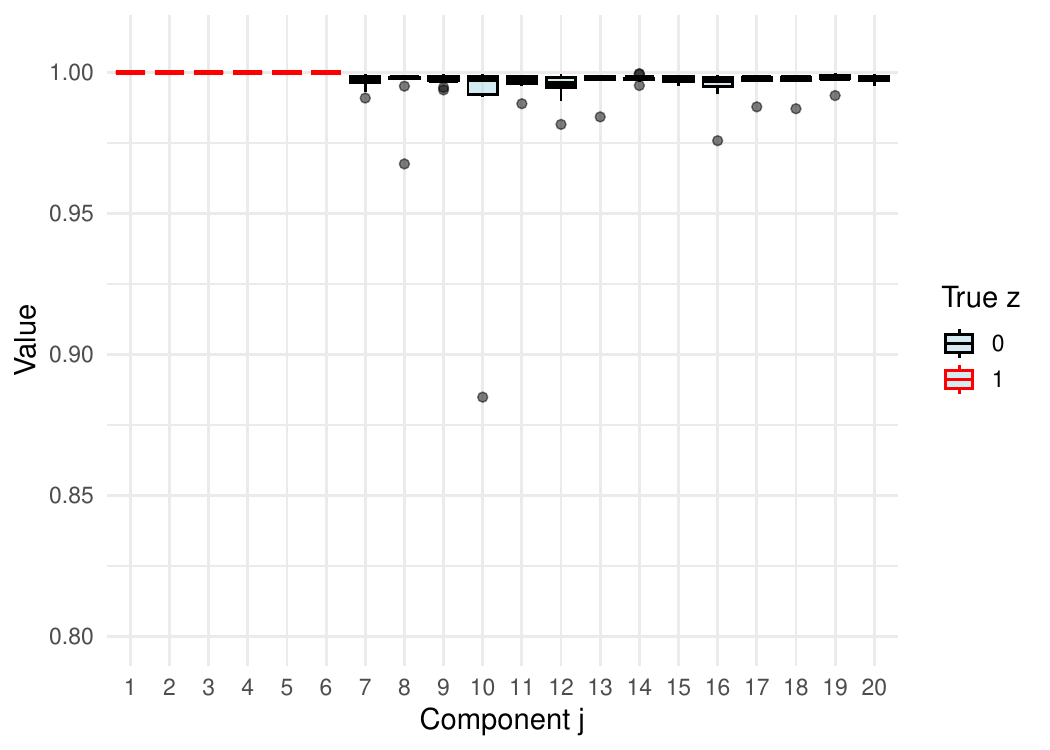}
        \label{fig:poi_acc_cd_n1000}}
        \caption{Poisson Model. Component-wise accuracy of the inclusion variable $\z$ across multiple simulations ($p=20, \ d=10$, and $2c=4$ correlated variables). Each boxplot refers to one component $\zj$.}
        \label{fig:settingC_Poisson_z_accuracy}
\end{figure}

\begin{figure}
    \centering
    \includegraphics[width=0.7\textwidth]{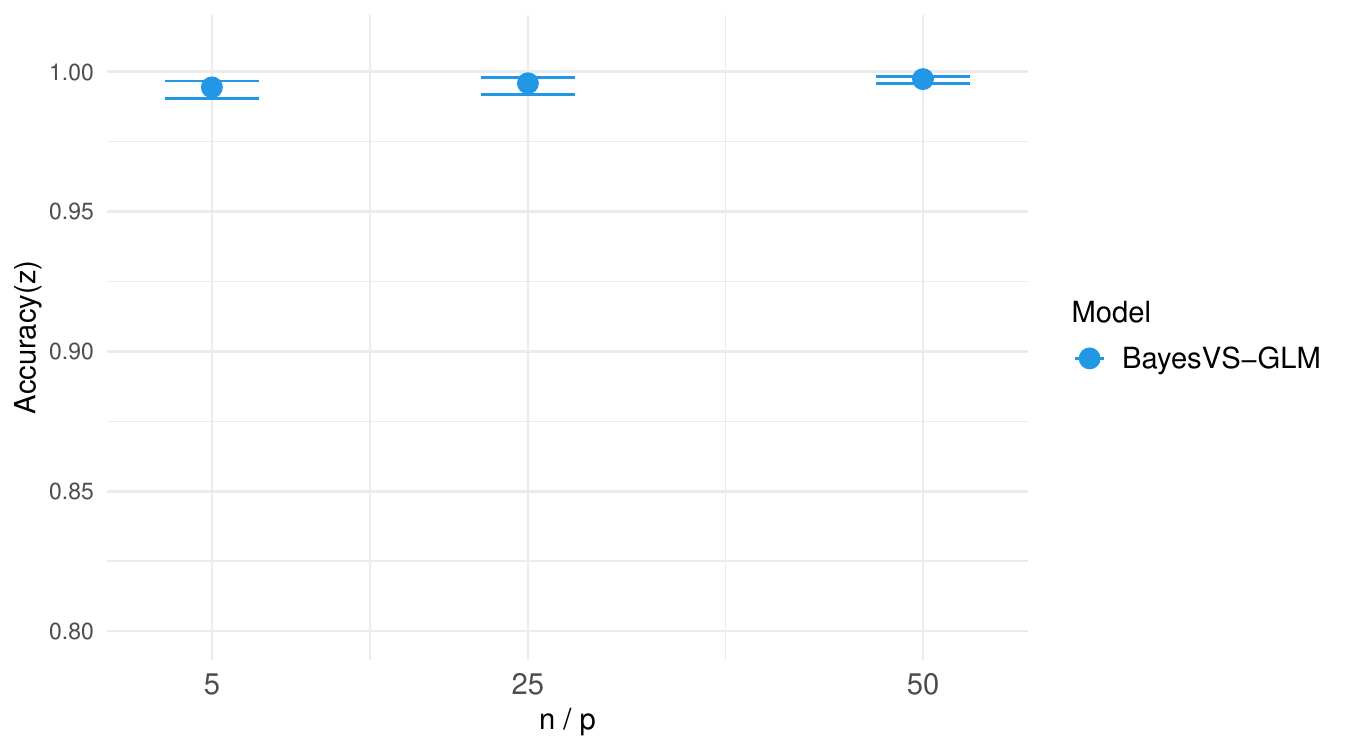}
    \caption{Poisson Model. With $p=20$ total covariates, $d=10$ noisy and $2c=4$ correlated variables. The accuracy of the inclusion variable $\z$
     increases for larger values of the $n/p$ ratio.}
    \label{fig:settingC_z_acc_np_curve}
\end{figure}

Shifting our focus to the recovery of the regression coefficients $\vbeta$, we are particularly interested in two aspects:
$(i)$ whether the model correctly estimates the coefficients associated with the relevant/active covariates, 
and $(ii)$ whether the coefficients associated to the irrelevant covariates are sampled from the prior, as they are not entering the model.
These two properties provide insight of the quality of both covariate selection and coefficient estimation.

Figures~\ref{fig:poi_beta_mse}-\ref{fig:poi_beta_rmse} 
report respectively the Mean Squared Error (MSE) of the estimated coefficients $\vbeta \circ \z$ relative to the ground truth $\vbeta^{*}$, and Relative Mean Squared Error (RelMSE) computed only on the active coefficients, as $n/p$ increases.
 Supplementary Material Section \ref{apd_subsec:poisson} contains additional plots on the posterior distribution of the estimated coefficients $\vbeta \circ \z$.

Remarkably, across both metrics and across the whole range of $n/p$, our method (\texttt{BayesVS-GLM}) achieves essentially the same performance as the oracle model (\texttt{Bayes-GLM(Oracle)}), despite not having access to the true inclusion vector $\z^*$. This shows that our variable selection procedure is able to recover the active set with enough accuracy to match the best possible performance attainable when the ground truth support is known.
\texttt{BayesVS-GLM} consistently outperforms the \texttt{MLE} and \texttt{Horseshoe} alternatives, particularly in smaller-sample regimes; while it performs similarly to \texttt{SpikeSlab} in this setting. 
\begin{figure}
    \centering
    \subfloat[$MSE(\vbeta, \vbeta^{*})$]{
         \includegraphics[width=0.4\textwidth]{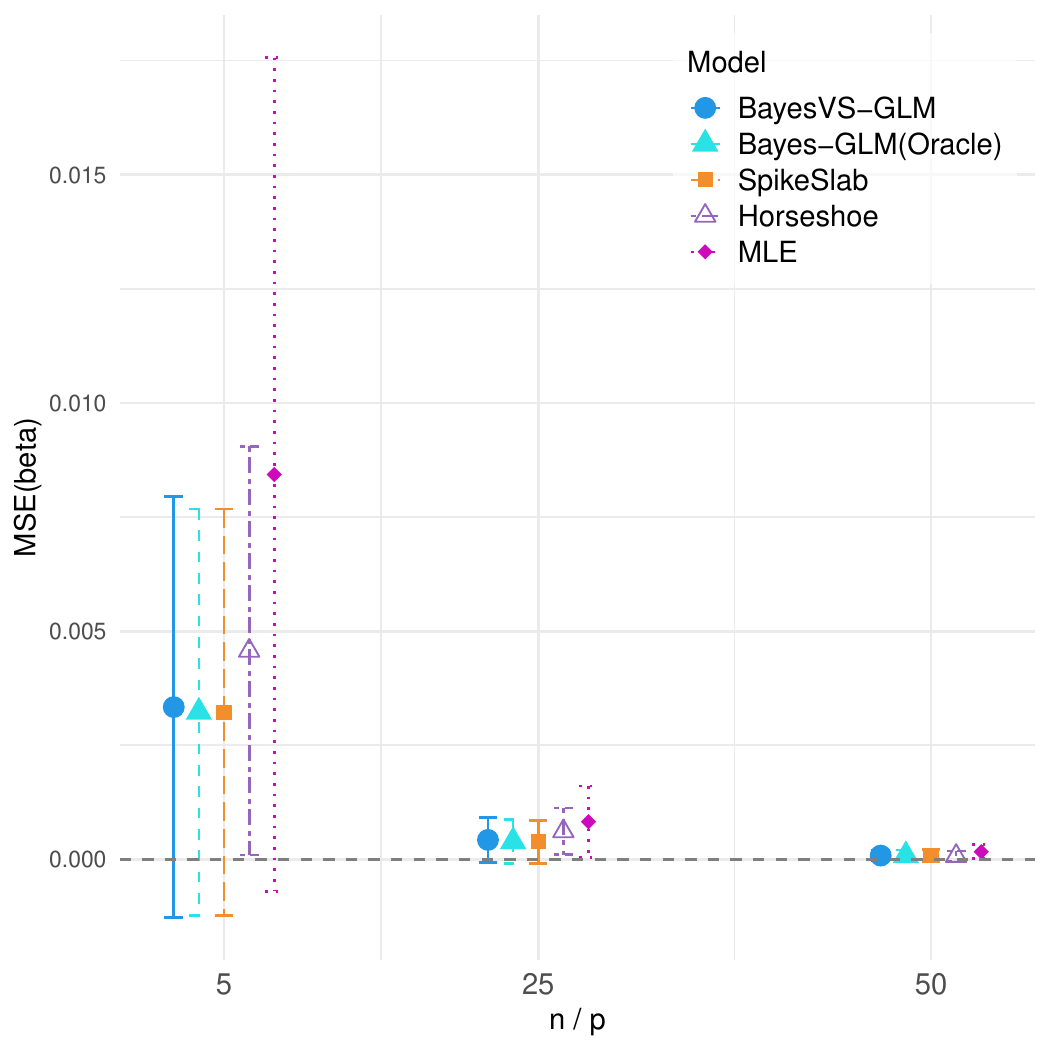}
         \label{fig:poi_beta_mse}}
     \hfill
     \subfloat[$RelMSE(\vbeta^{(1)}, \vbeta^{*(1)})$]{
         \includegraphics[width=0.4\textwidth]{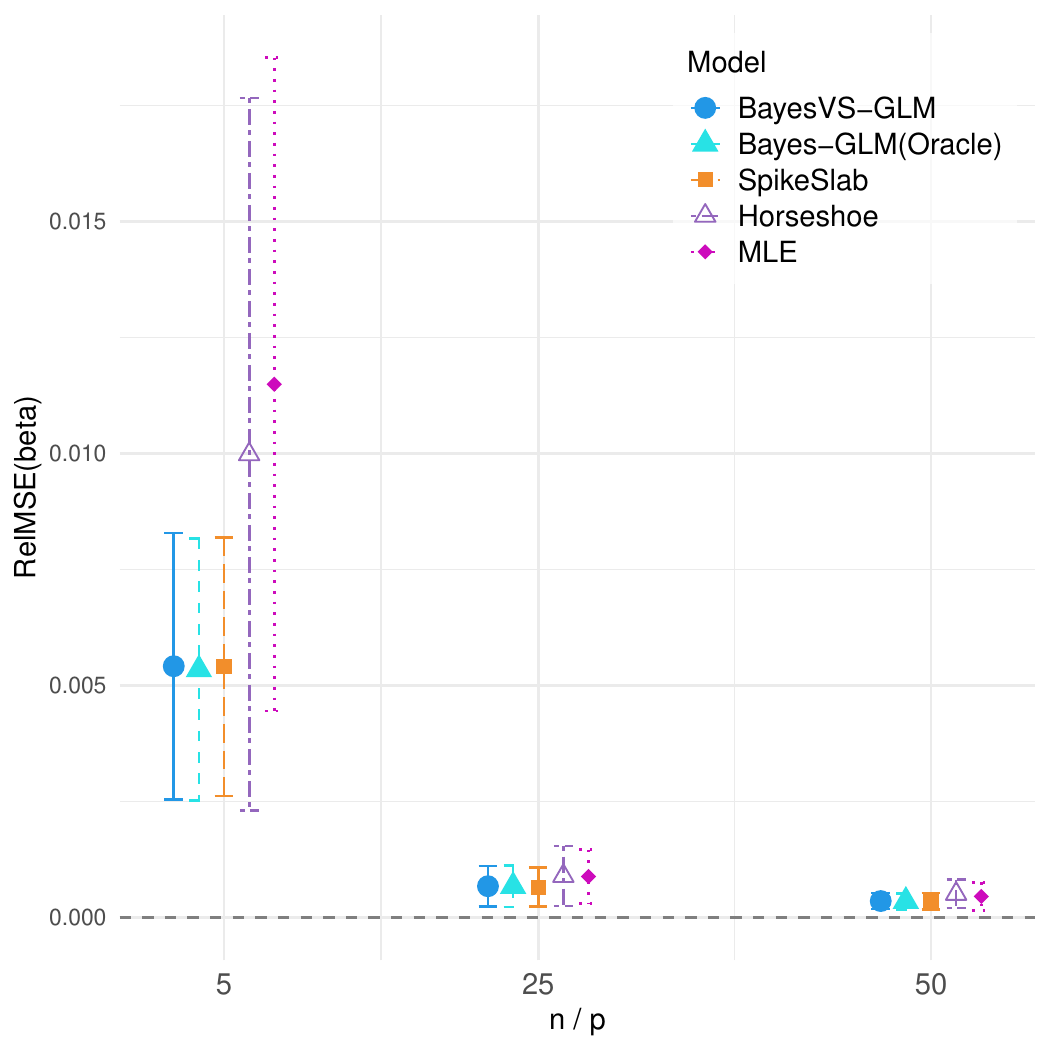}
         \label{fig:poi_beta_rmse}}
    \caption{Poisson Model. Error metrics of the coefficients $\vbeta$, for $p=20$, $d=10$ and $2c=4$, and increasing ratio $n/p$. Our method (\texttt{BayesVS-GLM}) closely matches the oracle version and the \texttt{SpikeSlab} baseline, while it outperforms \texttt{MLE} and \texttt{Horseshoe} in small sample size regimes.}
    \label{fig:settingC_poi_beta_error_np_curve}
\end{figure}

We also examine the behavior of the posterior for inactive coefficients. Figure~\ref{fig:settingC_irrelevant_betas} displays marginal histograms of selected components of $\vbeta^{(0)}$, \ie those corresponding to $\zj^* = 0$. As desired, the posterior distribution remains diffuse and matching the prior specification.

\begin{figure}
     \centering
     \subfloat[$\beta_{7}$]{
         \includegraphics[width=0.4\textwidth]{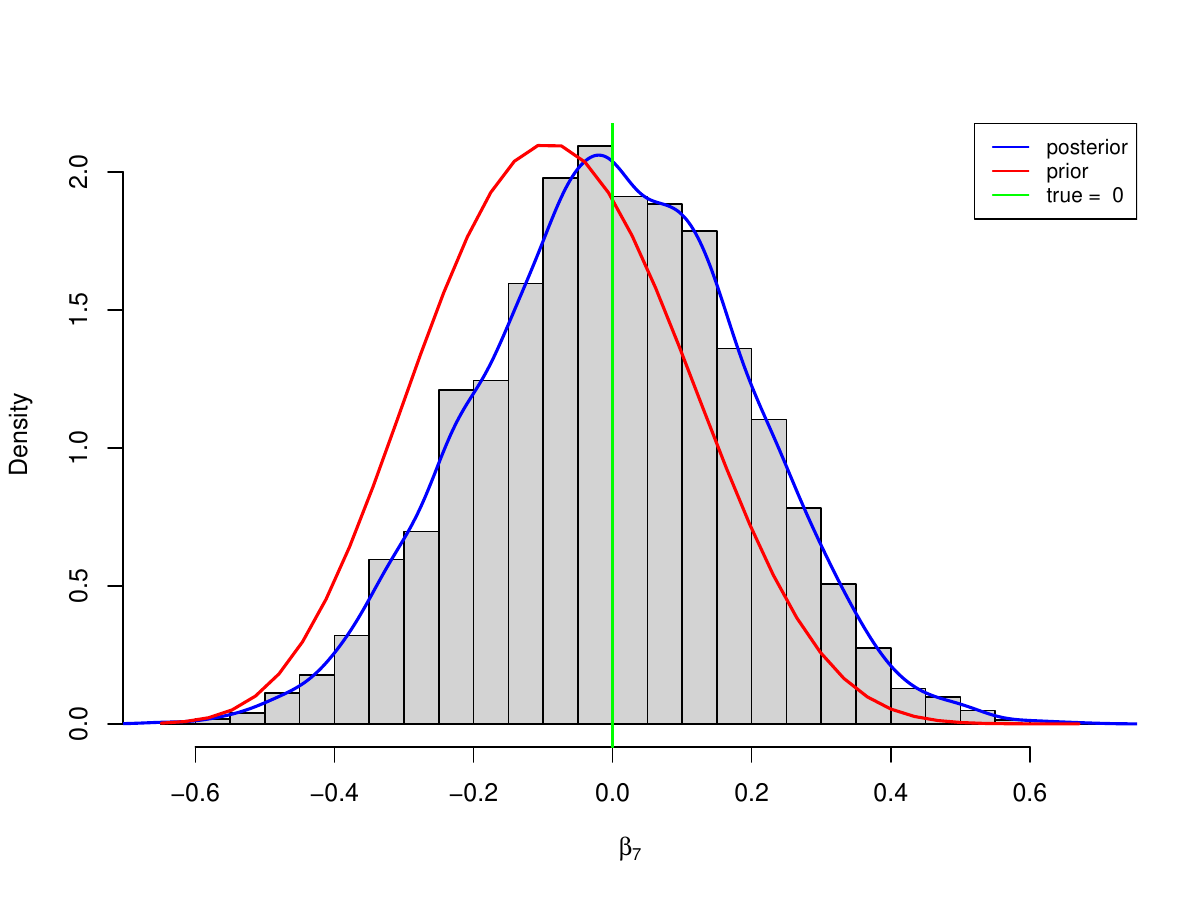} \label{fig:poi_beta7_c2d10}}
     \hfill
     \subfloat[$\beta_{20}$]{
         \includegraphics[width=0.4\textwidth]{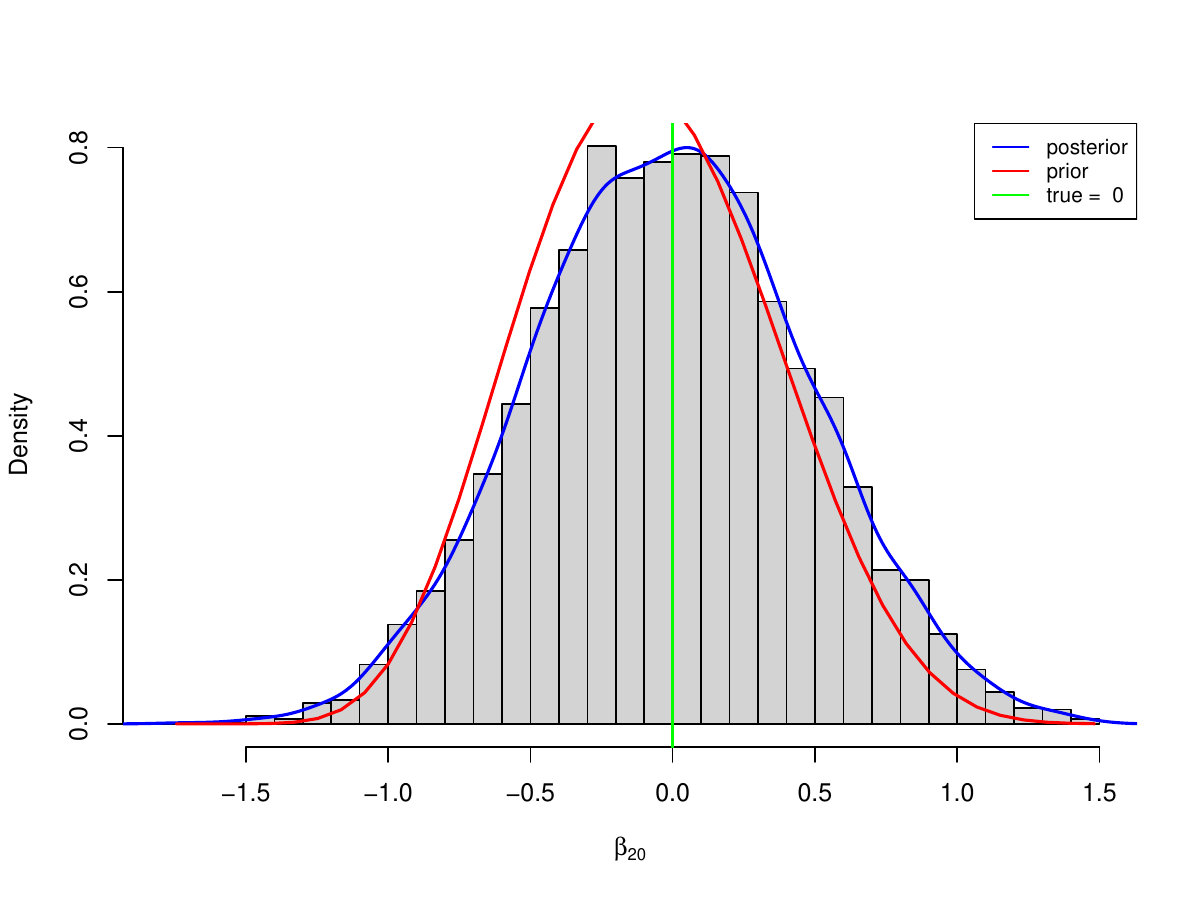} \label{fig:poi_beta20_c2d10}}
      \caption{Poisson Model. Posterior distribution of some non-active coefficients ($\vbeta^{(0)}$), for $n=1000$, $p=20$, $2c=4$ correlated and $d=10$ noisy variables.}
        \label{fig:settingC_irrelevant_betas}
\end{figure}

To further test the ability of our method to discriminate important covariates,
even when their effect sizes are small, 
we repeat the same simulation for a regime with
\textit{small regression coefficients}, by sampling 
$\vbeta^{*(1)}$ from $\text{Unif}([0.1, 0.5]^k)$.
For this experiment, with $\vbeta^{*(1)} = (0.472, 0.470, 0.372, 0.312, 0.280, 0.261)$,
we observe that 
\texttt{BayesVS-GLM} achieves better results than \texttt{SpikeSlab}, both in terms of accuracy of $\z$ (Figure \ref{fig:settingC_z_acc_np_curve_smallbeta}) and estimation of $\vbeta$ (Figure \ref{fig:poi_beta_mse_smallbeta}) and $\vbeta^{(1)}$ (Figure \ref{fig:poi_beta_rmse_smallbeta}),
in the small sample size regimes ($n=100$).
We acknowledge that \texttt{Horseshoe} appears to perform better coefficient estimation, albeit not performing model selection. 

This experiment highlights a benefit of the proposed method in the low-data, low-effect regime when compared to \texttt{SpikeSlab}: when a true coefficient's component is small, its associated $\beta_j$ can be drawn by \texttt{SpikeSlab}'s chain close to the ``spike'' component early, making it difficult for the sampler to later recognize the covariate as active.

\begin{figure}
    \centering
    \includegraphics[width=0.45\textwidth]{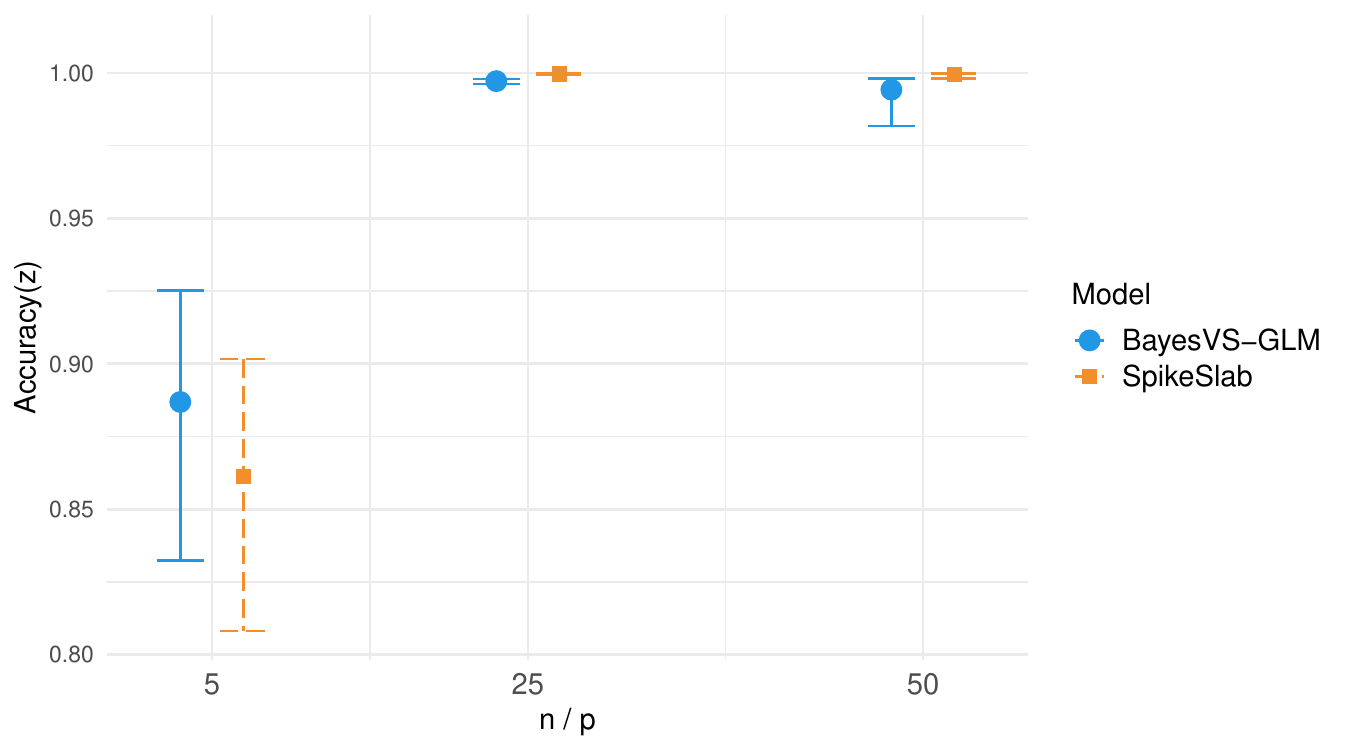}
    \caption{Poisson Model. With $p=20$ total covariates, $d=10$ noisy and $2c=4$ correlated variables, in \textit{small regression coefficients regime}. Our method achieves higher accuracy in recovering the inclusion vector $\z$ than \texttt{SpikeSlab} for small values of $n$. The methods reach comparable accuracy as $n/p$.}
    \label{fig:settingC_z_acc_np_curve_smallbeta}
\end{figure}
\begin{figure}
    \centering
    \subfloat[$MSE(\vbeta, \vbeta^{*})$]{
         \includegraphics[width=0.4\textwidth]{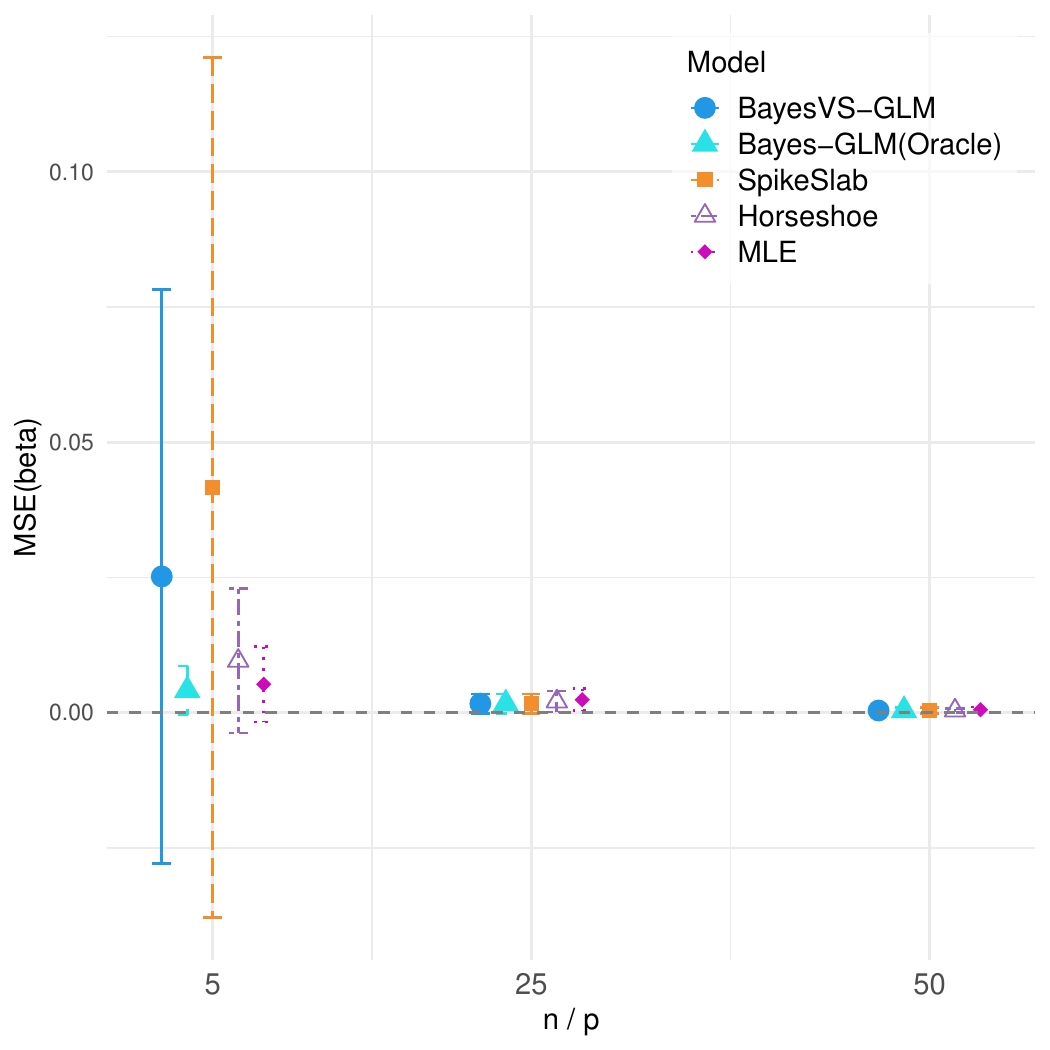}
         \label{fig:poi_beta_mse_smallbeta}}
     \hfill
     \subfloat[$RelMSE(\vbeta^{(1)}, \vbeta^{*(1)})$]{
         \includegraphics[width=0.4\textwidth]{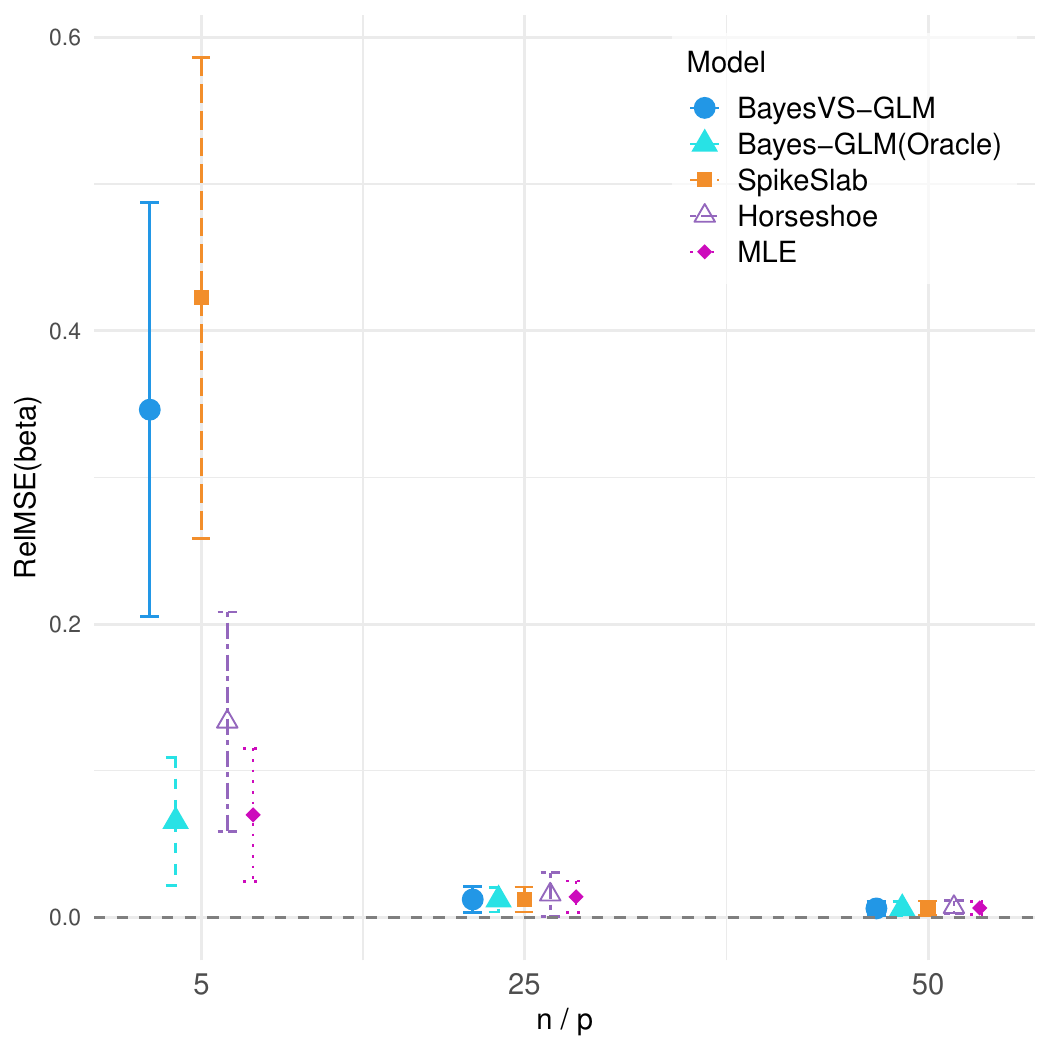}
         \label{fig:poi_beta_rmse_smallbeta}}
    \caption{Poisson Model. Error metrics of the coefficients $\vbeta$ in \textit{small regression coefficients regime}, for increasing ratio $n/p$ ($p=20$, $d=10$, $2c=4$). Our method (\texttt{BayesVS-GLM}) outperforms \texttt{SpikeSlab} baseline, in small sample size regimes.}
    \label{fig:settingC_poi_beta_error_np_curve_smallbeta}
\end{figure}

\subsection{Real data experiments}
\label{subsec:results_realdata}
In this section we illustrate the practical applicability of the proposed Bayesian variable selection procedure on real-world datasets,
covering data modeled via Logistic Model. Results for Poisson and Linear Models are available in the Supplemental Material Sections~\ref{subsec:real_dataset_poisson}-\ref{subsec:real_dataset_gaussian}.

Unlike the synthetic setting of Section~\ref{subsec:results_synthetic}, 
real datasets do not come with a known ground truth for the active set of predictors, and hence no oracle baseline is available for comparison. 
As in the synthetic experiments, posterior inference for all Bayesian methods is carried out using Gibbs Sampling with the same number of iterations and burn-in used throughout the paper, whenever applicable.

\subsubsection{Logistic Model for the Heart Disease dataset}
\label{subsubsec:real_dataset_binomial}
In this example we assess the performance of our model (\texttt{BayesVS-GLM}) on the \texttt{Heart Disease Data}, a multivariate dataset that provides medical indicators related to heart-disease diagnosis.
The task is a binary classification and consists in predicting the presence or absence of heart disease, based on the available clinical and demographic attributes.
It consists of $14$ attributes (categorical and continuous), and $920$ patient records, across four institutions \citep{uci_heart_disease}.

Following the preprocessing steps detailed by \citet{ogier2022flamby, rodriguez2024centralizedfederatedheartdisease}, we removed the observations that contain at least one missing value and excluded the variables with more than $30\%$ missing entries, namely \texttt{slope}, \texttt{ca} and \texttt{thal}. Redundant features such as \texttt{id} and \texttt{dataset} were also discarded. The response variable \texttt{num}, originally taking values $\{0,1,2,3,4\}$ was binarized into the variable \texttt{target}: all nonzero values were mapped to $1$, indicating the presence of heart disease, and $0$ otherwise.

After preprocessing, the dataset contained $740$ records with the following variables: 
patient age in years (\texttt{age}); 
binary indicator for gender (\texttt{sex}: \texttt{male} or \texttt{female}); 
indicator of the chest pain type (\texttt{cp}: \texttt{typical angina}, \texttt{atypical angina}, \texttt{non-anginal}, \texttt{asymptomatic}); 
resting blood pressure (in mmHg) on admission to the hospital (\texttt{trestbps}); 
the serum cholesterol in mg/dl (\texttt{chol}); 
a binary indicator for Fasting Blood Sugar (FBS) above $120$ mg/dl (\texttt{fbs}: 1 if \texttt{true}, 0 otherwise);
indicator of the resting electrocardiographic results (\texttt{restecg}: \texttt{normal}, \texttt{stt abnormality}, \texttt{lv hypertrophy});
maximum heart rate achieved (\texttt{thalch});
deviation of the ST segment (the section between the end of the S wave and the beginning of the T wave) on an electrocardiogram (ECG), induced by exercise relative to rest (\texttt{oldpeak});
binary indicator for the presence of heart disease (\texttt{target}: $1$ for presence, $0$ for absence).
After normalizing the continuous covariates and encoding categorical variables in dichotomous scale, we conducted a $30$-fold cross-validation, each time splitting the data into $90\%$ for training and $10\%$ for testing to evaluate model performance on unseen data. 
We estimated three models under identical conditions:
\texttt{BayesVS-GLM}, our proposed Bayesian GLM with variable selection; \texttt{BayesGLM}, a Bayesian GLM with all variables included (no sparsity); and its frequentist counterpart \texttt{MLE}.
The prior hyperparameters for \texttt{BayesVS-GLM} model were fixed to $\DparamScalar_0 = 0.01$, $\DparamVector_0 \overset{\iid}{\sim} Bern(0.5)$, $\alpha=1$.

To illustrate the performance of \texttt{BayesVS-GLM} in variable selection and estimation of model coefficients, Figure \ref{fig:heart_betaz} shows the posterior distribution of $\vbeta \circ \z$, obtained in the whole dataset,
while Figure \ref{fig:heart_z} reports the posterior inclusion probabilities over $30$ cv-repetitions.
\begin{figure}
  \centering
  \subfloat[~]{
     \includegraphics[height=6cm]{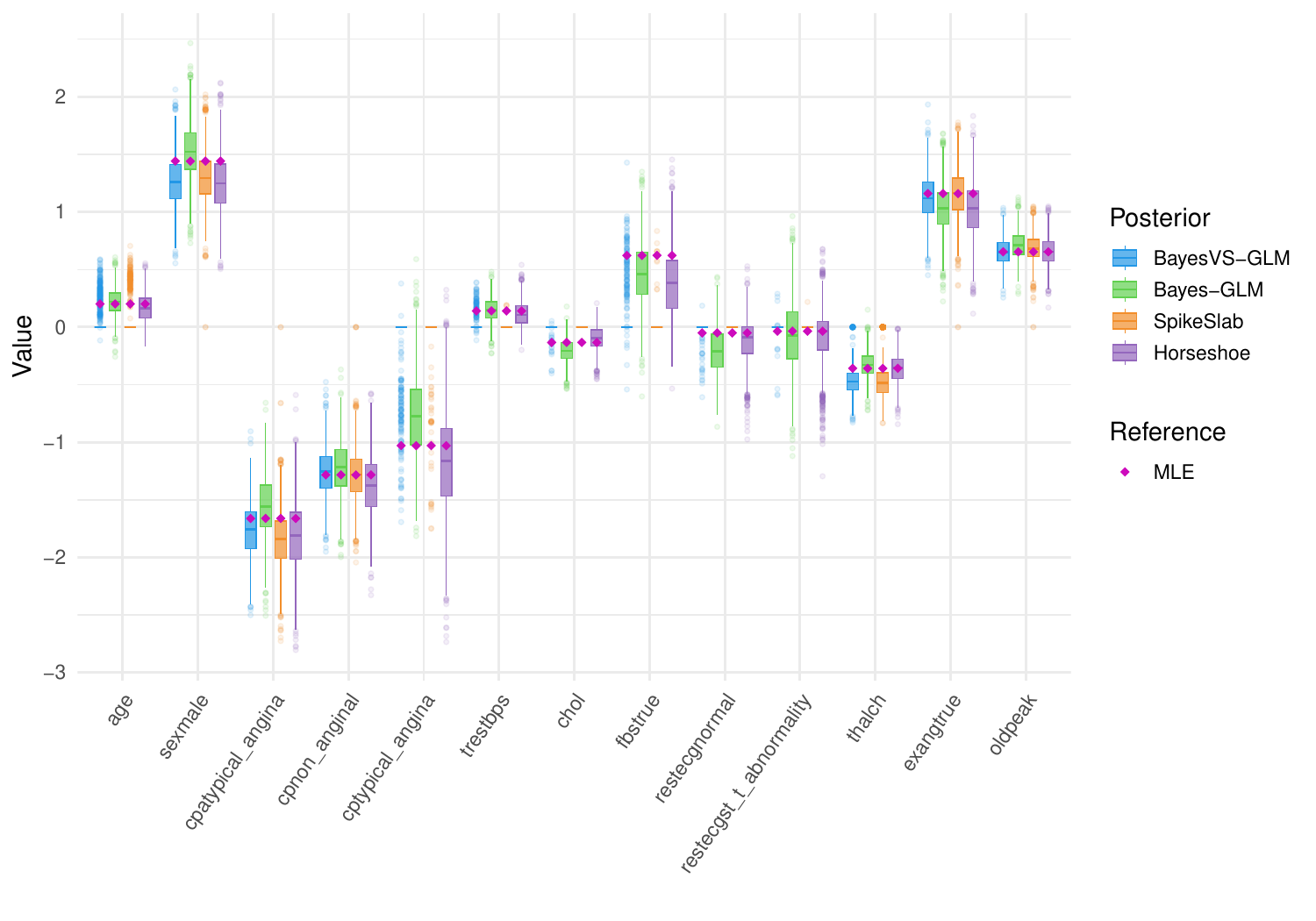}
     \label{fig:heart_betaz}}
  \subfloat[~]{
    \includegraphics[height=6cm]{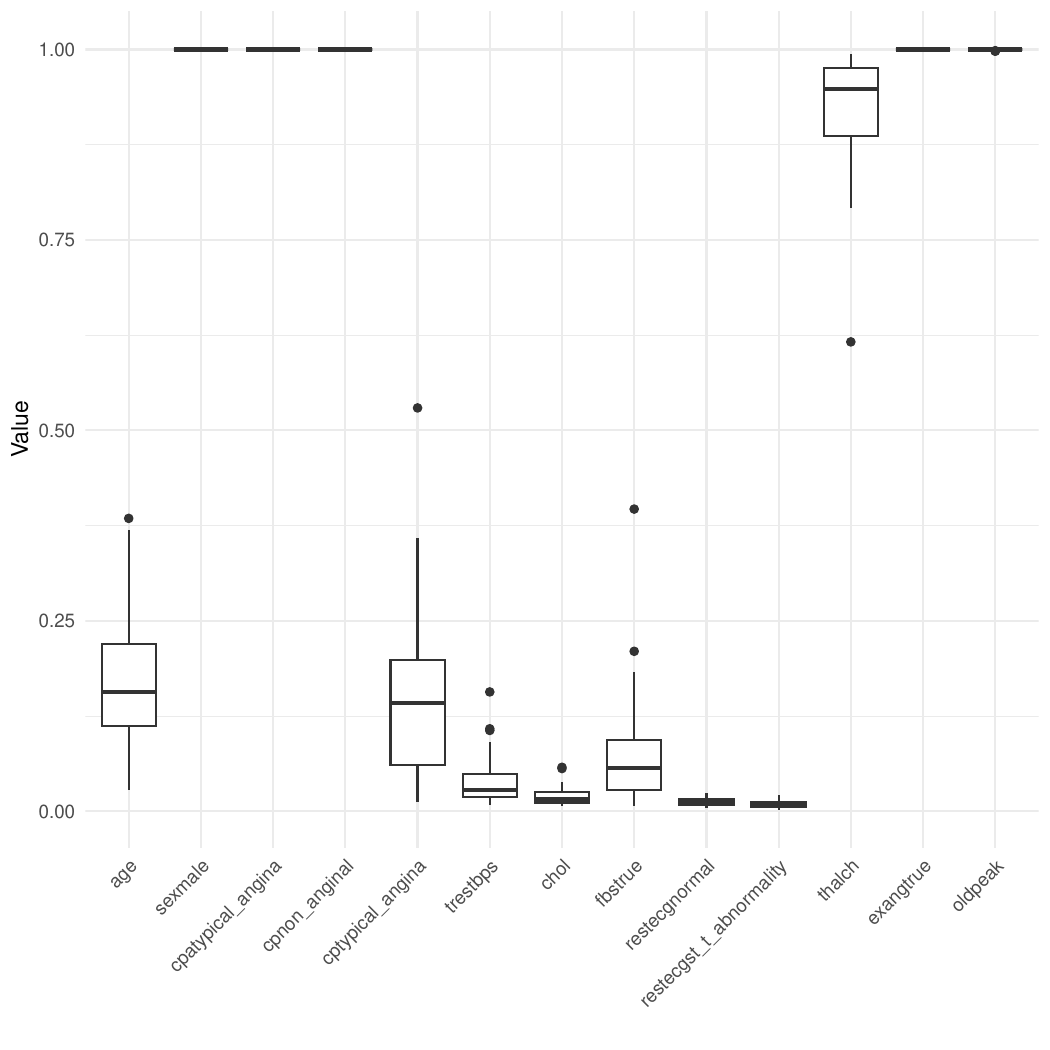}
    \label{fig:heart_z}}
    \caption{Heart Disease dataset. (a) Posterior distribution of $\vbeta \circ \z$ (\texttt{BayesVS-GLM}) and $\vbeta$ (baselines), on the whole dataset. (b) \texttt{BayesVS-GLM}'s posterior inclusion probabilities over $30$ folds.}
\end{figure}
From these results we deduce that variables as \texttt{cp}, \texttt{sex}, \texttt{thalch}, \texttt{exang} and \texttt{oldpeak} are important, while others as \texttt{fbs}, \texttt{chol}, and \texttt{restecg} are low importance. These findings align with studies of \citet{rodriguez2024centralizedfederatedheartdisease, teja2025optimizing}, which also identified chest pain type (\texttt{cp}) and the exercise-induced angina (\texttt{exang}) as good indicators of heart disease. Moreover, elevated \texttt{oldpeak} values, indicating ST-segment depression during exercise, are consistent with ischemic conditions, confirming its clinical relevance as a risk factor \citep{rodriguez2024centralizedfederatedheartdisease}.
The inclusion probability of \texttt{age} (around $0.4-0.5$) indicates a moderate but uncertain contribution. While age is known to correlate with heart disease, it seems to provide limited additional information once more specific covariates, such as \texttt{cp} and \texttt{oldpeak}, are included in the model. The low relevance of \texttt{trestbps} (resting blood pressure) may be due to the fact that it reflects overall health rather than direct cardiac function, and it can be influenced by medication or measurement variability. Similarly, \texttt{restecg} adds little to the model compared with variables that capture exercise-induced responses (\texttt{exang}, \texttt{oldpeak}), which appear to carry more importance. Since the number of variables is $p=13$, then a model selection procedure might in principle explore all the $2^{13} - 1$ models. Of all these models, during the MCMC iterations, our Gibbs sampler selected $25$. The most frequent model (covers about $70\%$ of all the explored models) contains the same variables highlighted in Figure \ref{fig:heart_z} while the $10$ most frequent ones, covering the $93.2\%$ (see Table \ref{tab:most_frequent_model} in the Supplementary Material) mainly differ by the variables \texttt{age}, \texttt{cp\_typical\_angina} and \texttt{chol}. A detailed view on pairwise association among variables is reported in Supplementary Material Section \ref{apd:subsec_real_dataset_logistic}.

To evaluate prediction accuracy on the test sets, we use several classification-specific metrics, namely Balanced Accuracy Detection Prevalence, Detection Rate, F1 score, Negative Predicted Value, Positive Predicted Value, Precision, Prevalence, Recall, Sensitivity, Specificity. 
Details on the definitions of these metrics can be found in \citet{kuhn2008building}.
Figure~\ref{fig:error_heart} reports Balanced Accuracy and F1 score across folds, while a more complete table with all metrics can be found in Supplementary Material Section \ref{apd:results_real_data}. 
\begin{figure}
    \centering
    \includegraphics[width=0.6\linewidth]{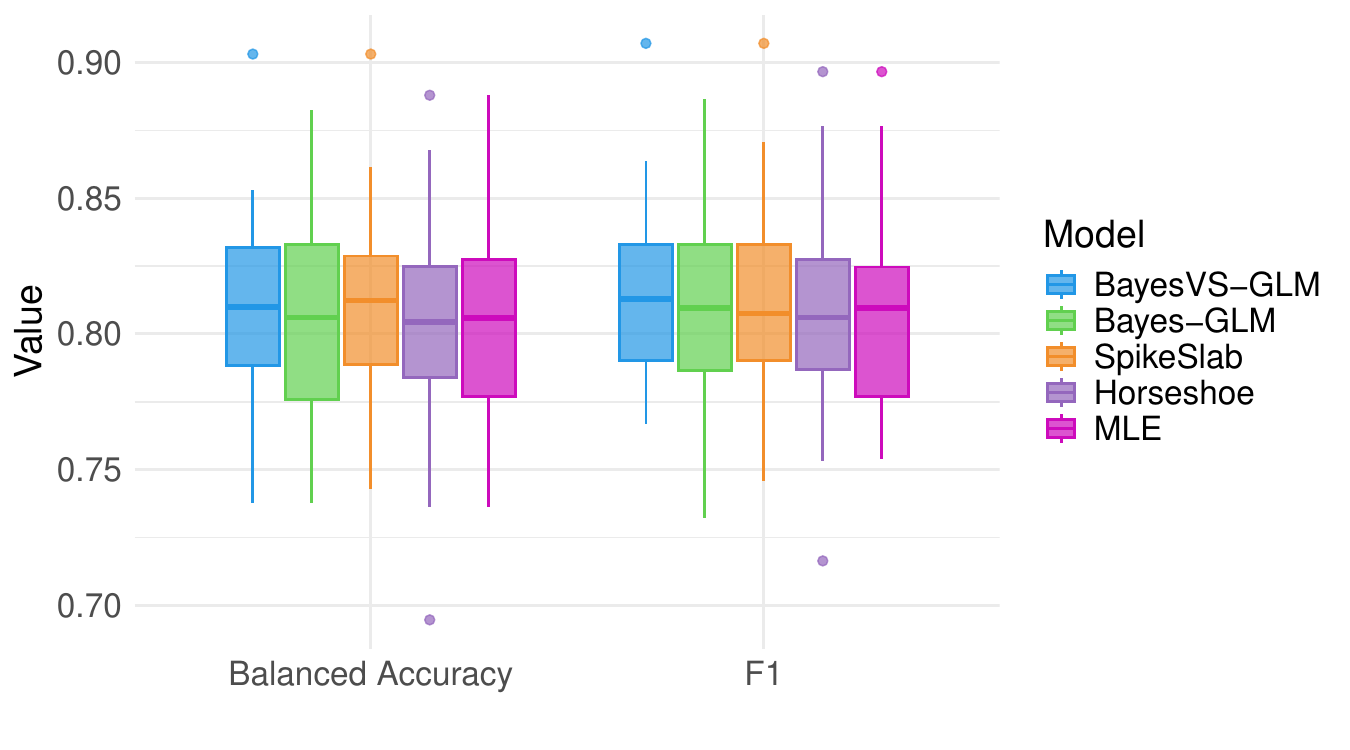}
    \caption{Heart Disease dataset. Prediction error metrics across folds: Balanced accuracy (left) and F1 Score (right), for the evaluated methods.}
    \label{fig:error_heart}
\end{figure}
Overall, \texttt{BayesVS-GLM} attains a comparable or slightly improved predictive accuracy compared to all baselines, while identifying relevant covariates consistent with established clinical findings \citep{rodriguez2024centralizedfederatedheartdisease}, \citet{teja2025optimizing}. 

\section{Conclusion}
\label{sec:conclusion}
We introduced, analyzed and validated a Bayesian, hierarchical conjugate framework to perform concurrent variable selection and posterior inference in Generalized Linear Models.
A Gibbs Sampler was developed, using conjugate priors for GLM regression coefficients $\vbeta$ and covariate indicator variables $\z$, to perform efficient posterior inference of latent indicator variables and GLM parameters.
We proved posterior consistency of the variable selection indicator $\z$, as well as for the active regression coefficients $\vbeta^{(1)}$.
The posterior variability incorporates both sources of uncertainty: that of model selection and parameters estimation. 

We support our theoretical findings with empirical results on both simulated and real-data use-cases, showcasing successful covariate selection and estimation, \ie accurate sparsity recovery and parameter estimation, across various GLM instances.

An \textsf{R} package is developed and available upon request to the authors.

\section{Acknowledgements}
The publication was produced with funding from the Italian Ministry of University and Research as part of the Call for Proposals of the PRIN 2022 - Project title “MEMIMR: Measurement Errors and Missing Information in Meta-Regression” - Project No.2022FZY9PM - CUP C53C24000740006.

L. Filippozzi and C. Agostinelli are members of the Gruppo Nazionale per l'Analisi Matematica, la Probabilità e le loro Applicazioni (GNAMPA) of the Istituto Nazionale di Alta Matematica (INdAM).

I. Urteaga acknowledges the support of Grant RYC2023-045922-I funded by \\ MICIU/AEI/10.13039/501100011033 and by ESF+,
as well as the  Basque Government through the BERC 2022-2025 program
and the Ministry of Science and Innovation's BCAM Severo Ochoa accreditation
CEX2021-001142-S/MICIN/AEI/10.13039/501100011033.

\appendix

\section*{Supplementary Material}

This Supplementary Material contains additional details on the posterior derivations (Section \ref{apd:posterior_general}), the proofs (Section \ref{apd:proofs}), procedures for computation and sampling of the conjugate prior and posterior distributions of the GLM's regression coefficients (Section \ref{sec:sampling}),
and extra experimental results (Sections \ref{apd:trace_poisson}, \ref{apd:results_poisson_gaussian_binomial}, \ref{apd:results_real_data} and \ref{apd:results_correlation}).

\section{GLM conjugate posteriors}
\label{apd:posterior_general}
This section contains the detailed derivations of the posterior distributions for the proposed GLM model, 
which complement the main results presented in the paper.

Let $\vect X$ be $n\times p$ matrix of data, and consider the GLM with identity link function and general likelihood given by
\begin{align*}
    \yi \mid\thetai, \tau \  & \sim \ f_Y(\yi| \thetai, \tau) =  \exp \left\{ \frac{\yi \thetai - b(\thetai)}{A_i(\tau)} + c(\yi, \tau)\right\} \\
                                & \thetai = \etai \mid \x_i, \z, \vbeta \ = \ \x_i^\top (\vbeta \circ \z) = \sumj \xij\betaj \zj \\ 
                                & \mu := \mathbb{E}[Y|X ]  \ = g^{-1}(\veta) \ \\
    \z  \mid\vc \             & \sim \ \prod_{j=1}^p Bern (\cj) \\ 
    \vc \mid\alpha \          & \sim \ \prod_{j=1}^p Beta\left( \frac{\alpha}{p}, 1 \right)  \\ 
    \vbeta \mid \z \            & \sim \ p(\vbeta \mid \z) \\
    \tau & \sim p(\tau)
\end{align*}
Following \citet{chen2003conjugate_prior_glm}, we choose conjugate priors for $\vbeta$ and $\tau$  consistent with the exponential family structure of the likelihood. In a classical GLM, for fixed $\tau$, a conjugate prior for the regression coefficients $\vbeta$—denoted $\mathcal{D}(\DparamScalar_0, \DparamVector_0)$—takes the form of Eq. \eqref{eq:prior_beta_exp_family}:
\begin{align*}
    p(\vbeta \mid \tau, \DparamScalar_0, \DparamVector_0) &\propto \exp \left\{ 
    \DparamScalar_0 \tau \left(  \DparamVector_0^\top \theta(\veta) - J^\top b(\theta(\veta))\right) 
    \right\} 
    =\exp \left\{ 
    \DparamScalar_0 \tau \left(  \DparamVector_0^\top \X \vbeta - J^\top b(\X \vbeta)\right) 
    \right\}
\end{align*}
where $\DparamScalar_0 \in \R^+$ and $\DparamVector_0 \in \R^{n \times 1}$ are  hyperparameters, and $J \in \R^{n \times 1}$ is a vector of ones. For the sake of simplicity, from now on we will drop the dependency on the hyperparameters $ \DparamScalar_0, \DparamVector_0$ and write  $p(\vbeta \mid \tau, \DparamScalar_0, \DparamVector_0) = p(\vbeta \mid \tau)$. 
When $\tau$ is unknown, a conjugate joint prior on $(\vbeta, \tau)$ is given by
\begin{align*}
    \pi(\vbeta, \tau) & = p(\vbeta \mid \tau) p(\tau) \ \propto \exp \left\{ \DparamScalar_0 \tau \left(  \DparamVector_0^\top \X \vbeta - J^\top b(\X \vbeta) + J^\top c(\DparamVector_0, \tau) \right) 
    \right\}\pi(\tau)
\end{align*}
where $\pi(\tau)$ must be chosen to ensure conjugacy, depending on the particular form of the exponential family (\eg an Inverse Gamma in the Linear Model). 
First of all notice that the following equivalent forms hold:
\begin{align}
    \X \vbeta \ &=  \ \X \left( \vbeta \circ \left(\z + (J_p - \z)\right) \right)  \nonumber \\
    &=  \ \X \left( \vbeta \circ \z\right) +  \X \left( \vbeta \circ (J_p - \z) \right) 
    \X \vbeta \ \nonumber \\
    &=\ \X^{(1)} \vbeta^{(1)} + \X^{(0)} \vbeta^{(0)} \label{eq:equivalence_xbeta}
\end{align}
where $J_p = (1, 1, \ldots, 1) \in \mathbb{R}^p$, and we denote as $\X^{(k)} \ (\vbeta^{(k)})$ the submatrix (subvector) containing only the columns (components) $j$ such $z_j=k$.\\

To ensure conjugacy of the distribution $p(\vbeta \mid \z, \tau)$ we adapt the prior in Eq. \ref{eq:prior_beta_exp_family} to be dependent on $\z$ as:
{\small 
\begin{align*}
    p(\vbeta \mid \z ) & = p\left(\vbeta^{(0)}, \vbeta^{(1)} \mid \z \right) =  p \left( \vbeta^{(1)} \mid \z \right) p\left(\vbeta^{(0)} \mid \z \right) \\
    &= \exp\left\{ \DparamScalar_0 \tau \left(\DparamVector_0^\top \X^{(1)} \vbeta^{(1)} - \vect{J}_n^{\top} b(\X^{(1)} \vbeta^{(1)}) \right) \right\} \exp\left\{ \DparamScalar_0 \tau \left(\DparamVector_0^\top \X^{(0)} \vbeta^{(0)}- \vect{J}_n^\top b(\X^{(0)} \vbeta^{(0)}) \right) \right\} = \\
    &= \exp\left\{ \DparamScalar_0 \tau \left(\DparamVector_0^\top \X \vbeta - \vect{J}_n^\top b\left(\X^{(1)} \vbeta^{(1)} \right) - \vect{J}_n^\top b\left(\X^{(0)} \vbeta^{(0)} \right) \right) \right\} \ , 
\end{align*}}
where the last equality is ensured by Eq. \eqref{eq:equivalence_xbeta}. 
Using these priors, the joint distribution results as:
\begin{align*}
    p(\z &, \vc, \vbeta, \tau, \y \mid\vect X)  = p( \y  \mid\X, \z, \vbeta, \vc, \tau)p(\z \mid \vc) p(\vc \mid \alpha) p(\vbeta, \tau \mid \z)  \ = \\
    &\propto  p( \y  \mid\X, \z, \vbeta, \vc, \tau) 
    \left[ \prod_{j=1}^p {\cj}^{\zj} {(1- \cj)}^{1- \zj}\right] \left[ \prod_{j=1}^p {\cj}^{ \frac{\alpha}{p} -1 } \right]p(\vbeta \mid \z, \tau) p( \tau)
\end{align*}

\paragraph{Posterior of $z_h$}

Let's denote $z_h$ the $h$-th component of the random vector $\z$, and $\z_{-h} = \left( z_1, \ldots, z_{h-1}, z_{h+1}, \ldots z_p\right)$. We are now interested in the full conditional of $z_h$:
\begin{align*}
    p(z_h \mid \X, \y, \z_{-h}, \vc, \vbeta, \tau ) &\propto p(\y \mid \vbeta, \z, \tau, \X )p(z_h \mid \z_{-h}, \vc) p(\z_{-h} \mid \vc) p(\vbeta \mid \z, \tau) \\
     &\propto p(\y \mid \vbeta, \z, \tau, \X )p(z_h \mid \z_{-h}, c_h) p(\vbeta \mid \z, \tau)
\end{align*}
Recall that $\etai = \sumj \xij\betaj \zj$, 
therefore, for $s=0, \ 1$: 
{\small 
\begin{align*}
    p(z_h  = s \mid \X, \y, \z_{-h}, \vc, \vbeta, \tau ) &\propto p(\y \mid \vbeta, \z, \tau, \X )p(z_h = s \mid \z_{-h}, c_h) p(\vbeta \mid z_h=s, \z_{-h}, \tau) \\
    & \propto 
    \underbrace{\exp \left. \left\{ \frac{\yi \etai - b(\etai)}{A_i(\tau)} \right\}  \right\rvert_{z_h  = s} }_{=:P^s}
    \cdot \  c_h^{s}( 1 -c_h)^{1-s}  \cdot  p(\vbeta \mid z_h=s, \z_{-h}, \tau)\\
    &\propto P^s c_h^{s}( 1 -c_h)^{1-s}   p(\vbeta \mid z_h=s, \z_{-h}, \tau) 
\end{align*}}
Hence 
\begin{align*}
    p(z_h  = s \mid \X, \y, \z_{-h}, c_h, \vbeta, \tau ) &=
   \frac{P^s c_h^{s}(1 -c_h)^{1-s}   p(\vbeta \mid z_h=s, \z_{-h}, \tau)}{P^1 c_h   p(\vbeta \mid z_h=1, \z_{-h}, \tau) + P^0 (1-c_h)   p(\vbeta \mid z_h=0, \z_{-h}, \tau)} \ , 
\end{align*}
meaning that:
\begin{equation*}
    z_h \mid \X, \y, \z_{-h}, c_h, \vbeta, \tau  \ \sim  \ 
    Bern\left( 
    \frac{c_h\cdot P^1 p(\vbeta \mid z_h=s, \z_{-h}, \tau)}{c_h P^1 p(\vbeta \mid z_h=1, \z_{-h}, \tau) + (1- c_h) P^0 p(\vbeta \mid z_h=0, \z_{-h}, \tau)}
    \right) 
\end{equation*}

\paragraph{Posterior of $\vbeta$}
First of all, suppose we can order the vectors $\vbeta$ and $\z$ in such a way that all the zero entries of $\z$ are in the first half of the vector. Them we can write
\[
\vbeta \circ \z = \begin{pmatrix}
    \vbeta^{(0)} \\
    \vbeta^{(1)}
\end{pmatrix} \circ
\begin{pmatrix}
    \vect 0 \\
    \vect 1
\end{pmatrix} = 
\begin{pmatrix}
    \vect 0 \\
    \vbeta^{(1)}
\end{pmatrix} = \Tilde{I} \vbeta 
\]
where $\Tilde{I}$ is the following block matrix: 
\begin{equation}
    \label{eq:tilde_I}
    \Tilde{I} := \begin{pNiceArray}{cc|cccc}
      \Block{2-2}<\Large>{\mathbf{0}} && \Block{2-4}<\Large>{\mathbf{0}} \\
      \\
      \hline
      \Block{4-2}<\Large>{\mathbf{0}} && 1 & 0 & \ldots & 0  \\
      && 0 & 1 & \ldots & 0 \\
      && & & \ddots  \\
      && 0 & 0 & \ldots & 1
    \end{pNiceArray}
\end{equation}
Therefore, the linear predictor $\veta$ can be re-written as: 
\[
\veta = \left[ \sumj \xij \betaj \zj \right]_{i=1}^n = \ \X (\Tilde{I} \vbeta) = (\X \Tilde{I}) \vbeta = \Tilde{\X}\vbeta
\]
where $\Tilde{\X} = \X \Tilde{I}$.
Recall that the prior over $\vbeta \mid \z$ is:
\begin{align*}
    p(\vbeta \mid \z) & = p\left(\vbeta^{(0)}, \vbeta^{(1)} \mid \z \right) =  p \left( \vbeta^{(1)} \mid \z \right) p\left(\vbeta^{(0)} \mid \z \right) \\
    &= \exp\left\{ \DparamScalar_0 \tau \left(\DparamVector_0^\top \X^{(1)} \vbeta^{(1)} - \vect{J}_n^\top b(\X^{(1)} \vbeta^{(1)}) \right) \right\} \exp\left\{ \DparamScalar_0 \tau \left(\DparamVector_0^\top \X^{(0)} \vbeta^{(0)}- \vect{J}_n^\top b(\X^{(0)} \vbeta^{(0)}) \right) \right\} \\
    &= \exp\left\{ \DparamScalar_0 \tau \left(\DparamVector_0^\top \X \vbeta - \vect{J}_n^\top b\left(\X^{(1)} \vbeta^{(1)} \right) - \vect{J}_n^\top b\left(\X^{(0)} \vbeta^{(0)} \right) \right) \right\}
\end{align*}
Going back to the main computation:
{
\small
\begin{align*}
    p(\vbeta & \mid \X, \y, \tau, \z)  \propto p(\y |\X, \z,  \vbeta) p(\vbeta \mid \z, \tau) \\
    & \propto \prodi \exp \left\{ \frac{\etai \yi - b(\etai)}{A_i(\tau)} \right\} \cdot 
    \exp\left\{ \DparamScalar_0 \tau \left[\DparamVector_0^\top \X \vbeta - \vect{J}_n^{\top} b\left(\X^{(1)} \vbeta^{(1)} \right) - \vect{J}_n^{\top} b\left(\X^{(0)} \vbeta^{(0)} \right) \right] \right\}  \\
    &= \exp \left\{ \sumi \frac{\etai \yi - b(\etai)}{A_i(\tau)} \right\} \cdot 
    \exp\left\{ \DparamScalar_0 \tau \left[\DparamVector_0^\top \X \vbeta - \vect{J}_n^{\top} b\left(\X^{(1)} \vbeta^{(1)} \right) - \vect{J}_n^{\top} b\left(\X^{(0)} \vbeta^{(0)} \right) \right] \right\} \\
    &=\exp \left\{ \frac{\y^\top (\Tilde{\X}\vbeta) - \vect{J}_n^{\top} b(\Tilde{\X}\vbeta)}{A(\tau)} \right\}  \cdot
    \exp\left\{ \DparamScalar_0 \tau \left[\DparamVector_0^\top \X \vbeta - \vect{J}_n^{\top} b\left(\X^{(1)} \vbeta^{(1)} \right) - \vect{J}_n^{\top} b\left(\X^{(0)} \vbeta^{(0)} \right) \right] \right\}  \\
    &=\exp \left\{ \frac{1}{A(\tau)} \left[ \y^\top (\Tilde{\X}\vbeta) - \vect{J}_n^\top b(\Tilde{\X}\vbeta) \right] + \DparamScalar_0 \tau \left[\DparamVector_0^\top \X \vbeta - \vect{J}_n^\top b\left(\X^{(1)} \vbeta^{(1)} \right) - \vect{J}_n^\top b\left(\X^{(0)} \vbeta^{(0)} \right) \right] \right\} \\ 
    &=\exp \left\{ \frac{1}{A(\tau)} \left[ \y^\top (\X^{(1)}\vbeta^{(1)}) - \vect{J}_n^\top b(\X^{(1)}\vbeta^{(1)}) \right] + \DparamScalar_0 \tau \left[\DparamVector_0^\top \X \vbeta - \vect{J}_n^\top b\left(\X^{(1)} \vbeta^{(1)} \right) - \vect{J}_n^\top b\left(\X^{(0)} \vbeta^{(0)} \right) \right] \right\} \\
    &= \exp \left\{ \frac{1}{A(\tau)} \left[ \y^\top (X^{(1)}\vbeta^{(1)}) - \vect{J}_n^\top b(X^{(1)}\vbeta^{(1)}) \right] + \DparamScalar_0 \tau \left[\DparamVector_0^\top \X^{(1)} \vbeta^{(1)} - \vect{J}_n^\top b(\X^{(1)} \vbeta^{(1)}) \right] \right. + \\
    &  \qquad \qquad +  \left. \DparamScalar_0 \tau \left[\DparamVector_0^\top \X^{(0)} \vbeta^{(0)} - \vect{J}_n^\top b\left(\X^{(0)} \vbeta^{(0)} \right) \right] \right\} \ . 
\end{align*}}
It is then possible to factorize the posterior of $\vbeta$ into: 
\[
p(\vbeta \mid \X, \y, \tau, \z) \propto p(\vbeta^{(0)} \mid \z, \y, \tau) p(\vbeta^{(1)} \mid \z, \y, \tau), 
\]
where: 
{\small 
\begin{align*}
    p(\vbeta^{(0)} \mid \z, \y,  \tau) 
    & \propto  \exp \left\{ \DparamScalar_0 \tau \left(\DparamVector_0^\top \X^{(0)} \vbeta^{(0)} - \vect{J}_n^\top b\left(\X^{(0)} \vbeta^{(0)} \right) \right) \right\}  \\ 
    p(\vbeta^{(1)} \mid  \z, \y, \tau) 
    & \propto  \exp \left\{ \frac{1}{A(\tau)} \left( \y^\top (X^{(1)}\vbeta^{(1)}) - \vect{J}_n^\top b(X^{(1)}\vbeta^{(1)}) \right) + \DparamScalar_0 \tau \left(\DparamVector_0^\top \X^{(1)} \vbeta^{(1)} - \vect{J}_n^\top b(\X^{(1)} \vbeta^{(1)}) \right) \right\}   \\
    &= \exp \left\{ \left( \frac{1}{A(\tau)} \y^\top + \DparamScalar_0 \tau \DparamVector_0^\top \right) \X^{(1)}\vbeta^{(1)} - \left( \frac{1}{A(\tau)} + \DparamScalar_0 \tau   \right)\vect{J}_n^\top b(\X^{(1)}\vbeta^{(1)})  \right\}   \\
    &=\exp \left\{ \left( \frac{ \y^\top + \DparamScalar_0 \tau A(\tau) \DparamVector_0^\top }{A(\tau)}\right) \X^{(1)}\vbeta^{(1)} - \left( \frac{1 + \DparamScalar_0 \tau A(\tau)}{A(\tau)} \right)\vect{J}_n^\top b(\X^{(1)}\vbeta^{(1)})  \right\}  \\
    &=\exp \left\{ \left( \frac{1 + \DparamScalar_0 \tau A(\tau)}{A(\tau)} \right) \left[ \left( \frac{\y + \DparamScalar_0 \tau A(\tau) \DparamVector_0 }{1 + \DparamScalar_0 \tau A(\tau)} \right)^\top \X^{(1)}\vbeta^{(1)} - \vect{J}_n^\top b(\X^{(1)}\vbeta^{(1)}) \right] \right\} \ , 
\end{align*}}
that corresponds to 
\[
\vbeta^{(0)} \mid \z^{(0)}, \y, \DparamScalar_0, \DparamVector_0, \tau \ \sim \ \mathcal{D} \left( \DparamScalar_0, \DparamVector_0 \right)
\]
\[
\vbeta^{(1)} \mid \z^{(1)}, \y, \DparamScalar_0, \DparamVector_0, \tau \ \sim \ \mathcal{D} \left( \hat{\DparamScalar}_0, \hat{\DparamVector}_0\right)
\]
with 
\[
 \hat{\DparamScalar}_0 = \frac{1 + \DparamScalar_0 \tau A(\tau)}{A(\tau)} , \qquad \hat{\DparamVector}_0 = \frac{\y + \DparamScalar_0 \tau A(\tau) \DparamVector_0 }{1 + \DparamScalar_0 \tau A(\tau)} \ .   
\]

\subsection{\texorpdfstring{Integrating $\vc$ out}{Integrating c out}}
\label{apd:glm_general_withoutc}
\begin{equation*}
        p(z_h, c_h \mid \X, \y, \z_{-h}, \vc_{-h}, \tau, \beta) \ \propto p( \y  \mid\z, \vc, \vbeta, \tau,  \vect X)p(z_h, c_h \mid\z_{-h}) p(\vbeta \mid \z, \tau )
\end{equation*}
Hence:
\begin{align*}
    p(z_h| \X, \y, \z_{-h}, \tau, \beta) \ 
    &\propto \int_{0}^{1} p( \y  \mid\z, \vbeta, \tau,  \vect X)p(z_h, c_h \mid\z_{-h}) p(\vbeta \mid \z, \tau )  d c_h  \  \\
    &\propto \int_{0}^{1} p( \y  \mid\z, \vbeta, \tau,  \vect X)p(z_h\mid\z_{-h}) p(\vc_{h}) p(\vbeta \mid \z, \tau )   d c_h  \  \\
    &\propto p( \y  \mid\z, \vbeta, \tau,  \vect X) p(\vbeta \mid \z, \tau )  \int_{0}^{1}  {c_h}^{z_h} {(1- c_h)}^{1- z_h} {c_h}^{ \frac{\alpha}{p} -1 }  d c_h  \\ 
    &= p( \y  \mid\z, \vbeta, \tau,  \vect X) p(\vbeta \mid \z, \tau ) \int_{0}^{1}  {c_h}^{\left(z_h + \frac{\alpha}{p} \right) -1} {(1- c_h)}^{(2- z_h) -1}d c_h  \\  
    &= p( \y  \mid\z, \vbeta, \tau,  \vect X) p(\vbeta \mid \z, \tau )\mathcal{B}\left( z_h + \frac{\alpha}{p}, 2 - z_h\right) 
\end{align*}
where $\mathcal{B}(a,b) = \frac{\Gamma(a)\Gamma(b)}{\Gamma(a+b)}$ indicates the Beta Function. 
Hence: 
\begin{align*}
     p(z_h = 1| \z_{-h}, \tau, \beta, \X, \y)  & \ \propto P^1 p(\vbeta \mid z_h =1, \z_{-h}, \tau ) \mathcal{B}\left( 1+ \frac{\alpha}{p}, 1\right)  \\
     p(z_h = 0| \z_{-h}, \tau, \beta, \X, \y)  & \ \propto P^0 p(\vbeta \mid z_h =0, \z_{-h}, \tau )  \mathcal{B}\left( \frac{\alpha}{p}, 2 \right) 
\end{align*}
{\small \begin{align*}
     z_h \mid \z_{-h}, \tau, \beta, \X, \y \ & \sim  
     Bern\left(  \frac{\displaystyle\mathcal{B}\left( 1+ \frac{\alpha}{p}, 1\right) P^1 p(\vbeta \mid z_h =1, \z_{-h}, \tau ) }{ \displaystyle \mathcal{B}\left( 1+ \frac{\alpha}{p}, 1\right) P^1 p(\vbeta \mid z_h =1, \z_{-h}, \tau ) + \mathcal{B}\left( \frac{\alpha}{p}, 2 \right) P^0 p(\vbeta \mid z_h =0, \z_{-h}, \tau ) }  \right) 
\end{align*}}
where we can easily simplify
\begin{align*}
    \mathcal{B}\left( 1+ \frac{\alpha}{p}, 1\right) 
    & = \frac{\Gamma\left( 1+ \frac{\alpha}{p} \right) \Gamma(1)}{\Gamma\left( 2+ \frac{\alpha}{p} \right)} = \left( 1+ \frac{\alpha}{p} \right)^{-1} \\
    \mathcal{B}\left( \frac{\alpha}{p}, 2\right) 
    & = \frac{\Gamma\left( \frac{\alpha}{p} \right) \Gamma(2)}{\Gamma\left( 2+ \frac{\alpha}{p} \right)} = \left( 1+ \frac{\alpha}{p} \right)^{-1}\left( \frac{p}{\alpha} \right)
\end{align*}
so we get
\[
    z_h \mid \z_{-h}, \tau, \beta, \X, \y \ 
    \sim  Bern\left( \frac{ P^1 p(\vbeta \mid z_h =1, \z_{-h}, \tau )  }{ P^1 p(\vbeta \mid z_h =1, \z_{-h}, \tau )  +  \displaystyle \frac{p}{\alpha} P^0 p(\vbeta \mid z_h =0, \z_{-h}, \tau ) }  \right) \ . 
\]

\subsection{Linear Model: conjugate posterior}
\label{apd:posterior_linear}
In this section we derive the posterior distributions for the Gaussian linear model, treated as a special case of the general GLM formulation introduced earlier.
Let $\vect X$ be an $n\times p$ design matrix and consider the model
\[
Y_i \mid \mu_i, \tau \sim \mathcal{N}(\mu_i,\; \tau), 
\qquad 
\mu_i = \theta_i = \x_i^\top (\vbeta \circ \z).
\]
The corresponding likelihood is
\[
p(\y \mid \X, \vbeta, \z, \tau) \propto \prod_{i=1}^n \exp \left\{ - \frac{\left( \vect \yi - \sum_{j=1}^p \xij\betaj \zj  \right)^2}{2\tau }\right\} \ . 
\]
The distribution $\mathcal{D}(\DparamScalar_0, \DparamVector_0)$, in this case, takes the form of a Normal distribution:
\begin{align*}
    p(\vbeta \mid \DparamScalar_0, \DparamVector_0) &\propto \exp \left\{ \DparamScalar_0 \tau \left[ \DparamVector_0 ^\top \X \vbeta - \frac{( \X \vbeta)^2}{2}\right] \right\} 
    \propto \exp \left\{ - \frac{\DparamScalar_0 \tau}{2} \left( \vbeta - \vmu \right)^\top (\X^{\top} \X) \left( \vbeta - \vmu \right) \right\} 
\end{align*}
so that $$\vbeta \sim \mathcal{N}_p \left( \vmu, \Sigma \right)$$
with 
\begin{equation}
\label{eq:moment_D_normal_prior}
    \vmu = (\X^{\top} \X)^{-1} X^{\top} \DparamVector_0 \ ,  \qquad \Sigma =  \frac{(\X^{\top} \X)^{-1} }{\DparamScalar_0 \tau} \ .
\end{equation}
Following the construction in Section \ref{apd:posterior_general}, the prior over $\vbeta \mid \z$ decomposes as
\begin{align*}
    p(\vbeta \mid \z) & = p\left(\vbeta^{(0)}, \vbeta^{(1)} \mid \z \right) =  p \left( \vbeta^{(1)} \mid \z \right) p\left(\vbeta^{(0)} \mid \z \right) \\
    &\propto \exp \left\{ - \frac{\left( \vbeta^{(1)} - \vmu^{(1)} \right)^\top {\Sigma^{(1)}}^{-1}  \left( \vbeta^{(1)} - \vmu^{(1)} \right)}{2 } \right\}  \times \\
    & \qquad \qquad \times
    \exp \left\{ - \frac{\left( \vbeta^{(0)} - \vmu^{(0)} \right)^\top {\Sigma^{(0)}}^{-1}  \left( \vbeta^{(0)} - \vmu^{(0)} \right) }{2}\right\} 
\end{align*}
where  $\vmu^{(s)}$ and $\Sigma^{(s)}$ are given by in \eqref{eq:moment_D_normal_prior}, for  $s=0,1$.
\paragraph{Posterior of $\vbeta$}
Using the prior decomposition
\[
\vbeta = (\vbeta^{(0)},\, \vbeta^{(1)}), 
\qquad 
\X(\vbeta \circ \z)=\X^{(1)}\vbeta^{(1)},
\]
only the active coefficients $\vbeta^{(1)}$ contribute to the likelihood. Consequently, the posterior also decomposes: for the inactive coefficients $\vbeta^{(0)}$ the posterior coincides with the prior, while for $\vbeta^{(1)}$ we obtain a Normal posterior with updated parameters:
\begin{align*}
    \vbeta^{(0)} \mid \X,\y,\z,\tau & \sim \ \mathcal{N} \left( \vmu^{(0)}, \Sigma^{(0)}\right) \\
    \vbeta^{(1)} \mid \X,\y,\z,\tau & \sim \ \mathcal{N} \left( \vmu_n, \Sigma_n \right)  
\end{align*}
with
$\vmu_n = \left(\X^{(1)\top} \X^{(1)} \right)^{-1} \X^{(1)\top} \displaystyle \frac{\y + \DparamScalar_0 \tau^2 \DparamVector_0 }{1 + \DparamScalar_0 \tau^2} $ , and 
$\Sigma_n =\left(\X^{(1)\top} \X^{(1)} \right)^{-1}  \frac{1}{\DparamScalar_0 \tau^2 + 1 } $.

Notice that conjugacy only requires specifying a Normal prior for $\vbeta$, 
and the prior parameters do not need to depend on $\X$ as in \eqref{eq:moment_D_normal_prior}. For instance, if we denote the prior over $\vbeta$ as $\mathcal{N}(\vmu_0, \Sigma_0)$, and define $ \Sigma = I_p \cdot \frac{1}{\tau}$, then we posterior paramters will be 
\[
{\Sigma}_n := \left( \Tilde{\X}^\top \Sigma^{-1}  \Tilde{\X} + \Sigma_0^{-1} \right)^{-1}
\]
\[
{\vmu}_n := \left[ 
\left(\y^\top \Sigma^{-1} \Tilde{\X} + \vmu_0^\top \Sigma_0^{-1}\right){\Sigma}_n 
\right]^{\top} \ = \ {\Sigma}_n^\top \left(\y^\top \Sigma^{-1} \Tilde{\X} + \vmu_0^\top \Sigma_0^{-1}\right)^\top \ . 
\]

\paragraph{Posterior of $\tau$}
For conjugacy, we adopt an Inverse-Gamma prior on $\tau$. \\
If $\tau \sim InvGamma(a,b)$, then
\[
p(\tau) = \frac{b^a}{\Gamma(a)} \tau^{-a-1} \exp\left(- \frac{b}{\tau}\right) \ . 
\] 
Combining prior and likelihood yields 
\begin{align*}
    p(\tau \mid \X, \y, \z, \vbeta ) &\propto  p( \vect y   \mid \z, \vc, \vbeta, \tau,  \x) p(\tau) \ = \\
    &= 
   \tau^{-a-1} \exp\left(- \frac{b}{\tau}\right) 
   \left[ \prodi 
   \frac{1}{\sqrt{2\pi \tau}} 
   \exp \left\{ - \frac{\left( \vect \yi - \sum_{j=1}^p \xij\betaj \zj  \right)^2}{2\tau }\right\} \right]\\
    &\propto 
    \frac{1}{(\sqrt{2\pi \tau})^n} \frac{1}{\tau^{a+1}} 
    \exp \left\{ - \frac{\sumi \left( \vect \yi - \sum_{j=1}^p \xij\betaj \zj  \right)^2}{2\tau } - \frac{b}{\tau} 
    \right\} \\
    &\propto 
    \frac{1}{\tau^{n/2 + a+1}} 
    \exp \left\{ - \frac{b + 1/2 \sumi \left( \vect \yi - \sum_{j=1}^p \xij\betaj \zj \right)^2  }{\tau } 
    \right\}
\end{align*}
Hence: 
\begin{equation*}
    \tau  \mid  \X, \y, \z, \vbeta \ \sim InvGamma\left( a_{\text{post}} , b_{\text{post}} \right) 
\end{equation*}
with 
\begin{equation*}
    a_{\text{post}} = \frac{n}{2} + a \qquad b_{\text{post}} = b + \frac{1}{2} \sumi \left( \vect \yi - \sum_{j=1}^p \xij\betaj \zj \right)^2 
\end{equation*}

\paragraph{Posterior of $z_h$}
Following Section \ref{apd:glm_general_withoutc}, the posterior distribution of $z_h$ results:
\[
    z_h \mid \z_{-h}, \tau, \beta, \X, \y \ 
    \sim  Bern\left( \frac{ P^1 p(\vbeta \mid z_h =1, \z_{-h}, \tau )  }{ P^1 p(\vbeta \mid z_h =1, \z_{-h}, \tau )  +  \displaystyle \frac{p}{\alpha} P^0 p(\vbeta \mid z_h =0, \z_{-h}, \tau ) }  \right) \ . 
\]
where $P^s$ is the likelihood of the model computed with $z_h=s$, for $s=0,1 $:
\[
P^s = \exp \left. \left\{ - \sum_{i=1}^n \frac{\left( \yi - \sum_{j=1}^p \xij\betaj \zj  \right)^2}{2\tau }\right\}\right\rvert_{z_h  = s}
\]

\subsection{Poisson Model: conjugate posterior}
\label{apd:posterior_poisson}
This section contains the detailed derivations of the posterior distributions for the Poisson model, a special case of the GLM framework discussed in the main text.

Let $\vect X$ be $n\times p$ matrix of data, and consider the GLM with link function the Log and Poisson likelihood given by
\begin{align}
    \yi \   |\  \lambda_i & \sim Poi(\cdot | \lambda_i) \qquad i=1...n \nonumber \\
    \etai = g(\lambda_i) & = \log(\lambda_i) \Rightarrow \lambda_i= \exp(\etai) \nonumber \\
    \label{eq:eta_ourmodel} \etai \ | \ \z, \x_i, \vect  \beta &= \sumj \xij\betaj \zj 
\end{align}
The joint distribution is
{\small 
\begin{align*}
    p(\z, \vc, & \vect \beta, \vect y \ | \ \vect X)  = p( \vect y  \ | \ \z, \vc, \vect \beta, \phi,  \vect X)p(\z \ | \  \vc) p(\vc \ | \  \alpha) p(\vect \beta) p( \phi) \ = \\
    &\propto 
    \left[ \prodi \exp \left\{ \log(\lambda_i)\yi - \lambda_i\right\} \right] 
    \left[ \prod_{j=1}^p {\cj}^{\zj} {(1- \cj)}^{1- \zj}\right] \left[ \prod_{j=1}^p {\cj}^{ \frac{\alpha}{p} -1 } \right]p(\vect \beta) = \\
    &\propto 
    \left[ \prodi \exp \left\{  \sumj \xij \betaj \zj \yi - \exp\left\{  \sumj \xij\betaj \zj \right\} \right\} \right] 
    \left[ \prod_{j=1}^p {\cj}^{\zj} {(1- \cj)}^{1- \zj}\right] \left[ \prod_{j=1}^p {\cj}^{ \frac{\alpha}{p} -1 } \right]p(\vect \beta) \\
\end{align*}}
In this model, $\tau = 1$ and $b(\theta) = \exp(\theta)$, so the distribution $\mathcal{D}(\DparamScalar_0, \DparamVector_0)$ takes the form
\begin{align*}
    p(\vbeta \mid \DparamScalar_0, \DparamVector_0) &\propto \exp \left\{ \DparamScalar_0 \left[ \DparamVector_0 ^\top \X \vbeta - \vect{J}_n^\top \exp( \X \vbeta) \right] \right\} 
\end{align*}
Following Section \ref{apd:glm_general_withoutc}, we derive the posterior distributions.
\paragraph{Posterior of $\vbeta$}
Taking the prior over $\vbeta \mid \z $ factorizes as
{\small 
\begin{align*}
    p(\vbeta \mid \z) & = p\left(\vbeta^{(0)}, \vbeta^{(1)} \mid \z \right) =  p \left( \vbeta^{(1)} \mid \z \right) p\left(\vbeta^{(0)} \mid \z \right) \\
    &= \exp\left\{ \DparamScalar_0  \left(\DparamVector_0^\top \X^{(1)} \vbeta^{(1)} - \vect{J}_n^\top \exp (\X^{(1)} \vbeta^{(1)}) \right) \right\} \exp\left\{ \DparamScalar_0  \left(\DparamVector_0^\top \X^{(0)} \vbeta^{(0)}- \vect{J}_n^\top \exp (\X^{(0)} \vbeta^{(0)}) \right) \right\} = \\
    &= \exp\left\{ \DparamScalar_0  \left(\DparamVector_0^\top \X \vbeta - \vect{J}_n^\top \exp \left( \X^{(1)} \vbeta^{(1)} \right) - \vect{J}_n^\top \exp \left( \X^{(0)} \vbeta^{(0)} \right) \right) \right\} \ . 
\end{align*}}
Combining prior and likelihood yields
\begin{align*}
    p(\vbeta^{(0)} \mid \z, \y, \X) 
    & \propto  \exp \left\{ \DparamScalar_0 \left(\DparamVector_0^\top \X^{(0)} \vbeta^{(0)} - \vect{J}_n^\top \exp \left(\X^{(0)} \vbeta^{(0)} \right) \right) \right\}  \\ 
    p(\vbeta^{(1)} \mid  \z, \y, \X) 
    &\propto\exp \left\{ \left( 1 + \DparamScalar_0 \right) \left[ \left( \frac{\y + \DparamScalar_0 \DparamVector_0 }{1 + \DparamScalar_0} \right)^\top \X^{(1)}\vbeta^{(1)} - \vect{J}_n^\top \exp(\X^{(1)}\vbeta^{(1)}) \right] \right\} \ .
\end{align*}
Therefore
\[
\vbeta^{(0)} \mid \z^{(0)}, \y, \DparamScalar_0, \DparamVector_0, \ \sim \ \mathcal{D} \left( \DparamScalar_0, \DparamVector_0 \right)
\]
\[
\vbeta^{(1)} \mid \z^{(1)}, \y, \DparamScalar_0, \DparamVector_0, \ \sim \ \mathcal{D} \left( \hat{\DparamScalar}_0, \hat{\DparamVector}_0\right)
\]
with posterior parameters
$ \hat{\DparamScalar}_0 = 1 + \DparamScalar_0$, and   $\hat{\DparamVector}_0 = \frac{\y + \DparamScalar_0 \DparamVector_0 }{1 + \DparamScalar_0 } $.    

\paragraph{Posterior of $z_h$}
\[
    z_h \mid \z_{-h}, \beta, \X, \y \ 
    \sim  Bern\left( \frac{ P^1 p(\vbeta \mid z_h =1, \z_{-h} )  }{ P^1 p(\vbeta \mid z_h =1, \z_{-h} )  +  \displaystyle \frac{p}{\alpha} P^0 p(\vbeta \mid z_h =0, \z_{-h}, ) }  \right) \ . 
\]
where $P^s$ is the likelihood of the model computed with $z_h=s$, for $s=0,1 $:
\[
P^s = \exp \left. \left\{ \sum_{i=1}^n 
    \left[ \etai  \yi - \exp\left(  \etai  \right)  \right]
    \right\}\right\rvert_{z_h  = s}
\]

\section{Proof of Theorem \ref{thm:consistency_z}}
\label{apd:proofs}
\setcounter{subsection}{-1}

This section is devoted to the proof of the main result of the paper, Theorem \ref{thm:consistency_z}. 
The proofs are a generalization of the techniques presented in \citep{narisetty2019Skinny} and \citep{lee2021bayesian} 
to Generalized Linear Models with the proposed conjugate priors and suitable regularity assumptions.

While the notation used in the main text is efficient to illustrate the model, methodology, and algorithms, we adopt a slightly different notation for the proof below. Our choice is motivated by the fact that this new notation keeps the proof more clean and the arguments easier to follow.
We stress that the underlying model and all the resulting considerations remain fully consistent with those of the main text. 

\paragraph{Preliminaries.}
Before we establish our main results, the following notations are needed.
For any $a, \ b \in \mathbb{R}, \ a \vee b$ and $a \wedge b$ denote the maximum and minimum of $a$ and $b$, respectively.  We will use $C, \tilde{C}, C^\prime, C^{\prime\prime}, C^*$ as generic constants  that can take different values depending on the context. For any sequences $a_n, \ b_n$, we will denote $a_n= O(b_n)$ if there exists a constant $C>0$ such that $a_n = C b_n$ as $n \rightarrow \infty$; $a_n= o(b_n)$ if $a_n/b_n \rightarrow 0$ as $n \rightarrow \infty$; and $a_n \approx b_n$ if $a_n/b_n \rightarrow 1$ as $n \rightarrow \infty$.
For the sake of simplicity, we suppress the bold notation for vectors and matrices.

Given a sample size $n$, we denote by $p_n$ the total number of covariates, Condition \ref{cond:4:d} gives the rate at which it can grow and this can be faster than $O(n)$. $X \in \mathbb{R}^{n \times p_n}$ denotes our data matrix. However, only models with at most $m_n$ covariates are considered and $m_n$ grows at a rate $O(n^\delta)$ with $0 < \delta < 1$, see Condition \ref{cond:4:e} for the precise rate. 

In the following we will use $k$ to denote a generic model and $t$ to denote the true model. Depending on the context, a model is treated both as the set containing only the indices of active covariates or as a $p_n \times 1$ binary vector (in this sense $k$ and $t$ correspond respectively  to a generic $\z$ and $\z^*$ in the main text). The size of a model $k$ is $\lvert k \rvert$. For any vector $v \in \mathbb{R}^{p_n \times 1}$, $v(k)$ denotes the $\lvert k \rvert \times 1$ vector containing the components of $v$ corresponding to  model $k$. In a similar way, for a given matrix $X \in \mathbb{R}^{n \times p_n}$, $X_k \in \mathbb{R}^{n \times \lvert k \rvert}$ denote the sub-matrix of $X$ containing the columns indexed by $k$ and $H_n(k)$ the Hessian matrix.

We denote as $\beta_0(t)$  the true regression vector (defined in Condition \ref{cond:3} below) which corresponds to $\vbeta^{*(1)}$ in the main text. For any model $k \supset t$, $\beta_0(k)$ indicates the $\lvert k \rvert \times 1$ vector including $\beta_0(t)$ for $t$ and zeroes for $k \cap t^C$. For any model $k$, $\hat{\beta}(k)$ denotes the maximum  likelihood estimator (MLE) of ${\beta}(k)$.

\paragraph{Assumptions.}
Consider the following conditions
\begin{enumerate}[label=\textbf{C\arabic*}]
    \item \textit{On the regularity of the family.} 
    \label{cond:1}
    \begin{enumerate}[label=(\alph*), ref=\textbf{\theenumi(\alph*)}]
        \item \label{cond:1:a} The exponential family is log-concave;
        \item \label{cond:1:b} The exponential family is $\varpi$-subExponential with $\varpi \in (\frac{1}{2}, 2]$;
        \item \label{cond:1:c} The maximum likelihood estimates exists, it is unique and it is a continuous function of the vector observations $y$. 
    \end{enumerate}  
    \item \textit{On the regularity of the design.} 
    \label{cond:2}
    \begin{enumerate}[label=(\alph*), ref=\textbf{\theenumi(\alph*)}]
        \item \label{cond:2:a} The predictors are bounded,  \ie $\max\{ \vert x_{ij} \vert, 1 \le i \le n, 1 \le j \le p \} \le C$, for some $0 < C < \infty$;
        \item \label{cond:2:b} For any $m \le p_n$
        \begin{align*}
          0 < \lambda & \le \min_{k: \vert k \vert \le m} \lambda_{\min}\left( n^{-1} H_n(\beta_0(k)) \right) \\
          & \le \max_{k: \vert k \vert \le m} \lambda_{\max}\left( n^{-1} X_k^\top X_k \right) = \Lambda_m < 4\lambda
        \end{align*}
        where, $\lambda_{\min}$ and $\lambda_{\max}$ are the minimum and the maximum eigenvalues of their arguments;
        \item \label{cond:2:c} There exists $K_1 > 0$ such that $\lvert  X_k^\top (\xi_0 - \mu(\beta_0)) \rvert \le K_1 n$;
        \item \label{cond:2:d} There exists $K_2 > 0$ such that $\lvert  \xi_0^\top X_k \beta_0 - J^\top b(X_k \beta_0)  \rvert \le K_2 n$.
    \end{enumerate}
    \item \textit{On the true model and signal strength}. 
    \label{cond:3} We assume that there exists constant $C >0$ such that 
    \begin{equation*}  
        \min_{1 \le j \le \vert t \vert} \vert \beta_{0j}(t) \vert \ge 2C\sqrt{\frac{2\lvert t \rvert \Lambda_{2\lvert t \rvert} (\log p_n)^\omega}{n^\gamma}}
    \end{equation*}
    where $\beta_0(t) = \{ \beta_{0j}(t) \}_{j=1}^{\vert t \vert}$ is the nonzero coefficients of $\beta$ under the true model and $\omega, \gamma$ are defined in \ref{cond:4};
    \item \textit{On the rate of growing of the dimension of the models.} \label{cond:4}
    For $\varpi \in \left( \frac{1}{2}, 2 \right]$
    \begin{enumerate}[label=(\alph*), ref=\textbf{\theenumi(\alph*)}]
        \item \label{cond:4:a} $\omega = 2/\varpi$;
        \item \label{cond:4:b} $0 < \gamma < \frac{2 \varpi - 1}{\varpi}$;
        \item \label{cond:4:c} $0 < \phi$ such that $\phi \omega < \gamma < \phi \omega + 1$ or equivalently $\max\left\{ \frac{\gamma - 1}{\omega}, 0 \right\} < \phi < \frac{\gamma}{\omega}$;
        \item \label{cond:4:d} $\log p_n = O(n^\phi)$;
        \item \label{cond:4:e} $m_n = O\left( \left( \frac{n^\gamma}{(\log p_n)^\omega} \right)^\delta \wedge (\log p_n)^{\delta^\prime}   \right) \wedge  p_n$ for some $0 < \delta < \min \left\{ \frac{2 - \frac{1}{\varpi} - \gamma}{(\gamma - \phi \omega)(\frac{3}{\varpi} - 1)}, 1 \right\}$ and $0 < \delta^\prime < 1$.
    \end{enumerate}
\end{enumerate}

\begin{remark}
$ $
\begin{itemize}  
\item Condition \ref{cond:1:a} is true for many members in the exponential family, such as those with canonical links \citep{diaconis1979conjugate}, and for many with noncanonical links \citep{Wedderburn1976}.
\item Condition \ref{cond:1:b},
which is needed in Theorem \ref{thm:sambale2023} and in Lemma \ref{lemma:bound_beta},
is verified for several members of the exponential family, such as Gaussian, Bernoulli, Poisson, Gamma, and Inverse-Gaussian. See Lemma \ref{lemma:sub-gauss}, Corollary \ref{corollary:sub-gauss} and Remark \ref{rem:locally-lipschitz}.
\item Conditions under which \ref{cond:1:c} holds are discussed, \eg in \citep{Wedderburn1976} for some members of the exponential family regarding existence and uniqueness, while continuity can be achieved under Berge's maximum theorem, see for instance \citet[Theorem 17.31]{AliprantisBorder2006}.  
\end{itemize}  
\end{remark}

\paragraph{Model subsets.}
We divide the set of candidate models of cardinality at most $m_n$ into three subsets
\begin{enumerate}[label=$\mathbf{M_{\arabic*}}$]
\item \textit{Over-fitted models}: $M_1 = \{ k: k \supsetneq t, \vert k \vert \le m_n \} $, 
\ie the models of dimension not greater than $m_n$ which include all the active covariates plus one or more inactive covariates.
\item \textit{Large models}: $M_2 = \{ k: k \nsupseteq t, \lvert t \rvert  < \vert k \vert \le m_n \}$,
\ie the models which miss at least one active covariate but are larger than the true model. 
\item \textit{Under-fitted models}: $M_3 = \{ k: k \nsupseteq t, \vert k \vert \le \lvert t \rvert  \}$,
\ie the models of moderate dimension which miss at least one active covariate.
\end{enumerate}

What we need to prove is, in main manuscript notation, that
\begin{equation*}
P(\z = \z^* \mid \X, \y, \tau) \longrightarrow 1 \quad \text{as } n \to \infty .
\end{equation*}
which is equivalent in this proof notation to $P(Z = t \mid y, X) \longrightarrow 1$.

To this aim, we define 
\begin{equation*}
PR(k,t) = \frac{P(Z = k \mid y, X)}{P(Z = t \mid y, X)} \ .
\end{equation*}
and we will prove that
\begin{equation*}
\sum_{k \in M_u} PR(k,t) \longrightarrow 0 \text{ as } n \rightarrow + \infty
\end{equation*}
for $u=1$ in Section \ref{apd:subsec:m1}, and then for $u=2, 3$ in Sections \ref{apd:subsec:m2}-\ref{apd:subsec:m3}. 

\subsection{Preliminary results}
We now present and prove some preliminary results, which will be useful for the rest of the proof.

The following result provides the joint posterior that corresponds to our Gibbs sampler.
This is similar to \cite[Theorem 1]{narisetty2014spike}.
\begin{lemma}
\label{lemma:jointgibbs}
The joint posterior of $\beta$, $Z$ and $\tau$ corresponding to the introduced Gibbs sampler in Equations \eqref{eq:post_glm_beta0}--\eqref{eq:post_glm_tau} is given by
\begin{align}
P(\beta, Z, \tau \vert y, X, a_0, \xi_0) & = P(\beta^{(0)} \vert Z, \tau, y, X, a_0, \xi_0) \label{jointposterior} \\
& \times P(\beta^{(1)} \vert Z, \tau, y, X, a_0, \xi_0) \nonumber \\
& \times \left( \frac{p_n}{\alpha} \right)^{\sum_{j=1}^{p_n} (1 - Z_j)}  \nonumber \\
& \times P(\tau \vert y, X, a_0, \xi_0) \ . \nonumber
\end{align}
\end{lemma}

\begin{proof}
The conditional distribution of $\beta$ under
\eqref{jointposterior} is given by
\begin{equation*}
P(\beta \vert Z, \tau, y, X, a_0, \xi_0) = P(\beta^{(0)} \vert Z, \tau, y, X, a_0, \xi_0) P(\beta^{(1)} \vert Z, \tau, y, X, a_0, \xi_0)
\end{equation*}
which coincides with the one reported in our algorithm. Now, the conditional distribution of $Z$ under \eqref{jointposterior} is given by
\begin{align} \label{jointposteriorzeta}
P(Z \vert \beta, \tau, y, X, a_0, \xi_0) & = P(\beta \vert Z, \tau, y, X, a_0, \xi_0) \left( \frac{p_n}{\alpha} \right)^{\sum_{j=1}^{p_n} (1 - Z_j)} \\
& = P(\beta^{(0)} \vert Z, \tau, y, X, a_0, \xi_0) P(\beta^{(1)} \vert Z, \tau, y, X, a_0, \xi_0) \left( \frac{p_n}{\alpha} \right)^{\sum_{j=1}^{p_n} (1 - Z_j)} \ . \nonumber
\end{align}
Furthermore, the conditional of each $Z_j$ based on the above expression can be derived as
\begin{equation*}
R := \frac{P(Z_j=1 \vert Z_{-j}=u, \beta, \tau, y, X, a_0, \xi_0)}{P(Z_j=0 \vert Z_{-j}=u, \beta, \tau, y, X, a_0, \xi_0)} = \frac{P(Z_j=1, Z_{-j}=u \vert \beta, \tau, y, X, a_0, \xi_0)}{P(Z_j=0, Z_{-j}=u \vert \beta, \tau, y, X, a_0, \xi_0)} \ .
\end{equation*}
The above ratio is comparing model $u_1 = (Z_j=1, Z_{-j}=u)$ with model $u_0 = (Z_j=0, Z_{-j}=u)$. Then, following \eqref{jointposteriorzeta}, we obtain
\begin{align*}
R & = \frac{P(\beta \vert Z_j=1, Z_{-j}=u, \tau, y, X, a_0, \xi_0)}{P(\beta \vert Z_j=0, Z_{-j}=u, \tau, y, X, a_0, \xi_0) \frac{p_n}{\alpha}} \\
& = \frac{P(\beta^{(0)} \vert Z=u_1, \tau, y, X, a_0, \xi_0) P(\beta^{(1)} \vert Z=u_1, \tau, y, X, a_0, \xi_0)}{P(\beta^{(0)} \vert Z=u_0, \tau, y, X, a_0, \xi_0) P(\beta^{(1)} \vert Z=u_0, \tau, y, X, a_0, \xi_0) \frac{p_n}{\alpha}} \\
& = \frac{P^1 P(\beta \vert Z=u_1, a_0, \xi_0)}{P^0 P(\beta \vert Z=u_0, a_0, \xi_0)  \frac{p_n}{\alpha}}
\end{align*}
with $P^s$ 
being the likelihood evaluated when $Z_j = s, \ s \in \{0,1\}$.
It is clear that the posterior of $\tau$ are the same. This concludes the proof.
\end{proof}  

We recall here some basics definitions we will need in the following. 
\begin{definition}[Lipschitz function]
A function $f(x)$ is locally Lipschitz at $x_0$ if there exists $0 < \delta(x_0)$ and $0 < L(x_0) < \infty$ such that for all $x \in \{x: \vert x - x_0 \vert \le \delta(x_0) \}$ we have 
\begin{equation*}
\vert f(x) - f(x_0) \vert \le L(x_0) \vert x - x_0 \vert \ .
\end{equation*}
\end{definition}

\begin{definition}[SubGaussian and SubExponential distribution]
A random variable $X$ is $\varpi$-subGaussian if:
\begin{equation*}
\mathbb{E}\left[ \exp(\lambda (X - \mathbb{E}[X])) \right] \le  \exp \left( \frac{\varpi^2 \lambda^2}{2}\right) \quad \forall \lambda
\end{equation*}
While, $X$ is $\varpi$-subExponential if:
\begin{equation*}
\mathbb{E}\left[ \exp(\lambda (X - \mathbb{E}[X])) \right] \le  \exp \left(\varpi^2 \lambda^2\right) \quad \forall -1/\varpi \le \lambda \le 1/\varpi
\end{equation*}
\end{definition}
The following lemma gives a sufficient condition for the condition \ref{cond:1:b} to hold.

\begin{lemma}[SubGaussianity of certain members of  the exponential family] \label{lemma:sub-gauss}
Consider the exponential family
\begin{equation*}
\exp \left\{ \frac{y \theta - b(\theta)}{A(\tau)} + c(y, \tau) \right\}
\end{equation*}
and assume the function $b^\prime(\cdot)$ is $L$-Lipschitz. Then, the exponential family is subGaussian with parameter $\varpi = \sqrt{2 L}$.
\end{lemma}
\begin{proof}
We provide the proof in the univariate exponential family in its natural observation/natural parametrization. For each $\lambda \in \mathbb{R}$ the moment generating function is given by
\begin{align*}
M_\theta(\lambda) & = \int \exp(\lambda y) \exp \left\{ \frac{y \theta - b(\theta)}{A(\tau)} + c(y, \tau) \right\} d y \\
& = \exp(b(\lambda + \theta) - b(\theta)) \ .
\end{align*}  
We know that $\mathbb{E}_\theta(y) = b^\prime(\theta)$ hence, for any real $\lambda$,
\begin{equation*}
  \mathbb{E}_\theta(\exp(\lambda (y - \mathbb{E}_\theta(y)))) = \exp(b(\lambda + \theta) - b(\theta) - \lambda b^\prime(\theta)) \ ,
\end{equation*}
and by the mean-value theorem, for some $a \in (0,1)$ depending on $\lambda$ and $\theta$, we have
\begin{equation*}
  b(\lambda + \theta) - b(\theta) = \lambda b^\prime(a \lambda + \theta)
\end{equation*}
hence by the Lipschitz condition on $ b^\prime(\cdot)$,
\begin{equation*}
  \vert b^\prime(a \lambda + \theta) -  b^\prime(\theta) \vert \le L a \lambda \le L \lambda \ .
\end{equation*}
Thus,
\begin{equation*}
b(\lambda + \theta) - b(\theta) - \lambda b^\prime(\theta) = \lambda \left( b^\prime(a \lambda + \theta) - b^\prime(\theta) \right) \vert \le L \lambda^2
\end{equation*}
and hence
\begin{equation*}
  \mathbb{E}_\theta(\exp(\lambda (y - \mathbb{E}_\theta(y)))) \le \exp(L \lambda^2)
\end{equation*}
 \ie, the random variable $y$ in the exponential family with $b(\cdot)$ Lipschitz function, is (uniformly) subGaussian for all $\theta$.
\end{proof}  

\begin{corollary} \label{corollary:sub-gauss}
Let $y = (y_1, \ldots, y_n)$ be independent random variables with exponential family distribution with natural parameters $\theta_1, \ldots, \theta_n$ and assume $b(\cdot)$ is $L$-Lipschitz function. Let $\mu_i = b^\prime(\theta_i)$ and $\lambda = (\lambda_1, \ldots, \lambda_n) \in \mathbb{R}^n$ then
\begin{equation*}
\mathbb{E}(\exp(\lambda^\top (y - \mu))) \le \exp(L \lVert \lambda \rVert^2)
\end{equation*}  
\end{corollary}  

\begin{proof}
Since the $y_i$s are independent are recall that $\mu_i = b^\prime(\theta_i)$ then
\begin{equation*}
\mathbb{E}(\exp(\lambda^\top (y - \mu))) = \prod_{i=1}^n \mathbb{E}(\exp(\lambda_i (y_i - \mu_i))) \le \exp(L \sum_{i=1}^n \lambda_i^2) = \exp(L \lVert \lambda \rVert^2)
\end{equation*}
where the inequality holds for Lemma \ref{lemma:sub-gauss}.
 
\end{proof}  

\begin{remark} \label{rem:locally-lipschitz}
In a similar way, it can be proved that, if the function $b^\prime(\cdot)$ is \textit{locally} Lipschitz, the exponential family is $\varpi$-subExponential for some $\varpi$,
\eg Poisson is a $1$-subExponential distribution.
The theory we developed will work for the members of the exponential family which are $\varpi$-subExponential with $\varpi \in \left(\frac{1}{2}, 2 \right]$ as requested in condition \ref{cond:1:b}.
\end{remark}

In the next Lemma we provide a bound regarding the Hessian matrix.
\begin{lemma} 
\label{lemma:hessian}
Let $c > 0$ be any fixed constant and conditions \ref{cond:2}--\ref{cond:3} holds. For a model $k \in M_1$, for all $\beta(k)$, such that $\lVert \beta(k) - \beta_0(k) \rVert \le \sqrt{c \vert k \vert \Lambda_{\lvert k \rvert} \log(p_n)/n}$ and $\tau(k)$ such that $\vert \tau(k) - \tau_0(k) \vert = O(n^{-1/2})$ assume that
\begin{equation} \label{cond:H_n}
  \bar{H}_n(k) = \frac{A(\tau_0(k)) v(\mu_0(k)) (g^\prime(\mu_0(k)))^2}{A(\tau(k)) v(\mu(k)) (g^\prime(\mu(k)))^2} \to 1
\end{equation}
as $n \to \infty$ where $\mu(k) = g^{-1}(x^\top \beta(k))$ and similarly for $\mu_0(k)$; $v(\mu(k))$ is the variance function given by the model. There exists $\varepsilon_n \to 0$ such that  
\begin{equation*}
(1 - \varepsilon_n) H_n(\beta_0(k)) \le  H_n(\beta(k)) \le (1 + \varepsilon_n) H_n(\beta_0(k)) \ .  
\end{equation*}
\end{lemma}  

\begin{proof}
Recall that $H_n(\beta(k)) = X_k^\top \Sigma(\beta(k)) X_k$ where $\Sigma(\beta(k)) = \operatorname{diag}(\sigma_1^2(\beta(k)), \ldots, \sigma_n^2(\beta(k)))$ and $\sigma_i^2(\beta(k)) = A(\tau(k)) v(\mu(k)) (g^\prime(\mu(k)))^2$ where we suppress the index $i$ on the left hand side, therefore, to prove the lemma, it is sufficient to show that
\begin{equation*}
(1 - \varepsilon_n) \sigma_i^2(\beta_0(k)) \le \sigma_i^2(\beta(k)) \le (1 + \varepsilon_n) \sigma_i^2(\beta_0(k)) \qquad 1 \le i \le n \ .  
\end{equation*}
Since the assumption above and by the definition of $\sigma_i^2(\beta_0(k))$ lemma holds.

\end{proof}

\begin{remark}
We are now going to check that condition in Equation \eqref{cond:H_n} holds for some important exponential-family models using the natural link function.
\begin{itemize}
\item \textit{Poisson}. In this case $A(\tau) = 1$ and hence Equation \ref{cond:H_n} becomes
\begin{equation*}
\bar{H}_n(k) = \frac{v(\mu_0(k)) (g^\prime(\mu_0(k)))^2}{v(\mu(k)) (g^\prime(\mu(k)))^2} = \frac{\mu_0(k)}{\mu(k)} \frac{\mu^2(k)}{\mu_0^2(k)} = \frac{\mu(k)}{\mu_0(k)} = \exp(x_i^\top (\beta(k) - \beta_0(k))) \to 1.
\end{equation*}  
as $n \to \infty$.
\item \textit{Bernoulli}. In this case $A(\tau) = 1$ and hence Equation \ref{cond:H_n} becomes
\begin{align*}
  \bar{H}_n(k) & = \frac{v(\mu_0(k)) (g^\prime(\mu_0(k)))^2}{v(\mu(k)) (g^\prime(\mu(k)))^2} = \frac{n \mu_0(k) (1 - \mu_0(k))}{n \mu(k) (1 - \mu(k))} \frac{(\mu(k) (1 - \mu(k)))^2}{(\mu_0(k) (1 - \mu_0(k)))^2} \\
  & = \frac{\mu(k) (1 - \mu(k))}{\mu_0(k) (1 - \mu_0(k))} \\
  & = \exp(\eta(k) - \eta_0(k)) \left( \frac{1 + \exp(\eta_0(k))}{1 + \exp(\eta(k))} \right)^2 \ .
\end{align*}
Recall that $(1 + \exp(a))/(1 + \exp(b)) \le \exp(\vert a - b \vert)$ then
\begin{equation*}
\bar{H}_n(k) \le \exp(3 \vert \eta(k) - \eta_0(k) \vert) = \exp(3 \vert x_i^\top (\beta(k) - \beta_0(k)) \vert) \to 1 \qquad n \to \infty \ . 
\end{equation*}

\item \textit{Gamma}. Here, $A(\tau) = \tau^{-1}$ hence $A(\tau_0(k))/A(\tau(k)) = \tau(k)/\tau_0(k) \to 1$ as $n \to \infty$ under the assumption of the lemma.
Furthermore,
\begin{align*}
  \bar{H}_n(k) & = \frac{v(\mu_0(k)) (g^\prime(\mu_0(k)))^2}{v(\mu(k)) (g^\prime(\mu(k)))^2} = \frac{\mu^2(k)}{\mu_0^2(k)} \\
  & = \frac{\eta_0^2(k)}{\eta^2(k)} = \left( \frac{x_i^\top \beta_0(k)}{x_i^\top \beta(k)} \right)^2 \to 1 \qquad n \to \infty \ .
\end{align*}

\item \textit{Inverse-Gaussian}. For an $IG(\mu, \lambda)$ and choosing the parametrization where $\theta=-1/(2\mu^2)$ and $\tau = \lambda$, it holds that $A(\tau) = \tau^{-1}$ hence $A(\tau_0(k))/A(\tau(k)) = \tau(k)/\tau_0(k) \to 1$ as $n \to \infty$ under the assumption of the Lemma. Furthermore, being $\mu = g^{-1}(\eta) = \frac{1}{\sqrt{ - 2\eta}}$, we have
\begin{align*}
  \bar{H}_n(k) & = \frac{v(\mu_0(k)) (g^\prime(\mu_0(k)))^2}{v(\mu(k)) (g^\prime(\mu(k)))^2} 
   = \frac{\mu_0^3(k)}{\mu^3(k)} \frac{(1/\mu_0^3(k))^2}{(1/\mu^3(k))^2} \\
  & = \frac{\mu_0^3(k)}{\mu^3(k)} \frac{\mu^6(k)}{\mu_0^6(k)} \  =  \ \frac{\mu^3(k)}{\mu_0^3(k)} \\
  & =   \left(\frac{\eta(k)}{\eta_0(k)}\right)^{3/2} = \left( \frac{x_i^\top \beta(k)}{x_i^\top \beta_0(k)} \right)^{3/2} \to 1 \qquad n \to \infty \ .
\end{align*}
\end{itemize}
\end{remark}  

Next, we recall an extension of the Hanson-Wright Inequality \citep[see  \eg][Theorem 1.1]{rudelson2013hansonwrightinequality} which holds for $\varpi$-subExponential distributions.
\begin{theorem}[{\citet[Theorem 2.1]{sambale2023}}]
\label{thm:sambale2023}
For any $\varpi \in (0,2]$, let $Y = (Y_1, \ldots, Y_n)$ is a random vector with independent components such that $\mathbb{E}[Y_i] = 0$, and $\lVert Y_i \rVert_{\psi_\varpi} \le K$ for any $i$. Let $A$ be a symmetric matrix. Then, for any $t \ge 0$, 
\begin{equation*}
P\left( \lvert  Y^\top A Y - \mathbb{E}[Y^\top A Y]  \rvert  > t \right) \le 2\exp\left\{ - \frac{1}{C_\varpi} \min \left( \frac{t^2}{K^4 \lVert A \rVert_{F}^2}, \ \left(\frac{t}{K^2 \lVert A \rVert_{op}} \right)^{\frac{\varpi}{2}} \right) \right\}
\end{equation*}
where $C_\varpi$ is a constant that depends on $\varpi$ and $\lVert \cdot \rVert_{F}$ and $\lVert \cdot \rVert_{op}$ denote the Frobenius and the operator norms, respectively. 
\end{theorem}

\begin{remark}
    Recall that the Frobenius and the operator norms are defined as follows:
    \begin{equation}
    \label{eq:frob_op_norm}
        \lVert A \rVert_{F} = \sqrt{\sum_{i,j} \lvert a_{ij} \rvert^2}
    \qquad 
    \lVert A \rVert_{op} = \max_{ \lVert x \rVert_2 \le 1} \lVert Ax \rVert_2
    \end{equation}
    Moreover, in general it holds that: 
    \begin{equation}
    \label{eq:frob_greater_op}
        \lVert A \rVert_{F} \ge \lVert A \rVert_{op}
    \end{equation}
\end{remark}

\begin{remark}
    Notice that the assumption  $\lVert Y_i\rVert_{\psi_\varpi} \le K$ of Theorem \ref{thm:sambale2023} is equivalent to asking that the vector $Y$ is $\varpi$-subExponential and there exists a constant $c$ such that $\varpi= c K$. 
\end{remark}

Recall that $\Sigma(\beta(k)) = \operatorname{diag}(\sigma_1^2(\beta(k)), \ldots, \sigma_n^2(\beta(k)))$ and $\sigma_i^2(\beta(k)) = A(\tau(k)) v(\mu(k)) (g^\prime(\mu(k)))^2$, introduced in the Proof of Lemma \ref{lemma:hessian}. We specialize Theorem \ref{thm:sambale2023} to our case as follows. 
\begin{lemma} 
\label{lemma:concentration}
Let $Y = (Y_1, \ldots, Y_n)$ be independent random variables with exponential family distribution with expectation vector $\mu = (\mu_1, \ldots, \mu_n)$. Suppose that $Y$ is $\varpi$-subExponential with parameter $\varpi \in (0, 2]$ and the corresponding $K$ as in Theorem \ref{thm:sambale2023}. Assume there exists a $\sigma^2$ such that
\begin{equation*}
\Sigma(\beta(k)) \le \sigma^2 I \ . 
\end{equation*}
for all $n$. Then for any $h \ge 0$, we have: 
\begin{equation*}
P\left( (Y - \mu)^\top X_k^\top X_k (Y - \mu) \ge h \right) \le \exp\left\{ - \frac{1}{C_\varpi} \left( \frac{h - \sigma^2 m n \lambda}{K^2 \sqrt{ m n \Lambda_m}} \right)^{\frac{\varpi}{2}} \right\}
\end{equation*}
where $\lambda, \ \Lambda_m$ are defined in Condition \ref{cond:2:b}.
\end{lemma}
\begin{proof}
For the sake of simplicity, let's denote $A = X_k^\top X_k$, and $Z = Y - \mu$. It follows that: 
\begin{align*}
\mathbb{E}[Z^\top A Z] & = \operatorname{tr}\left( \mathbb{E}[Z^\top A Z]  \right) = \mathbb{E}[ \operatorname{tr}\left( Z^\top A Z  \right) ] = \mathbb{E}[ \operatorname{tr}\left( A Z Z^\top \right) ] = \operatorname{tr}\left( A \mathbb{E}[ Z Z^\top ]  \right)  \\
& = \operatorname{tr}\left( A \Sigma(\beta(k))  \right) \\ 
& \le \sigma^2 \operatorname{tr}\left( A \right)
\end{align*}
Also we use the fact that
\begin{equation*}
[X_k^\top \Sigma(\beta(k)) X_k]_i = \sum_{j=1}^n x_{ji}^2 \sigma_j^2(\beta(k)) \le \sigma^2 \sum_{j=1}^n x_{ji}^2 = \ \sigma^2  [X_k^\top X_k]_i \ .
\end{equation*}  
for all $i=1, \ldots, n$. Therefore we have
\begin{enumerate}[(a)]
\item $\operatorname{tr}(A) = \sum_{i=1}^{\lvert k \rvert} \lambda(A) = \sum_{i=1}^{\lvert k \rvert} \lambda(X_k^\top X_k) \le \lvert k \rvert \lambda_{\max}(X_k^\top X_k) \le m \lambda_{\max}(X_k^\top X_k)$ 
\item $\mathbb{E}[Z^\top A Z] \le \sigma^2 \operatorname{tr}(A)$ 
\item $\lVert A \rVert_F \le \sqrt{\sum_i \lambda_i(X_k^\top X_k)} \le \sqrt{\lvert k \rvert \lambda_{\max}(X_k^\top X_k)} \le \sqrt{m \lambda_{\max}(X_k^\top X_k)}$
\end{enumerate}
and there exists a constant $K \propto \varpi$ such that for any $t \ge 0$ we have
\begin{align*}
P\left( Z^\top A Z \ge \sigma^2 m \lambda_{\max}(X_k^\top X_k) + t \right) & \overset{(a)}{\le} P\left( Z^\top A Z \ge \sigma^2 \operatorname{tr}(A) + t \right) \\
& \overset{(b)}{\le} P\left( Z^\top A Z \ge \mathbb{E}[Z^\top A Z] + t \right) \\
& = P\left( Z^\top A Z - \mathbb{E}[Z^\top A Z] \ge t \right) \\
& \overset{\text{Thm \ref{thm:sambale2023}}}{\le} \exp\left\{ - \frac{1}{C_\varpi} \min \left( \frac{t^2}{K^4 \lVert A \rVert_{F}^2}, \ \left(\frac{t}{K^2 \lVert A \rVert_{op}} \right)^{\frac{\varpi}{2}} \right) \right\} \\
& \overset{Eq. \eqref{eq:frob_greater_op}}{\le}\exp\left\{ -\frac{1}{C_\varpi} \min \left( \frac{t^2}{K^4 \lVert A \rVert_{F}^2}, \ \left( \frac{t}{K^2 \lVert A \rVert_{F}} \right)^{\frac{\varpi}{2}} \right) \right\} \\
& \le \exp\left\{ - \frac{1}{C_\varpi} \left( \frac{t}{K^2 \lVert A \rVert_{F}} \right)^{\frac{\varpi}{2}} \right\} \\
& \overset{(c)}{\le} \exp\left\{ -\frac{1}{C_\varpi} \left( \frac{t}{K^2 \sqrt{m \lambda_{\max}(X_k^\top X_k)}} \right)^{\frac{\varpi}{2}} \right\} 
\end{align*}
Hence
\begin{equation*}
P\left( Z^\top A Z \ge \sigma^2 m \lambda_{\max}(X_k^\top X_k) + t \right) \le \exp\left\{ -\frac{1}{C_\varpi} \left( \frac{t}{K^2 \sqrt{m\lambda_{\max}(X_k^\top X_k)}} \right)^{\frac{\varpi}{2}} \right\}
\end{equation*}  
Taking $h = \sigma^2 m \lambda_{\max}(X_k^\top X_k) + t $:
\begin{align*}
P\left( Z^\top A Z \ge h\right) & \le \exp\left\{ -\frac{1}{C_\varpi} \left( \frac{ h - \sigma^2 m \lambda_{\max}(X_k^\top X_k)}{K^2 \sqrt{m\lambda_{\max}(X_k^\top X_k)}} \right)^{\frac{\varpi}{2}} \right\}  \\
& \le \sup_{k \in M_1 : \lvert k \rvert=m} \exp\left\{ -\frac{1}{C_\varpi} \left( \frac{ h - \sigma^2 m \lambda_{\max}(X_k^\top X_k)  }{K^2 \sqrt{m \lambda_{\max}(X_k^\top X_k)}} \right)^{\frac{\varpi}{2}} \right\} \\
& = \exp\left\{ - \frac{1}{C_\varpi} \left( \frac{ h - \sigma^2 m n \lambda  }{K^2 \sqrt{m n \Lambda_m}} \right)^{\frac{\varpi}{2}} \right\} 
\end{align*}
where $\lambda, \Lambda_m$ are defined in Condition \ref{cond:2:b}.
\end{proof}

\subsection*{Likelihood behavior}
We now study the behavior of the maximum likelihood estimator when both data and prior information are used. Given the prior's hyperparameters $\xi_0$ and $a_0 = a_0(n)$ such that $a_0(n) \to 0$ as $n\to \infty$ we define 
\begin{equation*}
  \xi(a_0) = \frac{y + A(\tau) \tau a_0 \xi_0}{1 + A(\tau) \tau a_0} \quad \text{ and } \quad  a(a_0) = \frac{1 + A(\tau) \tau a_0}{A(\tau)} \ .
\end{equation*}  
We can think of $(\xi(a_0),a(a_0))$ as pseudo observations that combine the information from the prior and the likelihood in a new likelihood with the same shape since conjugacy. 

We define 
\begin{equation}
\label{eq:loglik_xi_a}
\ell_n(\beta; \xi(a_0), a(a_0)) = a(a_0) (\xi(a_0)^\top X \beta - J^\top b(X \beta))
\end{equation}
to be the log-likelihood,  \ie for each $(\beta, a_0) \in \mathbb{R}^{p}\times \mathbb{R}^{+} \cup \{0\}$ we have $(\beta, a_0) \mapsto \ell_n(\beta; \xi(a_0), a(a_0))$. As the sequence $a_0(n) \to 0$, for $n \to +\infty$, the likelihood in Eq. \eqref{eq:loglik_xi_a} will converge to the classical one:
\begin{equation}
\label{eq:loglik_classic}
\ell_n(\beta; y) = \frac{y^\top X \beta - J^\top b(X \beta)}{A(\tau)} \ .
\end{equation}
Let $\hat{\beta}(k; a_0)$ denote the maximum likelihood estimates of Eq. \eqref{eq:loglik_xi_a}, based on $(\xi(a_0), \ a(a_0))$ for model $k$, and let $\hat{\beta}(k)$ denote the MLE of Eq. \eqref{eq:loglik_classic}. Clearly $\hat{\beta}(k) = \hat{\beta}(k; 0)$ which is based on $(y, 1/A(\tau))$.

By Taylor expansion (with Lagrange remainder) of $\ell_n(\beta; \xi(a_0), a(a_0))$ around a $(\beta_{\dagger}, 0)$, with $\beta_{\dagger}$ being a generic $\beta$ 
and $\tilde{\beta}$ and $\tilde{a}_0$ are such that $\lVert \tilde{\beta} - \beta_{\dagger} \rVert \le \lVert \beta - \beta_{\dagger} \rVert$  and $\vert \tilde{a}_0 \vert \le \vert a_0 \vert$: 
\begin{align*}
\ell_n(\beta; \xi(a_0), a(a_0)) = & \ 
\ell_n(\beta_{\dagger}; \xi(0), a(0)) 
+ \left. \begin{bmatrix}
    \frac{\partial}{\partial \beta} \ell_n(\beta; \xi(a_0), a(a_0)) \\
    \frac{\partial}{\partial a_0} \ell_n(\beta; \xi(a_0), a(a_0)) \\
\end{bmatrix} \right \rvert_{\beta = \beta_{\dagger}, a_0=0} 
\begin{pmatrix}
    \beta- \beta_{\dagger} \\
    a_0 - 0
\end{pmatrix}^\top + \\
&\qquad - \frac{1}{2} \begin{pmatrix}
    \beta- \beta_{\dagger} \\
    a_0 - 0
\end{pmatrix}^\top 
\left.\begin{bmatrix}
    H_{\beta} & H^\top_{a_0\beta} \\
    H_{a_0\beta} & H_{a_0}
\end{bmatrix}\right \rvert_{\beta = \tilde{\beta}, a_0=\tilde{a_0}}  
\begin{pmatrix}
    \beta- \beta_{\dagger} \\
    a_0 - 0
\end{pmatrix}
\end{align*}
Taking into account that:
\begin{align*}
\left. \frac{\partial \xi(a_0)}{\partial a_0} \right \rvert_{a_0=0} & = \left. \frac{(A(\tau) \tau \xi_0)(1 + A(\tau)\tau a_0) - A(\tau)\tau (y + A(\tau) \tau a_0 \xi_0)}{(1 + A(\tau)\tau a_0)^2} \right \rvert_{a_0=0} \\
& = \left. \frac{A(\tau)\tau (\xi_0 - y)}{(1 + A(\tau)\tau a_0)^2} \right \rvert_{a_0=0} = A(\tau) \tau (\xi_0 - y) \\
\left. \frac{\partial^2 \xi(a_0)}{\partial a_0^2} \right \rvert_{a_0=0} & = \left. -2 A(\tau) \tau (\xi_0 - y) \frac{A(\tau)\tau }{(1 + A(\tau)\tau a_0)^{3}} \right \rvert_{a_0=0} = -2 A(\tau)^2 \tau^2 (\xi_0 - y)
\end{align*}
and
\begin{align*}
\frac{\partial a(a_0)}{\partial a_0} & = \tau \\
\frac{\partial^2 a(a_0)}{\partial a_0^2} & = 0
\end{align*}
we can derive the following:
\begin{align*}
s_n(\beta_{\dagger}; \xi(0), a(0)) & := \left. \frac{\partial}{\partial \beta} \ell_n(\beta; \xi(a_0), a(a_0)) \right \vert_{\beta=\beta_{\dagger}, a_0=0} \\
& = \left. a(a_0) X_k^\top (\xi(a_0) - \mu(\beta))  \right \vert_{\beta=\beta_{\dagger}, a_0=0} = \frac{X_k^\top (y- \mu(\beta_{\dagger}))}{A(\tau)}  \\
s_n^*(\beta_{\dagger}; \xi(0), a(0)) & := \left. \frac{\partial}{\partial a_0} \ell_n(\beta; \xi(a_0), a(a_0)) \right \vert_{\beta=\beta_{\dagger},  a_0=0} \\
& = \left. \frac{\partial a(a_0)}{\partial a_0} \left(\xi(a_0)^\top X_k\beta - J^\top b(X_k\beta)\right) + a(a_0) \frac{\partial \xi(a_0)^\top }{\partial a_0}  X_k\beta \right \vert_{\beta=\beta_{\dagger},  a_0=0} \\
& = \ \tau \left(y^\top X_k\beta_{\dagger} - J^\top b(X_k\beta_{\dagger})\right) + \tau (\xi_0 - y)^\top  X_k\beta_{\dagger} \\
& =\tau \left( \xi_0^\top X_k\beta_{\dagger} - J^\top b(X_k\beta_{\dagger}) \right)
\end{align*}
\begin{align}
H_{\beta} & := - \frac{\partial^2}{\partial \beta \partial \beta^\top } \ell_n(\beta; \xi(a_0), a(a_0)) 
= \frac{\partial}{\partial \beta} a(a_0) X_k^\top (\xi(a_0) - \mu(\beta))  \label{eq:hessian_beta_posterior} \\
&=  a(a_0) X_k^\top \Sigma X_k 
= \frac{1 + A(\tau)\tau a_0}{A(\tau)} X_k^\top \Sigma X_k \nonumber
\end{align}
where $\Sigma = \operatorname{diag}(b^{\prime \prime}(X_k \beta))$.
\begin{align}
H_{a_0} & := - \frac{\partial^2}{\partial a_0^2} \ell_n(\beta; \xi(a_0), a(a_0))  \label{eq:hessian_a0_posterior} \\  
& = -  \left[ 2\tau \frac{\partial \xi(a_0)^\top}{\partial a_0} + a(a_0) \frac{\partial^2 \xi(a_0)^\top}{\partial a_0^2}\right] X_k \beta \nonumber \\ 
& = -  \left[ 2\tau  \frac{A(\tau)\tau (\xi_0 - y)^\top}{(1 + A(\tau)\tau a_0)^2}  -2 a(a_0)(\xi_0 - y)^\top \frac{A(\tau)^2\tau^2 }{(1 + A(\tau)\tau a_0)^{3}} \right] X_k \beta \nonumber \\
& = - 2  \left[  \frac{A(\tau)\tau^2 }{(1 + A(\tau)\tau a_0)^2}  - a(a_0) \frac{A(\tau)^2\tau^2 }{(1 + A(\tau)\tau a_0)^{3}} \right] (\xi_0 - y)^\top X_k \beta \nonumber \\ 
& = - 2  \left[ \frac{A(\tau) \tau^2}{(1 + A(\tau) \tau^2 a_0)^2} - \frac{1 + A(\tau) \tau a_0}{A(\tau)} \frac{A(\tau)^2 \tau^2}{(1 + A(\tau) \tau a_0)^3} \right] (\xi_0 - y)^\top X_k \beta \nonumber \\
& = - 2  \left[ \frac{A(\tau) \tau^2}{(1 + A(\tau) \tau^2 a_0)^2} - \frac{A(\tau) \tau^2}{(1 + A(\tau) \tau a_0)^2} \right] (\xi_0 - y)^\top X_k \beta \nonumber \\ 
& = 0 \qquad \forall \beta, a_0 \nonumber
\end{align}
\begin{align}
H_{a_0 \beta}^\top & := - \frac{\partial^2}{\partial a_0  \partial \beta} \ell_n(\beta; \xi(a_0), a(a_0)) \label{eq:hessian_a0beta_posterior} \\
& = - \frac{\partial}{\partial a_0}   a(a_0) X_k^\top (\xi(a_0) - \mu(\beta)) \nonumber \\ 
& = - \left[ \frac{\partial a(a_0)}{\partial a_0}  X_k^\top \left( \xi(a_0) -  \mu(\beta)\right) + a(a_0)X_k^\top \frac{\partial \xi(a_0)}{\partial a_0}\right] \nonumber \\ 
& = - \left[ \tau   X_k^\top \left( \xi(a_0) -  \mu(\beta)\right) + a(a_0) X_k^\top \frac{A(\tau)\tau (\xi_0 - y)}{(1 + A(\tau)\tau a_0)^2}\right] \nonumber \\ 
& =  - \left[ \tau   X_k^\top \left( \xi(a_0) -  \mu(\beta)\right) + \frac{1 + A(\tau) \tau a_0}{A(\tau)} X_k^\top \frac{A(\tau)\tau (\xi_0 - y)}{(1 + A(\tau)\tau a_0)^2}\right]  \nonumber \\ 
& =  - \left[ \tau   X_k^\top \left( \frac{y + A(\tau) \tau a_0 \xi_0}{1 + A(\tau) \tau a_0}  -  \mu(\beta)\right) + \tau X_k^\top \frac{\xi_0 - y}{(1 + A(\tau)\tau a_0)}\right] \nonumber \\
&= - \tau   X_k^\top  \left[  \frac{y + A(\tau) \tau a_0 \xi_0 + \xi_0 - y}{1 + A(\tau) \tau a_0} -  \mu(\beta)
\right] \nonumber \\
&= \tau   X_k^\top  \left(  \mu(\beta) - \xi_0 \right) \qquad \forall a_0 \nonumber
\end{align}
Going back to the taylor expansion: 
\begin{align*}
\ell_n(\beta;  \xi(a_0), & a(a_0))  = \\
&=  \ell_n(\beta_{\dagger}; \xi(0), a(0)) 
+ \left. \begin{bmatrix}
    \frac{\partial}{\partial \beta} \ell_n(\beta; \xi(a_0), a(a_0)) \\
    \frac{\partial}{\partial a_0} \ell_n(\beta; \xi(a_0), a(a_0)) \\
\end{bmatrix} \right \rvert_{\beta = \beta_{\dagger}, a_0=0}^\top
\begin{pmatrix}
    \beta- \beta_{\dagger} \\
    a_0 - 0
\end{pmatrix} + \\
&\qquad - \frac{1}{2} \begin{pmatrix}
    \beta- \beta_{\dagger} \\
    a_0 - 0
\end{pmatrix}^\top 
\left.\begin{bmatrix}
    H_{\beta} & H^\top_{a_0\beta} \\
    H_{a_0\beta} & H_{a_0}
\end{bmatrix}\right \rvert_{\beta = \tilde{\beta}, a_0=\tilde{a_0}}  
\begin{pmatrix}
    \beta- \beta_{\dagger} \\
    a_0 - 0
\end{pmatrix} \\
& =  \ \ell_n(\beta_{\dagger}; \xi(0), a(0)) 
+ s_n(\beta_{\dagger}; \xi(0), a(0))^\top \left( \beta- \beta_{\dagger} \right) 
+  a_0s_n^*(\beta_{\dagger}; \xi(0), a(0)) + \\
&\qquad  
- \frac{1}{2} \left( \beta- \beta_{\dagger} \right)^\top H_{\tilde{\beta}} \left( \beta- \beta_{\dagger} \right) 
- a_0 H_{\tilde{a_0}\tilde{\beta}}\left( \beta- \beta_{\dagger} \right) 
\end{align*}

\begin{lemma}
\label{lemma:hessian_abeta}
Assume $\mu(\beta)$ is a continuous function of $\beta$, $a_0 \to 0$ as $n \to \infty$. Denote $H_{00} = H_{0\beta_0} = \tau X_k^\top (\mu(\beta_0) - \xi_0 )$ and $\tilde{a_0}$ and $\tilde{\beta}$ as in the expression above. Then, there exits an $\epsilon_n \to 0$ such that
\begin{equation*}
(1 - \epsilon_n) H_{00} \le H_{\tilde{a_0}\tilde{\beta}} \le (1 + \epsilon_n) H_{00} 
\end{equation*}
\end{lemma}

\begin{proof}
Since $\mu(\beta)$ is a continuous function of $\beta$ then $H_{a\beta}$ is a continuous function of $a$ and $\beta$.

\end{proof}

\subsection{Over-fitted models}
\label{apd:subsec:m1}
We are now going to provide additional results and the proof of the main result (Theorem  \ref{thm:m1:sum_pr_to_0}) for over-fitted models $M_1 = \{ k: k \supsetneq t, \vert k \vert \le m_n \}$. 

\begin{lemma}
\label{lemma:bound_beta}
Let $a_0 = o({c_n}^2)$, with  $c_n = C \sqrt{\frac{ m \Lambda_m (\log p_n)^\omega}{n^\gamma \lambda^2 (1 - \varepsilon)^2}}$ and $\Lambda_m$ defined as in \ref{cond:2:b}.  Under Conditions \ref{cond:1}--\ref{cond:4}, we have
\begin{equation*}
\sup_{k \in M_1 : \vert k \vert = m} \lVert \hat{\beta}(k;a_0)  - \beta_0(k) \rVert = O\left( \sqrt{\frac{m \Lambda_m (\log p_n)^\omega}{n^\gamma}} \right) 
\end{equation*}
uniformly for all $m \le m_n$. 
\end{lemma}

\begin{remark}[Rate of convergence]
\label{rem:convergence_rates} Notice that a direct consequence of the above Lemma is that the rate of convergence of the estimator $\hat{\beta}(k;a_0)$ is exactly \begin{equation*}
(\gamma - \phi \omega)(\delta - 1) \frac{1}{2} \ .
\end{equation*}
We consider now some important examples under which conditions \ref{cond:4} hold and we compute the corresponding rate of convergence.
\begin{itemize}
\item For all members of the exponential family which are subGaussian ( \ie $\varpi = 2$) conditions \ref{cond:4} implies $\omega = 1$, $\gamma < \frac{3}{2}$,  $\max\left\{(\gamma - 1), 0 \right\} < \phi < \gamma$ and $0 < \delta < \frac{2- 1/\varpi - \gamma}{(\gamma - \phi)(\frac{3}{2} - 1)} = \frac{3 - 2 \gamma}{\gamma - \phi}$. Let $\upsilon > 0$ be small and we can set $\gamma = \frac{3}{2} - \upsilon$, $\phi = \frac{1}{2}$ and $\delta = \frac{\upsilon}{1 - \upsilon}$ which leads to $(\gamma - \phi \omega)(\delta - 1) \frac{1}{2} = -\frac{1}{2} + \upsilon$ that is we get almost root-$n$ rate of convergence.

We notice that \citet{narisetty2019Skinny} used $\omega=1$, $\gamma=1$, and $\phi<1$ and $0 < \delta < 2$, hence for large $\phi$ they obtain very small convergence rate. However, taking $\phi=\frac{1}{2}$ (as in \citet{yang2022nonlocalpriors}) and small $\delta$ their rate of convergence is $\frac{1}{4}(-1 + \delta)$, which is still sub-optimal with respect to our results.

\item For all members of the exponential family which are $1$-subExponential ( \ie $\varpi = 1$) conditions \ref{cond:4} imply $\omega = 2$, $\gamma < 1$,  $\max\left\{\frac{(\gamma - 1)}{\omega}, 0 \right\} < \phi < \frac{\gamma}{\omega}$. Let $0 < \upsilon < \frac{1}{6}$ be small and we can set $\gamma = 1 - \upsilon$, $\phi = \upsilon$ and $\delta = \upsilon$ which leads to $(\gamma - \phi \omega)(\delta - 1) \frac{1}{2} = -\frac{1}{2} + 2 \upsilon - \frac{3}{2} \upsilon^2$ that is we get almost root-$n$ rate of convergence.

\item For all members of the exponential family which are $\frac{2}{3}$-subExponential ( \ie $\varpi = 2/3$) conditions \ref{cond:4} imply $\omega = 3$, $\gamma < \frac{1}{2}$,  $0 < \phi < \frac{1}{6}$. Let $0 < \upsilon < \frac{1+\sqrt{65}}{16} \approx 0.5664$ be small and we can set $\gamma = \frac{1}{2} - \upsilon$, $\phi = \upsilon$ and $\delta = \upsilon$ which leads to $(\gamma - \phi \omega)(\delta - 1) \frac{1}{2} = -\frac{1}{4} + 2 \upsilon - 2 \upsilon^2$ that is we get almost $n^{\frac{1}{4}}$ rate of convergence.
\end{itemize}
\end{remark}

\begin{proof}
Consider the log-likelihood in Eq. \eqref{eq:loglik_xi_a} and the previous derivations of the Taylor expansion (with Lagrange remainder) of $\ell_n(\beta; \xi(a_0), a(a_0))$ around a $(\beta_{\dagger}, 0)$, with $\beta_{\dagger}$ being a generic $\beta$
and $\tilde{\beta}$ and $\tilde{a}_0$ are such that $\lVert \tilde{\beta} - \beta_{\dagger} \rVert \le \lVert \beta - \beta_{\dagger} \rVert$  and $\vert \tilde{a}_0 \vert \le \vert a_0 \vert$.
\\
Let $\beta(k) = \beta_{\dagger}(k) + c_n u$, where $u \in \mathbb{R}^{\vert k \vert}$ and $u^\top u = 1$, $c_n = C \sqrt{\frac{ m \Lambda_m (\log p_n)^\omega}{n^\gamma \lambda^2 (1 - \varepsilon)^2}}$, $C=8 C_\varpi^{2/\varpi} 2^{2/\varpi} K^{4/\varpi}$ and $m := \vert k \vert$. Notice that conditions \ref{cond:4:c}--\ref{cond:4:e} are sufficient for $c_n \to 0$ as $n \to \infty$ since
\begin{equation}
\label{eq:rate_cn}
\frac{m (\log p_n)^\omega}{n^\gamma} = \frac{(\log p_n)^\omega}{n^\gamma} \left(\frac{n^\gamma}{(\log p_n)^\omega}\right)^\delta = \left( \frac{(\log p_n)^\omega}{n^\gamma} \right)^{1 - \delta} = n^{(\phi \omega - \gamma) (1 - \delta)} 
\end{equation}
which goes to zero when $0 < \phi \omega < \gamma$ and $0 < \delta < 1$. For
simplicity of notation, let
\begin{equation*}
s_n(\beta_{\dagger}; \xi(0), a(0)) := s_n(\beta_{\dagger}) \qquad s_n^*(\beta_{\dagger}; \xi(0), a(0))  := s_n^*(\beta_{\dagger}) 
\end{equation*}
and take 
\begin{equation*}
  \beta_{\dagger} = \beta_0(k), \qquad \beta = \beta(k) \ .
\end{equation*}
Then, we have
\begin{align*}
\ell_n(\beta(k); \xi(a_0), a(a_0))  & -  \ell_n(\beta_0(k); \xi(0), a(0)) = \\
& = c_n s_n(\beta_0(k))^\top u +  a_0s_n^*(\beta_0(k)) - \frac{1}{2} c_n^2 u^\top H_{\tilde{\beta}} u - a_0 c_n H_{\tilde{a_0}\tilde{\beta}}^\top u \\
& \overset{*}{\le} c_n s_n(\beta_0(k))^\top u +  a_0s_n^*(\beta_0(k)) - \frac{1}{2} c_n^2 u^\top  (1 - \varepsilon) H_{\beta_0(k)} u  - a_0 c_n H_{\tilde{a_0}\tilde{\beta}}^\top u 
\end{align*}
where $(*)$ is  due to Lemma \ref{lemma:hessian}. Hence, for some $u$
\begin{align*}
P( \ell_n &(\beta(k); \xi(a_0), a(a_0)) -  \ell_n(\beta_0(k); \xi(0), a(0)) > 0 ) \\
& \le P\left( c_n s_n(\beta_0(k))^\top u \ge  \frac{1}{2} c_n^2 (1 - \varepsilon) u^\top (1 - \varepsilon) H_{\beta_0(k)} u - a_0s_n^*(\beta_0(k)) + a_0 c_n H_{\tilde{a_0}\tilde{\beta}}^\top u \right) \\
& = P\left( s_n(\beta_0(k))^\top u \ge \frac{1}{2} c_n (1 - \varepsilon) u^\top  H_{\beta_0(k)} u - \frac{a_0}{c_n} s_n^*(\beta_0(k)) + a_0 H_{\tilde{a_0}\tilde{\beta}}^\top u \right) \\
& = P\left( \lVert s_n(\beta_0(k)) \rVert^2 \ge \left( \frac{1}{2} c_n (1 - \varepsilon) u^\top H_{\beta_0(k)} u - \frac{a_0}{c_n} s_n^*(\beta_0(k)) + a_0 H_{\tilde{a_0}\tilde{\beta}}^\top u \right)^2 \right) 
\end{align*}
Consider the right hand side of the above inequality. For large enough $n$ and using Conditions \ref{cond:2:b}-\ref{cond:2:c}-\ref{cond:2:d} and Lemma \ref{lemma:hessian_abeta}, we have:
\begin{align*}
W & = \left( \frac{1}{2} c_n (1 - \varepsilon) u^\top   H_{\beta_0(k)} u - \frac{a_0}{c_n} s_n^*(\beta_0(k)) + a_0 H_{\tilde{a_0}\tilde{\beta}}^\top u \right)^2 \\
& \ge \left( \frac{1}{2} c_n (1 - \varepsilon) n \lambda - \frac{a_0}{c_n} n K_2 - a_0 n (1+\epsilon) K_1 \right)^2 \\
& = \frac{1}{4} c_n^2 (1 - \varepsilon)^2 n^2 \lambda^2 + \frac{a_0^2}{c_n^2} n^2 K_2^2 + a_0^2 n^2 (1+\epsilon)^2 K_1^2 \\
& \qquad - a_0 n^2 (1 - \varepsilon) \lambda K_2 \\ 
& \qquad - a_0 c_n n^2 (1-\varepsilon) (1+\epsilon) \lambda K_1 \\
& \qquad + 2 \frac{a_0^2}{c_n} n^2 (1 + \epsilon) K_1 K_2 \\
& \ge \frac{1}{4} c_n^2 (1 - \varepsilon)^2 n^2 \lambda^2 \\
& \qquad - a_0 n^2 (1 - \varepsilon) \lambda K_2 \\ 
& \qquad - a_0 c_n n^2 (1-\varepsilon) (1+\epsilon) \lambda K_1
\end{align*}
where in the last expression we keep the leading positive term and all the negative terms. Since, under the assumption $a_0 = o({c_n}^2)$, we have, as $n \to \infty$ 
\begin{equation*}
\frac{a_0 n^2 (1 - \varepsilon) \lambda K_2}{\frac{1}{16} c_n^2 (1 - \varepsilon)^2 n^2 \lambda^2} = 16 \frac{K_2}{(1 - \varepsilon) \lambda} \frac{a_0}{c_n^2} \to 0
\end{equation*}
and
\begin{equation*}
\frac{a_0 c_n n^2 (1-\varepsilon) (1+\epsilon) \lambda K_1}{\frac{1}{16} c_n^2 (1 - \varepsilon)^2 n^2 \lambda^2} = 16 \frac{(1+\epsilon) K_1}{(1 - \varepsilon) \lambda^2} \frac{a_0}{c_n} \to 0
\end{equation*}
hence for large enough $n$ we have 
\begin{equation*}
- \frac{1}{8} c_n^2 (1 - \varepsilon)^2 n^2 \lambda^2 \le - a_0 n^2 (1 - \varepsilon) \lambda K_2 - a_0 c_n n^2 (1-\varepsilon) (1+\epsilon) \lambda K_1
\end{equation*}
so finally
\begin{equation*}
W \ge \tilde{W} = \frac{1}{8} c_n^2 (1 - \varepsilon)^2 n^2 \lambda^2 \ .
\end{equation*}
which leads to
\begin{align}
    P( \ell_n(\beta(k); \xi(a_0), a(a_0)) -  \ell_n(\beta_0(k); \xi(0), a(0)) > 0 ) 
    & \le P\left( \lVert s_n(\beta_0(k)) \rVert^2 \ge W \right) \nonumber \\
    & \le P\left( \lVert s_n(\beta_0(k)) \rVert^2 \ge \tilde{W} \right)  \nonumber \\
    & = P\left( \lVert s_n(\beta_0(k)) \rVert^2 \ge \frac{1}{8} c_n^2 (1 - \varepsilon)^2 n^2 \lambda^2 \right)
    \label{eq:chain_to_apply_lemmaconcentration}
\end{align}
Since $c_n = C \sqrt{\frac{ m \Lambda_m (\log p_n)^\omega}{n^\gamma \lambda^2 (1 - \varepsilon)^2}}$ we have
\begin{equation*}
\tilde{W}  = \frac{1}{8} c_n^2 (1 - \varepsilon)^2 n^2 \lambda^2 
= \frac{C}{8} m \Lambda_m (\log p_n)^\omega n^{2 - \gamma}
\end{equation*}
Using Lemma \ref{lemma:concentration} we would like to choose $h$ such that 
\begin{equation}
\label{eq:condition_on_h}
    \frac{1}{C_\varpi} \frac{(h - \sigma^2 m n \lambda)^{\varpi/2}}{K^2 \sqrt{m n \Lambda_m}} = 2 m \log p_n
\end{equation}
and such that $h \le \tilde{W}$, at least for $n$ large enough. From Eq. \eqref{eq:condition_on_h}, it follows that:
\begin{equation*}
h = (C_\varpi 2 m \log p_n)^{2/\varpi} K^{4/\varpi} (m n \Lambda_m)^{1/\varpi} + \sigma^2 m n \lambda \ .
\end{equation*}
Finally we need to show that for large enough $n$ $h/\tilde{W} \to 0$. Since $\omega = 2/\varpi$ by Condition \ref{cond:4:a}
\begin{align*}
    \frac{h}{\tilde{W}} 
    & = \frac{m^{3/\varpi} n^{1/\varpi} \Lambda_m^{1/\varpi} (\log p_n)^{2/\varpi}}{m n^{2-\gamma} \Lambda_m (\log p_n)^{\omega}} + C_1 \frac{m n \lambda}{m n^{2 - \gamma} \Lambda_m (\log p_n)^{\omega}} \\
    & = m^{3/\varpi - 1} n^{1/\varpi - 2 + \gamma} \Lambda_m^{1/\varpi - 1} + C_2 n^{\gamma - 1} (\log p_n)^{-\omega} 
\end{align*}
for some constants $C_1$ and $C_2$. We analyze the two terms separately. For the first term a necessary condition is $1/\varpi - 2 + \gamma < 0$ which is equivalent to Condition \ref{cond:4:b}, that is, $0 < \gamma < \frac{2 \varpi - 1}{\varpi}$. Furthermore under Conditions \ref{cond:4:d}-\ref{cond:4:e} we need
\begin{equation*}
    (\gamma - \phi \omega) \delta \left( \frac{3}{\varpi} - 1 \right) + \frac{1}{\varpi} - 2 + \gamma < 0
\end{equation*}
which implies
\begin{equation*}
    \delta < \Delta = \frac{2- \frac{1}{\varpi} - \gamma}{(\gamma - \phi \omega) \left( \frac{3}{\varpi} - 1 \right)} \ .
\end{equation*}
However, Condition \ref{cond:4:b} guaranties the positiveness of the numerator of $\Delta$, while $\frac{3}{\varpi} - 1 > 0$ for all $\varpi \in (\frac{1}{2}, 2]$ and $\phi \omega < \gamma $ under Condition \ref{cond:4:c}. Hence, $\Delta > 0$. For the second term, using Condition \ref{cond:4:d}, we have
\begin{equation*}  
n^{\gamma - 1} (\log p_n)^{-\omega} = n^{\gamma - 1} n^{-\phi \omega} = n^{\gamma - 1 - \phi \omega} 
\end{equation*}
since we need this term to go to zero we must have $\gamma < 1 + \phi \omega$ which is satisfied under Condition \ref{cond:4:c}. Hence, going back to Eq. \eqref{eq:chain_to_apply_lemmaconcentration}, and applying Lemma \ref{lemma:concentration}, under Conditions \ref{cond:4} we have
\begin{align*}
     P( \ell_n(\beta(k); \xi(a_0), a(a_0)) -  \ell_n(\beta_0(k); \xi(0), a(0)) > 0 ) 
    & \le P\left( \lVert s_n(\beta_0(k)) \rVert^2 \ge \tilde{W} \right)   \\
     & \le P\left( \lVert s_n(\beta_0(k)) \rVert^2 \ge h \right)  \\
    & \le \exp(-2 m \log p_n).
\end{align*}
That is
\begin{align*}
 P( \ell_n(\beta(k); \xi(a_0), a(a_0)) -  \ell_n(\beta_0(k); \xi(0), a(0)) < 0 )  
 &\ge 1 - \exp(-2 m \log p_n) \\
 &= 1 - {p_n}^{-2 m} \ .
\end{align*}
The concavity of $\ell_n$ implies that
\begin{equation*}
P\left( \lVert \hat{\beta}(k;a_0) - \beta_0(k) \rVert \le c_n  \right) \ge 1 - {p_n}^{-2 m}
\end{equation*}
and by taking a union bound over all models $k \supset t$ with size at most $m_n$, we have
\begin{equation*}
P\left( \sup_{k \in M_1: \lvert k \rvert = m} \lVert \hat{\beta}(k;a_0) - \beta_0(k) \rVert > c_n, \quad \text{for any} \ m \le m_n \right) \le \sum_{\lvert t \rvert \le m \le m_n} {p_n}^{-2 m} {p_n}^m \to 0 \ ,
\end{equation*}
which proves the Lemma.
 
\end{proof}
We now state and prove a result similar to Lemma 7.3 in \citep{lee2021bayesian}, that will be necessary in the proof of Theorem \ref{thm:m1:sum_pr_to_0}).
\begin{lemma}
\label{lemma:inequality_bn} 
    Under conditions \ref{cond:2}-\ref{cond:3} and $a_0 \rightarrow 0 \ as \ n \rightarrow \infty$. Let $\phi$ as defined in \ref{cond:4}, and $\psi < \phi$. Assume that the eigenvalue
    $\tilde{\Lambda}_m = \sup_{k \in M_1: \lvert k \rvert=m} \lambda_{\max}\left( P_k - P_t \right)$ bounded as $\tilde{\lambda} < \tilde{\Lambda}_m < \infty$. Then:
    \begin{equation*}
        \ell(\hat{\beta}(k); \xi(a_0), a(a_0))  -  \ell(\hat{\beta}(t); \xi(a_0), a(a_0)) \le b_n(m - \lvert t \rvert) + C
    \end{equation*}
    where $b_n = \displaystyle \frac{\sigma^2 \tilde{\lambda}}{2(1-\varepsilon)} \frac{\log p_n  }{n^{\psi}}$ and $\hat{\beta}(\cdot)$ denote $\hat{\beta}(\cdot; a_0)$ for simplicity of notation. 
\end{lemma}
\begin{proof}
Consider the log-likelihood in Eq. \eqref{eq:loglik_xi_a} and the previous derivations of the Taylor expansion (with Lagrange remainder) of $\ell_n(\beta; \xi(a_0), a(a_0))$ around a $(\beta_{\dagger}, 0)$, with $\beta_{\dagger}$ being a generic $\beta$ and $\tilde{\beta}$ and $\tilde{a}_0$ are such that $\lVert \tilde{\beta} - \beta_{\dagger} \rVert \le \lVert \beta - \beta_{\dagger} \rVert$  and $\vert \tilde{a}_0 \vert \le \vert a_0 \vert$.
\begin{align}
    \ell_n(\beta; \xi(a_0), a(a_0))  =  \ \ell_n(\beta_{\dagger}; &  \xi(0), a(0)) 
    + \left( \beta- \beta_{\dagger} \right)^\top s_n(\beta_{\dagger})
    +  a_0s_n^*(\beta_{\dagger}) + \label{eq:loglik_taylor_expansion}\\
    &\qquad  
    - \frac{1}{2} \left( \beta- \beta_{\dagger} \right)^\top H_{\tilde{\beta}} \left( \beta- \beta_{\dagger} \right) 
    - a_0 H_{\tilde{a_0}\tilde{\beta}}\left( \beta- \beta_{\dagger} \right) \nonumber
\end{align}
where, for simplicity of notation, we denoted
\begin{equation*}
    s_n(\beta_{\dagger}) := s_n(\beta_{\dagger}; \xi(0), a(0)) \qquad s_n^*(\beta_{\dagger}):= s_n^*(\beta_{\dagger}; \xi(0), a(0))   \ .
\end{equation*}
\\
Taking Eq. \eqref{eq:loglik_taylor_expansion} with $\beta_{\dagger} = \beta_0(k)$ and $\beta = \hat{\beta}(k)$, we get: 
\begin{align*}
\ell_n&(\hat{\beta}(k);  \xi(a_0), a(a_0)) - \ell_n(\beta_0(k);  \xi(0), a(0)) = \\
&= \left( \hat{\beta}(k) - \beta_0(k) \right)^\top s_n(\beta_0(k)) 
+  a_0s_n^*( \beta_0(k)) - \frac{1}{2} \left( \hat{\beta}(k) - \beta_0(k) \right)^\top H_{\tilde{\beta}} \left( \hat{\beta}(k) - \beta_0(k) \right) \\
&\qquad
- a_0 H_{\tilde{a_0}\tilde{\beta}}\left( \hat{\beta}(k) - \beta_0(k) \right) \\
&\overset{*}{\le} \left( \hat{\beta}(k) - \beta_0(k) \right)^\top s_n(\beta_0(k)) 
+  a_0s_n^*( \beta_0(k)) - \frac{1-\varepsilon}{2} \left( \hat{\beta}(k) - \beta_0(k) \right)^\top H_{\beta_0(k)} \left( \hat{\beta}(k) - \beta_0(k) \right) \\
&\qquad
- a_0 H_{\tilde{a_0}\tilde{\beta}}\left( \hat{\beta}(k) - \beta_0(k) \right) 
\end{align*}
where $(*)$ is  due to Lemma \ref{lemma:hessian}.
Consider now the first order Taylor expansion of $s_n(\cdot)$ around $\beta_0(k)$, with Lagrange remainder for some $\tilde{\tilde{\beta}}$:
\begin{equation*}
s_n(\beta) = s_n(\beta_0(k)) - H_{\tilde{\tilde{\beta}}} (\beta - \beta_0(k))
\end{equation*}
hence:
\begin{equation*}
\hat{\beta}(k) - \beta_0(k) = H_{\tilde{\tilde{\beta}}}^{-1}s_n(\beta_0(k)) \ . 
\end{equation*}
Therefore:
\begin{align*}
\ell_n(\hat{\beta}(k); & \xi(a_0), a(a_0)) - \ell_n(\beta_0(k);  \xi(0), a(0)) = \\
&\le \left( \hat{\beta}(k) - \beta_0(k) \right)^\top s_n(\beta_0(k))
- \frac{1-\varepsilon}{2} \left( \hat{\beta}(k) - \beta_0(k) \right)^\top H_{\beta_0(k)} \left( \hat{\beta}(k) - \beta_0(k) \right) \\
&\qquad
+  a_0s_n^*( \beta_0(k)) 
- a_0 H_{\tilde{a_0}\tilde{\beta}}\left( \hat{\beta}(k) - \beta_0(k) \right)  \\
&= s_n(\beta_0(k))^\top H_{\tilde{\tilde{\beta}}}^{-1} s_n(\beta_0(k))
- \frac{1-\varepsilon}{2} s_n(\beta_0(k))^\top H_{\tilde{\tilde{\beta}}}^{-1}H_{\beta_0(k)} H_{\tilde{\tilde{\beta}}}^{-1}s_n(\beta_0(k)) \\
&\qquad
+  a_0s_n^*( \beta_0(k)) 
- a_0 H_{\tilde{a_0}\tilde{\beta}}\left( \hat{\beta}(k) - \beta_0(k) \right) \\
&\le \frac{1}{1-\varepsilon} s_n(\beta_0(k))^\top H_{\beta_0(k)}^{-1} s_n(\beta_0(k))
- \frac{1}{2(1-\varepsilon)} s_n(\beta_0(k))^\top H_{\beta_0(k)} ^{-1} s_n(\beta_0(k)) \\
&\qquad
+  a_0s_n^*( \beta_0(k)) 
- a_0 H_{\tilde{a_0}\tilde{\beta}}\left( \hat{\beta}(k) - \beta_0(k) \right) \\
&= \frac{1}{2(1-\varepsilon)}  s_n(\beta_0(k))^\top H_{\beta_0(k)}^{-1} s_n(\beta_0(k)) \\
&\qquad
+  a_0s_n^*( \beta_0(k)) 
- a_0 H_{\tilde{a_0}\tilde{\beta}}\left( \hat{\beta}(k) - \beta_0(k) \right) \\
&= \frac{1}{2(1-\varepsilon)}  \frac{(y- \mu(\beta_0(k)))^\top X_k}{A(\tau)}H_{\beta_0(k)}^{-1} \frac{X_k^\top (y- \mu(\beta_0(k)))}{A(\tau)} \\
&\qquad
+  a_0s_n^*( \beta_0(k)) 
- a_0 H_{\tilde{a_0}\tilde{\beta}}\left( \hat{\beta}(k) - \beta_0(k) \right) \\
&= \frac{1}{2(1-\varepsilon)}   U^\top P_k U 
+  a_0s_n^*( \beta_0(k)) 
- a_0 H_{\tilde{a_0}\tilde{\beta}}\left( \hat{\beta}(k) - \beta_0(k) \right) \\
\end{align*}
where the quadratic form has been rewritten as
\begin{equation*}
    \frac{(y- \mu(\beta_0(k)))^\top X_k}{A(\tau)}H_{\beta_0(k)}^{-1} \frac{X_k^\top (y- \mu(\beta_0(k)))}{A(\tau)}  =: U^\top P_k U
\end{equation*}
with 
\begin{equation*}
    U = \frac{(y- \mu(\beta_0(k)))^\top}{A(\tau)} \qquad P_k = X_k H_{\beta_0(k)}^{-1} X_k^\top
\end{equation*}
\\
In a similar way with $\beta_{\dagger} = \beta_0(t)$ and $\beta = \hat{\beta}(t)$, we get: 
\begin{align*}
\ell_n(\hat{\beta}(t); & \xi(a_0), a(a_0)) - \ell_n(\beta_0(t);  \xi(0), a(0)) = \\
&\ge \left( \hat{\beta}(t) - \beta_0(t) \right)^\top s_n(\beta_0(t))
- \frac{1+\varepsilon}{2} \left( \hat{\beta}(t) - \beta_0(t) \right)^\top H_{\beta_0(t)} \left( \hat{\beta}(t) - \beta_0(t) \right) \\
&\qquad
+  a_0s_n^*( \beta_0(t)) 
- a_0 H_{\tilde{a_0}\tilde{\beta}}\left( \hat{\beta}(t) - \beta_0(t) \right)  \\
&\ge \frac{1}{2(1+\varepsilon)}  s_n(\beta_0(t))^\top H_{\beta_0(t)}^{-1} s_n(\beta_0(t))
+  a_0s_n^*( \beta_0(t)) 
- a_0 H_{\tilde{a_0}\tilde{\beta}}\left( \hat{\beta}(t) - \beta_0(t) \right) \\
&= \frac{1}{2(1+\varepsilon)}   U^\top P_t U 
+  a_0s_n^*( \beta_0(t)) 
- a_0 H_{\tilde{a_0}\tilde{\beta}}\left( \hat{\beta}(t) - \beta_0(t) \right) \\
\end{align*}
with 
\begin{equation*}
    U = \frac{(y- \mu(\beta_0(t)))^\top}{A(\tau)} \qquad P_t = X_t H_{\beta_0(t)}^{-1} X_t^\top \ . 
\end{equation*}
Notice that the two quantities defined as $U$ are essentially the same because $\mu(\beta_0(t)) = \mu(\beta_0(k))$, by definition of $\beta_0(\cdot)$. 
Now take:
\begin{align*}
    &\ell(\hat{\beta}(k); \xi(a_0), a(a_0))  -  \ell(\hat{\beta}(t); \xi(a_0), a(a_0)) \\
    &= \left[ \ell(\hat{\beta}(k); \xi(a_0), a(a_0))  - \ell_n(\beta_0(k);  \xi(0), a(0))  \right] 
    - \left[ \ell(\hat{\beta}(t); \xi(a_0), a(a_0)) - \ell_n(\beta_0(t);  \xi(0), a(0))  \right] \\
    &= \frac{1}{2(1-\varepsilon)}   U^\top P_k U 
    - \frac{1}{2(1+\varepsilon)}   U^\top P_t U + \\
    &\qquad +  a_0s_n^*( \beta_0(k)) 
    - a_0 H_{\tilde{a_0}\tilde{\beta}}\left( \hat{\beta}(k) - \beta_0(k) \right) 
    -  a_0s_n^*( \beta_0(t)) 
    + a_0 H_{\tilde{a_0}\tilde{\beta}}\left( \hat{\beta}(t) - \beta_0(t) \right) \\
    &= \frac{1}{2(1-\varepsilon)}   U^\top \left( P_k - P_t \right) U + \left[ 
    - \frac{1}{2(1+\varepsilon)}  + \frac{1}{2(1-\varepsilon)} \right]  U^\top P_t U + \\
    &\qquad +  a_0s_n^*( \beta_0(k)) 
    - a_0 H_{\tilde{a_0}\tilde{\beta}}\left( \hat{\beta}(k) - \beta_0(k) \right) 
    -  a_0s_n^*( \beta_0(t)) 
    + a_0 H_{\tilde{a_0}\tilde{\beta}}\left( \hat{\beta}(t) - \beta_0(t) \right) \\
    &= \frac{1}{2(1-\varepsilon)}   U^\top \left( P_k - P_t \right) U + \frac{\varepsilon}{(1+\varepsilon)(1-\varepsilon)}  U^\top P_t U 
    +  a_0 \left[ s_n^*( \beta_0(k)) -  s_n^*( \beta_0(t)) \right] +  \\
    &\qquad 
    + a_0 \left[  H_{\tilde{a_0}\tilde{\beta}}\left( \hat{\beta}(t) - \beta_0(t) \right)  - H_{\tilde{a_0}\tilde{\beta}}\left( \hat{\beta}(k) - \beta_0(k) \right) \right]
\end{align*}

\noindent We analyze the four terms separately. First of all, we recall that $s_n^*(\beta) = \tau \left( \xi_0^\top X \beta - J^\top b(X^\top \beta)\right)$. It is clear that $s_n^*( \beta_0(k)) =  s_n^*( \beta_0(t))$, hence
\begin{equation*}
 a_0 \left[ s_n^*( \beta_0(k)) -  s_n^*( \beta_0(t)) \right] = 0
\end{equation*}
In a similar way, being $H_{0\beta} = \tau (\mu(\beta) - \xi_0)^\top X$, and using Lemma \ref{lemma:hessian_abeta}, it holds:
\begin{align*}
    a_0 &\left[  H_{\tilde{a_0}\tilde{\beta}}\left( \hat{\beta}(t) - \beta_0(t) \right)  - H_{\tilde{a_0}\tilde{\beta}}\left( \hat{\beta}(k) - \beta_0(k) \right) \right] \\
    & \qquad \le a_0 \left[  (1+\bar{\epsilon})H_{0\beta_0(t)}\left( \hat{\beta}(t) - \beta_0(t) \right)  - (1+\bar{\epsilon})H_{0\beta_0(k)}\left( \hat{\beta}(k) - \beta_0(k) \right) \right] \\
    & \qquad = a_0(1+\bar{\epsilon})\tau  \left[ (\mu(\beta_0(t)) - \xi_0)^\top X_t\left( \hat{\beta}(t) - \beta_0(t) \right) - (\mu(\beta_0(k)) - \xi_0)^\top X_k\left( \hat{\beta}(k) - \beta_0(k) \right) \right] \\
    & \qquad = 0
\end{align*}
We are therefore left with two terms: 
\begin{align*}
    &\ell(\hat{\beta}(k); \xi(a_0), a(a_0))  -  \ell(\hat{\beta}(t); \xi(a_0), a(a_0)) \le \frac{1}{2(1-\varepsilon)}   U^\top \left( P_k - P_t \right) U + \frac{\varepsilon}{(1+\varepsilon)(1-\varepsilon)}  U^\top P_t U 
\end{align*}
With all the hypothesis of Lemma \ref{lemma:concentration}, we can state that, for any $\varpi \in (0, 2]$ and $h = 2(1-\varepsilon) (\lvert k \rvert- \lvert t \rvert) b_n$:

\begin{align*}
    P\left( U^\top (P_k - P_t) U \ge h \right)  
    &= P\left( U^\top (P_k - P_t) U \ge 2(1-\varepsilon)b_n (\lvert k \rvert - \lvert t \rvert) \right)   \\
    & \le \exp\left\{ - \frac{1}{C_\varpi} \left( \frac{h - \sigma^2 (\lvert k \rvert - \lvert t \rvert)  \tilde{\lambda} }{K^2 \sqrt{ (\lvert k \rvert - \lvert t \rvert)  \tilde{\Lambda}_m }} \right)^{\frac{\varpi}{2}} \right\} \\
    & = \exp\left\{ - \frac{1}{C_\varpi} \left( \frac{2(1-\varepsilon) b_n (\lvert k \rvert - \lvert t \rvert)  - \sigma^2 (\lvert k \rvert - \lvert t \rvert)  \tilde{\lambda}}{K^2 \sqrt{ (\lvert k \rvert - \lvert t \rvert) \tilde{\Lambda}_m }} \right)^{\frac{\varpi}{2}} \right\} \\
    & = \exp\left\{ - \frac{1}{C_\varpi} \left( \frac{ \sigma^2 \log p_n \tilde{\lambda} n^{-\psi} (\lvert k \rvert - \lvert t \rvert)  - \sigma^2 (\lvert k \rvert - \lvert t \rvert) \tilde{\lambda}}{K^2 \sqrt{ (\lvert k \rvert - \lvert t \rvert) \tilde{\Lambda}_m }} \right)^{\frac{\varpi}{2}} \right\} \\
    & = \exp\left\{ - \frac{1}{C_\varpi} \left( C^*(\lvert k \rvert - \lvert t \rvert)^{1/2}\left( n^{-\psi}\log p_n - 1\right) \right)^{\varpi/2}   \right\} \\
    & \approx \exp\left\{ - \frac{1}{C_\varpi} \left( C^*(\lvert k \rvert - \lvert t \rvert)^{1/2} n^{-\psi}\log p_n \right)^{\varpi/2}   \right\} \\
    & = \exp\left\{  \log \left(  {p_n}^{-C (\lvert k \rvert - \lvert t \rvert)^{1/2}n^{-\psi}}    \right)^{\varpi/2} \right\} \ = \  {p_n}^{-C \left[(\lvert k \rvert - \lvert t \rvert)^{1/2}n^{-\psi}\right]^{\varpi/4}}  
\end{align*}
where $\tilde{\lambda} < \tilde{\Lambda}_m = \sup_{k \in M_1: \lvert k \rvert=m} \lambda_{\max}\left( P_k - P_t \right)$. This proves that
\begin{equation*}
\frac{1}{2(1-\varepsilon)}U^\top (P_k - P_t) U \le  (\lvert k \rvert - \lvert t \rvert) b_n \text{ with probability} \rightarrow 1
\end{equation*}
for $b_n = \frac{\sigma^2}{2(1-\varepsilon)} \tilde{\lambda}  \log p_n n^{-\psi}$.
The second term
\begin{equation*}
    \frac{\varepsilon}{(1+\varepsilon)(1-\varepsilon)}  U^\top P_t U \ , 
\end{equation*}
can be controlled in the same way, using $b_n =  \frac{\sigma^2}{(1-\varepsilon)(1+\varepsilon)}  \tilde{\lambda} \log p_n n^{-\psi}$ and $h=(1-\varepsilon)(1+\varepsilon) b_n \lvert t \rvert$: 
\begin{align*}
    P\left( U^\top P_t U \ge h \right) 
    & \lessapprox   {p_n}^{-C (\lvert t \rvert^{1/2} n^{-\psi})^{\varpi/2}}  \ . 
\end{align*}
Therefore, with $\varepsilon = o(1)$, we get 
\begin{equation*}
    \frac{\varepsilon}{(1+\varepsilon)(1-\varepsilon)}  U^\top P_t U \ = o(1) 
\end{equation*}
That leads to
\begin{align*}
    &\ell(\hat{\beta}(k); \xi(a_0), a(a_0))  -  \ell(\hat{\beta}(t); \xi(a_0), a(a_0))  \le b_n(m - \lvert t \rvert) + C
\end{align*} 
\end{proof}

Since the proof of the following (and main) theorem is already convoluted, we prefer to present the following results when $\tau$ is known. Therefore from now on the posterior distribution of $Z$, obtained by Lemma \ref{lemma:jointgibbs} and integrating out the other variables conditioning on the data $(y,X)$
\begin{align*}
    P(Z \mid y, X) &= \int P(\beta, Z, \tau \vert y, X) d \beta d \tau \\
    & = \int P(\beta^{(0)} \vert Z, \tau, y, X)  P(\beta^{(1)} \vert Z, \tau, y, X)  \left( \frac{p_n}{\alpha} \right)^{\sum_{j=1}^{p_n} (1 - Z_j)}  P(\tau \vert y, X) \ d \beta d \tau \ 
\end{align*}
is reduced to
\begin{align*}
    P(Z \mid y, X) & = \int P(\beta^{(0)} \vert Z, \tau, y, X)  P(\beta^{(1)} \vert Z, \tau, y, X)  \left( \frac{p_n}{\alpha} \right)^{\sum_{j=1}^{p_n} (1 - Z_j)} \ d \beta  \ . 
\end{align*}
Recalling that the posterior of $\beta(k)$ can be seen as the likelihood with new data (see Eq. \eqref{eq:loglik_xi_a}), we can write
\begin{align}
  P(Z = k \mid y, X) & = \int_{\beta^{(0)}} \int_{\beta^{(1)}} P(\beta^{(0)} \vert Z=k, \tau, y, X)  P(\beta^{(1)} \vert Z=k, \tau, y, X)  \left( \frac{p_n}{\alpha} \right)^{p_n -\lvert k \rvert} \ d \beta \nonumber \\
  & \propto \left( \frac{p_n}{\alpha} \right)^{p_n -\lvert k \rvert} \frac{1}{C(k)} \int_{\beta(k)} \exp\left\{ \ell_n \left(\beta(k); \xi(a_0), a(a_0)) \right) \right\} \ d \beta(k) , \label{eq:posterior_model_k_with_normconst}
\end{align}
with $C(k)$ being the normalization constant coming from the prior over $\beta$
\[
C(k)  =  \int_{\beta(k)} \exp\left\{ a_0 \tau \left[\xi_0^\top X_k \beta(k) - J^\top b(X_k \beta(k)) \right] \right\}  d\beta(k) \ . 
\]
Note that this constant is different for each considered model. This dependence is irrelevant when $k$ is fixed, but it is essential when different models are compared.
If $\beta(k)$ is the active sub-vector coefficients and we define 
\begin{equation}
(\beta(k), a_0) \mapsto \ell_n(\beta(k); \xi_0, a_0) = a_0 \tau (\xi_0^\top X_k \beta(k) - J^\top b(X_k \beta(k))) \ ,
\end{equation}
with $a_0 = a_0(n)$ such that $a_0(n) \to 0$ as $n\to \infty$, then we can write:
\[
C(k) =  \int_{\beta(k)} \exp\left\{ \ell_n(\beta(k); \xi_0, a_0)  \right\} d\beta(k) \ . 
\]
By Taylor expansion (with Lagrange remainder) of $\ell_n(\beta(k); \xi_0, a_0)$ around a $(\beta_{\dagger}, 0)$, with $\beta_{\dagger}$ being a generic $\beta$
and $\tilde{\beta}$ and $\tilde{a}_0$ such that $\lVert \tilde{\beta} - \beta_{\dagger} \rVert \le \lVert \beta(k) - \beta_{\dagger} \rVert$  and $\vert \tilde{a}_0 \vert \le \vert a_0 \vert$:
\begin{align*}
\ell_n(\beta(k); \xi_0, a_0) & = \left. \begin{bmatrix}
    \frac{\partial}{\partial \beta} \ell_n(\beta(k); \xi_0, a_0) \\
    \frac{\partial}{\partial a_0} \ell_n(\beta(k); \xi_0, a_0) \\
\end{bmatrix} \right \rvert_{\beta(k)= \beta_{\dagger}, a_0=0}^\top
\begin{pmatrix}
    \beta(k) - \beta_{\dagger} \\
    a_0 - 0
\end{pmatrix} + \\
&\qquad - \frac{1}{2} \begin{pmatrix}
    \beta(k) - \beta_{\dagger} \\
    a_0 - 0
\end{pmatrix}^\top 
\left.\begin{bmatrix}
    \bar{H}_{\beta} & \bar{H}^\top_{a_0\beta} \\
    \bar{H}_{a_0\beta} & \bar{H}_{a_0}
\end{bmatrix}\right \rvert_{\beta(k)= \tilde{\beta}, a_0=\tilde{a_0}}  
\begin{pmatrix}
    \beta(k) - \beta_{\dagger} \\
    a_0 - 0
\end{pmatrix}
\end{align*}
with 
\begin{align*}
\bar{s}_n(\beta_{\dagger}; \xi_0, 0) :&=  \left. \frac{\partial}{\partial \beta} \ell_n(\beta(k); \xi_0, a_0) \right \vert_{\beta(k)=\beta_{\dagger}, a_0=0} 
= \left. a_0 \tau X_k^\top (\xi_0 - \mu(\beta))  \right \vert_{\beta(k)=\beta_{\dagger}, a_0=0} \\
&= 0   \\
\bar{s}_n^*(\beta_{\dagger}; \xi_0, 0) :&= \left. \frac{\partial}{\partial a_0} \ell_n(\beta(k); \xi_0, a_0) \right \vert_{\beta(k)=\beta_{\dagger},  a_0=0} 
= \left. \tau \left[ \xi_0^\top X_k \beta(k) - J^\top b(X_k \beta(k)) \right]  \right \vert_{\beta(k)=\beta_{\dagger},  a_0=0} \\
& = \tau \left[\xi_0^\top X_k \beta_{\dagger} - J^\top b(X_k \beta_{\dagger})\right] 
\end{align*} 
\begin{align*}
\bar{H}_{a_0} :&= - \frac{\partial^2}{\partial a_0^2} \ell_n(\beta(k); \xi_0, a_0)   
=  - \frac{\partial}{\partial a_0}  \tau (\xi_0^\top X_k \beta(k) - J^\top b(X_k \beta(k))) \\
&= 0 \\
\bar{H}_{a_0 \beta}^\top :&= - \frac{\partial^2}{\partial a_0  \partial \beta} \ell_n(\beta(k); \xi_0, a_0) 
= - \frac{\partial}{\partial a_0}   a_0 \tau X_k^\top (\xi_0 - \mu(\beta(k)))  \\ 
&= \tau X_k^\top (\mu(\beta(k)) - \xi_0) \\
\bar{H}_{\beta} :&= - \frac{\partial^2}{\partial \beta(k)\partial \beta(k)^\top } \ell_n(\beta(k); \xi_0, a_0) 
= \frac{\partial}{\partial \beta} a_0 \tau X_k^\top (\xi_0 - \mu(\beta(k))) \\
&=  a_0 \tau X_k^\top \Sigma(\beta(k)) X_k 
\end{align*}
where $\Sigma(\beta(k)) = \operatorname{diag}(b^{\prime \prime}(X_k \beta(k)))$. Notice that we use the notation $\bar{s}_n, \bar{s}_n^*, \bar{H}$ to distinguish them from the ones that refer to the posterior defined in the Eqs.\eqref{eq:hessian_beta_posterior}-\eqref{eq:hessian_a0_posterior}-\eqref{eq:hessian_a0beta_posterior}.
Going back to the Taylor expansion:
\begin{align*}
\ell_n(\beta(k); \xi_0, a_0) & = a_0\bar{s}_n^*(\beta_{\dagger}; \xi_0, 0) - \frac{1}{2} \left( \beta(k) - \beta_{\dagger} \right)^\top \bar{H}_{\tilde{\beta}} \left( \beta(k) - \beta_{\dagger} \right) 
- a_0 \bar{H}_{\tilde{a_0}\tilde{\beta}}\left( \beta(k) - \beta_{\dagger} \right) \\
& = a_0 \tau \left[\xi_0^\top X_k \beta_{\dagger} - J^\top b(X_k \beta_{\dagger})\right] - \frac{a_0 \tau}{2} \left( \beta(k) - \beta_{\dagger} \right)^\top  X_k^\top \Sigma(\tilde{\beta}) X_k \left( \beta(k) - \beta_{\dagger} \right) \\
& \qquad +  a_0 \tau (\beta(k) - \beta_{\dagger})^\top X_k^\top (\xi_0 - \mu(\tilde{\beta})) \ . 
\end{align*}
We will use this Taylor expansion in the following theorem, to account for the normalizing constant for different models.  

\begin{theorem}
\label{thm:m1:sum_pr_to_0}
    Consider the models  $M_1 = \{ k: k \supsetneq t, \vert k \vert \le p \}$. Define 
    \begin{equation*}
    PR(k,t) = \frac{P(Z = k \mid y, X)}{P(Z = t \mid y, X)} \ .
    \end{equation*}
    Under Conditions \ref{cond:1}--\ref{cond:4}  and the conditions of Lemma \ref{lemma:inequality_bn} it holds that
    \begin{equation*}
    \sum_{k \in M_1} PR(k,t) \rightarrow 0 \quad \text{ as } n \rightarrow +\infty 
    \end{equation*}
    with probability tending to $1$.
\end{theorem}
\begin{proof}
    Consider the posterior distribution of $Z$ given by Eq.~\eqref{eq:posterior_model_k_with_normconst}:
    \begin{equation*}
    P(Z \mid y, X) \propto \left( \frac{p_n}{\alpha} \right)^{p_n -\lvert k \rvert} \frac{1}{C(k)} \int_{\beta(k)} \exp\left\{ \ell_n \left(\beta(k); \xi(a_0), a(a_0)) \right) \right\} \ d \beta(k) .  \ 
    \end{equation*}

    Consider the Taylor expansion in Eq. \eqref{eq:loglik_taylor_expansion} centered in $(\beta_{\dagger}, a_0) = (\hat{\beta}(k), 0)$. For some $\tilde{\beta}$ such that $\lvert  \tilde{\beta} - \hat{\beta}(k)\ \rvert \le \lvert  {\beta} - \hat{\beta}(k) \rvert$, we have: 
    \begin{align*}
    \ell_n(\beta(k);&  \xi(a_0), a(a_0)) - \ell_n(\hat{\beta}(k);  \xi(0), a(0)) = \\
    &= \left( \beta(k) - \hat{\beta}(k) \right)^\top s_n(\hat{\beta}(k)) 
    +  a_0s_n^*( \hat{\beta}(k)) - \frac{1}{2} \left( \beta(k) - \hat{\beta}(k) \right)^\top H_{\tilde{\beta}} \left( \beta(k) - \hat{\beta}(k) \right) \\
    &\qquad
    - a_0 H_{\tilde{a_0}\tilde{\beta}}\left( \beta(k) - \hat{\beta}(k) \right) \\
    &=a_0s_n^*( \hat{\beta}(k)) - \frac{1}{2} \left( \beta(k) - \hat{\beta}(k) \right)^\top H_{\tilde{\beta}} \left( \beta(k) - \hat{\beta}(k) \right) - a_0 H_{\tilde{a_0}\tilde{\beta}}\left( \beta(k) - \hat{\beta}(k) \right) \\
    \end{align*}
    Using Lemmas \ref{lemma:hessian}, \ref{lemma:hessian_abeta}, and \ref{lemma:bound_beta}, for any model $k \in M_1$, we have:
    \begin{align}
    \ell_n (\beta(k);  \xi(a_0), a(a_0)) - \ell_n(\hat{\beta}(k);  \xi(0), a(0))  &\le  a_0s_n^*( \hat{\beta}(k)) + \label{eq:inequality_for_integral} \\
    &\quad   - \frac{1-\varepsilon}{2} \left( \beta(k) - \hat{\beta}(k) \right)^\top H_{\beta_0(k)} \left( \beta(k) - \hat{\beta}(k) \right) + \nonumber \\
     &\quad   - a_0 (1-\epsilon) H_{0\beta_0(k)}\left( \beta(k) - \hat{\beta}(k) \right)  \nonumber 
    \end{align}
    for all $\beta(k)$ such that $\lvert  \beta(k) - \beta_0(k)  \rvert < c_n$. Define the set $B = \left\{ \beta(k) : \lvert  \beta(k) - \hat{\beta}(k)  \rvert \le \frac{c_n}{2}\right\}$. For Lemma \ref{lemma:bound_beta}, $B \subset \left\{ \beta(k) : \lvert \beta(k) - \beta_0(k) \rvert \le c_n\right\}$ with probability going to $1$, for any $k \in M_1$.
    For large enough $n$, conditions \ref{cond:2:c} and \ref{cond:2:d}, imply respectively that 
    \begin{equation*}
     a_0s_n^*( \hat{\beta}(k))  \rightarrow 0 \ 
    \end{equation*}
    \begin{equation*}
    a_0(1-\epsilon) H_{0\beta_0(k)}\left( \beta(k) - \hat{\beta}(k) \right) \rightarrow 0 \ ;  
    \end{equation*}
    additionally
    \begin{align*}
        (1-\varepsilon)\left( \beta(k) - \hat{\beta}(k) \right)^\top H_{\beta_0(k)} \left( \beta(k) - \hat{\beta}(k) \right) 
        &\overset{\ref{cond:2:b}}{\le} (1-\varepsilon)\frac{c_n^2}{4}u^\top H_{\beta_0(k)} u \ \le \ (1-\varepsilon)\frac{c_n^2}{4} n \lambda u^\top u\\
        &\overset{\ref{cond:2:b}}{=} (1-\varepsilon)\frac{1}{4} n \lambda C \frac{\lvert k \rvert \Lambda_{\lvert k \rvert} (\log p_n)^\omega}{n^\gamma \lambda^2 (1- \varepsilon)^2} \\
        &\overset{Eq. \eqref{eq:rate_cn}}{=} \frac{1}{4} C \Lambda_{\lvert k \rvert} \frac{n^{(\phi \omega - \gamma)(1-\delta) + 1}}{\lambda (1- \varepsilon)} \ \rightarrow + \infty
    \end{align*}
    since by Remark \ref{rem:convergence_rates}, it holds: $-1 < (\phi \omega - \gamma)(1-\delta)  < 0$.
    Hence, with all the conditions above, $\forall \beta(k) \in B$
    \begin{equation}
    \label{eq:diff_ell_goes_minus_infty}
    \ell_n(\beta(k);  \xi(a_0), a(a_0)) - \ell_n(\hat{\beta}(k);  \xi(0), a(0)) \rightarrow - \infty
    \end{equation}
    By concavity of the log-likelihood $\ell_n$, Eq. \eqref{eq:diff_ell_goes_minus_infty}
    holds also for $\beta(k) \in B^c$. As a direct consequence, we can always assume that the difference of the log--likelihoods computed at $\beta(k) \in B^c$ is always less than the one computed at $\beta(k) \in B$; and therefore the inequality in Eq. \eqref{eq:inequality_for_integral} is valid for all $\beta(k)$. At this point
    {\small \begin{align*}
        \int_{\beta(k)} &\exp\left\{ \ell_n \left(\beta(k); \xi(a_0), a(a_0)) \right) \right\}  \ d \beta(k) \\
        &\le  \exp\left\{ \ell_n \left(\hat{\beta}(k); \xi(0), a(0)) \right) \right\} \times \\
        & \quad  \times \exp\left\{ a_0s_n^*( \hat{\beta}(k)) \right\} \times \\
        & \quad \times  \int_{\beta(k)} \exp\left\{   - \frac{1-\varepsilon}{2} \left( \beta(k) - \hat{\beta}(k) \right)^\top H_{\beta_0(k)} \left( \beta(k) - \hat{\beta}(k) \right) 
        - a_0(1-\epsilon) H_{0\beta_0(k)}\left( \beta(k) - \beta_0(k) \right)  \right\}  d \beta(k) \\
        &\approx \exp\left\{ \ell_n \left(\hat{\beta}(k); \xi(0), a(0)) \right) \right\} \times  \int_{\beta(k)} \exp\left\{   - \frac{1-\varepsilon}{2} \left( \beta(k) - \hat{\beta}(k) \right)^\top H_{\beta_0(k)} \left( \beta(k) - \hat{\beta}(k) \right)  \right\}  \ d \beta(k) \\
        &= \exp\left\{ \ell_n \left(\hat{\beta}(k); \xi(0), a(0)) \right) \right\} \times \operatorname{det}\left( (1-\varepsilon)H_{\beta_0(k)}\right)^{-1/2} (2\pi)^{\lvert k \rvert/2}  \\
        &= \exp\left\{ \ell_n \left(\hat{\beta}(k); \xi(0), a(0)) \right) \right\} \times \operatorname{det}\left( (1-\varepsilon)  \left[ \frac{1 + A(\tau)\tau a_0}{A(\tau)} X_k^\top \Sigma X_k \right]  \right)^{-1/2} (2\pi)^{\lvert k \rvert/2} 
    \end{align*}}
    Following similar arguments, for $n$ large enough:
    {\small \begin{align*}
       C(k) &= \int_{\beta(k)} \exp\left\{ \DparamScalar_0 \tau \left[\DparamVector_0^\top X_k \beta(k) - J^\top b(X_k \beta(k)) \right] \right\} d\beta(k) \ \\
        &= \int_{\beta(k)} \exp\left\{ \ell_n(\beta(k); \xi_0, a_0)  \right\} d\beta(k) \ \\
        &\ge  \exp\left\{ \ell_n \left(\beta_0(k); \xi_0, 0) \right) \right\}
        \times \exp\left\{ a_0\bar{s}_n^*( {\beta}_0(k)) \right\} \times \\
        & \quad \times  \int_{\beta(k)} \exp\left\{  - \frac{1+\varepsilon}{2} \left( \beta(k) - {\beta}_0(k) \right)^\top \bar{H}_{\beta_0(k)} \left( \beta(k) - {\beta}_0(k) \right) 
        - a_0(1+\epsilon) \bar{H}_{0\beta_0(k)}\left( \beta(k) - \beta_0(k) \right)  \right\}  d \beta(k) \\
        &\approx \int_{\beta(k)} \exp\left\{   - \frac{1+\varepsilon}{2} \left( \beta(k) - {\beta_0}(k) \right)^\top \bar{H}_{\beta_0(k)} \left( \beta(k) - {\beta_0}(k) \right)  \right\}  \ d \beta(k) \\
        &=  \operatorname{det}\left( (1+\varepsilon)\bar{H}_{\beta_0(k)}\right)^{-1/2} (2\pi)^{\lvert k \rvert/2} \\
        &= \operatorname{det}\left( (1+\varepsilon) \left[  a_0 \tau X_k^\top \Sigma X_k  \right]\right)^{-1/2} (2\pi)^{\lvert k \rvert/2}  \ .  
    \end{align*}}
    Hence
    \begin{align*}
    P(& Z = k  \mid y, X) \\
    &\propto  \left( \frac{p_n}{\alpha} \right)^{p_n -\lvert k \rvert}  \frac{1}{C(k)}\int_{\beta(k)} \exp\left\{ \ell_n \left(\beta(k); \xi(a_0), a(a_0)) \right) \right\}  \ d \beta(k)  \ \\
    &\lesssim  \left( \frac{p_n}{\alpha} \right)^{p_n -\lvert k \rvert} \exp\left\{ \ell_n \left(\hat{\beta}(k); \xi(0), a(0)) \right) \right\} \times\frac{ \operatorname{det}\left( (1-\varepsilon)H_{\beta_0(k)}\right)^{-1/2} (2\pi)^{\lvert k \rvert/2} }{\operatorname{det}\left( (1+\varepsilon)\bar{H}_{\beta_0(k)}\right)^{-1/2} (2\pi)^{\lvert k \rvert/2} } \\
    &=  \left( \frac{p_n}{\alpha} \right)^{p_n -\lvert k \rvert} \exp\left\{ \ell_n \left(\hat{\beta}(k); \xi(0), a(0)) \right) \right\} \times \left( \frac{1-\varepsilon}{1+\varepsilon}\right)^{-|k|/2}  \times \frac{\left( \frac{1 + A(\tau)\tau a_0}{A(\tau)}\right)^{-|k|/2} }{\left( a_0 \tau\right)^{-|k|/2} } \ . 
    \end{align*}
    In a similar way, we get 
    \begin{align*}
    P(& Z = t  \mid y, X) \\
    &\gtrsim  \left( \frac{p_n}{\alpha} \right)^{p_n -\lvert t \rvert} \exp\left\{ \ell_n \left(\hat{\beta}(t); \xi(0), a(0)) \right) \right\} \times \frac{ \operatorname{det}\left( (1+\varepsilon)H_{\beta_0(t)}\right)^{-1/2} (2\pi)^{\lvert t \rvert/2}  }{\operatorname{det}\left( (1-\varepsilon)\bar{H}_{\beta_0(t)}\right)^{-1/2} (2\pi)^{\lvert t \rvert/2} } \\
    &=  \left( \frac{p_n}{\alpha} \right)^{p_n -\lvert t \rvert} \exp\left\{ \ell_n \left(\hat{\beta}(t); \xi(0), a(0)) \right) \right\} \times \left( \frac{1+\varepsilon}{1-\varepsilon}\right)^{-|t|/2}  \times \frac{\left( \frac{1 + A(\tau)\tau a_0}{A(\tau)}\right)^{-|t|/2} }{\left( a_0 \tau\right)^{-|t|/2} }  \ . 
    \end{align*}
    We now take the ratio
    {\small \begin{align*}
        PR(k,t) 
        &:= \frac{P(Z = k \mid y, X)}{P(Z = t \mid y, X)}  \\
        &\lesssim \left( \frac{p_n}{\alpha} \right)^{\lvert t \rvert-\lvert k \rvert}  \frac{ \exp\left\{ \ell_n \left(\hat{\beta}(k); \xi(0), a(0)) \right) \right\} }{ \exp\left\{ \ell_n \left(\hat{\beta}(t); \xi(0), a(0)) \right) \right\} } \frac{\left( \frac{1 + A(\tau)\tau a_0}{A(\tau)}\right)^{-|k|/2} }{\left( a_0 \tau\right)^{-|k|/2} } \frac{\left( \frac{1 + A(\tau)\tau a_0}{A(\tau)}\right)^{|t|/2} }{\left( a_0 \tau\right)^{|t|/2} } \left( \frac{1+\varepsilon}{1-\varepsilon}\right)^{( \lvert k \rvert + \lvert t \rvert ) /2} \\
        &= C \left( \frac{p_n}{\alpha} \right)^{\lvert t \rvert-\lvert k \rvert} \exp\left\{ \ell_n \left(\hat{\beta}(k); \xi(0), a(0)) \right) - \ell_n \left(\hat{\beta}(t); \xi(0), a(0)) \right) \right\}  \left( \frac{ \frac{1 + A(\tau)\tau a_0}{A(\tau)} }{a_0 \tau} \right)^{(\lvert t \rvert - \lvert k \rvert)/2}  \\
        &\overset{(*)}{\le}  C \left( \frac{p_n}{\alpha} \right)^{\lvert t \rvert-\lvert k \rvert} \exp\left\{ b_n(\lvert k \rvert- \lvert t \rvert) + C^\prime \right\}  \left(  \frac{A(\tau)\tau a_0}{1 + A(\tau)\tau a_0}  \right)^{(\lvert k \rvert - \lvert t \rvert)/2} \tag{\theequation}\label{eq:inequality_Pkt_use_lemma_bn}\\
        &=  C \left( \frac{p_n}{\alpha} \right)^{\lvert t \rvert-\lvert k \rvert} \exp\left\{ \frac{\sigma^2}{2(1-\varepsilon)}  \tilde{\lambda} \log p_n n^{-\psi}(\lvert k \rvert- \lvert t \rvert) + C^\prime \right\}   \left(  \frac{A(\tau)\tau a_0}{1 + A(\tau)\tau a_0}  \right)^{(\lvert k \rvert - \lvert t \rvert)/2} 
    \end{align*}}
    where (*) holds for any $k \in M_1$ with probability going to $1$, and it is obtained using Lemma \ref{lemma:inequality_bn}.
    Moreover, with $a_0 = o(c_n^2)$, or equivalently $o(n^{(\phi \omega- \gamma)(1-\delta)})$, it holds that for large enough $n$:
    \[
     \left(  \frac{A(\tau)\tau a_0}{1 + A(\tau)\tau a_0}  \right)^{\frac{\lvert k \rvert-\lvert t \rvert}{2}}  \le
     \left(A(\tau)\tau a_0 \right)^{\frac{\lvert k \rvert-\lvert t \rvert}{2}} \lesssim
     \left(A(\tau)\tau n^{-1} \right)^{\frac{\lvert k \rvert-\lvert t \rvert}{2}} \ .
    \]
    Hence with probability going to $1$
    \begin{align*}
        PR(k,t) 
        &\le  C^{\prime\prime} \left( \frac{p_n}{\alpha} \right)^{\lvert t \rvert-\lvert k \rvert} {p_n}^{\frac{\sigma^2}{2(1-\varepsilon)}  \tilde{\lambda} n^{-\psi}(\lvert k \rvert- \lvert t \rvert)}  \left(A(\tau)\tau \right)^{\frac{\lvert k \rvert-\lvert t \rvert}{2}} n^{-\frac{\lvert k \rvert-\lvert t \rvert}{2}} \\
        &=  C^{\prime\prime} \left( \alpha \sqrt{A(\tau)\tau} \right)^{\lvert k \rvert-\lvert t \rvert}  {p_n}^{-(\lvert k \rvert-\lvert t \rvert)}  {p_n}^{C n^{-\psi}(\lvert k \rvert- \lvert t \rvert)} n^{-\frac{\lvert k \rvert-\lvert t \rvert}{2}}  \\
        &=  C^{\prime\prime} \left( \alpha \sqrt{A(\tau)\tau} \right)^{\lvert k \rvert-\lvert t \rvert} {p_n}^{(\lvert k \rvert- \lvert t \rvert)(Cn^{-\psi} - 1)} n^{-\frac{\lvert k \rvert-\lvert t \rvert}{2}} \\
        &\overset{(*)}{\approx}  C^{\prime\prime} \left( \alpha \sqrt{A(\tau)\tau} \right)^{\lvert k \rvert-\lvert t \rvert} {p_n}^{-(\lvert k \rvert- \lvert t \rvert)} n^{-\frac{\lvert k \rvert-\lvert t \rvert}{2}} 
    \end{align*}
    where the last approximation $(*)$ holds for large enough $n$.  Finally,
    \begin{align*}
        \sum_{k \in M_1} PR(k,t) 
        &\lessapprox \sum_{k \in M_1}  C^{\prime\prime} \left( \alpha \sqrt{A(\tau)\tau} \right)^{\lvert k \rvert-\lvert t \rvert} {p_n}^{-(\lvert k \rvert- \lvert t \rvert)} n^{-\frac{\lvert k \rvert-\lvert t \rvert}{2}} \\
        &= \sum_{\lvert k \rvert=\lvert t \rvert + 1}^{m_n} \binom{p_n - \lvert t \rvert}{\lvert k \rvert - \lvert t \rvert} C^{\prime\prime} \left( \alpha \sqrt{A(\tau)\tau} \right)^{\lvert k \rvert-\lvert t \rvert} {p_n}^{-(\lvert k \rvert- \lvert t \rvert)} n^{-\frac{\lvert k \rvert-\lvert t \rvert}{2}}.
    \end{align*}
    Using $\binom{p - \lvert t \rvert}{\lvert k \rvert - \lvert t \rvert}  \le p^{\lvert k \rvert-\lvert t \rvert}$, and recalling that $\alpha$ is a fixed hyperparameter, and $A(\tau), \tau$ are fixed:
    \begin{align*}
        \sum_{k \in M_1} PR(k,t) 
        &\lessapprox C \sum_{\lvert k \rvert=\lvert t \rvert + 1}^{m_n} \left( \alpha \sqrt{A(\tau)\tau} \right)^{\lvert k \rvert-\lvert t \rvert} n^{-\frac{\lvert k \rvert-\lvert t \rvert}{2}} \\
        &= C \sum_{\lvert k \rvert-\lvert t \rvert = 1}^{m_n-\lvert t \rvert} \left(\frac{\left( \alpha \sqrt{A(\tau)\tau} \right)}{n^{1/2}}\right)^{\lvert k \rvert-\lvert t \rvert} \\
        &= C \frac{\left(\frac{ \alpha \sqrt{A(\tau)\tau} }{n^{1/2}}\right)^{m_n-\lvert t \rvert + 1} - \frac{ \alpha \sqrt{A(\tau)\tau} }{n^{1/2}}}{\frac{ \alpha \sqrt{A(\tau)\tau} }{n^{1/2}} - 1 } \\
        &\approx   C \frac{n^{- \frac{m_n-\lvert t \rvert}{2}} - 1}{ 1 - n^{1/2}} \  \longrightarrow 0 \text{ as } n\rightarrow + \infty
    \end{align*}
    and this concludes the proof.
\end{proof}

\subsection{Large models}
\label{apd:subsec:m2}
We are going to prove an analogous result of Theorem  \ref{thm:m1:sum_pr_to_0} for large models ($M_2$). 

\begin{theorem}
\label{thm:m2:sum_pr_to_0}
    Consider the models  $M_2 = \{ k: k \nsupseteq t, \vert t \vert < \vert k \vert \le m_n \}$. 
    Under Conditions \ref{cond:1}--\ref{cond:4} it holds that
    \begin{equation*}
    \sum_{k \in M_2} PR(k,t) \rightarrow 0 \quad \text{ as } n \rightarrow +\infty 
    \end{equation*}
    with probability tending to $1$.
\end{theorem}
\begin{proof}
For any $k \in M_2$, let $k^* = k \cup t$. In this way $k^* \in M_1$ and all the results proved in Section \ref{apd:subsec:m1} are valid for $k^*$, 
in particular Lemma \ref{lemma:inequality_bn} implies that
\begin{align*}
    \ell_n \left(\hat{\beta}(k); \xi(0), a(0)) \right) - \ell_n \left(\hat{\beta}(t); \xi(0), a(0)) \right)  
    & \le \ell_n \left(\hat{\beta}(k^*); \xi(0), a(0)) \right) - \ell_n \left(\hat{\beta}(t); \xi(0), a(0)) \right) \\
    &\le b_n(\lvert k^* \rvert- \lvert t \rvert) + C^\prime 
\end{align*}
for some constant $C^\prime$, with probability going to 1. Therefore for some constant $C$,
{\small 
 \begin{align*}
    PR(k,t) 
        &\le C \left( \frac{p_n}{\alpha} \right)^{\lvert t \rvert-\lvert k \rvert} \exp\left\{ \ell_n \left(\hat{\beta}(k); \xi(0), a(0)) \right) - \ell_n \left(\hat{\beta}(t); \xi(0), a(0)) \right) \right\}  \left( \frac{ \frac{1 + A(\tau)\tau a_0}{A(\tau)} }{a_0 \tau} \right)^{(\lvert t \rvert - \lvert k \rvert)/2}  \\
        &\le C \left( \frac{p_n}{\alpha} \right)^{\lvert t \rvert-\lvert k \rvert} \exp\left\{ \ell_n \left(\hat{\beta}(k^*); \xi(0), a(0)) \right) - \ell_n \left(\hat{\beta}(t); \xi(0), a(0)) \right) \right\}  \left( \frac{ \frac{1 + A(\tau)\tau a_0}{A(\tau)} }{a_0 \tau} \right)^{(\lvert t \rvert - \lvert k \rvert)/2}  \\
        &\overset{(*)}{\le}  C \left( \frac{p_n}{\alpha} \right)^{\lvert t \rvert-\lvert k \rvert} \exp\left\{ b_n(\lvert k^*\rvert- \lvert t \rvert) + C^\prime \right\}  \left(  \frac{A(\tau)\tau a_0}{1 + A(\tau)\tau a_0}  \right)^{(\lvert k \rvert - \lvert t \rvert)/2} \\
    &\le C \left(\frac{p_n}{\alpha}\right)^{\lvert t \rvert - \lvert k \rvert} {p_n}^{\frac{\sigma^2}{2(1-\varepsilon)} \tilde{\Lambda}_m n^{-\psi}(\lvert k^* \rvert - \lvert t \rvert)} n^{-\frac{\lvert k \rvert - \lvert t \rvert}{2}} \left(A(\tau)\tau \right)^{\frac{\lvert k \rvert-\lvert t \rvert}{2}}
\end{align*} 
}
and, being $\lvert k^* \rvert = \lvert k \rvert + \lvert t \rvert - \lvert k\cap t \rvert$, we have
\begin{align*}
    \sum_{k \in M_2} PR(k,t) &\le
    \sum_{k \in M_2} C \left(\frac{p_n}{\alpha}\right)^{\lvert t \rvert - \lvert k \rvert} {p_n}^{\frac{\sigma^2}{2(1-\varepsilon)} \tilde{\Lambda}_m n^{-\psi}(\lvert k^* \rvert - \lvert t \rvert)} n^{-\frac{\lvert k \rvert - \lvert t \rvert}{2}}\left(A(\tau)\tau \right)^{\frac{\lvert k \rvert-\lvert t \rvert}{2}} \\
    &=
    \sum_{k \in M_2} C \left(\frac{\alpha \sqrt{A(\tau) \tau}}{n^{1/2}}\right)^{\lvert k \rvert - \lvert t \rvert} 
    {p_n}^{\lvert t \rvert - \lvert k \rvert} {p_n}^{C n^{-\psi}(\lvert k \rvert - \lvert k \cap t \rvert)}   \\
    &=
    \sum_{k \in M_2} C \left(\frac{\alpha \sqrt{A(\tau) \tau}}{n^{1/2}}\right)^{\lvert k \rvert - \lvert t \rvert} 
    {p_n}^{(C n^{-\psi} - 1 )\lvert k \rvert + \lvert t \rvert -C n^{-\psi} \lvert k \cap t \rvert}  \\
    &\le
    \sum_{k \in M_2} C \left(\frac{\alpha \sqrt{A(\tau) \tau}}{n^{1/2}}\right)^{\lvert k \rvert - \lvert t \rvert} 
    {p_n}^{(C n^{-\psi} - 1 )\lvert k \rvert + \lvert t \rvert} \\
    &\lessapprox  \sum_{k \in M_2} C \left(\frac{\alpha \sqrt{A(\tau) \tau}}{n^{1/2}}\right)^{\lvert k \rvert - \lvert t \rvert} 
    {p_n}^{-(\lvert k \rvert + \lvert t \rvert)}
\end{align*}
which yields directly to 
\begin{equation*}
    \sum_{k \in M_2} PR(k,t) \rightarrow 0 \quad \text{ as } n \rightarrow +\infty
\end{equation*}
with probability $1$, as in the Proof of Theorem \ref{thm:m1:sum_pr_to_0}. 
\end{proof}

\subsection{Under-fitted models}
\label{apd:subsec:m3}
We are going to prove an analogous result of Theorems~\ref{thm:m1:sum_pr_to_0}-~\ref{thm:m2:sum_pr_to_0} for under-fitted models ($M_3$).

\begin{theorem}
\label{thm:m3:sum_pr_to_0}
Consider the models  $M_3 = \{ k: k \nsupseteq t, \vert k \vert \le \lvert t \rvert \}$. Assume that $a_0=o({w_n}^2)$ with $w_n = \sqrt{\frac{2\lvert t \rvert \Lambda_{2\lvert t \rvert} (\log p_n)^\omega}{n^\gamma}}$. Under Conditions \ref{cond:1}--\ref{cond:4} it holds that
\begin{equation*}
  \sum_{k \in M_3} PR(k,t) \rightarrow 0 \quad \text{ as } n \rightarrow +\infty
\end{equation*}
with probability tending to $1$.
\end{theorem}
\begin{proof}
For any $k\in M_3$, we can define $k^* = k \cup t$ in such a way that $k^* \in M_1$, therefore applying Lemma \ref{lemma:bound_beta}, and noticing that $\lvert k^* \rvert = \lvert k \rvert + \lvert t \rvert  - \lvert k\cap t \rvert \le 2\lvert t \rvert $, we have
\begin{equation*}
\lVert \hat{\beta}(k^*; a_0)  - \beta_0(k^*) \rVert = C \sqrt{\frac{2\lvert t \rvert \Lambda_{2\lvert t \rvert} (\log p_n)^\omega}{n^\gamma}} =: C w_n
\end{equation*}
for some constant $C$. Assume also that $ \lVert \beta(k^*)  - \beta_0(k^*) \rVert = C w_n $. Therefore $\lVert {\beta}(k^*) - \hat{\beta}(k^*) \rVert \le 2C w_n $.
\\
Now consider 
{\small 
\begin{align*}
    &\ell_n({\beta}(k^*); \xi(a_0), a(a_0))  -  \ell_n(\beta_0(k^*); \xi(0), a(0)) = \\
    &= \left[ \ell_n({\beta}(k^*); \xi(a_0), a(a_0)) - \ell_n(\hat{\beta}(k^*);  \xi(a_0), a(a_0))  \right] 
    - \left[ \ell_n (\beta_0(k^*); \xi(0), a(0)) - \ell_n(\hat{\beta}(k^*);  \xi(a_0), a(a_0))  \right] \\
    &= \left[ \left( \beta(k^*) - \hat{\beta}(k^*) \right)^\top s_n(\hat{\beta}(k^*); a_0) +0 -\frac{1}{2}\left( \beta(k^*) - \hat{\beta}(k^*) \right)^\top  H_{\tilde{\beta}}\left( \beta(k^*) - \hat{\beta}(k^*) \right) - 0 \right]  + \\
    &\qquad - \left[ \left( \beta_0(k^*) - \hat{\beta}(k^*) \right)^\top s_n(\hat{\beta}(k^*); a_0) +a_0 s_n^*(\hat{\beta}(k^*);a_0) -\frac{1}{2}\left( \beta_0(k^*) - \hat{\beta}(k^*) \right)^\top  H_{\tilde{\tilde{\beta}}}\left( \beta_0(k^*) - \hat{\beta}(k^*) \right) + \right. \\
    & \qquad \left. -  a_0 H_{\tilde{\tilde{a_0}}\tilde{\tilde{\beta}}}\left( \beta_0(k^*) - \hat{\beta}(k^*) \right) \right] = \\
    &= \frac{1}{2}\left( \beta_0(k^*) - \hat{\beta}(k^*) \right)^\top  H_{\tilde{\tilde{\beta}}}\left( \beta_0(k^*) - \hat{\beta}(k^*) \right) -\frac{1}{2}\left( \beta(k^*) - \hat{\beta}(k^*) \right)^\top  H_{\tilde{\beta}}\left( \beta(k^*) - \hat{\beta}(k^*) \right) + \\
    &\qquad - a_0 s_n^*(\hat{\beta}(k^*);a_0) - a_0 H_{\tilde{\tilde{a_0}}\tilde{\tilde{\beta}}}\left( \beta_0(k^*) - \hat{\beta}(k^*) \right) \\
     &\le \frac{1+\varepsilon}{2}\left( \beta_0(k^*) - \hat{\beta}(k^*) \right)^\top  H_{\beta_0}\left( \beta_0(k^*) - \hat{\beta}(k^*) \right) -\frac{1-\varepsilon}{2}\left( \beta(k^*) - \hat{\beta}(k^*) \right)^\top  H_{\beta_0}\left( \beta(k^*) - \hat{\beta}(k^*) \right) + \\
    &\qquad - a_0 s_n^*(\hat{\beta}(k^*);a_0) - a_0 (1+\epsilon) H_{0\beta_0(k^*)}\left( \beta_0(k^*) - \hat{\beta}(k^*) \right)\\
    &\le \frac{1+\varepsilon}{2}C^2 w_n^2 n \lambda_{\max}-\frac{1-\varepsilon}{2}4C^2 w_n^2 n \lambda_{\min} 
    - a_0 s_n^*(\hat{\beta}(k^*);a_0) - a_0(1+\epsilon) H_{0\beta_0(k^*)}\left( \beta_0(k^*) - \hat{\beta}(k^*) \right) \\
    &= \frac{1}{2}C^2 w_n^2 n \left[ (1+\varepsilon )\lambda_{\max}-4(1-\varepsilon) \lambda_{\min} \right] 
    - a_0 s_n^*(\hat{\beta}(k^*);a_0) - a_0(1+\epsilon) H_{0\beta_0(k^*)}\left( \beta_0(k^*) - \hat{\beta}(k^*) \right) \\
    &= \frac{1}{2}C^2 w_n^2 n \left[ (1+\varepsilon )\lambda_{\max}-4(1-\varepsilon) \lambda_{\min} \right] 
    - a_0 s_n^*(\hat{\beta}(k^*);a_0) -  a_0(1+\epsilon) \tau X_{k^*}^\top \left( \mu(\beta_0(k^*)-\xi_0) \right) C w_n \\
    &= \frac{1}{2}C^2 w_n^2 n \left[ (1+\varepsilon )\lambda_{\max}-4(1-\varepsilon) \lambda_{\min} \right] + \\
    &\qquad - a_0 \tau \left[ \xi(a_0)^\top X_{k^*}\hat{\beta}(k^*) - J^\top b\left( \hat{\beta}(k^*) \right) \right]
     - a_0 \tau \left( \xi_0 - y \right)^\top X_{k^*}\hat{\beta}(k^*) \\
    &\qquad -  a_0(1+\epsilon) \tau X_{k^*}^\top \left( \mu(\beta_0(k^*)-\xi_0) \right) C w_n \\
    &= - \frac{1}{2}\tilde{C}^2 w_n^2 n + \\
    &\qquad - a_0 \tau \left[ \xi(a_0)^\top X_{k^*}\hat{\beta}(k^*) - J^\top b\left( \hat{\beta}(k^*) \right) 
     - \left( \xi_0 - y \right)^\top X_{k^*}\hat{\beta}(k^*) \right] \\
    &\qquad -  a_0(1+\epsilon) \tau X_{k^*}^\top \left( \mu(\beta_0(k^*)-\xi_0) \right) C w_n 
\end{align*}}
where $\lambda_{\min}:=\lambda_{\min}(X_{k^*}^\top X_{k^*})$, $\lambda_{\max}:=\lambda_{\max}(X_{k^*}^\top X_{k^*})$ and  the last equality holds for condition \ref{cond:2:b}:
$\left[ (1+\varepsilon )\lambda_{\max}-4(1-\varepsilon) \lambda_{\min} \right] <0$. Let's consider now the last two terms separately, remembering that by assumption $a_0 = o({w_n}^2)$. First, taking the Taylor expansion centered in $\beta_0(k^*)$ we can write
{\small
\begin{align*}
   \xi(a_0)^\top & X_{k^*}\hat{\beta}(k^*) - J^\top b\left( \hat{\beta}(k^*) \right)   - \left( \xi_0 - y \right)^\top X_{k^*}\hat{\beta}(k^*) \\
   &= \xi(a_0)^\top X_{k^*}{\beta}_0(k^*) - J^\top b\left( {\beta}_0(k^*) \right) - \left( \xi_0 - y \right)^\top X_{k^*}{\beta}_0(k^*)  + C\left( \hat{\beta}(k^*) - {\beta}_0(k^*) \right) = \\
   &= \xi_0^\top X_{k^*}{\beta}_0(k^*) - J^\top b\left( {\beta}_0(k^*) \right) - \left( \xi(a_0) - y \right)^\top X_{k^*}{\beta}_0(k^*)  + C \left( \hat{\beta}(k^*) - {\beta}_0(k^*) \right) = \\
   &\le K_2 n - \left( \xi(a_0) - y \right)^\top X_{k^*}{\beta}_0(k^*)  + C^\prime w_n = \\
   &\approx  K_2 n  + C^\prime w_n 
\end{align*}
hence}
{\small 
\begin{equation*}
\frac{\left| a_0 \tau \left[ \xi(a_0)^\top X_{k^*}\hat{\beta}(k^*) - J^\top b\left( \hat{\beta}(k^*) \right)  - \left( \xi_0 - y \right)^\top X_{k^*}\hat{\beta}(k^*) \right] \right|  }{\lvert - \frac{1}{2}\tilde{C}^2 w_n^2 n \rvert } \approx C \frac{a_0 ( n  +  w_n) }{w_n^2 n}  = C \frac{a_0  }{w_n^2} \rightarrow 0  \ . 
\end{equation*}
}
Lastly, for $n$ large enough and a constant $C$:
\begin{equation*}
\frac{\lvert a_0(1+\epsilon) \tau X_{k^*}^\top \left( \mu(\beta_0(k^*)-\xi_0) \right) C w_n \rvert}{\lvert - \frac{1}{2}\tilde{C}^2 w_n^2 n \rvert } \approx  C \frac{a_0  n w_n}{w_n^2 n}  = C \frac{a_0}{w_n} \rightarrow 0  
\end{equation*}

Therefore we get that for large enough $n$
\begin{equation*}
\ell_n\left({\beta}(k^*); \xi(a_0), a(a_0) \right)  -  \ell_n\left(\beta_0(k^*); \xi(0), a(0) \right) \lessapprox -\frac{1}{2}\tilde{C}^2 w_n^2 n 
\end{equation*}
for some positive constant $\tilde{C}$, with probability going to $1$.
\\
Define $\tilde{\beta}(k^*)= \left[ \hat{\beta}(k), \mathbf{0}_{k\cap t^C}  \right]$, so that $\ell(\tilde{\beta}(k^*); \xi(a_0), a(a_0)) = \ell(\hat{\beta}(k);\xi(a_0), a(a_0))$. Moreover by Condition \ref{cond:3}: $\lVert \tilde{\beta}(k^*)- \beta_0(k^*)\rVert > 2C w_n$, hence by concavity of $\ell_n$:
\begin{align*}
    \ell_n\left(\tilde{\beta}(k^*); \xi(a_0), a(a_0)\right)  -  \ell_n\left(\beta_0(k^*); \xi(0), a(0) \right) &\le \ell_n\left({\beta}(k^*); \xi(a_0), a(a_0) \right)  -  \ell_n\left(\beta_0(k^*); \xi(0), a(0)\right) \\&\lessapprox -\frac{1}{2}\tilde{C}^2 w_n^2 n 
\end{align*}
We now take the ratio: 
{\small 
\begin{align*}
    PR(k,t) 
    &:= \frac{P(Z = k \mid y, X)}{P(Z = t \mid y, X)}  \\
    &\le
    C \left( \frac{p_n}{\alpha} \right)^{\lvert t \rvert-\lvert k \rvert} \exp\left\{ \ell_n \left(\hat{\beta}(k); \xi(0), a(0)) \right) - \ell_n \left(\hat{\beta}(t); \xi(0), a(0)) \right) \right\}  \left( \frac{ \frac{1 + A(\tau)\tau a_0}{A(\tau)} }{a_0 \tau} \right)^{(\lvert t \rvert - \lvert k \rvert)/2} \\
    &\le
    C \left( \frac{p_n}{\alpha} \right)^{\lvert t \rvert-\lvert k \rvert} \exp\left\{ \ell_n\left(\tilde{\beta}(k^*); \xi(a_0), a(a_0)\right) - \ell_n \left(\hat{\beta}(t); \xi(0), a(0)) \right) \right\} \left(  \frac{A(\tau)\tau a_0}{1 + A(\tau)\tau a_0}  \right)^{(\lvert k \rvert - \lvert t \rvert)/2}  \\
    &\le
    C \left( \frac{p_n}{\alpha} \right)^{\lvert t \rvert-\lvert k \rvert} \exp\left\{ \ell_n\left(\tilde{\beta}(k^*); \xi(a_0), a(a_0)\right) - \ell_n \left(\hat{\beta}(t); \xi(0), a(0)) \right) \right\}   n^{\frac{\lvert t \rvert - \lvert k \rvert}{2}} \left(A(\tau)\tau \right)^{\frac{\lvert k \rvert-\lvert t \rvert}{2}} \\
    &\lessapprox
    C^\prime \left( \frac{p_n}{\alpha} \right)^{\lvert t \rvert-\lvert k \rvert} \exp\left\{ -\frac{1}{2}\tilde{C}^2 w_n^2 n  \right\}   n^{\frac{\lvert t \rvert - \lvert k \rvert}{2}} \left(A(\tau)\tau \right)^{\frac{\lvert k \rvert-\lvert t \rvert}{2}}  \\
    &=
    C^\prime \left( \frac{p_n}{\alpha} \right)^{\lvert t \rvert-\lvert k \rvert} \exp\left\{ -\frac{1}{2}C^2 2\lvert t \rvert \Lambda_{2\lvert t \rvert} (\log p_n)^\omega n^{1-\gamma}  \right\}   n^{\frac{\lvert t \rvert - \lvert k \rvert}{2}} \left(A(\tau)\tau \right)^{\frac{\lvert k \rvert-\lvert t \rvert}{2}} \\
    &=
    C^\prime \left( \frac{n^{1/2}}{\alpha \sqrt{A(\tau) \tau} } \right)^{\lvert t \rvert-\lvert k \rvert} {p_n}^{-C^{\prime \prime}\lvert t \rvert \omega n^{1-\gamma} }  {p_n}^{\lvert t \rvert - \lvert k \rvert} \\ 
    &=
    C^\prime \left( \frac{n^{1/2}}{\alpha \sqrt{A(\tau) \tau} } \right)^{\lvert t \rvert-\lvert k \rvert} {p_n}^{-C^{\prime \prime} \lvert t \rvert \omega n^{1-\gamma} + \lvert t \rvert - \lvert k \rvert} 
\end{align*}}
for some positive constant $C, C^{\prime}, C^{\prime \prime}$. Recall that, by Condition \ref{cond:4:d}, $p_n = O\left(\exp\{{n^\phi}\} \right)$ and $\lvert k \rvert$ is limited in $M_3$, therefore 
\begin{equation*}
\lim_{n\rightarrow + \infty} n^{\frac{\lvert t \rvert-\lvert k \rvert}{2}}{p_n}^{-f(n)+ \lvert t \rvert - \lvert k \rvert } = 0
\end{equation*}
where $f(n):=C^{\prime \prime} \lvert t \rvert \omega n^{1-\gamma}$. Hence also the finite sum 
\begin{equation*}
\sum_{k \in M_3}  PR(k,t) \rightarrow 0 
\end{equation*}
as $n\rightarrow \infty$, with probability going to $1$.
\end{proof}

\section{\texorpdfstring{Conjugate $\vbeta$ distributions: Computation and Sampling}{Conjugate regression coefficients' distributions: Computation and Sampling}}
\label{sec:sampling}
All conjugate posterior conditionals presented in equations~\eqref{eq:post_glm_beta0}--\eqref{eq:post_glm_tau} 
have closed-form expressions. However, efficient posterior inference via Gibbs sampling relies on the possibility of drawing samples from such posteriors.
In general, the conjugate distribution of the model  parameters' distribution $\mathcal{D}(\vbeta; \DparamScalar, \DparamVector)$ poses two main challenges. First, its density is typically known up to a normalization constant that cannot be computed in closed form for most exponential-family likelihoods. Second, sampling from $\mathcal{D}$ is generally intractable ---except for specific GLM settings, such as the Gaussian likelihood, for which the conjugate prior and posterior couple for $\vbeta$ are Normal and easy to sample from.

To address these difficulties, 
we draw samples from the intractable density $\mathcal{D}(\vbeta; \DparamScalar, \DparamVector)$ via a Laplace within Sampling-Importance Resampling (SIR) (Section \ref{subsec:sir}) specifically tailored to the Bayesian GLM in Eq.~\eqref{eq:glm_variable_selection_model_begin}. 
The Laplace approximation is described in Section~\ref{subsec:laplace_approx}). 

\subsection{\texorpdfstring{The Laplace Approximation to GLM's conjugate $\vbeta$ distributions}{The Laplace Approximation to GLM's conjugate regression coefficients' distributions}}
\label{subsec:laplace_approx}
Because the conjugate distribution $\mathcal{D}(\vbeta; \DparamScalar, \DparamVector)$ does not, in general, match a closed-form or known distribution,
we approximate the density of such distribution using a second-order Taylor expansion around its mode \cite[Chapter 4]{gelman1995_laplace_approx}.
With this Laplace approximation, we approximate the distribution of $\vbeta$ with a Gaussian,
\ie we can efficiently compute the (approximate) density of $\vbeta$ that appears in the posterior probability of the variable selection indicator $\z$ in Eq.~\eqref{eq:post_glm_zh}.

The theoretical justification for this Gaussian approximation is grounded on \citep[Theorem 2.3]{chen2003conjugate_prior_glm},
who demonstrated that the distribution $\mathcal{D}(\vbeta; \DparamScalar, \DparamVector)$ converges to a $p$ dimensional multivariate normal distribution as $n \rightarrow + \infty$. 
That is, the distribution of $\vbeta$ is 
asymptotically well approximated by a Gaussian centered at its mode $\tilde{\vbeta}$, with covariance given by the inverse of the observed negative Hessian evaluated at that point:
\[
\vbeta \approx \mathcal{N}\left( \tilde{\vbeta}, \mathbf{F}\left(\tilde{\vbeta} \right)^{-1} \right)
\]
where $\tilde{\vbeta} = \arg \max_{\vbeta} \log p(\vbeta; \DparamScalar, \DparamVector)$  and $\mathbf{F}(\vbeta) = - \nabla_{\vbeta}^2 \log p(\vbeta; \DparamScalar, \DparamVector) \rvert_{\vbeta = \tilde{\vbeta}}$.

Our task here is to derive explicit expressions for $\tilde{\vbeta}$ and $\mathbf{F}(\vbeta)$ for the general GLM as defined in Eq.~\eqref{eq:classic_glm}.
We begin with the gradient and Hessian of the log-posterior, which are required both to locate the mode $\tilde{\vbeta}$ and to evaluate the curvature at that point.
\begin{theorem}
    Let $\vbeta\sim \mathcal{D}(\DparamScalar, \DparamVector)$. Then the maximum likelihood estimator $\tilde{\vbeta}$ satisfies 
    \begin{equation}
    \label{eq:betahat_stationarity_condition}
        \X^\top \left(\DparamVector - b'(\X \tilde{\vbeta}) \right) =  \vect{0}_{p\times 1}
    \end{equation}
    and the Observed Fisher Information is: 
    \begin{equation*}
        \mathbf{F}\left(\vbeta \right) = \DparamScalar \tau \X^\top \operatorname{diag}(b''(\X \vbeta)) \X 
    \end{equation*}
\end{theorem}
\begin{proof}
    We start with the gradient
    \begin{align*}
        \vect{0}_{p\times 1} \ 
        &\overset{!}{=} \ \nabla_{\vbeta}  \log p(\vbeta | \DparamScalar, \DparamVector) = \nabla_{\vbeta}  \left[ \DparamScalar \tau \left(  \DparamVector^\top \X \vbeta - \vect{J}^\top b(\X \vbeta)\right)   \right] \\
        & = \DparamScalar \tau \left( \X^\top \DparamVector - \left[ b'(\vbeta^{\top} \X^\top) \X \right]^\top \right) \  =  \ \DparamScalar \tau \left( \X^\top \DparamVector - \X^\top  b'(\X \vbeta) \right)  \ ,
    \end{align*}
    to obtain Eq.~\eqref{eq:betahat_stationarity_condition}.
    For the hessian we have: 
    \begin{align*}
        \mathbf{F}(\vbeta) &:=  - \nabla_{\vbeta} \nabla_{\vbeta}  \log p(\vbeta | \DparamScalar, \DparamVector) \ = \ - \nabla_{\vbeta} \left[ \DparamScalar \tau \left( \X^\top \DparamVector - \X^\top  b'(\X \vbeta)\right) \right] = \\
        &= \ + \nabla_{\vbeta} \left[ \DparamScalar \tau \left( \X^\top  b'(\X \vbeta) \right)\right] = \DparamScalar \tau \X^\top \operatorname{diag}(b''(\X \vbeta)) \X \ . 
    \end{align*}
\end{proof}
In general, the solution $\tilde{\vbeta}$  to Eq.~\eqref{eq:betahat_stationarity_condition} does not admit closed-form expressions for arbitrary GLMs,
hence the need for a numerical optimization method.
Once $\tilde{\vbeta}$ is obtained, the observed Fisher Information 
$\mathbf{F}(\vbeta)$ can be evaluated at that point to complete the Laplace approximation.

Although any optimization method can be used to estimate the mode of the Laplace approximation to $\mathcal{D}(\vbeta; \DparamScalar, \DparamVector)$,
we suggest to use an \textit{Iterative Reweighted Least Squares} (IRLS) algorithm
approximating the nonlinear term $b'(\X\vbeta)$ via a first-order Taylor expansion.
Namely, we numerically solve Eq.~\eqref{eq:betahat_stationarity_condition} by linearizing $b'(\X\vbeta)$ around the current iterate $\vbeta_0$, \ie
\begin{align*}
    b'(\X\vbeta) &\approx b'(\X\vbeta_0) + \nabla_{\vbeta} b'(\X\vbeta)\Big|_{\vbeta = \vbeta_0} (\vbeta - \vbeta_0) \\
    &= b'(\X\vbeta_0) + \operatorname{diag}(b''(\X\vbeta_0)) \X (\vbeta - \vbeta_0) \;.
\end{align*}
Substituting this into the stationarity condition in  Eq.~\eqref{eq:betahat_stationarity_condition}:
\begin{align*}
    \X^\top \DparamVector &= \X^\top b'(\X\vbeta_0) + \X^\top \mathbf{D} \X (\vbeta - \vbeta_0) \\
    \X^\top \mathbf{D} \X\vbeta &= \X^\top \left[\DparamVector - b'(\X\vbeta_0)\right] + \X^\top \mathbf{D} \X\vbeta_0 \;,
\end{align*}
where $\mathbf{D} = \operatorname{diag}(b''(\X\vbeta_0))$,
we conclude
\begin{align*}
    \vbeta &= (\X^\top \mathbf{D} \X)^{-1} \X^\top \left[\DparamVector - b'(\X\vbeta_0)\right] + \vbeta_0 \;.
\end{align*}
The IRLS algorithm iterates the above updates until convergence, as described in Algorithm~\ref{alg:irls}.
\begin{algorithm}[H]
\caption{IRLS Algorithm Pseudocode for $\tilde{\vbeta}$} 
\label{alg:irls}
\begin{algorithmic}[1]
\State Initialize $\vbeta_0$
\Repeat
    \State $\mathbf{D} \gets \operatorname{diag}(b''(\X\vbeta_0))$
    \State $\mathbf{r}  \gets \DparamVector - b'(\X\vbeta_0)$
    \State $\vbeta_1 \gets (\X^\top \mathbf{D} \X)^{-1} \X^\top \mathbf{r} + \vbeta_0$
    \State $\text{diff} \gets \|\vbeta_1 - \vbeta_0\| / \|\vbeta_0\|$
    \State $\vbeta_0 \gets \vbeta_1$
\Until{diff $<$ tol}
\State $\tilde{\vbeta} \gets \vbeta_1$
\end{algorithmic}
\end{algorithm}

\subsection{\texorpdfstring{Sampling from GLM's conjugate $\vbeta$ distributions}{Sampling from GLM's conjugate regression coefficients' distributions}}
\label{subsec:sir}
When the $\mathcal{D}(\vbeta; \DparamScalar, \DparamVector)$ does not correspond to a known distribution
---such as in the Linear regression model---
one must rely on auxiliary sampling techniques, \eg Sampling Importance Resampling (SIR), Markov Chain Monte Carlo, etc.
We present here a \textit{Laplace within SIR}
technique
that leverages the Laplace approximation described above as a proposal
to sample from the conjugate prior/posterior GLM parameter distribution $\mathcal{D}(\vbeta; \DparamScalar, \DparamVector)$.

A SIR procedure draws from 
$\mathcal{D}(\vbeta; \DparamScalar, \DparamVector)$ by sampling an auxiliary proposal distribution $q(\vbeta)$ and performing importance weighting on the attained samples.
As described in Algorithm~\ref{alg:laplace_sir},
we propose a Laplace within SIR variant,
where we use proposal $q(\vbeta) = \mathcal{N}(\tilde{\vbeta}, \mathbf{F}( \tilde{\vbeta} )^{-1})$
with sufficient statistics as computed in the Laplace approximation described in Section~\ref{subsec:laplace_approx}.
We then weight the candidate samples according to the likelihood ratio between the target density and the proposal, which can be computed exactly up to a normalization constant.
With renormalization of the SIR weights,
consistency (as $M \rightarrow \infty$) and unbiased evidence estimation are guaranteed.

We implement this inference procedure,
and use it within the Gibbs sampler for our numerical experiments.

\begin{algorithm}[H]
\caption{Pseudocode for Laplace within SIR Algorithm} 
\label{alg:laplace_sir}
\begin{algorithmic}[1]
\State Fix sampling/resampling values $M$ and $N$
\State Fix $\DparamScalar, \DparamVector$ and define $\mathcal{D}(\DparamScalar, \DparamVector)$
\State For observed $\X$, compute $\tilde{\vbeta}$ and $\mathbf{F}(\tilde{\vbeta})$ using numerical optimization
\State Define the Laplace within SIR proposal distribution:
    $q(\vbeta) = \mathcal{N}\left( \tilde{\vbeta}, \mathbf{F}(\tilde{\vbeta})^{-1} \right)$
\For{$i = 1$ to $M$}
    \State Sample $\vbeta^{(i)} \sim q(\vbeta)$
    \State Compute importance weight $w^{(i)} = \displaystyle \frac{ \mathcal{D}\left(\vbeta^{(i)}; \DparamScalar_0, \DparamVector_0 \right)}{q\left( \vbeta^{(i)}\right)}$
\EndFor
\State Normalize weights: $w^{(i)} \gets w^{(i)} / \sum_{j=1}^M w^{(j)}$
\State Resample $N \ll M$ values from $\left\{\vbeta^{(i)}\right\}$ with probabilities $\left\{w^{(i)}\right\}$
\end{algorithmic}
\end{algorithm}

\section{Trace Plots and Mixing Behavior}
\label{apd:trace_poisson}
We complement the discussion in Remark~\ref{rem:mixing_trace} with empirical evidence on the mixing behavior of the Gibbs sampler. We consider the same Poisson GLM setting analyzed in Section~\ref{subsubsec:syntheticdata_poisson_settingC}.

Figure~\ref{fig:traceplot_z_n1000} reports trace plots of the selection indicators $z_j$, $j=1,\ldots,p$, for $5$ independent chains, each initialized at an independently drawn $\z \sim \text{Binom}(p, 1/2)$. 
Across chains, the indicators associated with active covariates rapidly and stably settle at $z_j=1$, while those associated with inactive covariates settle at $z_j=0$; no chain remains trapped in an incorrect configuration for an extended number of iterations, and all chains reach consistent posterior estimates regardless of their (random) starting configuration.

\begin{figure}
    \centering
    \includegraphics[width=\linewidth]{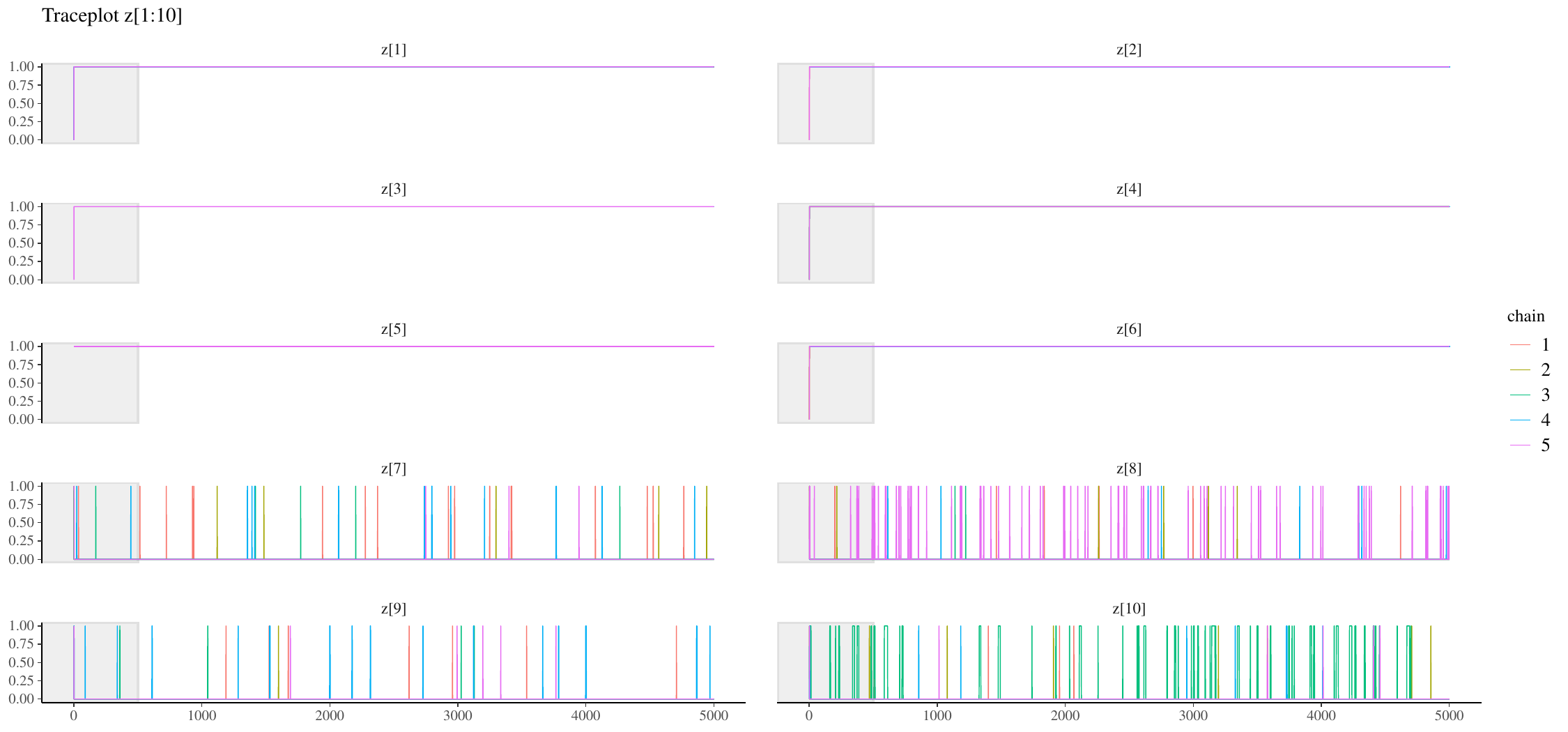}
    \includegraphics[width=\linewidth]{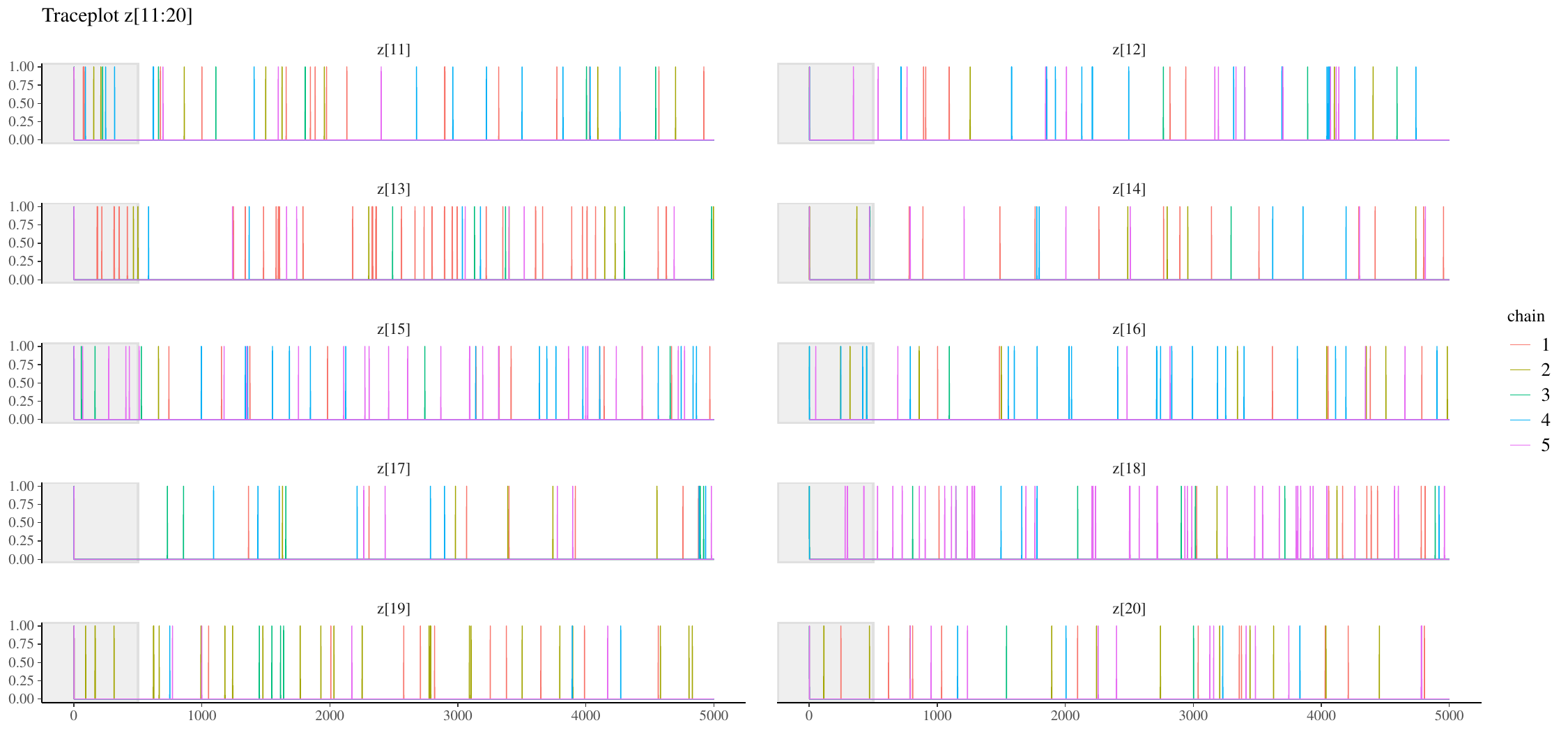}
    \caption{Poisson Model. Trace plots of the selection indicators $\z$ with $p=20$ covariates, $k=6$ active covariates, and $n=1000$ observations, across $5$ independent chains of $5000$ MCMC iterations each. The shaded region indicates the first $10\%$ of iterations, discarded as burn-in.}
    \label{fig:traceplot_z_n1000}
\end{figure}

\begin{figure}
    \centering
    \includegraphics[width=\linewidth]{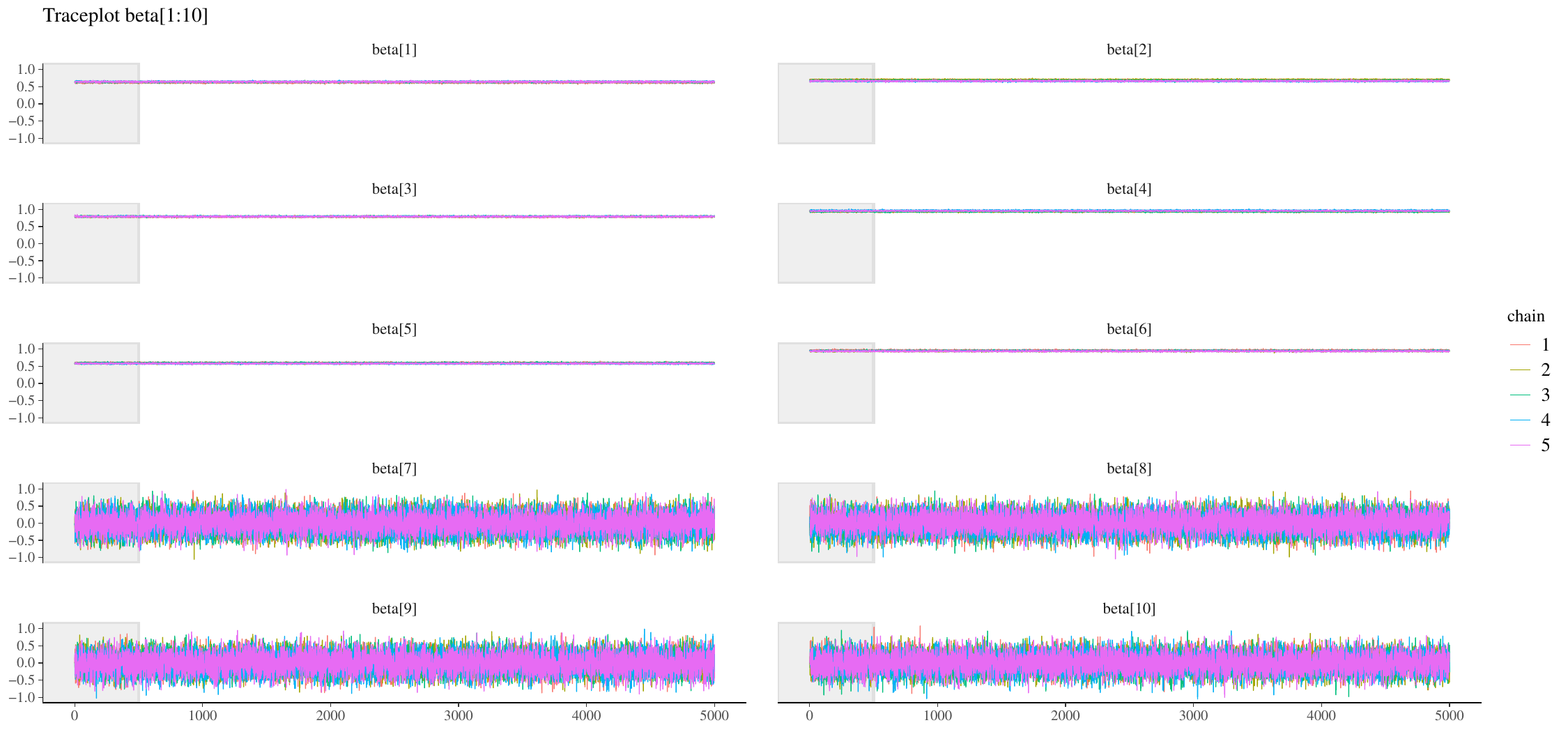}
    \includegraphics[width=\linewidth]{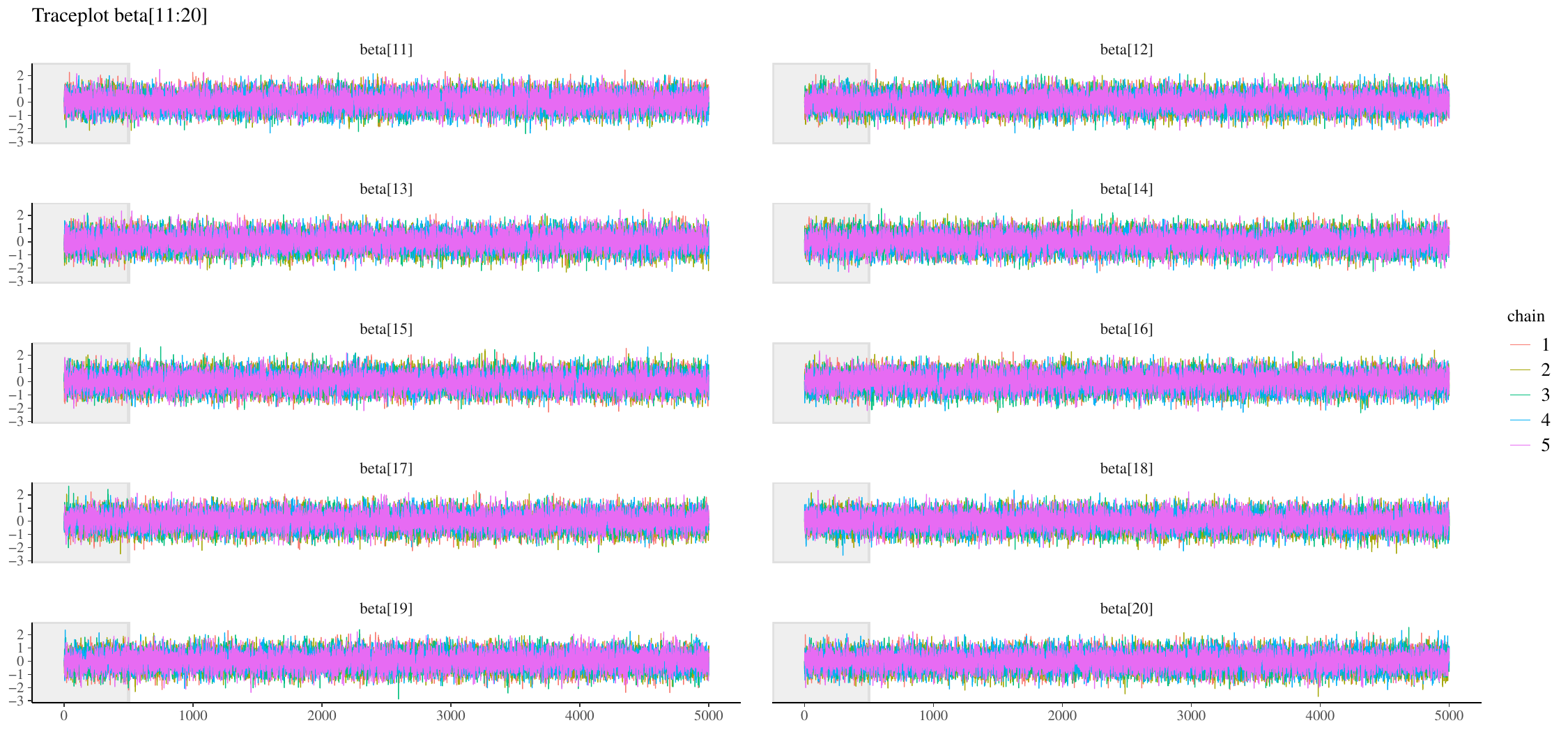}
    \caption{Poisson Model. Trace plots of the regression coefficients $\vbeta$ with $p=20$ covariates, $k=6$ active covariates, and $n=1000$ observations, across $5$ independent chains of $5000$ MCMC iterations each. The shaded region indicates the first $10\%$ of iterations, discarded as burn-in. Coefficients associated with active covariates concentrate around their true values, while coefficients associated with inactive covariates do not shrink to zero.}
    \label{fig:traceplot_beta_n1000}
\end{figure}

Figure~\ref{fig:traceplot_beta_n1000} reports the corresponding trace plots for the regression coefficients $\vbeta$. Chains for $\vbeta$ are initialized at the maximum likelihood estimate under the full model.  As expected from Eq.~\eqref{eq:post_glm_beta0}--\eqref{eq:post_glm_beta1}, coefficients corresponding to active covariates ($z_j=1$) concentrate tightly around their true values ($\vbeta^{*(1)}= (0.63,  0.68, 0.79, 0.95, 0.60, 0.95)$), while coefficients corresponding to inactive covariates ($z_j=0$) remain spread out according to the prior $\mathcal{D}(\cdot\,; \DparamScalar_0, \DparamVector_0)$, without exhibiting shrinkage toward zero. This confirms empirically the mechanism discussed in Remark~\ref{rem:mixing_trace}: since $\beta_j$ is never artificially pulled toward zero when $z_j=0$, it retains information relevant to reactivating the covariate, which in turn explains the absence of trapping observed in Figure~\ref{fig:traceplot_z_n1000}.

\section{Additional Simulation Results}
\label{apd:results_poisson_gaussian_binomial}

This section complements the results presented in the main text.
. We first report additional results for the Poisson GLM presented in Section \ref{subsec:results_synthetic}, which were omitted earlier for brevity. 
We then present the complete set of experiments of \textbf{Setting C} for the Linear and the Logistic Models.

All simulations follow the same design described in Section~\ref{subsec:results_synthetic}, with identical data-generating processes and number of independent runs per configuration.
The only differences concern the choice of prior hyperparameters.

\subsection{Poisson Model}
\label{apd_subsec:poisson}
We present here additional numerical results on Poisson Model that complements the results presented in the main text.
\begin{figure}[htb]
    \centering
    \includegraphics[height=6cm]{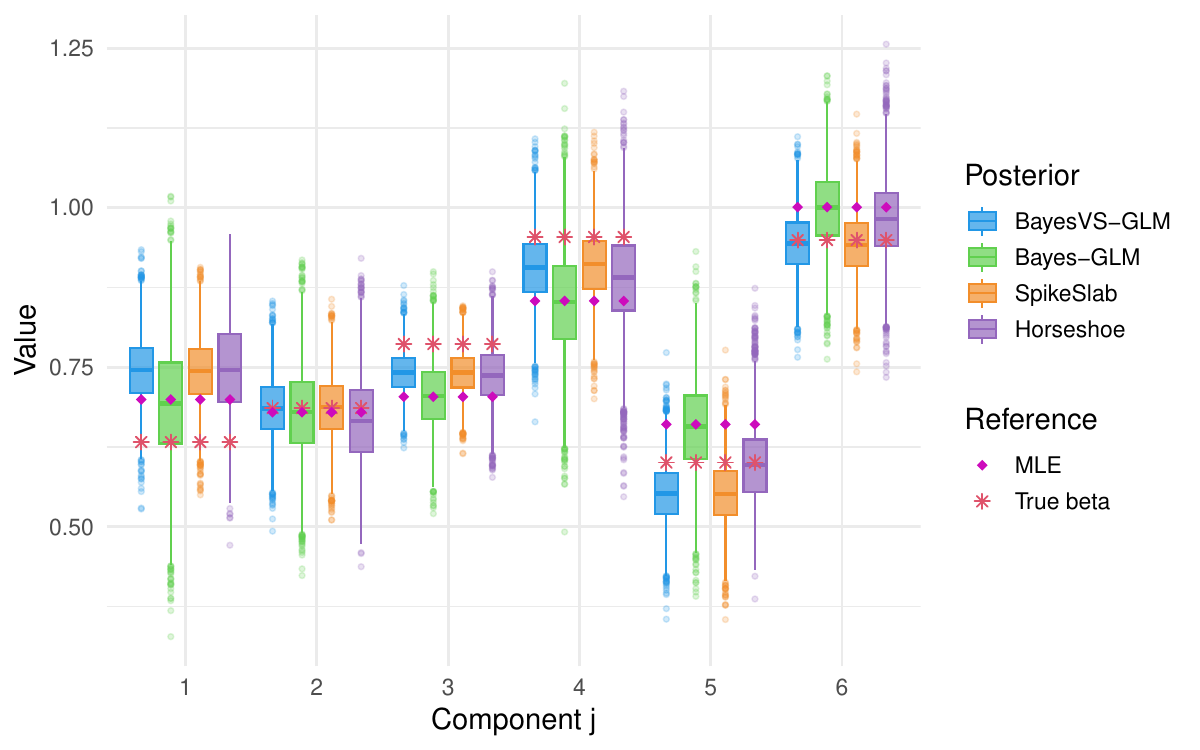}
    \includegraphics[height=6cm]{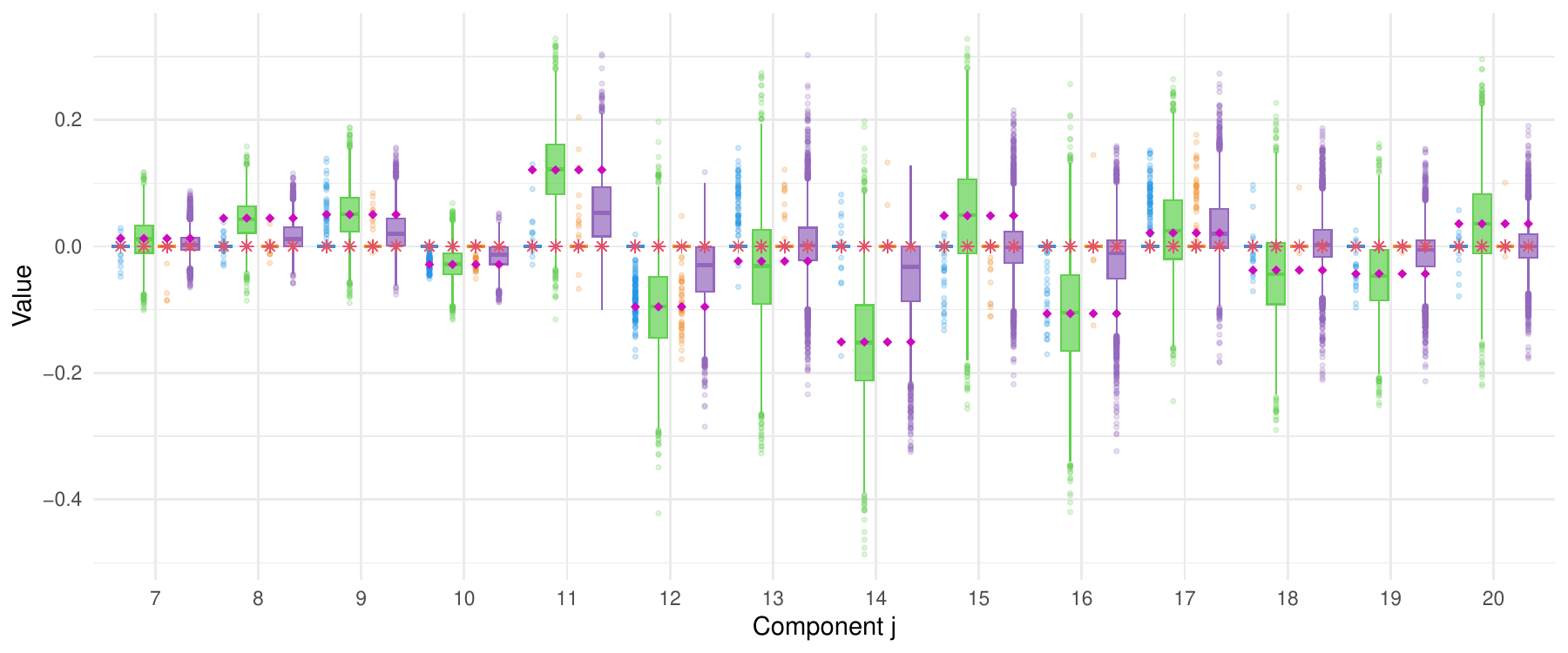}
    \caption{Poisson Model. Posterior distribution of $\vbeta \circ \z$, for $p=20, \ d=10$, $2c=4$ and $n=100$. Our method (\texttt{BayesVS-GLM}) reaches more accurate estimate than the \texttt{MLE} and the \texttt{Horseshoe}, especially in the non-active coefficients.}
    \label{fig:settingC_betas}
\end{figure}

Figure~\ref{fig:settingC_betas} reports the marginal distributions of $\vbeta \circ \z$ across posterior samples for a fixed simulation setup ($p = 20$, $d = 10$, $c=2$, $n = 100$).
For the non-active coefficients, our method (blue boxplots) yields posterior medians that coincide with the true value $\vbeta^* = \mathbf{0}$ (magenta star shape), indicating more accurate shrinkage toward zero.
By contrast, the \texttt{MLE} (pink diamond shape) aligns with the posterior medians of the samples obtained from the standard bayesian approach, \texttt{Bayes-GLM} (green boxplots), but both exhibit larger deviations from zero in inactive components, as \texttt{Horseshoe}.

\subsection{Linear Model}
\label{apd_subsec:linear}
We here present the results on synthetic data for the  canonical GLM with identity link and Gaussian likelihood:
\begin{align*}
    \yi \mid \mu_i, \tau &\sim \mathcal{N}(\mu_i, \tau), \quad && \etai = \mu_i = \sum_{j=1}^p x_{ji} \betaj \zj
\end{align*}
with priors:
\begin{align}
    \z \mid \vc &\sim \prod_{j=1}^p \text{Bern}(\cj), &
    \vc \mid \alpha &\sim \prod_{j=1}^p \text{Beta}\left( \frac{\alpha}{p}, 1 \right), \\
    \vbeta &\sim \mathcal{N}(\vmu_0, \Sigma_0), &
    \tau &\sim \text{InvGamma}(a, b). \nonumber
\end{align}
Thanks to conjugacy, all full conditional distributions are available in closed form.
Complete derivations are reported in Section \ref{apd:posterior_linear}.

The chosen hyperparameters of the prior distributions for the Linear case are: $\DparamScalar_0 = 10^{-4}$, $\DparamVector_0 \overset{\iid}{\sim} \mathcal{N}(1,4)$, and $\alpha=1$.

\subsubsection*{Setting C: informative, noise and correlated covariates}
In this configuration, the design matrix $\X$ contains $k=6$ informative, $d=10$ completely noisy, and $2c=4$ correlated covariates, for a total of $p=20$ covariates. 
Figure \ref{fig:settingC_Gaussian_z_accuracy} reports the component-wise accuracy of $\z$ over multiple simulation runs, for each
coordinate $j=1, \ldots, p$. The results are shown for two different numbers of training samples:
\begin{figure}[H]
     \centering
     \begin{subfigure}[b]{0.48\textwidth}
         \centering
         \includegraphics[width=\textwidth]{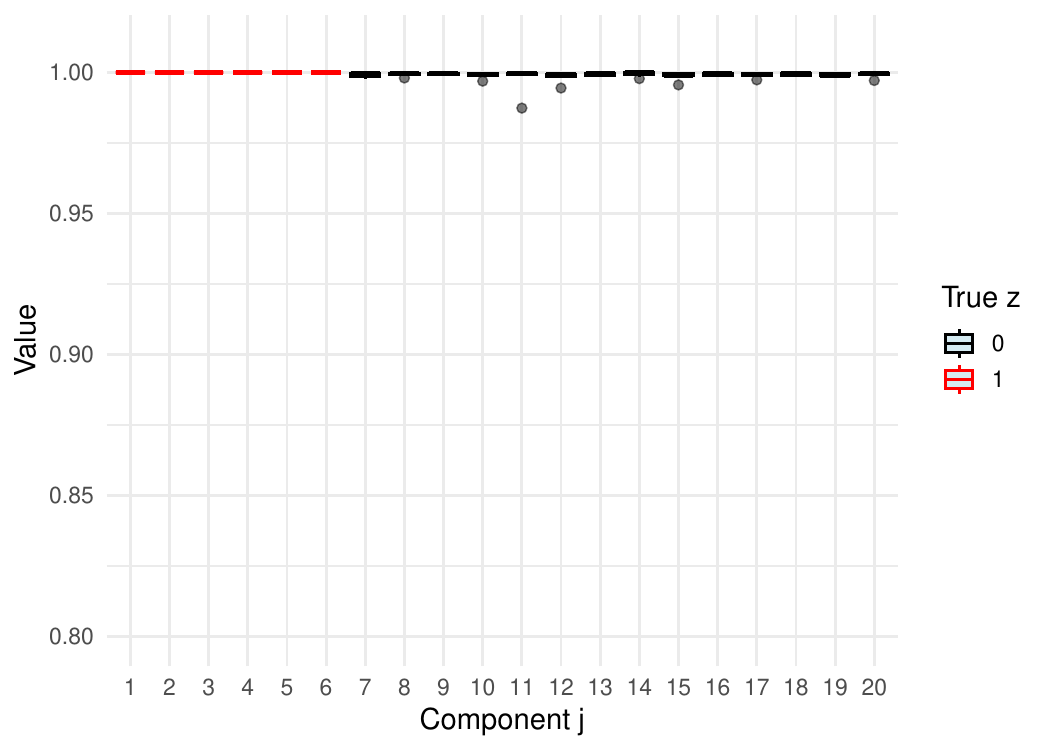}
         \caption{$n=100$}
         \label{fig:gau_acc_cd_n100}
     \end{subfigure}
     \hfill
     \begin{subfigure}[b]{0.48\linewidth}
         \centering
         \includegraphics[width=\textwidth]{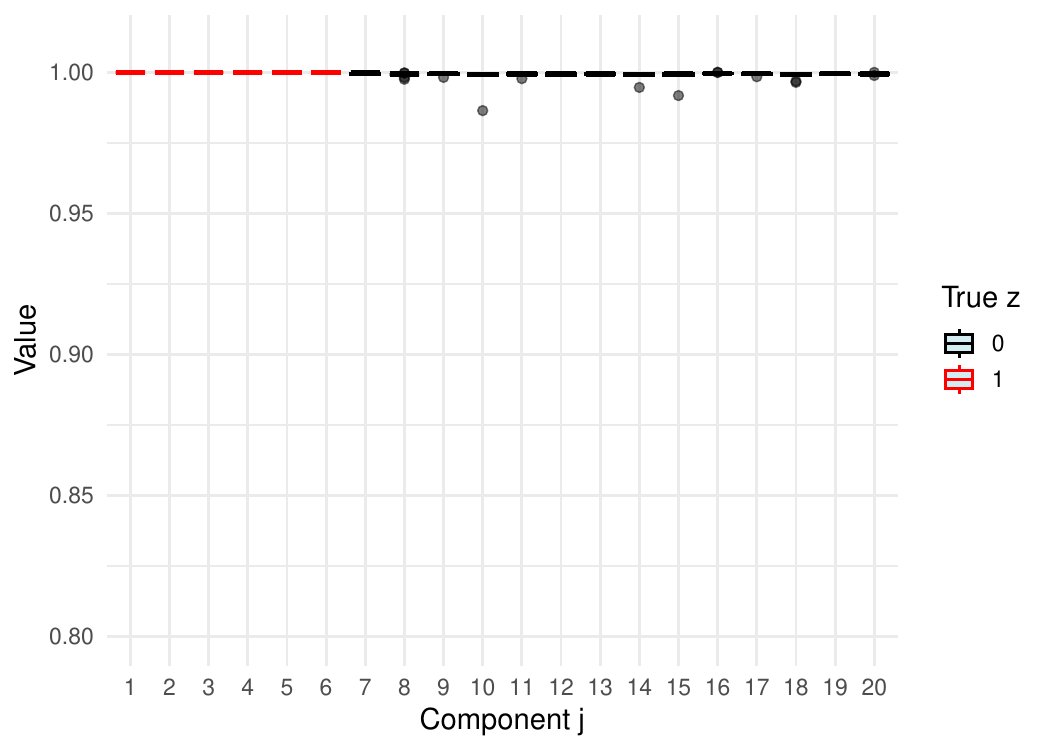}
         \caption{$n=1000$}
         \label{fig:gau_acc_cd_n1000}
     \end{subfigure}
        \caption{Linear Model. Component-wise accuracy of the inclusion variable $\z$ across multiple simulations ($p=20, \ d=10$, and $2c=4$ correlated variables). Each boxplot refers to one component $\zj$.}
        \label{fig:settingC_Gaussian_z_accuracy}
\end{figure}
Even with low number of training samples ($n=100$, in Fig. \ref{fig:gau_acc_cd_n100}), we reach very high levels of accuracy. 
In Figure~\ref{fig:settingC_gau_z_acc_np_curve}, we summarize the overall ability to retrieve $\z^*$ by computing the total selection accuracy (\ie the proportion of correctly recovered entries in $\z$) for each run, as $n$ increases. In each run, even for a small $n/p$ ratio, the accuracy is almost constant and equal to one. 
\begin{figure}[H]
    \centering
    \includegraphics[width=0.5\textwidth]{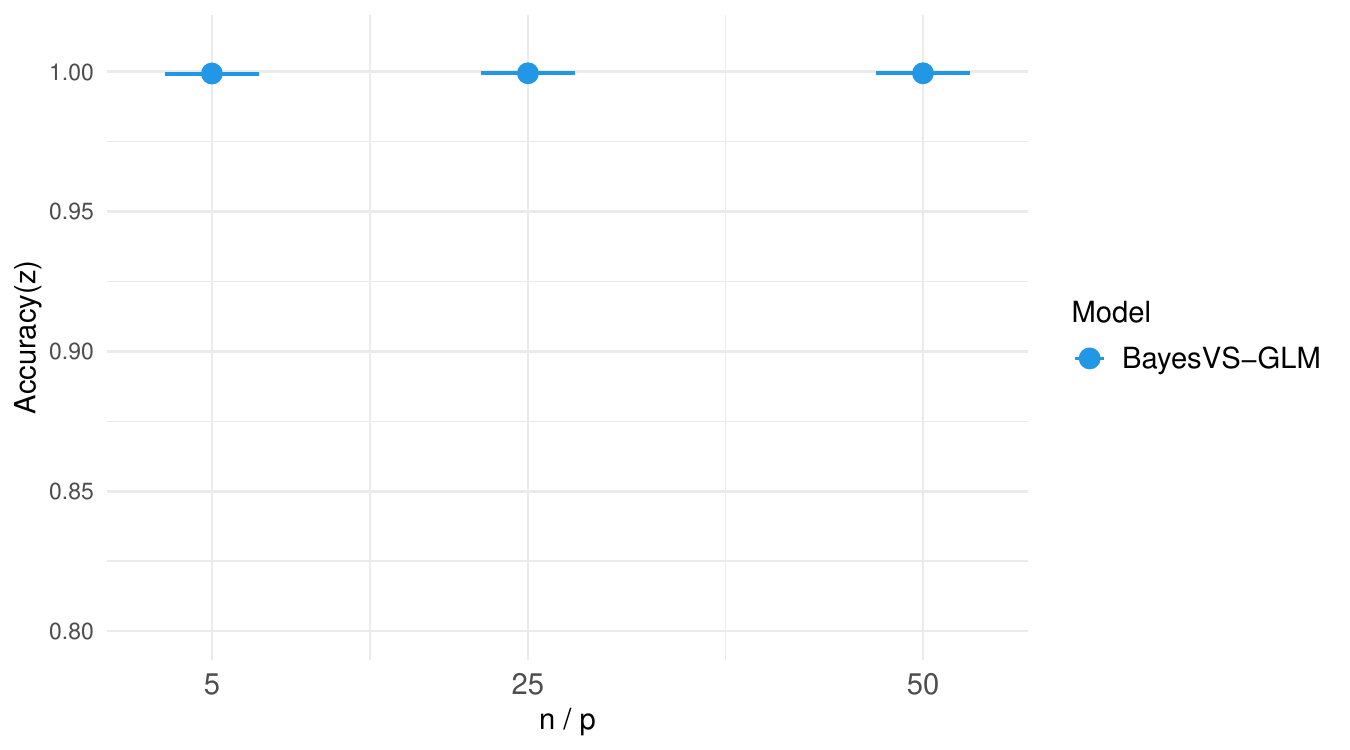}
    \caption{Linear Model. With $p=20$ total covariates, $d=10$ noisy and $2c=4$ correlated variables. Accuracy of the inclusion variable $\z$, and increasing ratio $n/p$.}
    \label{fig:settingC_gau_z_acc_np_curve}
\end{figure}

Shifting our focus to the recovery of the regression coefficients. 
Figure~\ref{fig:settingC_gau_betas} reports the marginal distributions of $\betaj  \zj$ across posterior samples for a fixed simulation setup ($p = 20$, $d = 10$, $2c=4$) for $n=100$.
\begin{figure}[bht!]
     \centering
     \begin{subfigure}{0.8\textwidth}
         \centering
         \includegraphics[height=6cm]{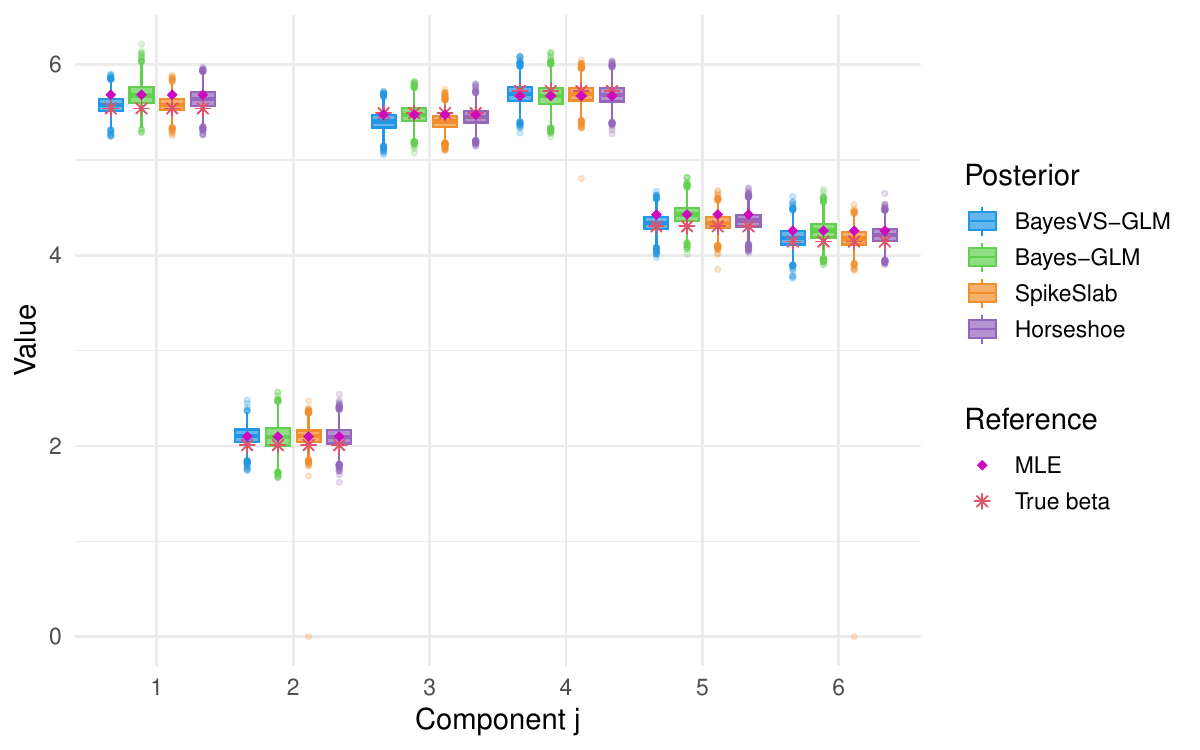}
         \caption{Active components.}
         \label{fig:gau_betaz_n100_active}
     \end{subfigure} 
     \begin{subfigure}{0.8\linewidth}
         \centering
         \includegraphics[height=6cm]{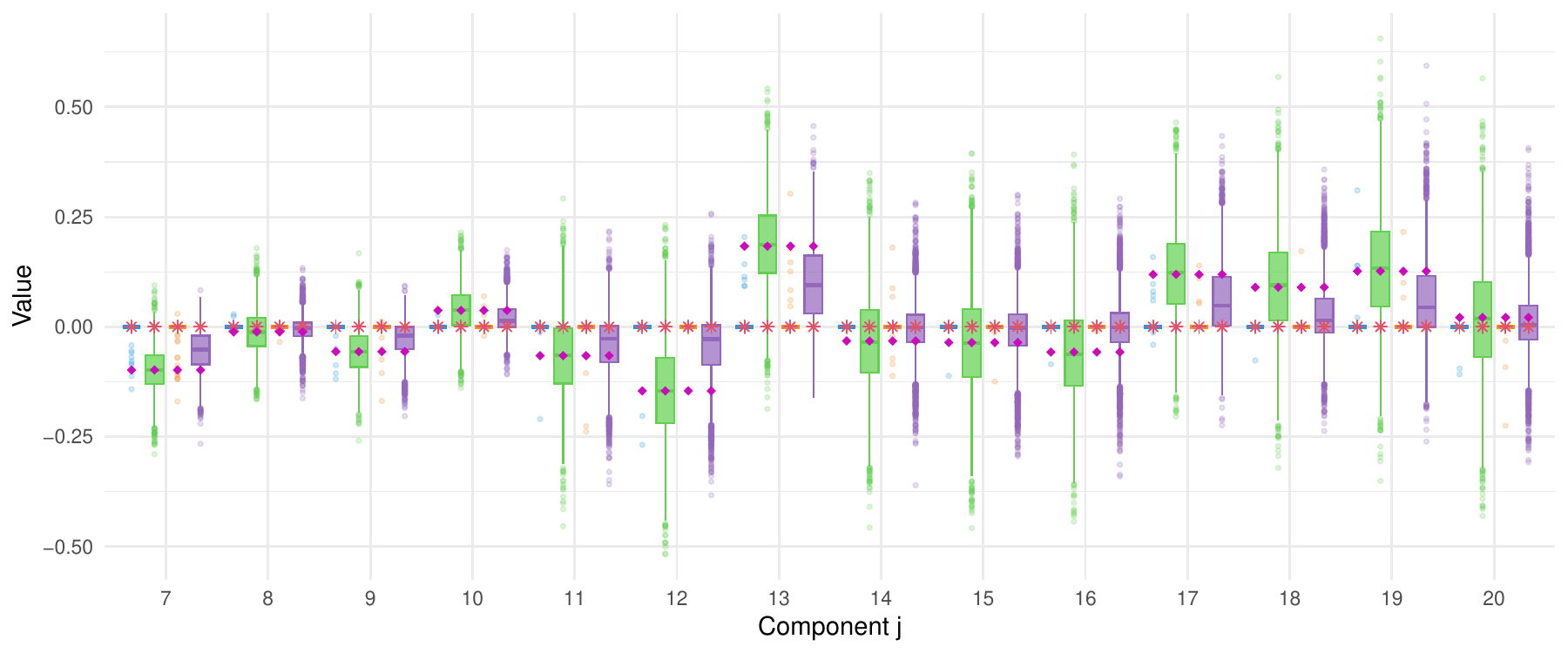}
         \caption{Inactive components.}
         \label{fig:gau_betaz_n100_nonactive}
     \end{subfigure}
        \caption{Linear Model. Posterior distribution of $\vbeta \circ \z$, for $p=20, \ d=10$ and $2c=4$ and $n=100$. Our method (\texttt{BayesVS-GLM}) reaches more accurate estimate than the \texttt{MLE} and the \texttt{Horseshoe}, especially in the non-active coefficients.}
        \label{fig:settingC_gau_betas}
\end{figure}
While we reach similar results to the state-of-the-art methods in the active components (Figure \ref{fig:gau_betaz_n100_active}), in the non-active ones (Figure \ref{fig:gau_betaz_n100_nonactive}) we reach more accurate estimate than the \texttt{Horseshoe} method, the \texttt{Bayes-GLM} and its frequentist counterpart \texttt{MLE}: the posterior distribution of $\vbeta^{(0)} \circ \z^{*{(0)}}$ of our method (\texttt{BayesVS-GLM}) is close to $\mathbf{0}$.

We then summarize estimation performance through a quantitative error metric. 
Figures~\ref{fig:gau_beta_mse}-\ref{fig:gau_beta_rmse} 
report respectively the Mean Squared Error (MSE) of the estimated coefficients $\vbeta \circ \z$ relative to the ground truth $\vbeta^{*}$, and Relative Mean Squared Error (RelMSE) computed only on the active coefficients, as $n$ increases for fixed values of $p, \ c$, and $d$. 
\begin{figure}[t]
    \centering
    \begin{subfigure}[b]{0.48\linewidth}
         \centering
         \includegraphics[width=\textwidth]{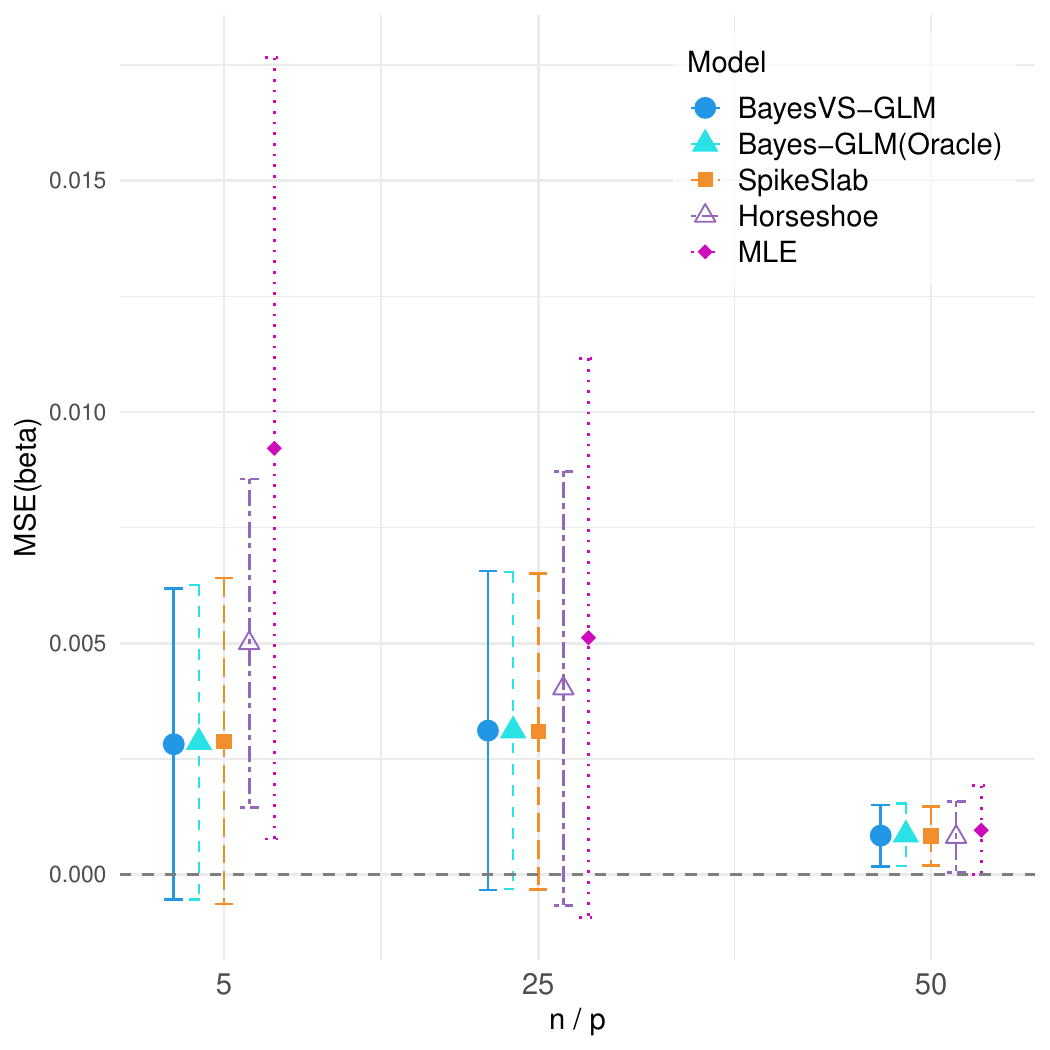}
         \caption{$MSE(\vbeta, \vbeta^{*})$}
         \label{fig:gau_beta_mse}
     \end{subfigure}
     \hfill
     \begin{subfigure}[b]{0.48\linewidth}
         \centering
         \includegraphics[width=\textwidth]{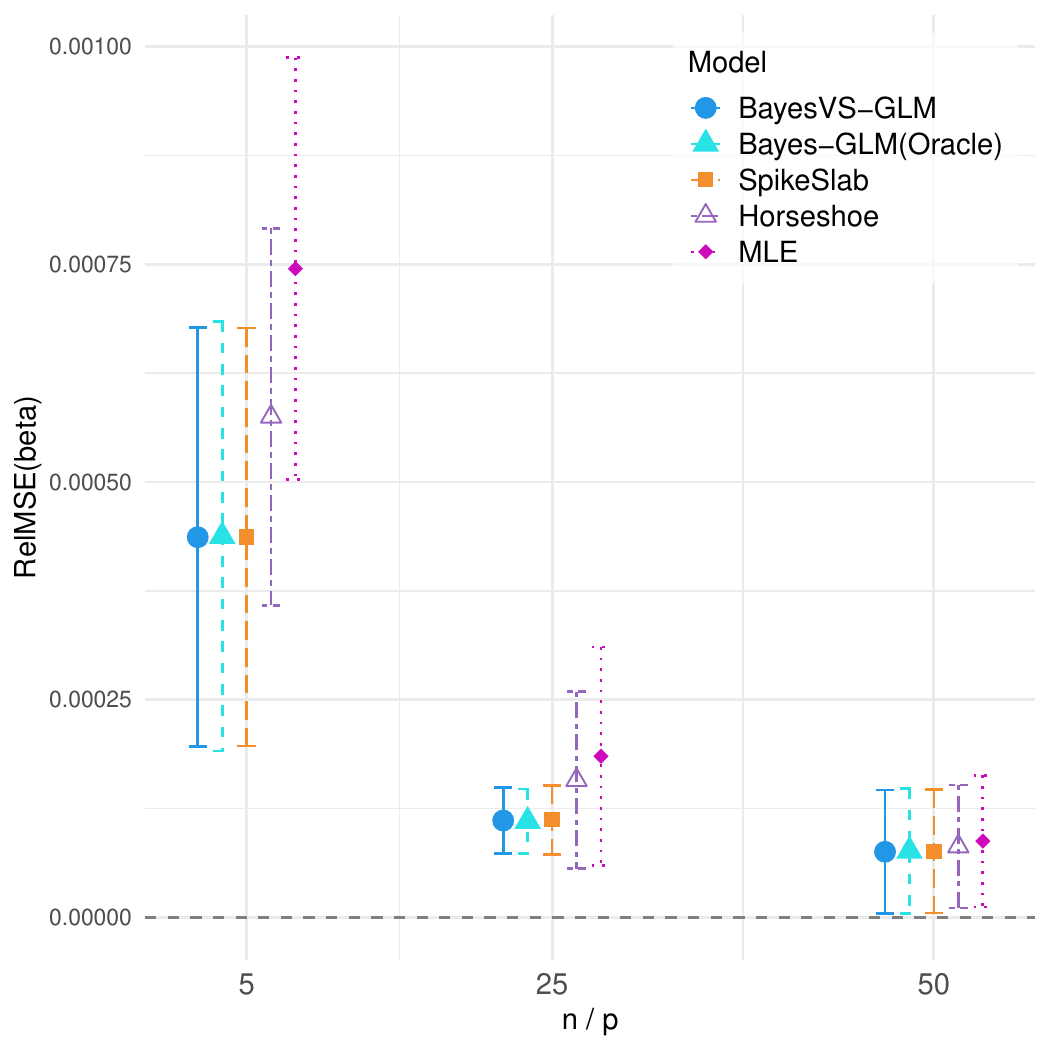}
         \caption{$RelMSE(\vbeta^{(1)}, \vbeta^{*(1)})$}
         \label{fig:gau_beta_rmse}
     \end{subfigure}
    \caption{Linear Model. Error metrics of the coefficients $\vbeta$, for $p=20$, $d=10$ and $2c=4$, and increasing ratio $n/p$. Our method (\texttt{BayesVS-GLM}) closely matches the oracle version and the \texttt{SpikeSlab} baseline, while it outperforms \texttt{MLE} and \texttt{Horseshoe} in small sample size regimes.}
    \label{fig:settingC_gau_beta_error_np_curve}
\end{figure}
Across both metrics and across the whole range of $n/p$, our method (\texttt{BayesVS-GLM}) achieves the same performance as the oracle model (\texttt{Bayes-GLM(Oracle)}), despite not having access to the true inclusion vector $\z^*$. Moreover it consistently outperforms both \texttt{MLE} and \texttt{Horseshoe}, maintaining comparable performance with \texttt{SpikeSlab}.

\subsection{Logistic Model}
\label{apd_subsec:binomial}
We here present the results on synthetic data for the  canonical GLM with identity link and Binomial likelihood:
\begin{align*}
    \yi \mid \mu_i, \tau &\sim Bin(\mu_i), \quad && \etai = \mu_i = \sum_{j=1}^p x_{ji} \betaj \zj
\end{align*}
with priors:
\begin{align}
    \z \mid \vc &\sim \prod_{j=1}^p \text{Bern}(\cj), &
    \vc \mid \alpha &\sim \prod_{j=1}^p \text{Beta}\left( \frac{\alpha}{p}, 1 \right), \\
    \vbeta &\sim \mathcal{D}(\vbeta; \DparamScalar_0, \DparamVector_0) \; .&  & \nonumber
\end{align}
The chosen hyperparameters of the prior distributions for the Linear case are: $\DparamScalar_0 = 10^{-2}$, $\DparamVector_0 \overset{\iid}{\sim} \mathcal{N}(0,1)$, and $\alpha=1$.
 
\subsubsection*{Setting C: informative, noise and correlated covariates}
In this configuration, the design matrix $\X$ contains $k=6$ informative, $d=5$ completely noisy, and $2c=2$ correlated covariates, for a total of $p=10$ covariates. 
Figure \ref{fig:settingC_Binomial_z_accuracy} reports the component-wise accuracy of $\z$ over multiple simulation runs, for each
coordinate $j=1, \ldots, p$. The results are shown for two different numbers of training samples:
\begin{figure}[H]
     \centering
     \begin{subfigure}[b]{0.48\textwidth}
         \centering
         \includegraphics[width=\textwidth]{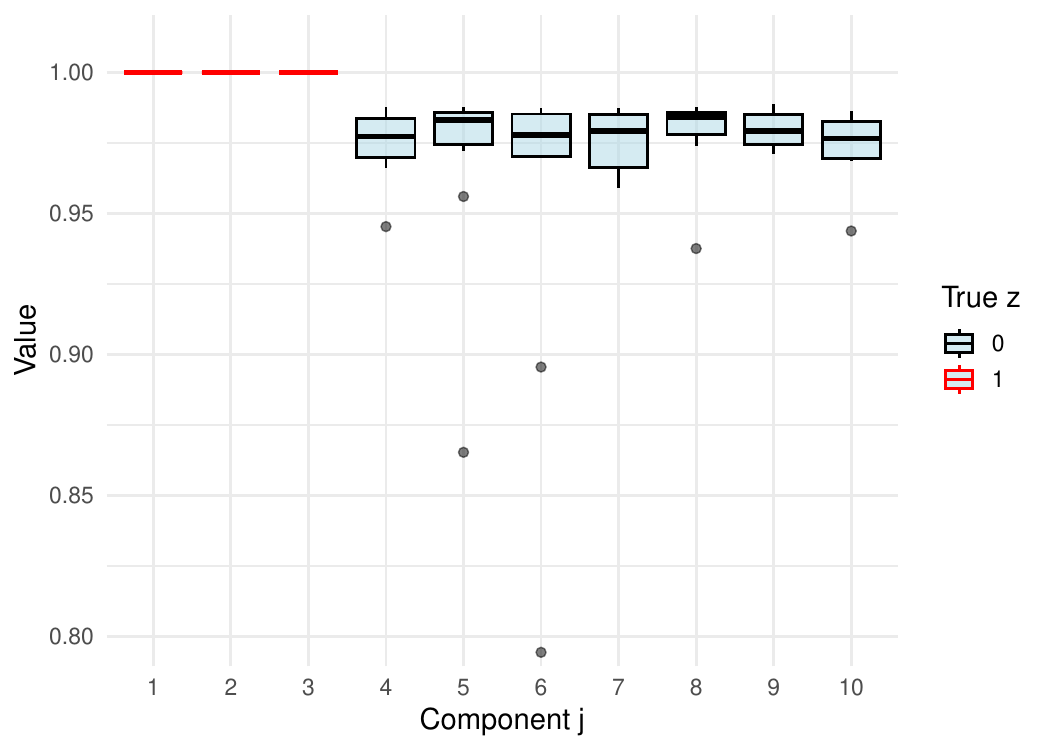}
         \caption{$n=500$}
         \label{fig:bin_acc_cd_n500}
     \end{subfigure}
     \hfill
     \begin{subfigure}[b]{0.48\linewidth}
         \centering
         \includegraphics[width=\textwidth]{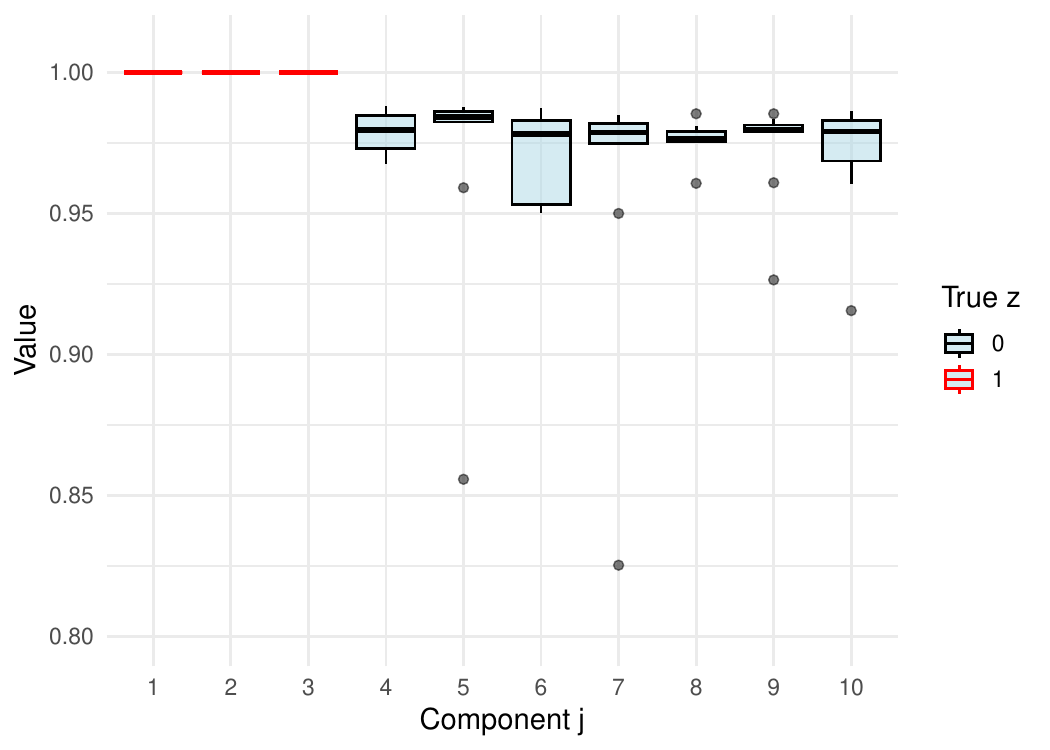}
         \caption{$n=1000$}
         \label{fig:bin_acc_cd_n1500}
     \end{subfigure}
        \caption{Logistic Model. Component-wise accuracy of the inclusion variable $\z$ across multiple simulations ($p=10, \ d=5$, and $2c=2$ correlated variables). Each boxplot refers to one component $\zj$.}
        \label{fig:settingC_Binomial_z_accuracy}
\end{figure}
Even with low number of training samples ($n=500$, in Fig. \ref{fig:bin_acc_cd_n500}), we reach very high levels of accuracy. 
Figure~\ref{fig:settingC_bin_z_acc_np_curve} reports the total selection accuracy as $n$ increases. 
\begin{figure}[tbh]
    \centering
    \includegraphics[width=0.5\textwidth]{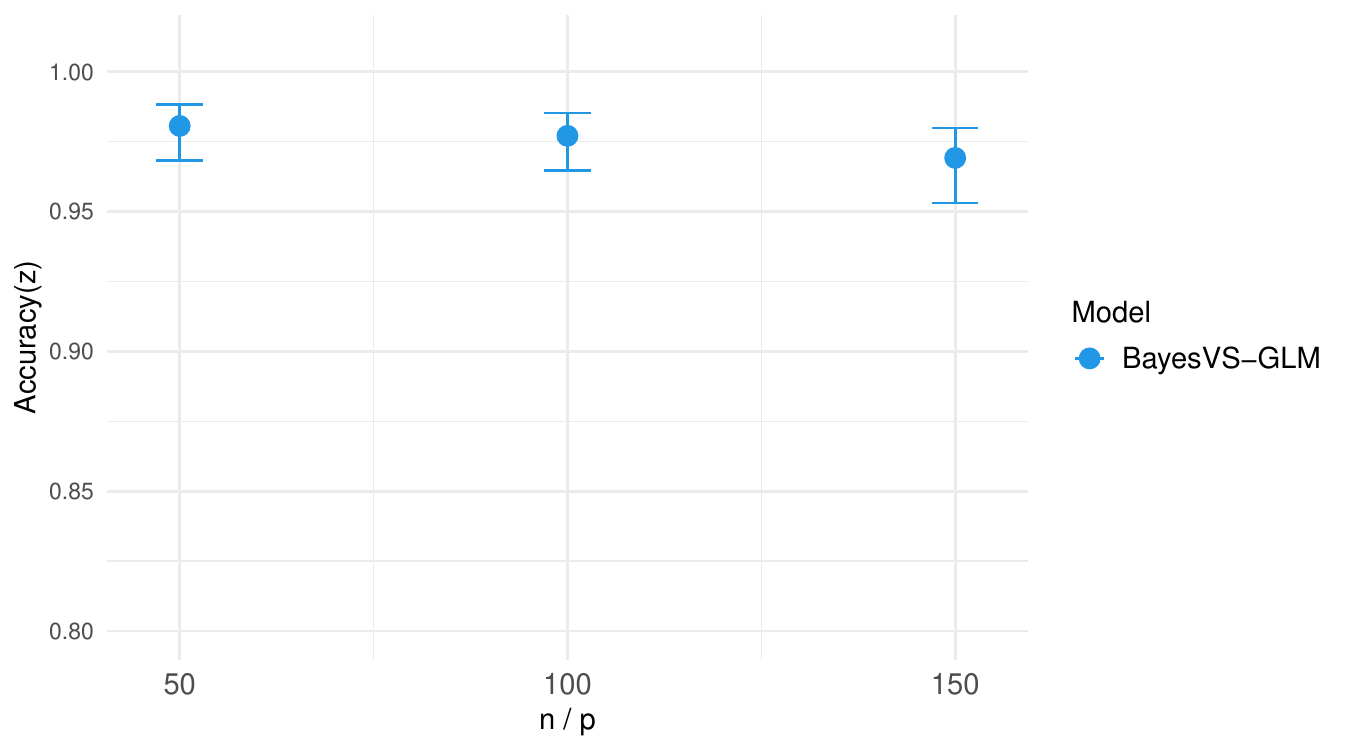}
    \caption{Logistic Model. With $p=10$ total covariates, $d=5$ noisy and $2c=2$ correlated variables. Accuracy of the inclusion variable $\z$, and increasing ratio $n/p$.}
    \label{fig:settingC_bin_z_acc_np_curve}
\end{figure}

Figure~\ref{fig:settingC_bin_betas} reports the marginal distributions of $\betaj  \zj$ across posterior samples for a fixed simulation setup ($p=10$, $d = 5$, $c=1$), 
 In the small training set scenario, on average we reach more accurate estimate than the classic MLE in the non-active coefficients: the posterior median of $\vbeta^{(0)} \circ \z^{*{(0)}}$is closer to $\mathbf{0}$ than the MLE.
\begin{figure}[ht!]
     \centering
     \begin{subfigure}{0.8\textwidth}
         \centering
         \includegraphics[height=5cm]{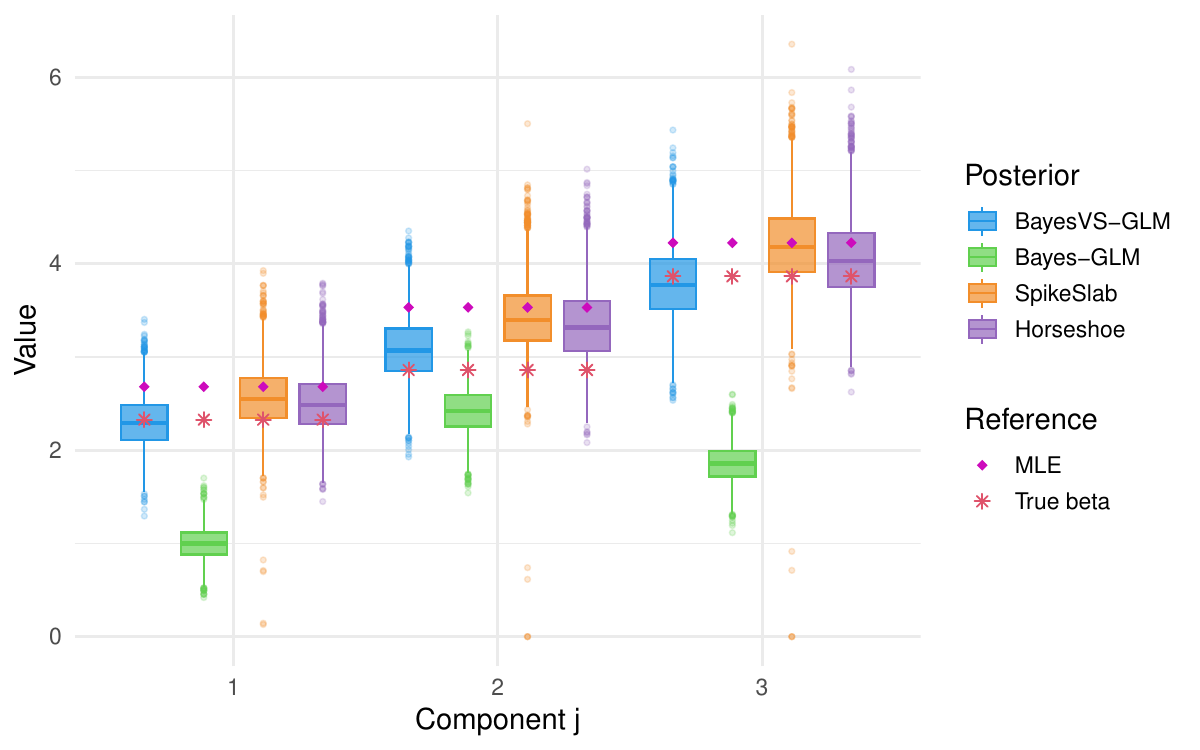}
         \caption{Active components.}
         \label{fig:bin_betaz_n500_active}
     \end{subfigure} 
     \begin{subfigure}{0.8\linewidth}
         \centering
         \includegraphics[height=5cm]{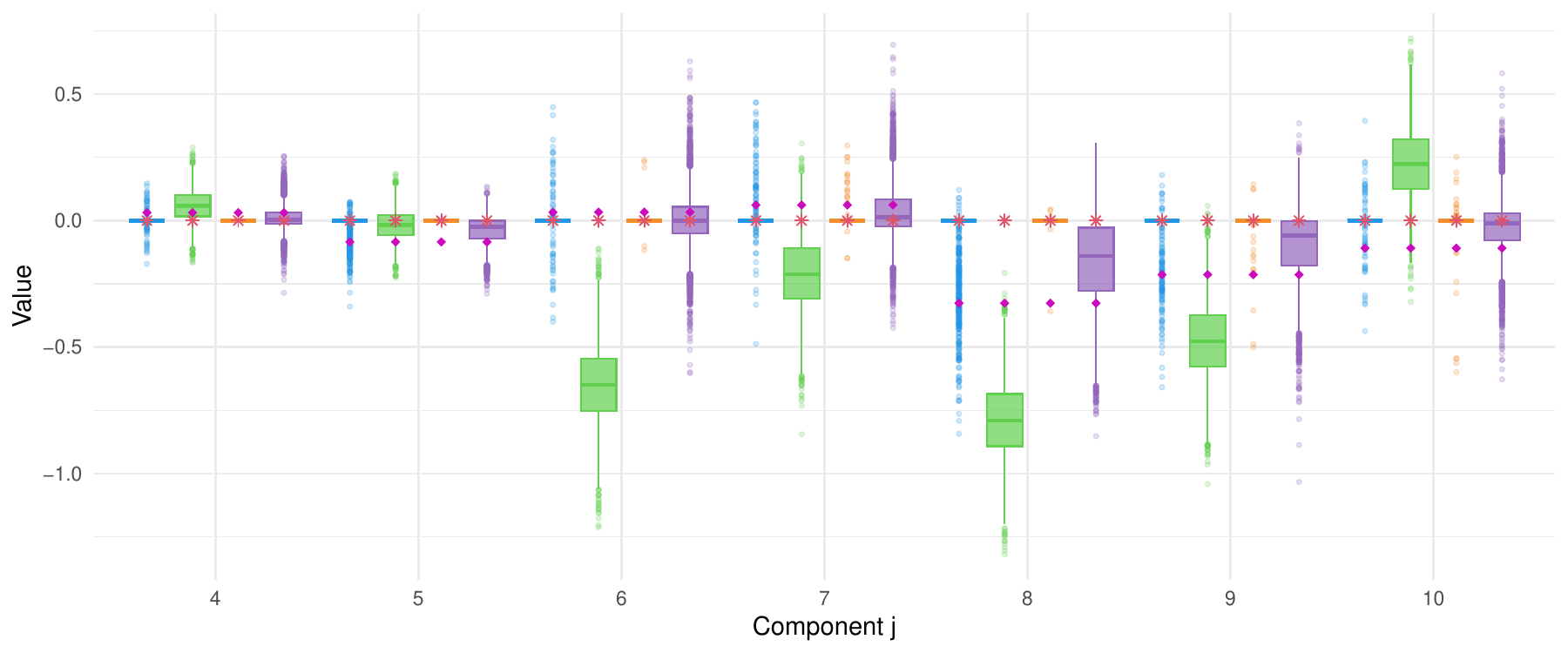}
         \caption{Inactive components.}
         \label{fig:bin_betaz_n500_nonactive}
     \end{subfigure}
        \caption{Logistic Model. Posterior distribution of $\vbeta \circ \z$, for $p=10, \ d=5, \ 2c=2$ and $n=500$. Our method (\texttt{BayesVS-GLM}) reaches comparable results with \texttt{SpikeSlab}, and more accurate estimate than \texttt{Horseshoe} in the non-active coefficients.}
        \label{fig:settingC_bin_betas}
\end{figure}
While we reach similar results to the state-of-the-art methods in the active components (Figure \ref{fig:bin_betaz_n500_active}), in the non-active ones (Figure \ref{fig:bin_betaz_n500_nonactive}) we reach more accurate estimate than the \texttt{Horseshoe} method, the \texttt{Bayes-GLM} and its frequentist counterpart \texttt{MLE}: the posterior distribution of $\vbeta^{(0)} \circ \z^{*{(0)}}$ of our method (\texttt{BayesVS-GLM}) is close to $\mathbf{0}$.

Figures~\ref{fig:bin_beta_mse}-\ref{fig:bin_beta_rmse} 
report respectively the Mean Squared Error (MSE) of the estimated coefficients $\vbeta \circ \z$ relative to the ground truth $\vbeta^{*}$, and Relative Mean Squared Error (RelMSE) computed only on the active coefficients, as $n$ increases for fixed values of $p, \ c$, and $d$. 
\begin{figure}[H]
    \centering
    \begin{subfigure}[b]{0.45\linewidth}
         \centering
         \includegraphics[width=\textwidth]{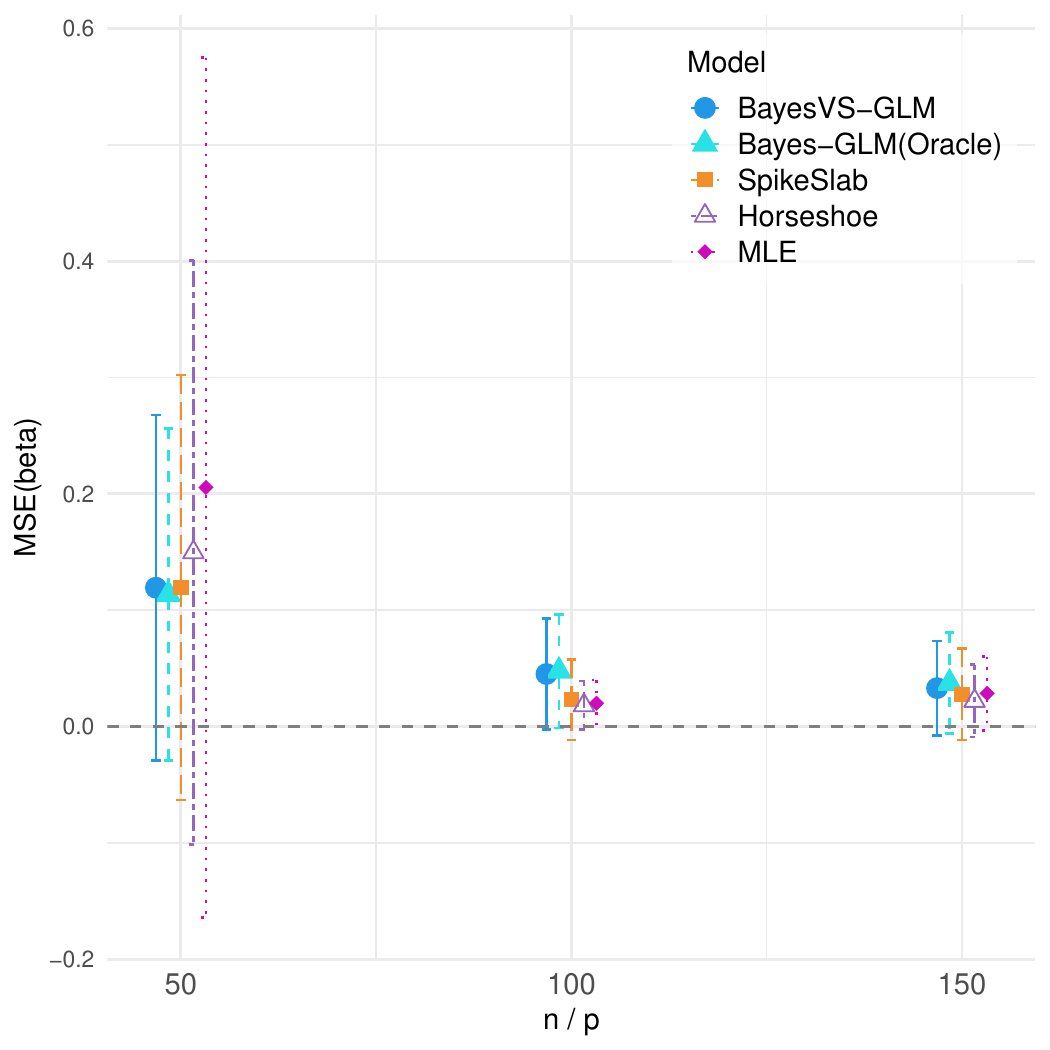}
         \caption{$MSE(\vbeta, \vbeta^{*})$}
         \label{fig:bin_beta_mse}
     \end{subfigure}
     \hfill
     \begin{subfigure}[b]{0.45\linewidth}
         \centering
         \includegraphics[width=\textwidth]{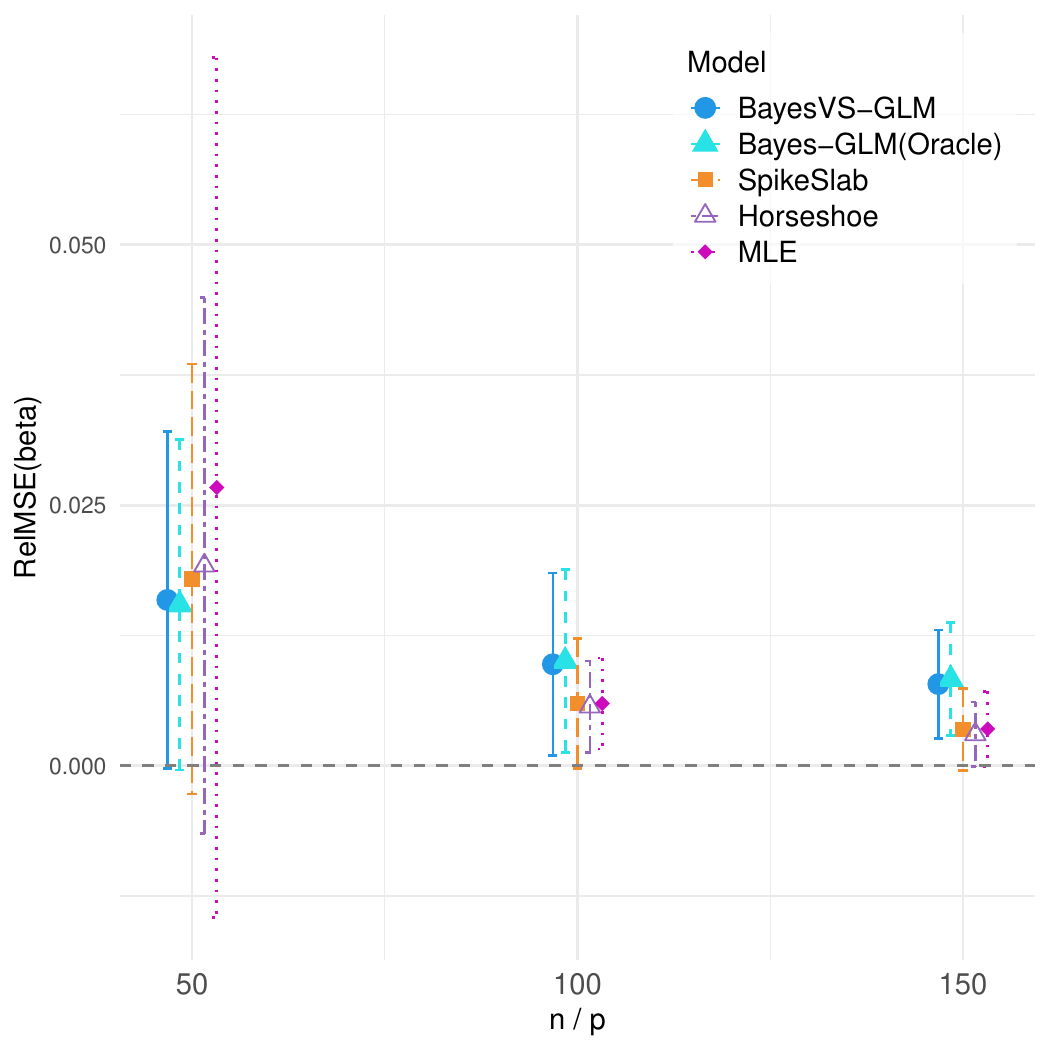}
         \caption{$RelMSE(\vbeta^{(1)}, \vbeta^{*(1)})$}
         \label{fig:bin_beta_rmse}
     \end{subfigure}
    \caption{Logistic Model. Error metrics of the coefficients $\vbeta$, for $p=20$, $d=10$ and $2c=4$, and increasing ratio $n/p$. Our method (\texttt{BayesVS-GLM}) closely matches the oracle version and the \texttt{SpikeSlab} baseline, while it outperforms \texttt{MLE} and \texttt{Horseshoe} in small sample size regimes.}
    \label{fig:settingC_bin_beta_error_np_curve}
\end{figure}

\section{Additional Results on Real dataset}
\label{apd:results_real_data}
We here report supplementary results on real data experiments presented in the paper. 

\subsection{Logistic Model for Heart Disease Dataset}
\label{apd:subsec_real_dataset_logistic}
To evaluate prediction accuracy on the test sets, we use multiple metrics, namely Balanced Accuracy, Detection Prevalence, Detection Rate, F1 score, Negative Predicted Value, Positive Predicted Value, Precision, Prevalence, Sensitivity, Specificity. 

Table~\ref{tab:error_heart} summarizes all metrics.
\begin{table}[ht]
\small
\centering
\begin{tabular}{llllll}
    \hline
    Metric & BayesVS-GLM & Bayes-GLM & Horseshoe & SpikeSlab & MLE \\ 
    \hline
    Balanced Accuracy & $\textbf{0.81} \pm \textbf{0.04}$ & $\textbf{0.81} \pm \textbf{0.04}$ & $0.8 \pm 0.04$ & $\textbf{0.81} \pm \textbf{0.04}$ & $0.8 \pm 0.04$ \\ 
    Detection Prevalence & $0.52 \pm 0.04$ & $0.52 \pm 0.04$ & $0.52 \pm 0.04$ & $0.52 \pm 0.04$ & $0.52 \pm 0.04$ \\ 
    Detection Rate & $0.42 \pm 0.03$ & $0.42 \pm 0.03$ & $0.42 \pm 0.04$ & $0.42 \pm 0.04$ & $0.42 \pm 0.03$ \\ 
    F1 & $\textbf{0.81} \pm \textbf{0.03}$ & $0.81 \pm 0.04$ & $0.81 \pm 0.04$ & $\textbf{0.81} \pm \textbf{0.03}$ & $0.81 \pm 0.04$ \\ 
    Neg Pred Value & $0.8 \pm 0.06$ & $0.8 \pm 0.06$ & $0.8 \pm 0.07$ & $0.8 \pm 0.06$ & $0.8 \pm 0.07$ \\ 
    Pos Pred Value & $0.81 \pm 0.06$ & $0.81 \pm 0.06$ & $0.81 \pm 0.06$ & $0.81 \pm 0.06$ & $0.81 \pm 0.05$ \\ 
    Precision & $\textbf{0.81} \pm \textbf{0.05}$ & $0.81 \pm 0.06$ & $0.81 \pm 0.06$ & $0.81 \pm 0.06$ & $0.81 \pm 0.06$ \\ 
    Prevalence & $0.52 \pm 0.05$ & $0.52 \pm 0.05$ & $0.52 \pm 0.05$ & $0.52 \pm 0.05$ & $0.52 \pm 0.05$ \\ 
    Sensitivity & $\textbf{0.82} \pm \textbf{0.05}$ & $0.81 \pm 0.05$ & $0.81 \pm 0.05$ & $0.81 \pm 0.05$ & $\textbf{0.82} \pm \textbf{0.06}$ \\ 
    Specificity & $\textbf{0.8} \pm \textbf{0.06}$ & $\textbf{0.8} \pm \textbf{0.06}$ & $0.79 \pm 0.06$ & $\textbf{0.8} \pm \textbf{0.06}$ & $0.79 \pm 0.06$ \\ 
   \hline
\end{tabular}
\caption{Heart Disease dataset. Prediction error metrics across $30$-fold cross-validation (mean $\pm$ standard deviation). In bold the best model. \texttt{BayesVS-GLM} reach slightly higher (or equivalent) performances of the classical \texttt{MLE} and the Bayesian baselines \texttt{BayesGLM}, \texttt{SpikeSlab}, \texttt{Horseshoe}.}
\label{tab:error_heart}
\end{table}


Table \ref{tab:most_frequent_model} reports the $10$ most frequent ones, covering the $93.2\%$ of all the explored models. These models always contain the same variables highlighted in Figure \ref{fig:heart_z} and they mainly differ by the variables \texttt{age}, \texttt{cp\_typical\_angina} and \texttt{chol}.

\begin{table}[ht]
\centering
\begin{tabular}{ll}
  \hline
  Model  & Freq \\ 
  \hline
  0111000000111 & 1267 \\ 
  0111000100111 & 149 \\ 
  1111000000111 & 122 \\ 
  0111010000111 & 73 \\ 
  0111100000111 & 61 \\ 
  0111000010111 & 22 \\ 
  0111001000111 & 20 \\ 
  1111100000111 & 14 \\ 
  1111000100111 & 13 \\ 
  0111100100111 & 11 \\ 
   \hline
\end{tabular}
\caption{Heart Disease dataset. Frequency of the first $10$ selected models, out of $1800$ MCMC iterations, in one repetition.}
\label{tab:most_frequent_model}
\end{table}


For more a more detailed view of the pairwise association among variables, we provide  Figure~\ref{fig:heart_pairwise_inclusion} that reports the relative frequencies of inclusion over all the MCMC replications.
\begin{figure}[H]
    \centering
    \includegraphics[width=0.6\linewidth]{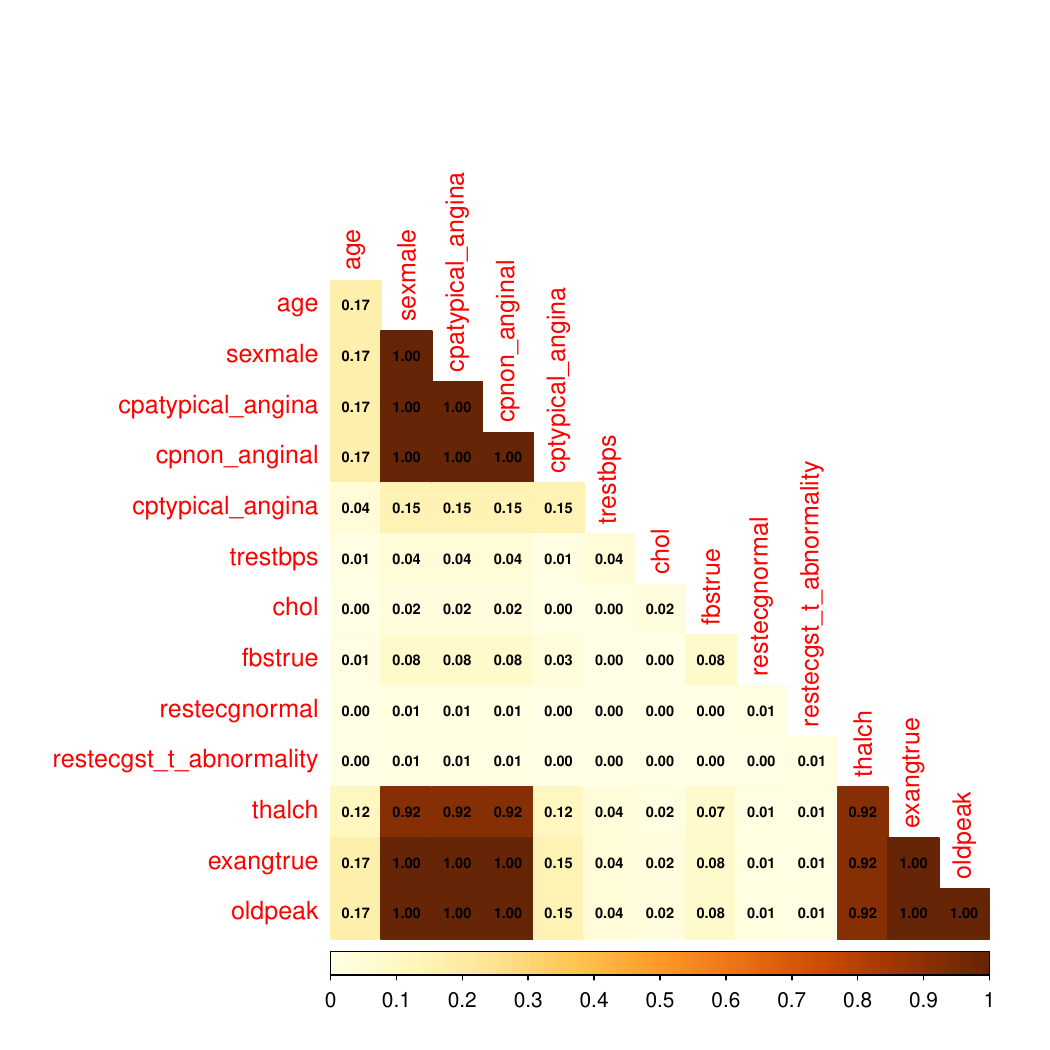}
    \caption{Heart Disease dataset. Pairwise relative frequencies of inclusion of our method.}
    \label{fig:heart_pairwise_inclusion}
\end{figure}

\subsection{Poisson Regression for the Crabs dataset}
\label{subsec:real_dataset_poisson}
We consider the data available in \citet[Table 4.3]{agresti2002categorical} that contains field measurements on \emph{female horseshoe crabs} collected to study factors associated with mating success. The \texttt{Crabs} dataset is available in the \textsf{R} package \texttt{glm2} and comprises $173$ records, each representing a single female crab.

The task is modeled as a regression problem where the response variable (\texttt{Satellites}) represents the count of male partners (in addition to the primary partner) accompanying each female and is analyzed using the Poisson model.
The dataset includes the following attributes: number of male crabs partners, excluding the primary partner (satellites); carapace width of the female measured in centimeters (\texttt{width}); binary indicator for dark shell coloring (\texttt{dark}: \texttt{yes} for dark, \texttt{no} otherwise); binary indicator for good spine condition (\texttt{goodspine}: \texttt{yes} for good condition, \texttt{no} otherwise); two indices for a random sample with replacement from $1:173$ (\texttt{rep1} and \texttt{rep2}).

We conducted a $30$-fold cross-validation, each time splitting the data into $90\%$ for training and $10\%$ for testing to evaluate model performance on unseen data. Both our method (\texttt{BayesVS-GLM}) and classical \texttt{MLE} models were estimated on the respective training folds under the same conditions.
The prior hyperparameters for \texttt{BayesVS-GLM} model were fixed to $\DparamScalar_0 = 0.01$, $\DparamVector_0 = 3 \vect{J}_n$, $\alpha=1$.

To illustrate how well the \texttt{BayesVS-GLM} performs variable selection and estimation of the model coefficients, we show a boxplot of the posterior distribution of $\vbeta \circ \z$ obtained from the whole dataset (Fig. \ref{fig:data_crabs_betaz}), as well as the Posterior inclusion probability over $30$ folds repetitions (Fig. \ref{fig:data_crabs_z}).
\begin{figure}[htb]
	\centering
	\begin{subfigure}{0.49\textwidth}
		\centering
		\includegraphics[height=6cm]{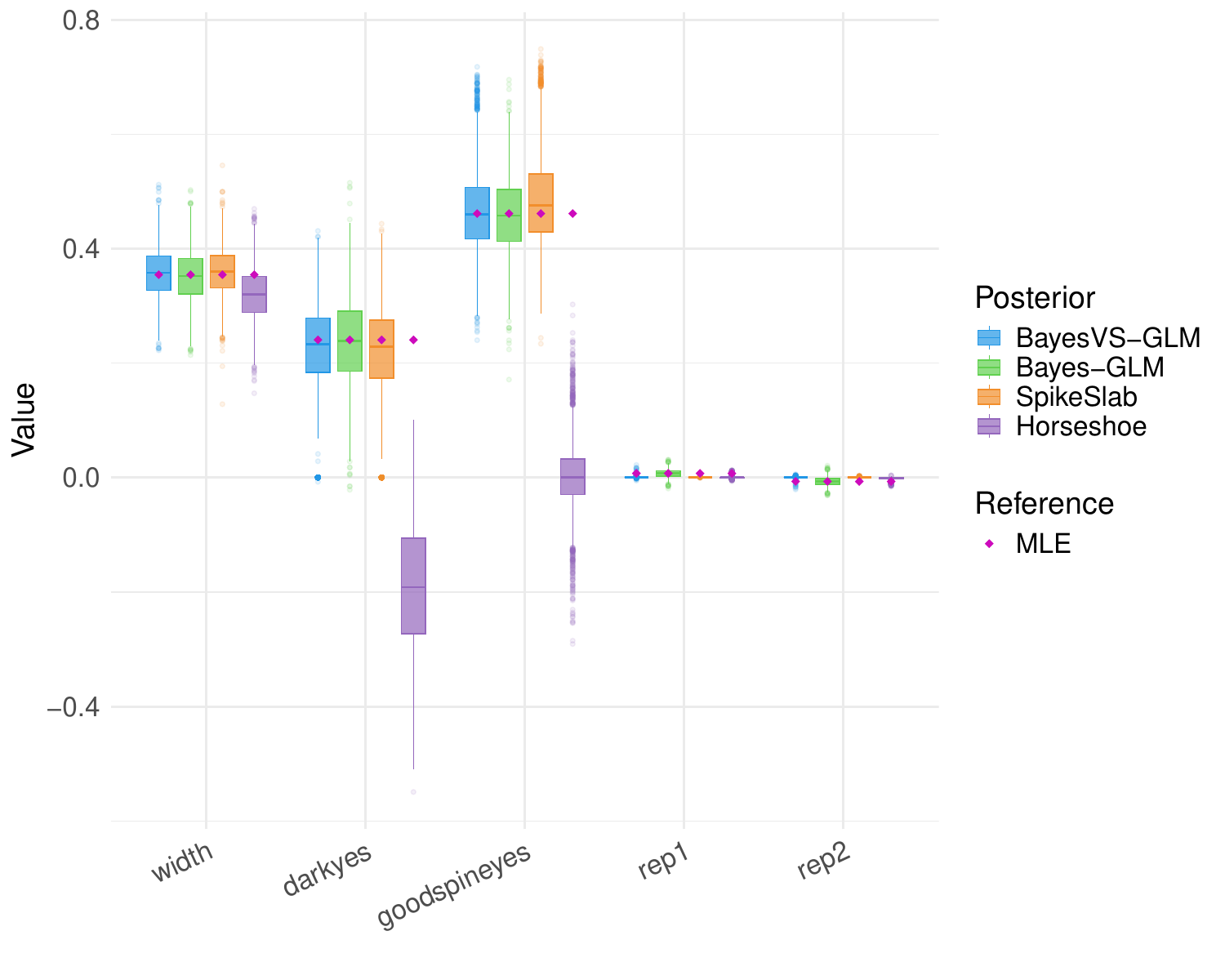}
		\caption{Posterior distribution of $\vbeta \circ \z$ (\texttt{BayesVS-GLM}) and $\vbeta$ (baselines), on the whole dataset.}
		\label{fig:data_crabs_betaz}
	\end{subfigure}
	\hfill
	\begin{subfigure}{0.49\linewidth}
		\centering
		\includegraphics[height=6cm]{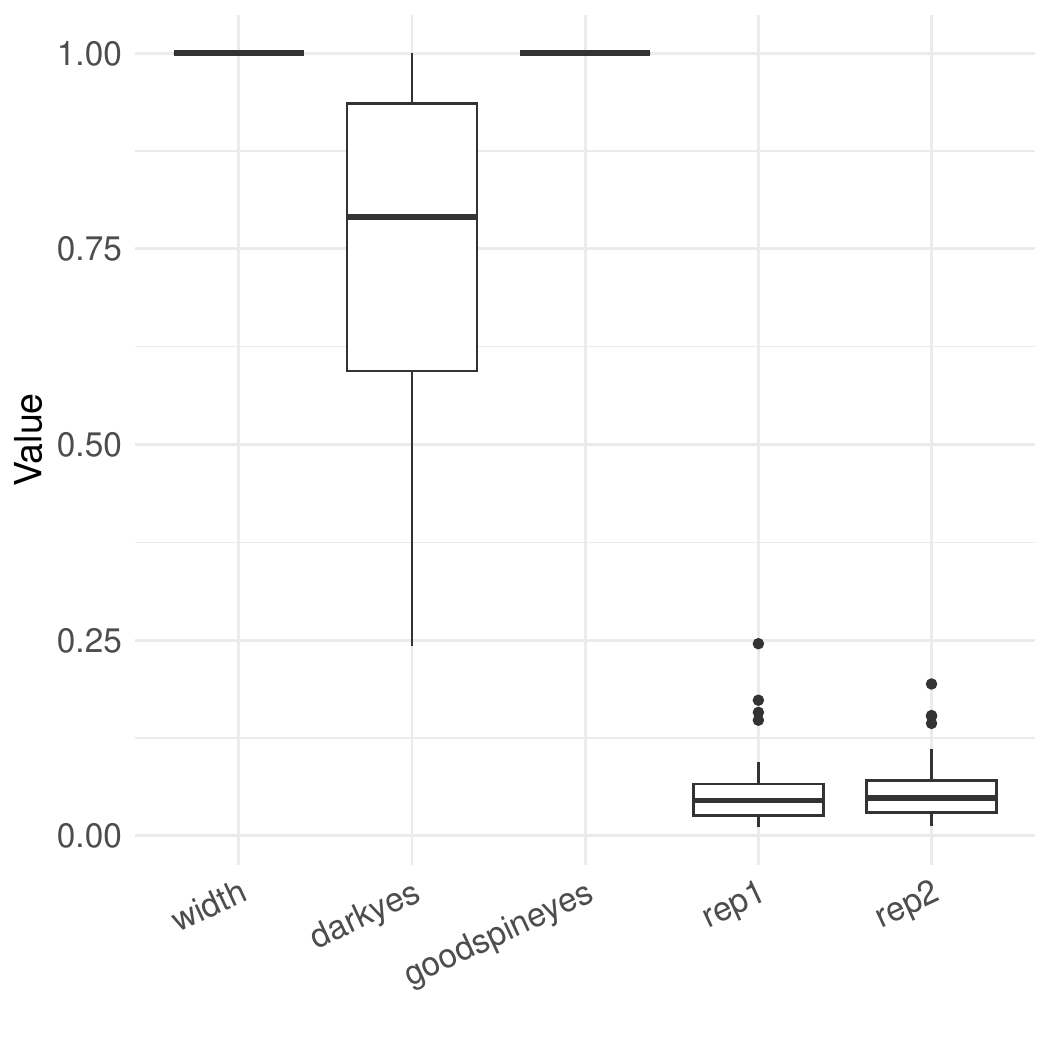}
		\caption{Posterior inclusion probabilities of our method, over $30$ folds repetitions.}
		\label{fig:data_crabs_z}
	\end{subfigure}
	\caption{Crabs dataset.}
	\label{fig:beta_z_crabs}
\end{figure}
In particular, the variables \texttt{rep1} and \texttt{rep2}, which represent random samples added as noise (as noted in the dataset documentation), are discarded through the variable selector $\z$. This indicates that, as we expected, the model effectively identifies and excludes pure noise covariates. Comparing the posterior distributions of the coefficients against the baselines (Fig.~\ref{fig:data_crabs_betaz}), \texttt{BayesVS-GLM} yields more accurate estimates than \texttt{Horseshoe}, while achieving results in par with \texttt{SpikeSlab}.

We calculated the Mean Absolute Error (MAE) and Root Mean Squared Error (RootMSE) on test set outcomes of interest to evaluate each model's prediction accuracy. The distribution of errors across folds is summarized in Figure \ref{fig:error_crabs}.
\begin{figure}[htb]
	\centering
	\includegraphics[width=0.7\linewidth]{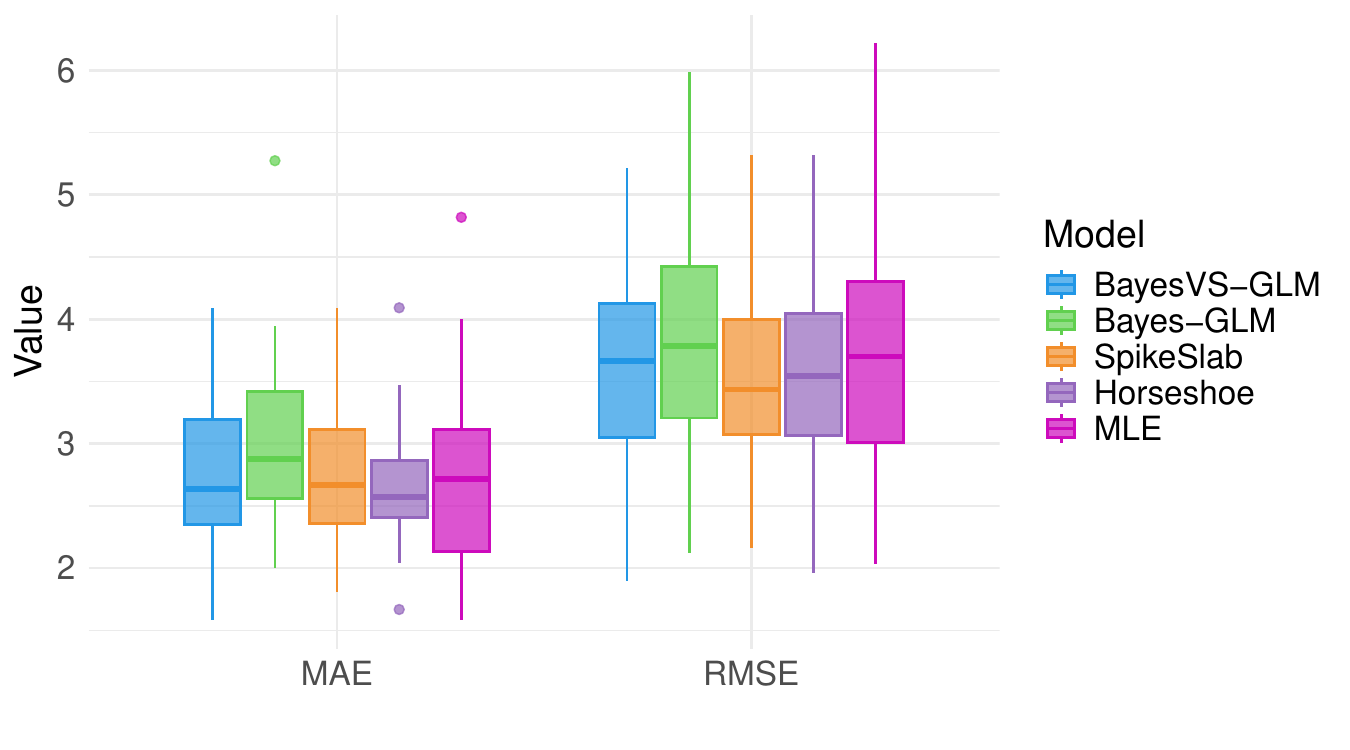}
	\caption{Crabs dataset. Prediction error metrics across folds.  Predictive accuracy for our method (\texttt{BayesVS-GLM})  in comparison to the classic Bayes GLM approach (\texttt{Bayes-GLM}), the maximum likelihood estimator (\texttt{MLE}), and state-of-the-art baselines (\texttt{SpikeSlab}, \texttt{Horseshoe}).}
	\label{fig:error_crabs}
\end{figure}

\subsection{Linear Model for the Pollution dataset}
\label{subsec:real_dataset_gaussian}
In this example, we evaluate the performance of our model on the \texttt{McDonald Pollution Dataset}, first introduced by \citet{mcdonald1973instabilities}, and commonly used for studying variable selection problems (see \eg \citep{ohara2009review}). The dataset is available in the \textsf{R} package \texttt{bestglm}.

The response variable is the age-adjusted mortality rates in $1963$, from $60$ metropolitan areas of the United States. The available covariates are:
average annual precipitation in inches (\texttt{prec});
average January temperature in °F (\texttt{jant});
average July temperature in °F (\texttt{jult});
percent of $1960$ SMSA population aged 65 or older (\texttt{ovr65});
average household size (\texttt{popn});
median years of schooling among individuals aged $22$ or older (\texttt{educ});
percent of housing units which are sound and with all facilities (\texttt{hous});
population density per sq. mile in urbanized areas, $1960$ (\texttt{dens});
percent non-white population in urbanized areas, $1960$ (\texttt{nonw});
percent employed in white collar occupations (\texttt{wwdrk});
percent of families with income below \$$3000$ (\texttt{poor});
relative nitric oxides pollution potential (\texttt{nox});
relative sulfur dioxide pollution potential (\texttt{sox});
annual average percent relative humidity at $1$pm (\texttt{humid});
total age-adjusted mortality rate per $100,000$ (\texttt{target}).

We removed the variable \texttt{hc}, (relative hydrocarbon pollution potential) due to its correlation of 0.99 with \texttt{nox}.
After scaling and centering each covariate, we conducted a $30$-fold cross-validation, each time splitting the data into $90\%$ for training and $10\%$ for testing to evaluate model performance on unseen data. 
Three models were estimated under identical settings:
\texttt{BayesVS-GLM}, our proposed Bayesian GLM with variable selection; \texttt{BayesGLM}, a Bayesian GLM with all inclusion variables fixed to one (no sparsity; its frequentist counterpart \texttt{MLE}, \ie the classical maximum likelihood estimator implemented in \textsf{R} via \texttt{glm}.

The prior hyperparameters for \texttt{BayesVS-GLM} model were set to $\DparamScalar_0 = 0.1$, $\DparamVector_0 \overset{\iid}{\sim} \mathcal{N}(0,1)$, $\alpha=p=15$.

To illustrate the performances of the \texttt{BayesVS-GLM} in  variable selection and estimation of the model coefficients, Figure \ref{fig:pollution_betaz} shows the posterior distribution of $\vbeta \circ \z$, obtained on the whole dataset,
while Figure \ref{fig:pollution_z} reports the posterior inclusion probabilities over the $30$ cv-repetitions.
\begin{figure}[H]
	\centering
	\includegraphics[width=\textwidth]{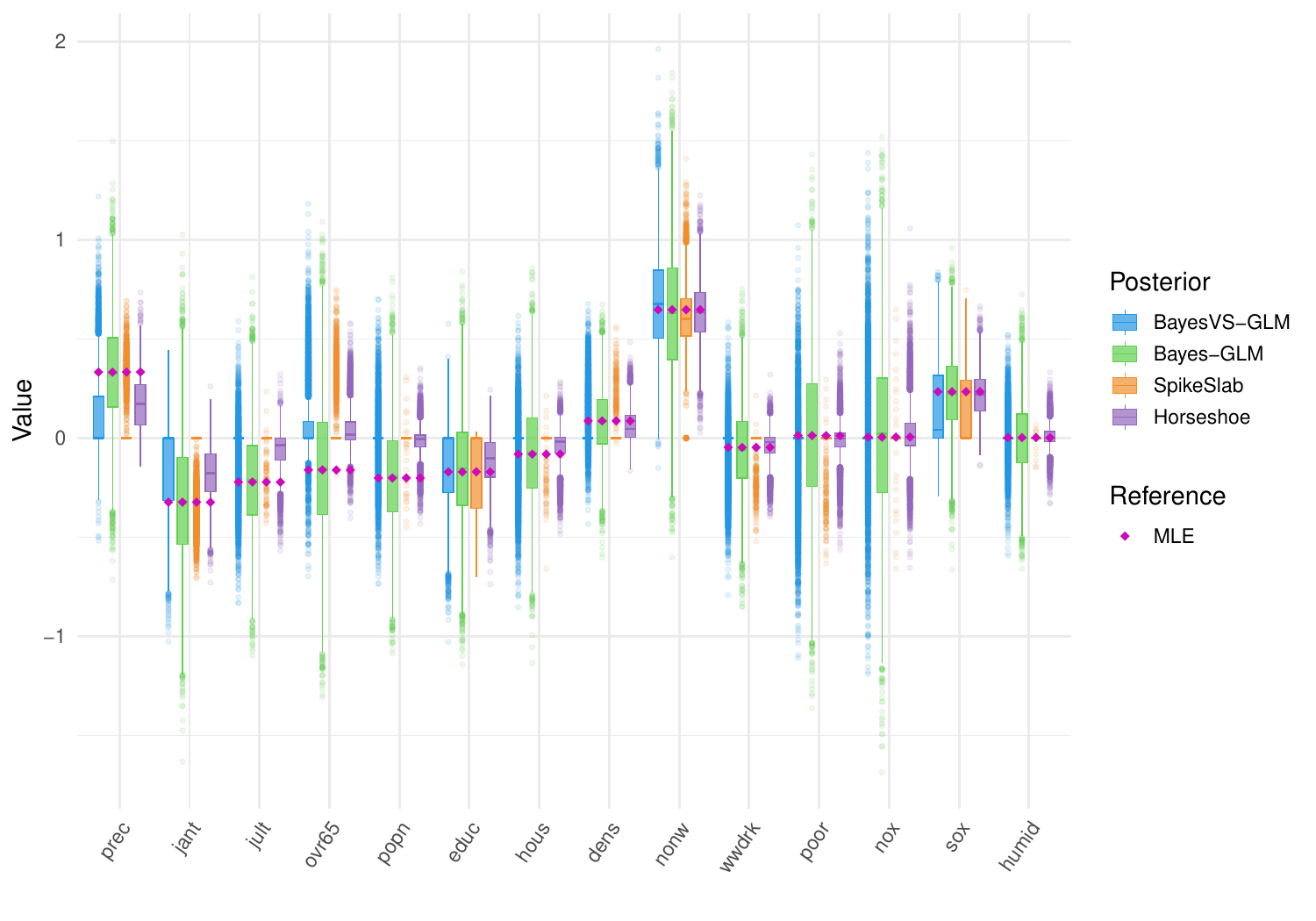}
	\caption{Pollution dataset. Posterior distribution of $\vbeta \circ \z$ (\texttt{BayesVS-GLM}) and $\vbeta$ (baselines) on the whole dataset.}
	\label{fig:pollution_betaz}
\end{figure}
\begin{figure}[H]
	\centering
	\includegraphics[width=0.5\linewidth]{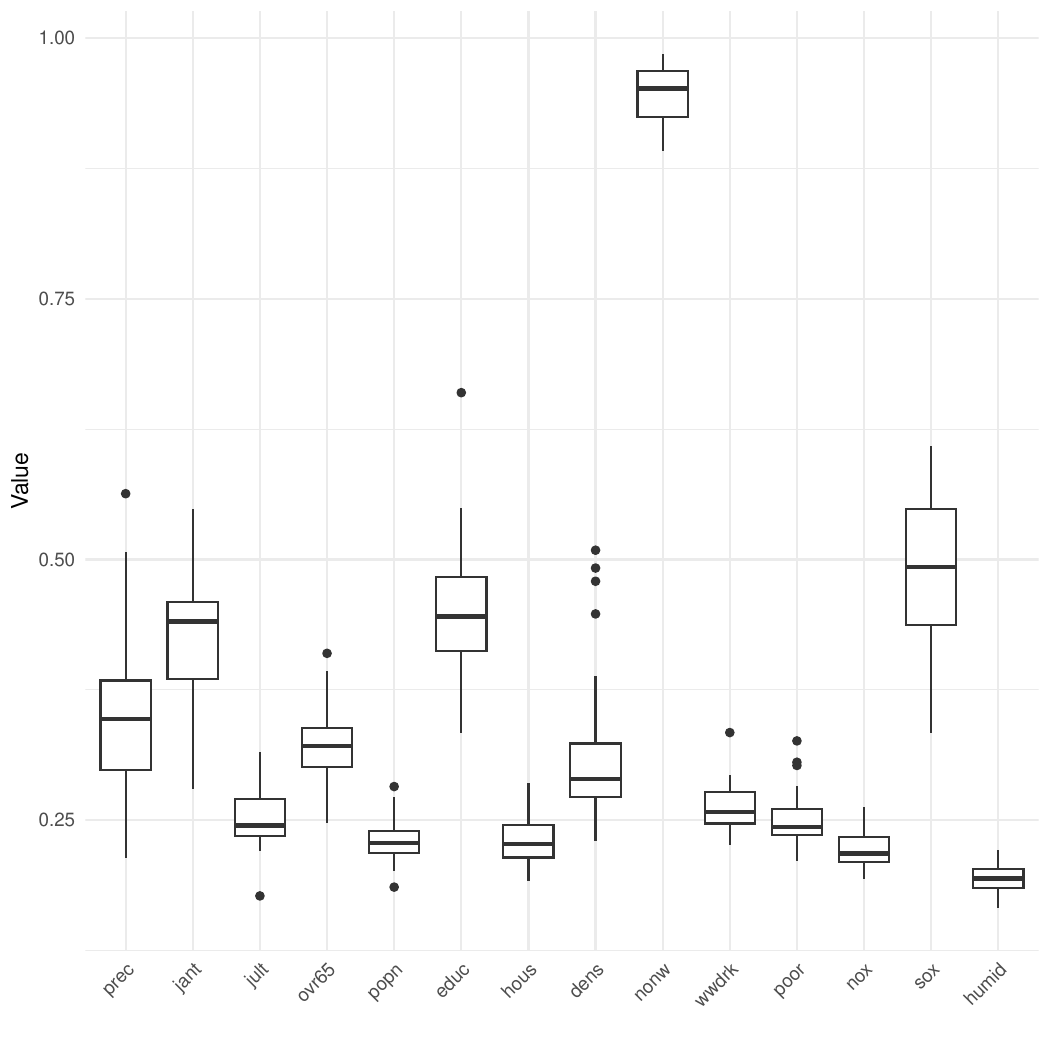}
	\caption{Pollution dataset. \texttt{BayesVS-GLM}'s posterior inclusion probabilities over $30$ folds.}
	\label{fig:pollution_z}
\end{figure}
Across the model space of dimension $2^{p} -1 = 32767$, the Gibbs sampler explored $3025$ distinct models during the MCMC run. After discarding those visited only $1$ or $2$ times, $954$ models remain.
The most frequent models consistently include \texttt{nonw}, while there's a disagreement regarding \texttt{sox} and \texttt{educ}, and even more variability for \texttt{prec} and \texttt{jant}. 

Although some of these variables overlap with those highlighted in \citet{ohara2009review} (who report higher marginal posterior probabilities for precipitation, January temperature, education, percentage non-white, and \texttt{sox}), our results do not replicate a clear separation between ``important'' and ``uninformative'' predictors. Instead, our posterior distribution reveals marked uncertainty across many inclusion indicators. 

For more a more detailed view of the pairwise association among variables, see Figure~\ref{fig:pollution_pairwise_inclusion} that reports the relative frequencies of inclusion over all the MCMC replications.
\begin{figure}[H]
	\centering
	\includegraphics[width=0.5\linewidth]{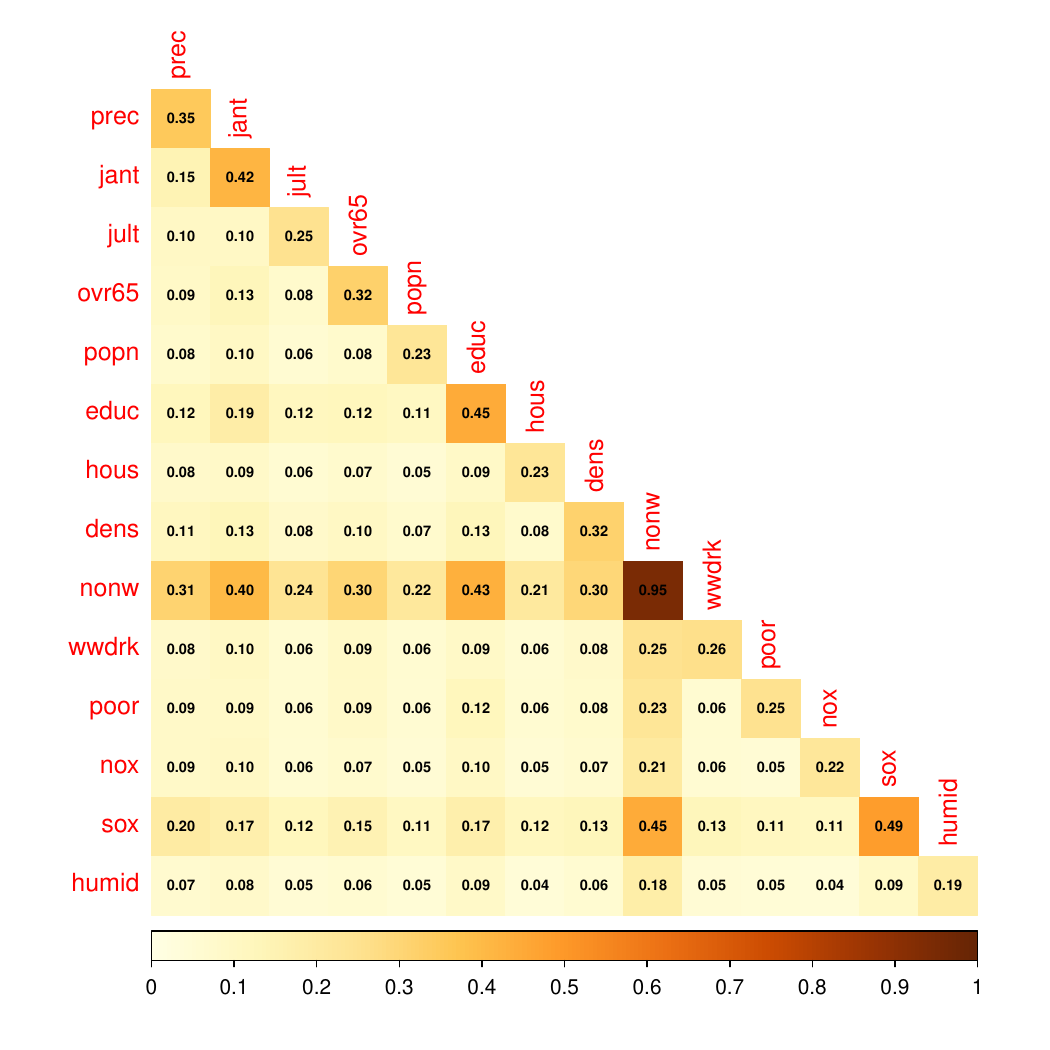}
	\caption{Pollution dataset. Pairwise relative frequencies of inclusion.}
	\label{fig:pollution_pairwise_inclusion}
\end{figure}

To assess predictive performance, we report Adjusted R-Squared (\texttt{AdjR2}), Mean Absolute Error (\texttt{MAE}), and Root Mean Squared Error (\texttt{RMSE}), aggregated across the 30 cross-validation splits (Figure~\ref{fig:error_pollution}).
Overall, \texttt{BayesVS-GLM} attains similar predictive accuracy to the evaluated baselines.
\begin{figure}[H]
	\centering
	\includegraphics[width=0.7\linewidth]{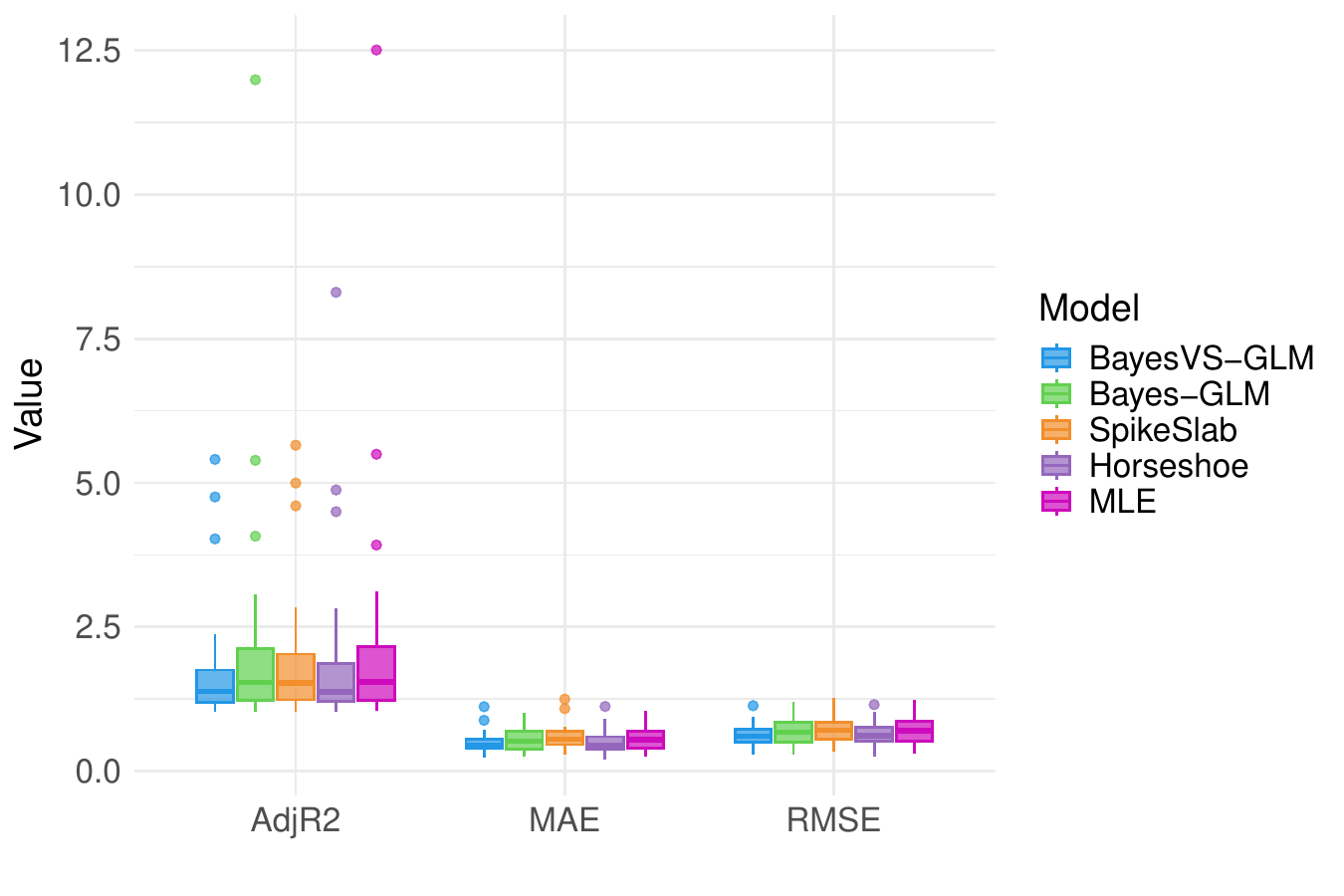}
	\caption{Pollution dataset. Prediction error metrics across folds: Adjusted R2 (left), Mean Absolute Error (center), and Root Means Squared Error (right), for the five methods.}
	\label{fig:error_pollution}
\end{figure}

\section{Effect of correlation on variable selection}
\label{apd:results_correlation}
To assess the impact of correlation among covariates on variable selection, we conduct the following controlled experiment.
We generate data according to:
\[
Y \sim f(\eta \mid \X), \ \eta = X_1 \beta, \ \X = (X_1, X_2)
\]
where $f$ is specified by the particular member of the exponential family, and  $X_2$ is constructed as
\[
X_2 = \gamma X_1 + (1-\gamma)W, \ W \sim \mathcal{N}(0,1)
\]
and $\gamma \in [0, 1]$ controls the correlation between $X_1$ and $X_2$.
When $\gamma = 0$, the two covariates are independent and $X_2$ carries no information about the response. As $\gamma \rightarrow 1$, $X_2$ becomes more correlated to $X_1$, making variable selection more difficult.

For each value of $\gamma$, we fit our model (\texttt{BayesVS-GLM}) over multiple random seeds.

\subsection{Linear Model}
The chosen hyperparameters of the prior distributions for the Linear case are: $\DparamScalar_0 = 10^{-4}$, $\DparamVector_0 \overset{\iid}{\sim} \mathcal{N}(1,4)$, and $\alpha=1$.
The sample size is $n=100$ and the posterior inference is carried out using Gibbs Sampling with $1000$ iterations and a burn-in of $10\%$, for $20$ multiple random seeds.

Figure~\ref{fig:correlation_boxplots_inclusionprob} reports boxplots of the posterior inclusion probabilities $\Pr(z_j = 1)$ for the two covariates across different values of $\gamma$.
\begin{figure}[ht!]
    \centering
    \begin{subfigure}[b]{0.3\linewidth}
        \centering
        \includegraphics[width=\textwidth]{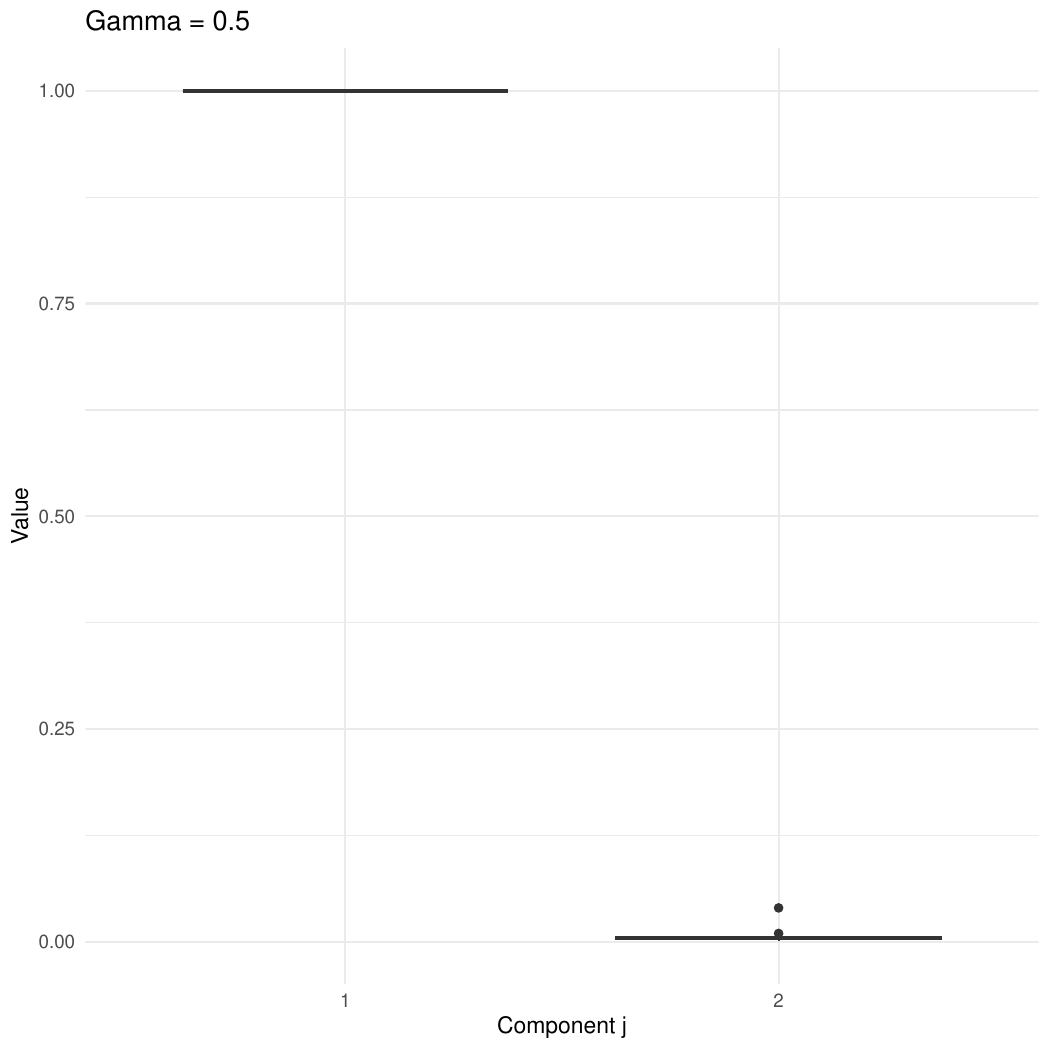}
        \caption{$\gamma = 0.5$}
    \end{subfigure}
    \hfill
    \begin{subfigure}[b]{0.3\linewidth}
        \centering
        \includegraphics[width=\textwidth]{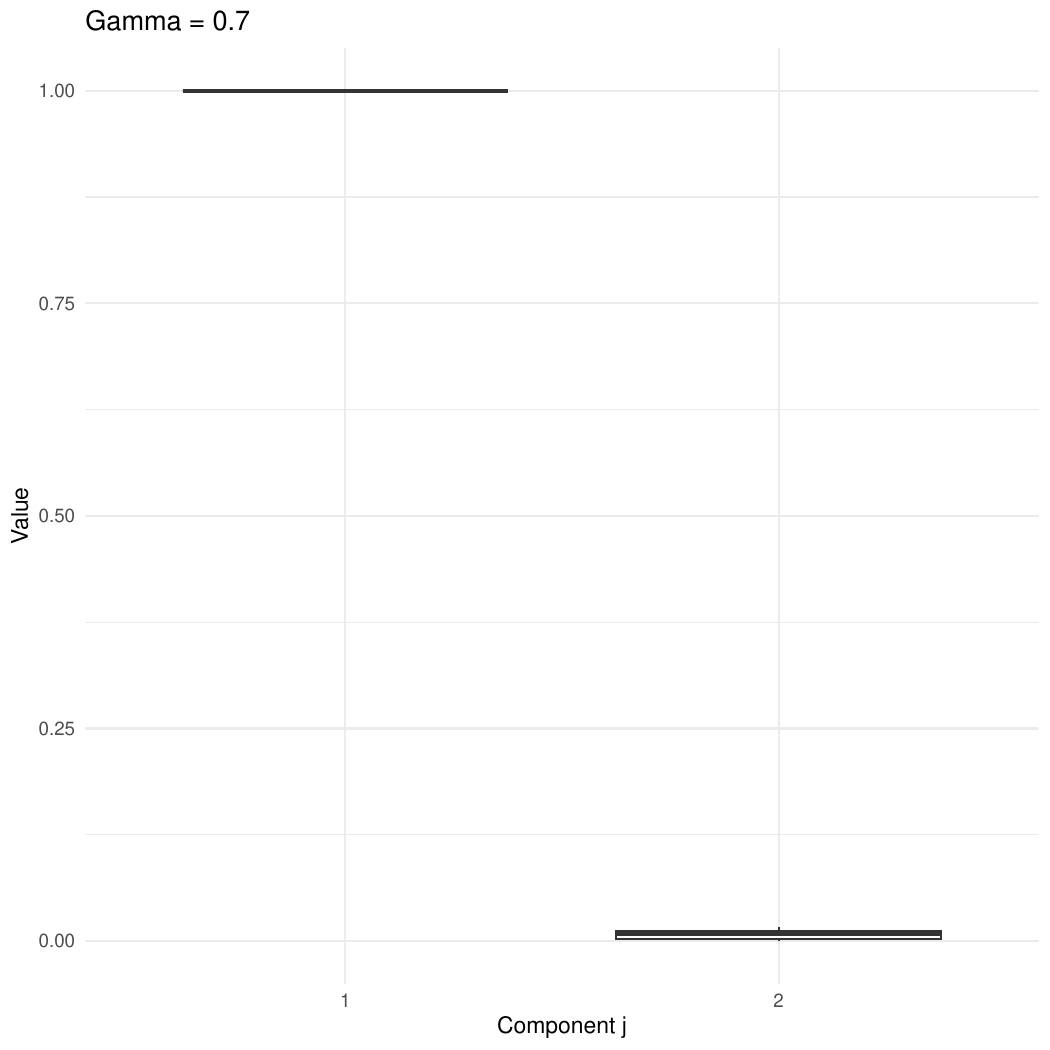}
        \caption{$\gamma = 0.7$}
    \end{subfigure}
    \hfill
    \begin{subfigure}[b]{0.3\linewidth}
        \centering
        \includegraphics[width=\textwidth]{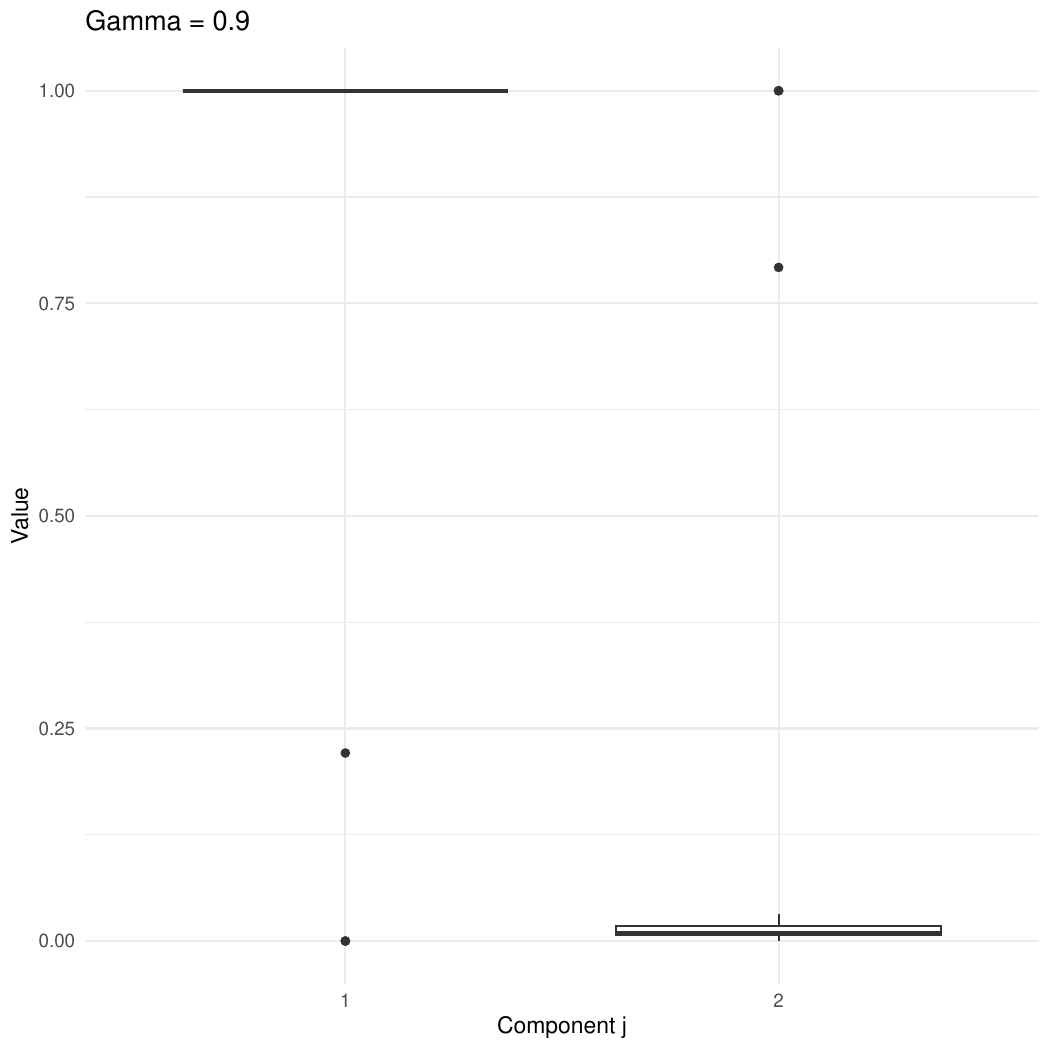}
        \caption{$\gamma = 0.9$}
    \end{subfigure}
    \hfill
    \begin{subfigure}[b]{0.3\linewidth}
        \centering
        \includegraphics[width=\textwidth]{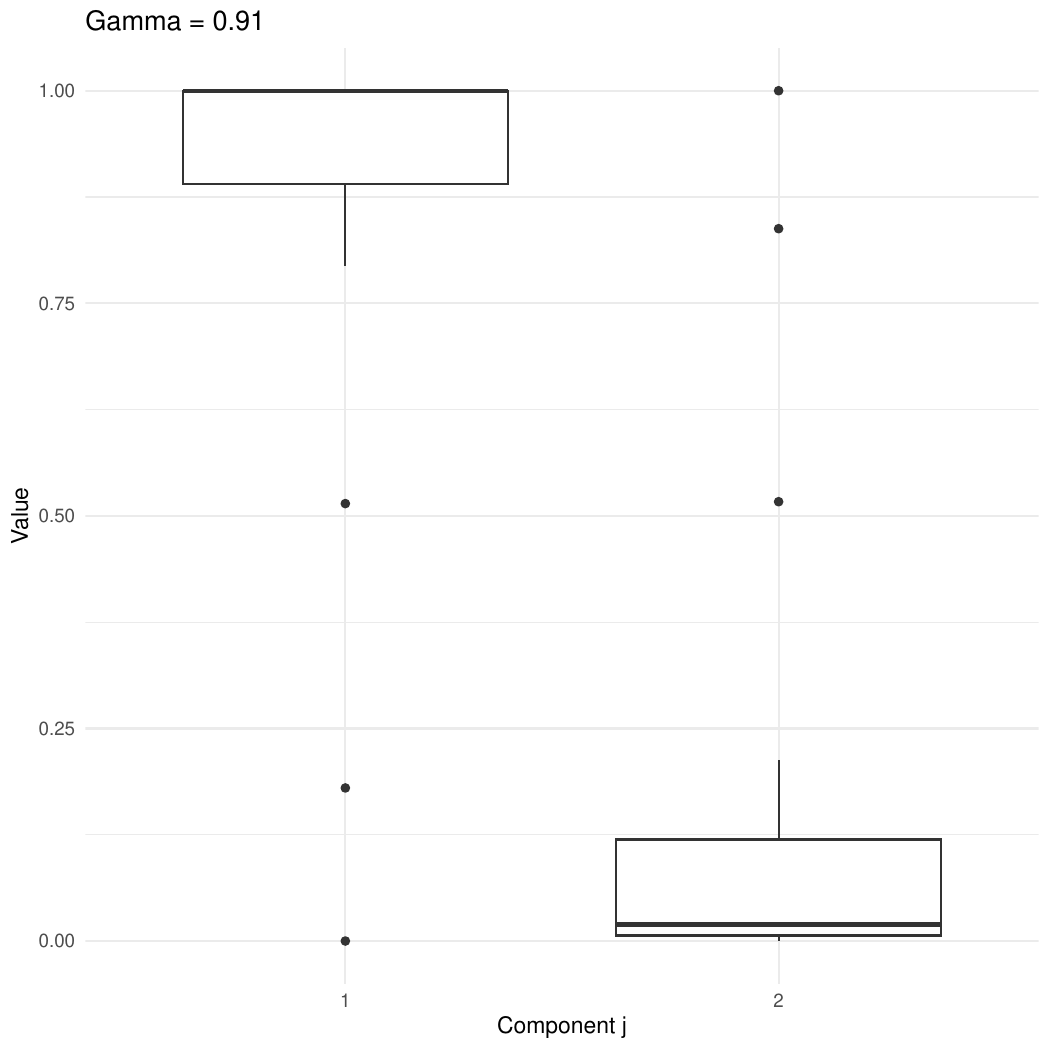}
        \caption{$\gamma = 0.91$}
    \end{subfigure}
    \hfill
    \begin{subfigure}[b]{0.3\linewidth}
        \centering
        \includegraphics[width=\textwidth]{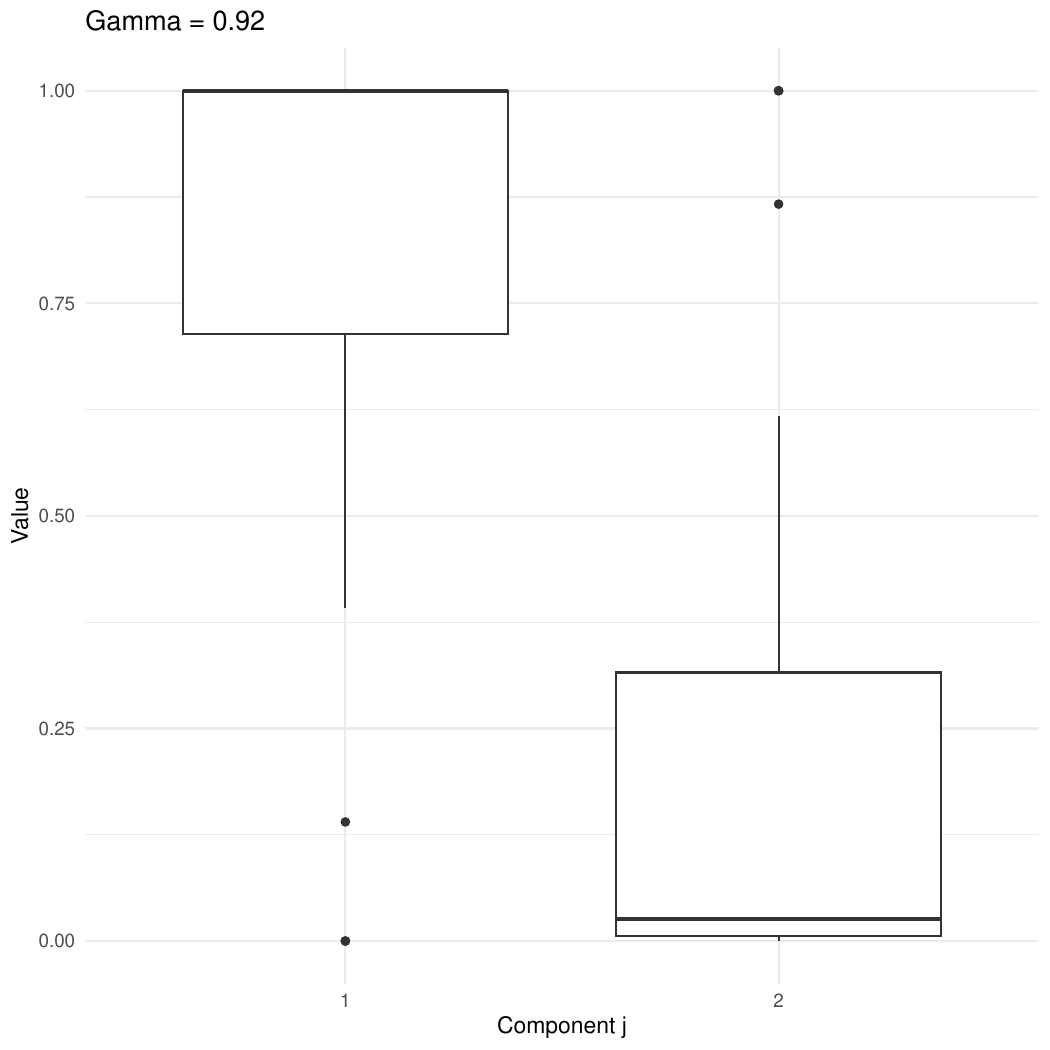}
        \caption{$\gamma = 0.92$}
    \end{subfigure}
    \hfill
    \begin{subfigure}[b]{0.3\linewidth}
        \centering
        \includegraphics[width=\textwidth]{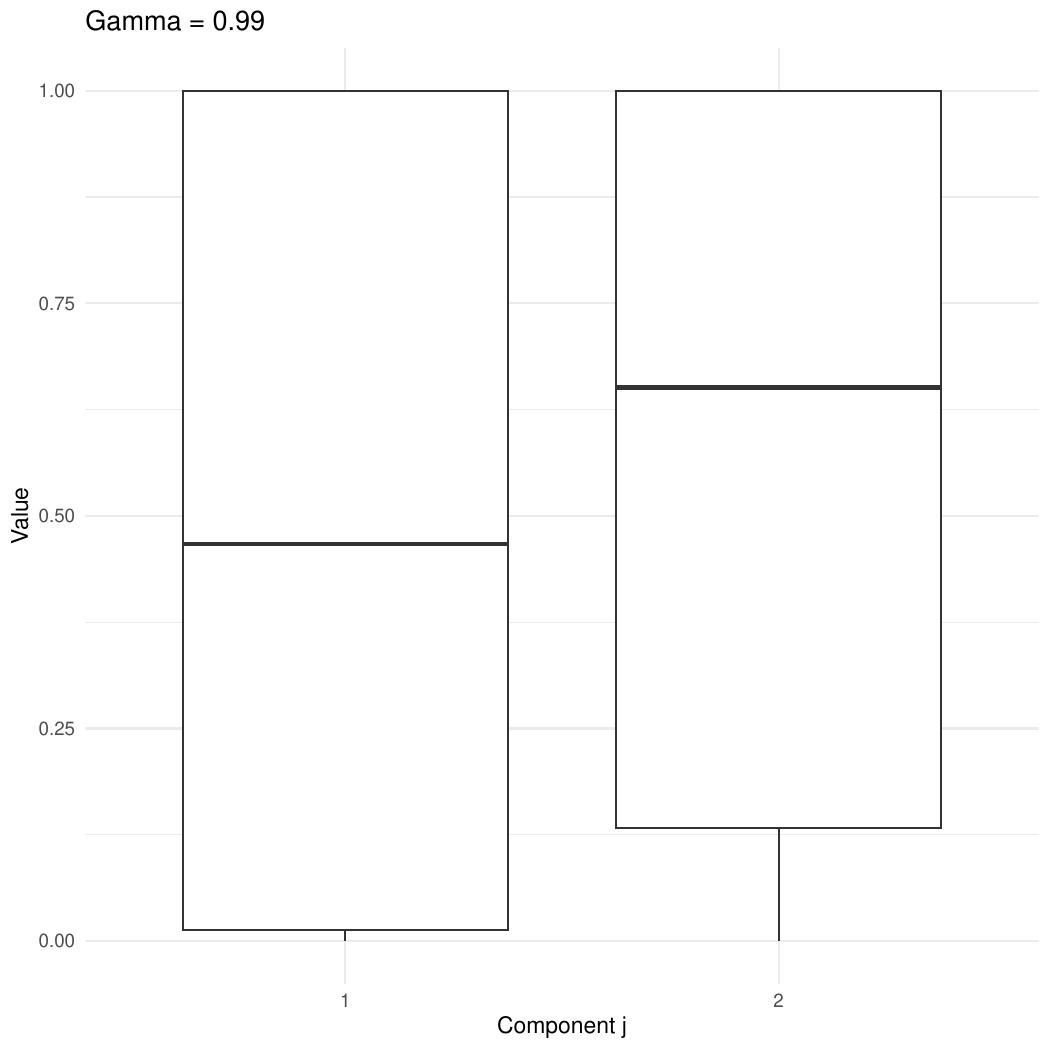}
        \caption{$\gamma = 0.99$}
    \end{subfigure}
    \caption{Linear Model. Posterior inclusion probabilities $\Pr(z_j = 1)$ for the two covariates across different values of $\gamma$. Each boxplot summarizes the results over multiple random seeds. As correlation increases, the distinction between informative ($X_1$) and redundant ($X_2$) covariates becomes less clear.}
    \label{fig:correlation_boxplots_inclusionprob}
\end{figure}

While boxplots in Figure~\ref{fig:correlation_boxplots_inclusionprob} illustrate the distribution of inclusion probabilities for each covariate, they do not directly quantify how overall selection performance varies with the degree of correlation.
To capture this relationship more explicitly, we report the selection accuracy (proportion of correctly recovered entries of $\mathbf z$) across seeds, for different values of $\gamma$ (Figure~\ref{fig:correlation_boxplots_acc}); and the average and median accuracy across repetitions as $\gamma$ increases (Figure~\ref{fig:correlation_acc_vs_gamma}).
\begin{figure}[ht!]
    \centering
    \begin{subfigure}[b]{0.3\linewidth}
        \centering
        \includegraphics[width=\textwidth]{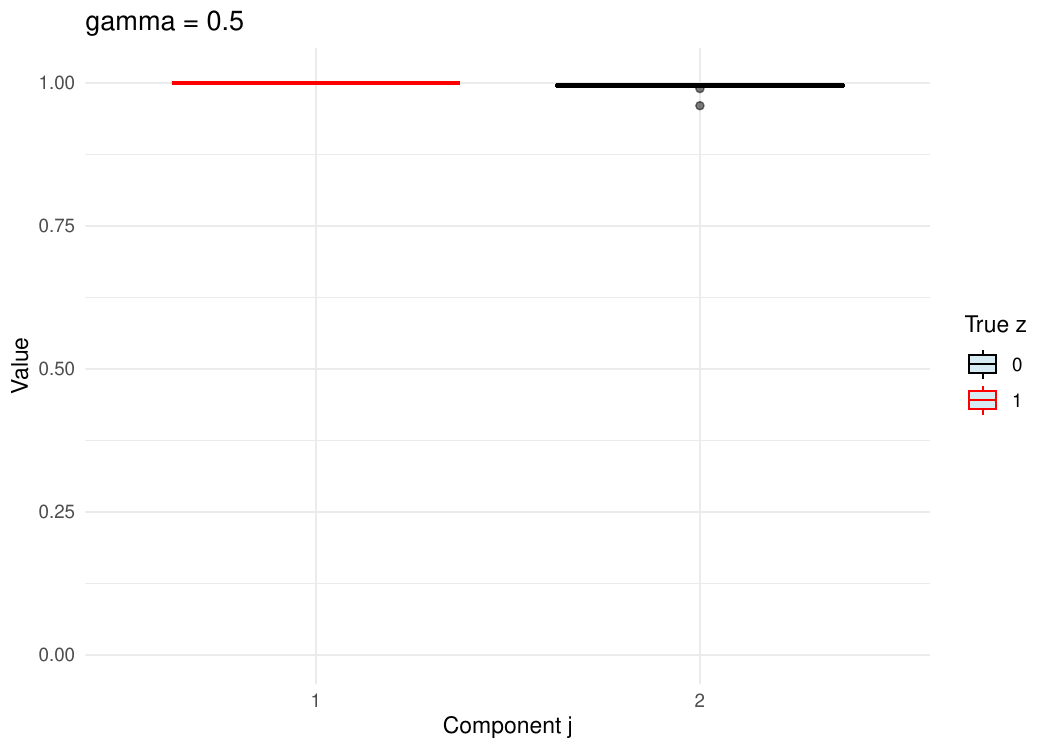}
        \caption{$\gamma = 0.5$}
    \end{subfigure}
    \hfill
    \begin{subfigure}[b]{0.3\linewidth}
        \centering
        \includegraphics[width=\textwidth]{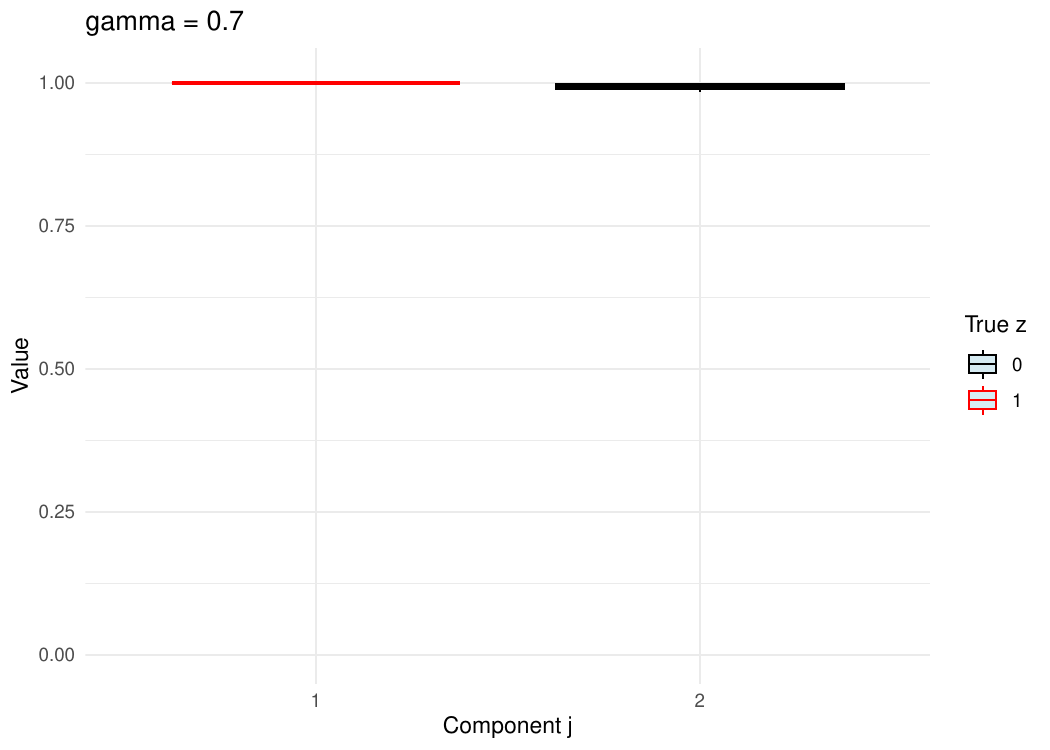}
        \caption{$\gamma = 0.7$}
    \end{subfigure}
    \hfill
    \begin{subfigure}[b]{0.3\linewidth}
        \centering
        \includegraphics[width=\textwidth]{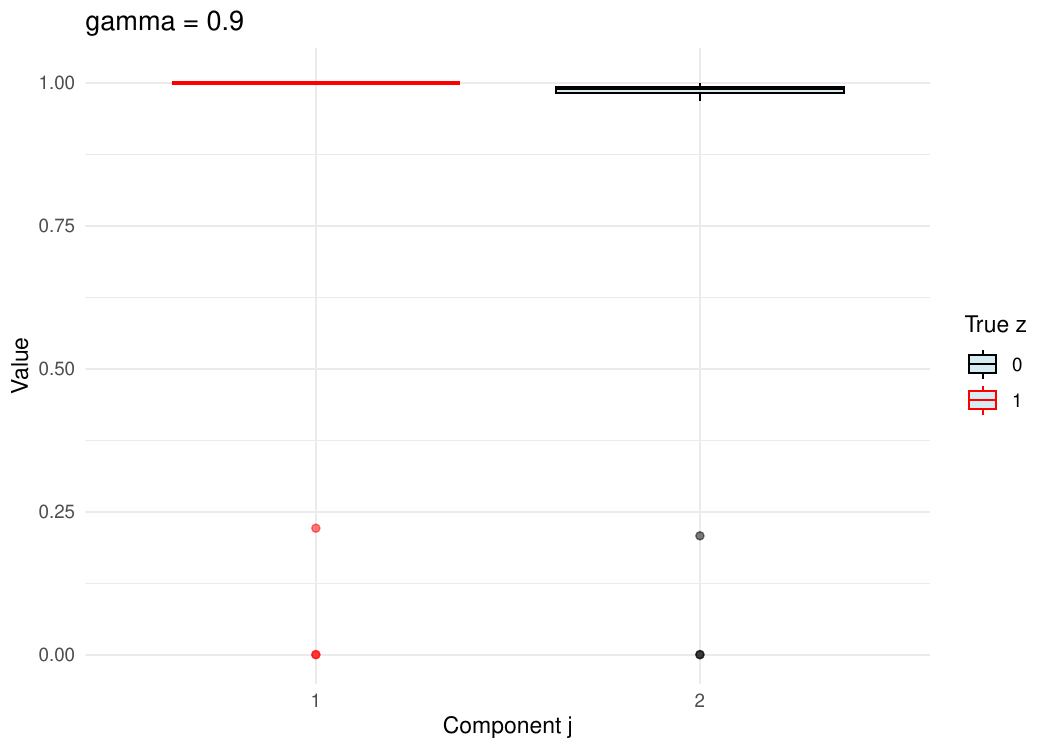}
        \caption{$\gamma = 0.9$}
    \end{subfigure}
    \hfill
    \begin{subfigure}[b]{0.3\linewidth}
        \centering
        \includegraphics[width=\textwidth]{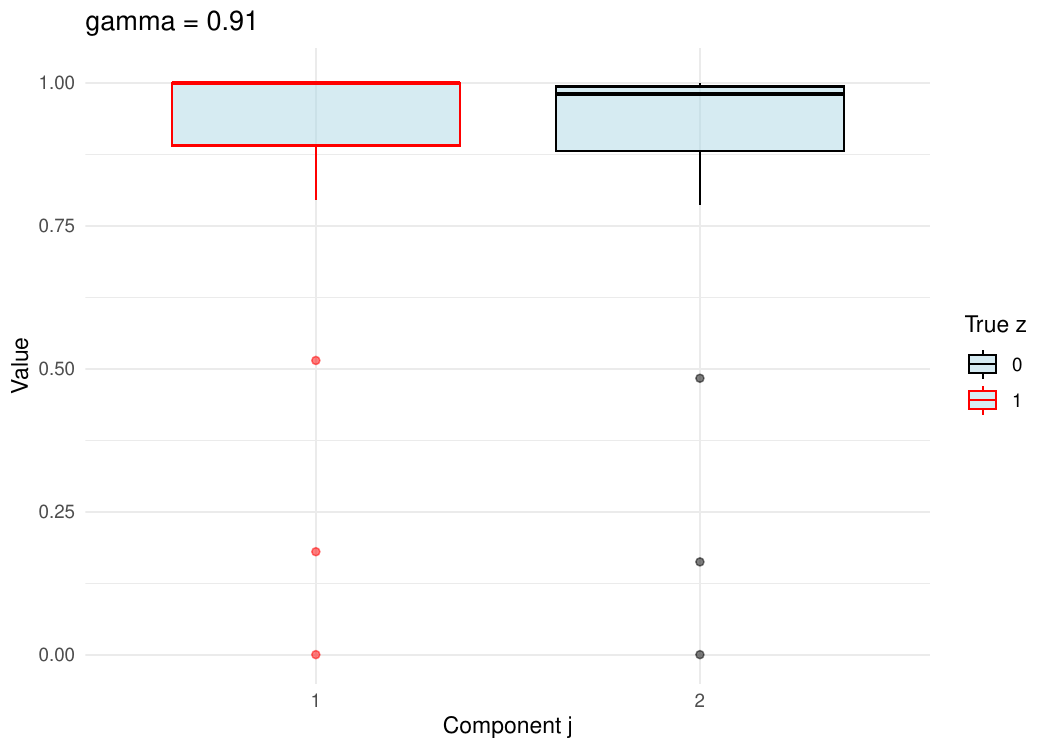}
        \caption{$\gamma = 0.91$}
    \end{subfigure}
    \hfill
    \begin{subfigure}[b]{0.3\linewidth}
        \centering
        \includegraphics[width=\textwidth]{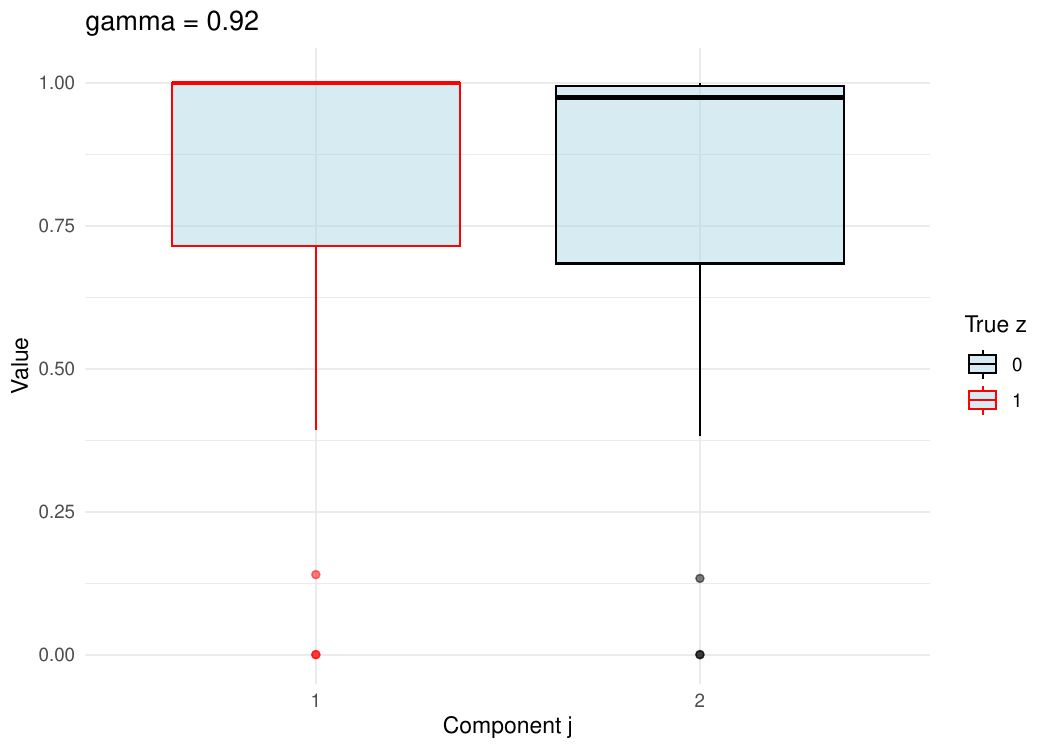}
        \caption{$\gamma = 0.92$}
    \end{subfigure}
    \hfill
    \begin{subfigure}[b]{0.3\linewidth}
        \centering
        \includegraphics[width=\textwidth]{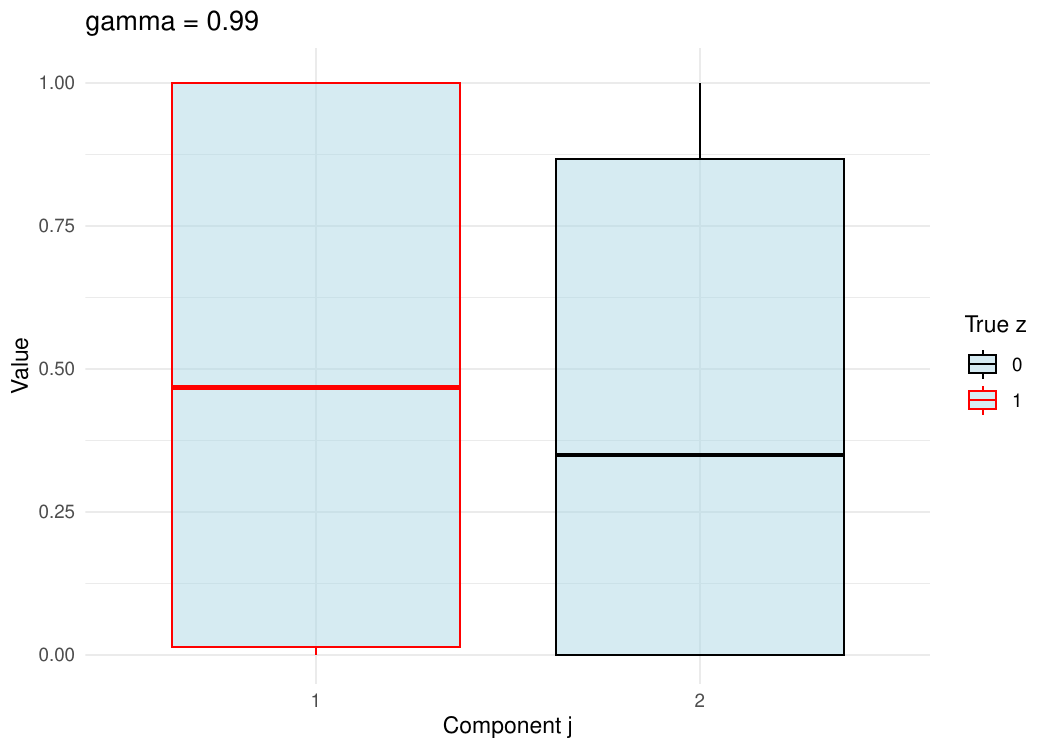}
        \caption{$\gamma = 0.99$}
    \end{subfigure}
    \caption{Linear Model. Proportion of correctly recovered entries of $\mathbf z$ for different values of $\gamma$. Each boxplot summarizes the results over multiple random seeds. As correlation increases, the distinction between informative ($X_1$) and redundant ($X_2$) covariates becomes less clear.}
    \label{fig:correlation_boxplots_acc}
\end{figure}

\begin{figure}[ht!]
    \centering
    \begin{subfigure}[b]{0.38\linewidth}
        \centering
        \includegraphics[width=\textwidth]{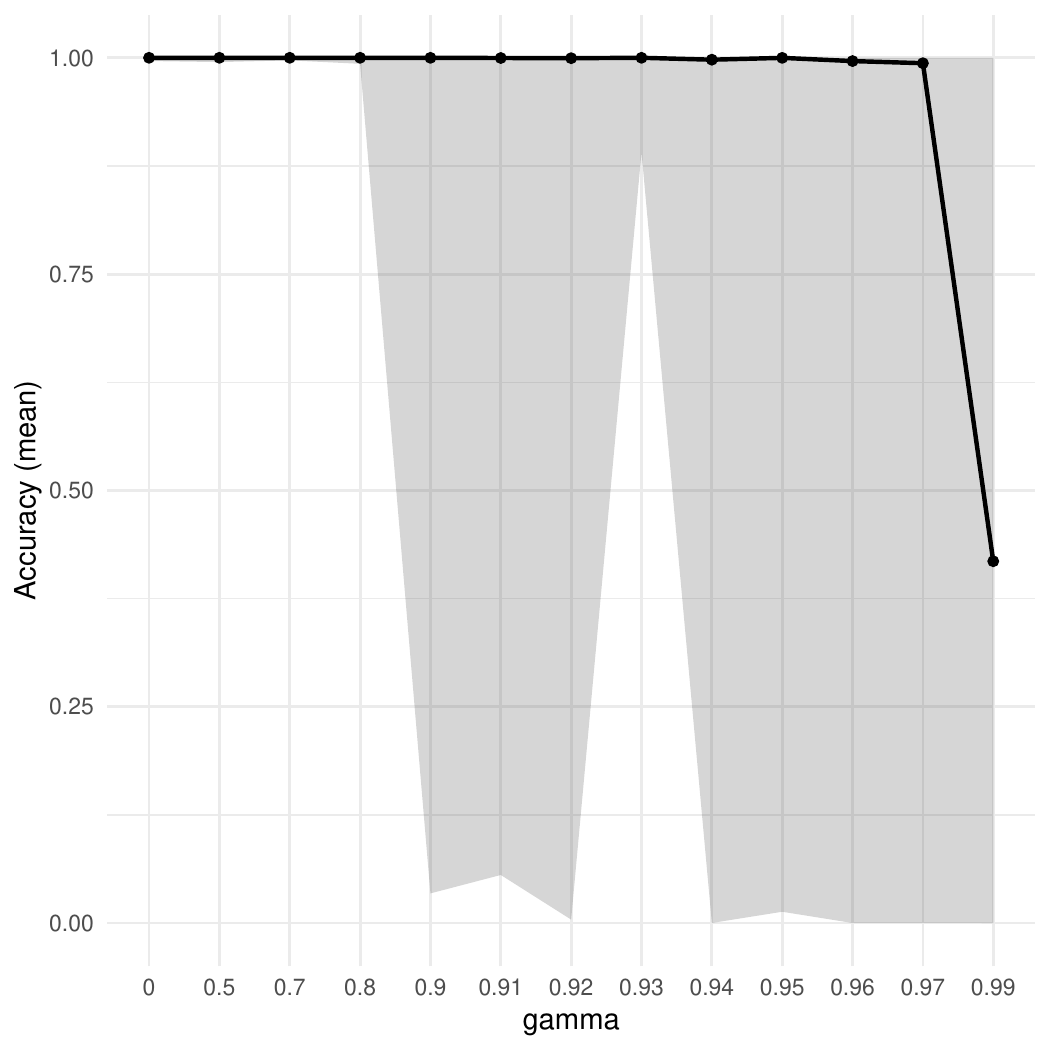}
        \caption{Mean accuracy $\pm$ standard deviation.}
    \end{subfigure}
    \hfill
    \begin{subfigure}[b]{0.38\linewidth}
        \centering
        \includegraphics[width=\textwidth]{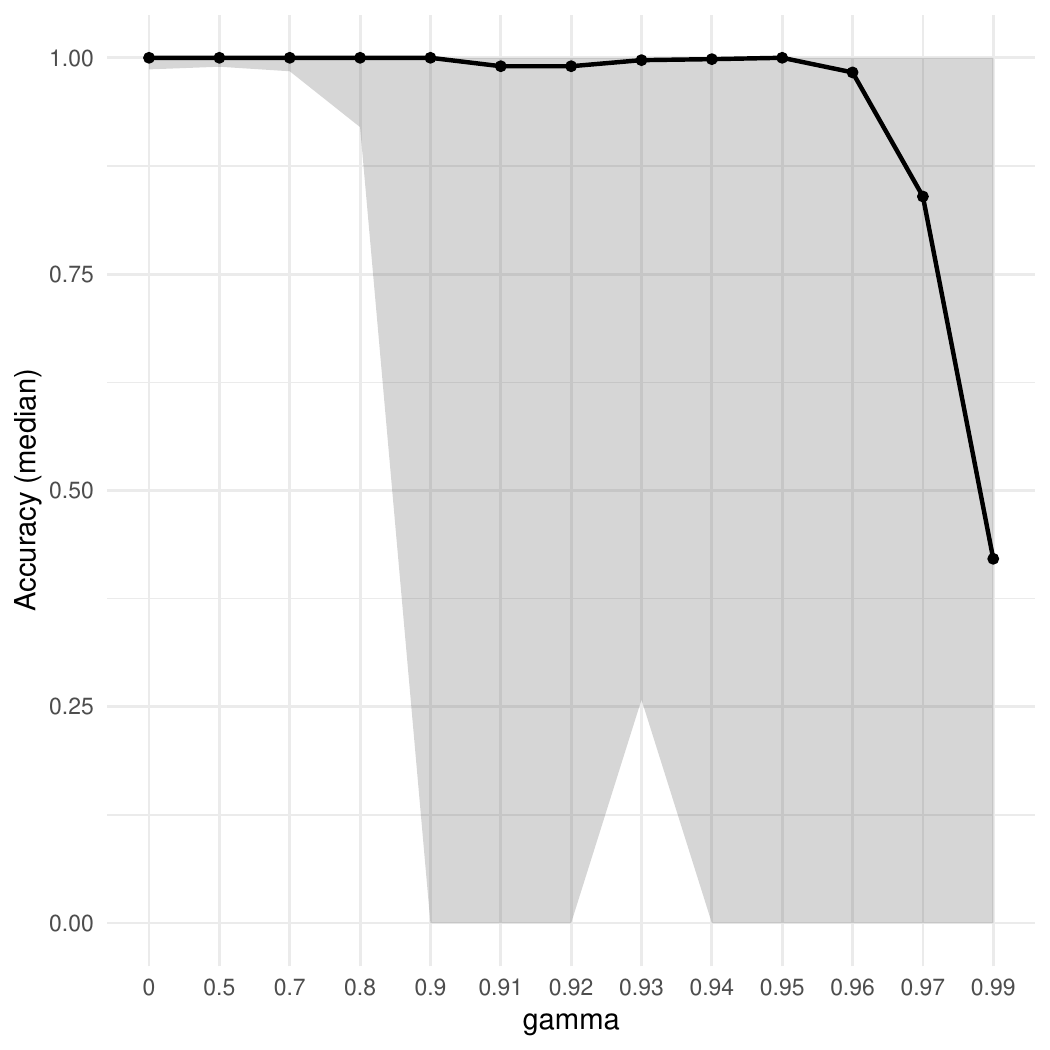}
        \caption{Median accuracy and interquartile range (IQR).}
    \end{subfigure}
    \caption{Linear Model. Overall selection accuracy (proportion of correctly identified entries of $\mathbf{z}$) as a function of the correlation parameter $\gamma$.
    Each point reports the mean or median accuracy over multiple random seeds, while the shaded area represents the corresponding variability. As $\gamma$ increases, the ability to correctly recover the true model decreases, reflecting the growing difficulty in distinguishing correlated predictors.}
    \label{fig:correlation_acc_vs_gamma}
\end{figure}
We observe that both the mean and the median accuracies remain stable until high correlation levels ($\gamma = 0.97$). However, the variability across random seed increases from $\gamma \ge 0.9$, with several runs failing below $0.5$ accuracy. 

\clearpage
\subsection{Poisson Model}
We now analyze a Poisson Model. The chosen hyperparameters of the prior distributions for this case are: $\DparamScalar_0 = 10^{-3}$, $\DparamVector_0 \overset{\iid}{\sim} \mathcal{N}(0,1)$, and $\alpha=1$.
The sample size is $n=100$ and the posterior inference is carried out using Gibbs Sampling with $1000$ iterations and a burn-in of $10\%$, for $20$ multiple random seeds.

Figure~\ref{fig:correlation_boxplots_inclusionprob_poi} reports boxplots of the posterior inclusion probabilities $\Pr(z_j = 1)$ for the two covariates across different values of $\gamma$.
\begin{figure}[ht!]
    \centering
    \begin{subfigure}[b]{0.3\linewidth}
        \centering
        \includegraphics[width=\textwidth]{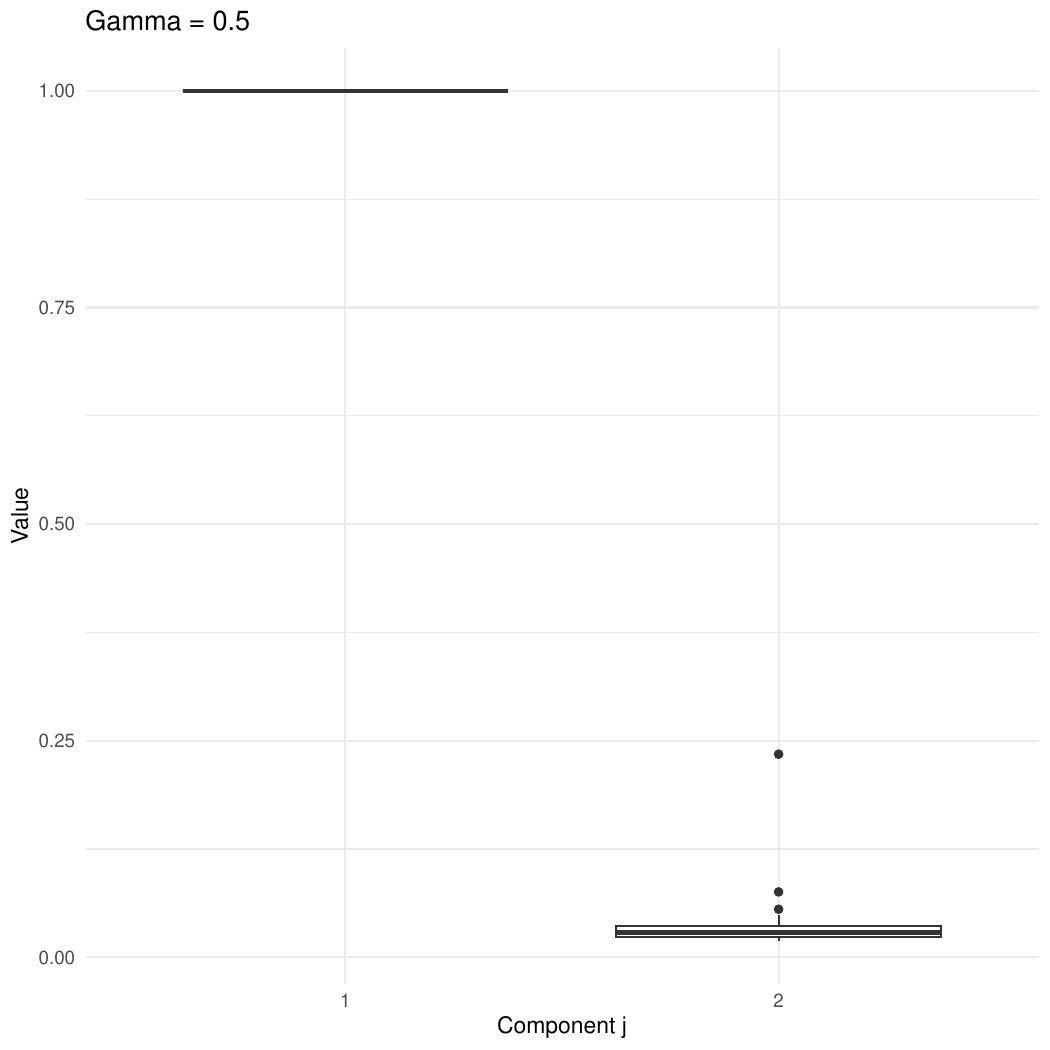}
        \caption{$\gamma = 0.5$}
    \end{subfigure}
    \hfill
    \begin{subfigure}[b]{0.3\linewidth}
        \centering
        \includegraphics[width=\textwidth]{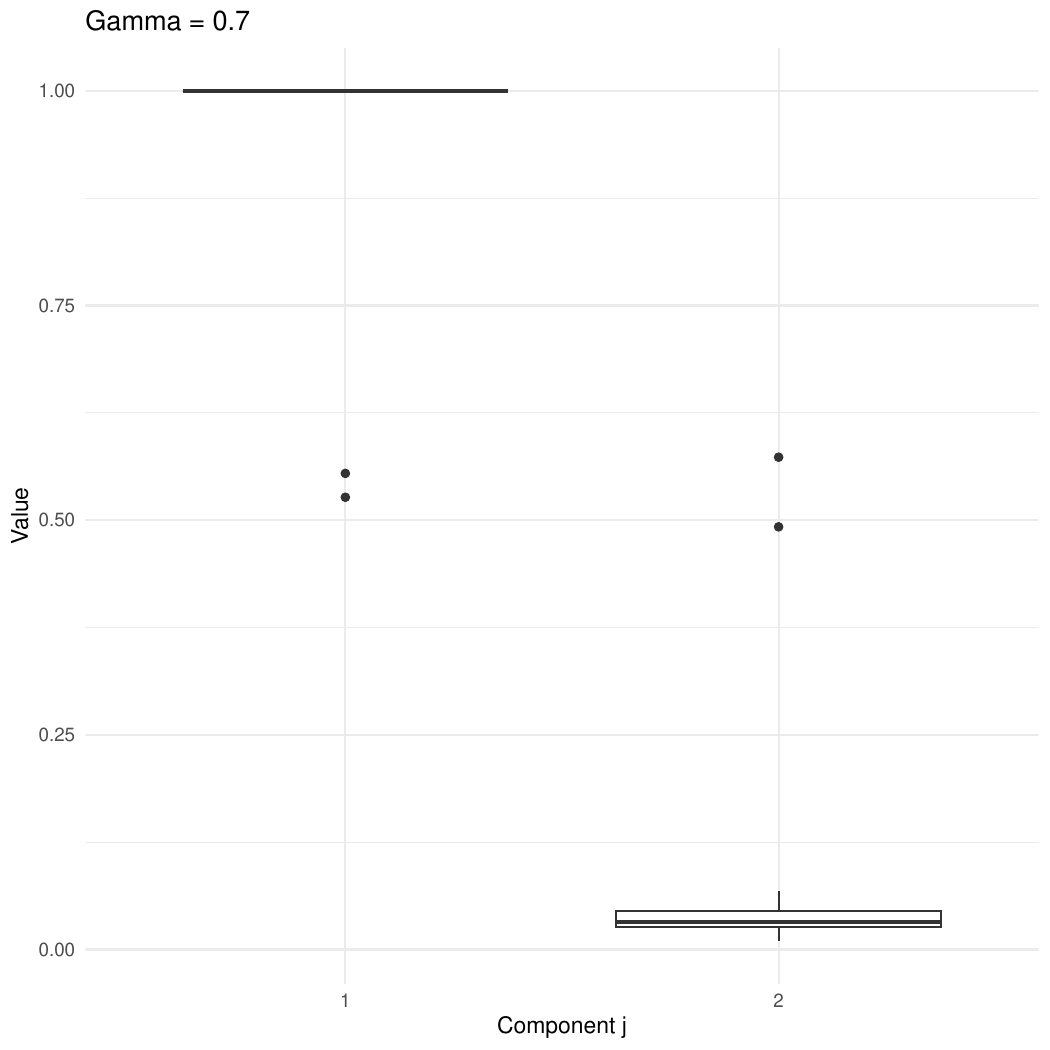}
        \caption{$\gamma = 0.7$}
    \end{subfigure}
    \hfill
    \begin{subfigure}[b]{0.3\linewidth}
        \centering
        \includegraphics[width=\textwidth]{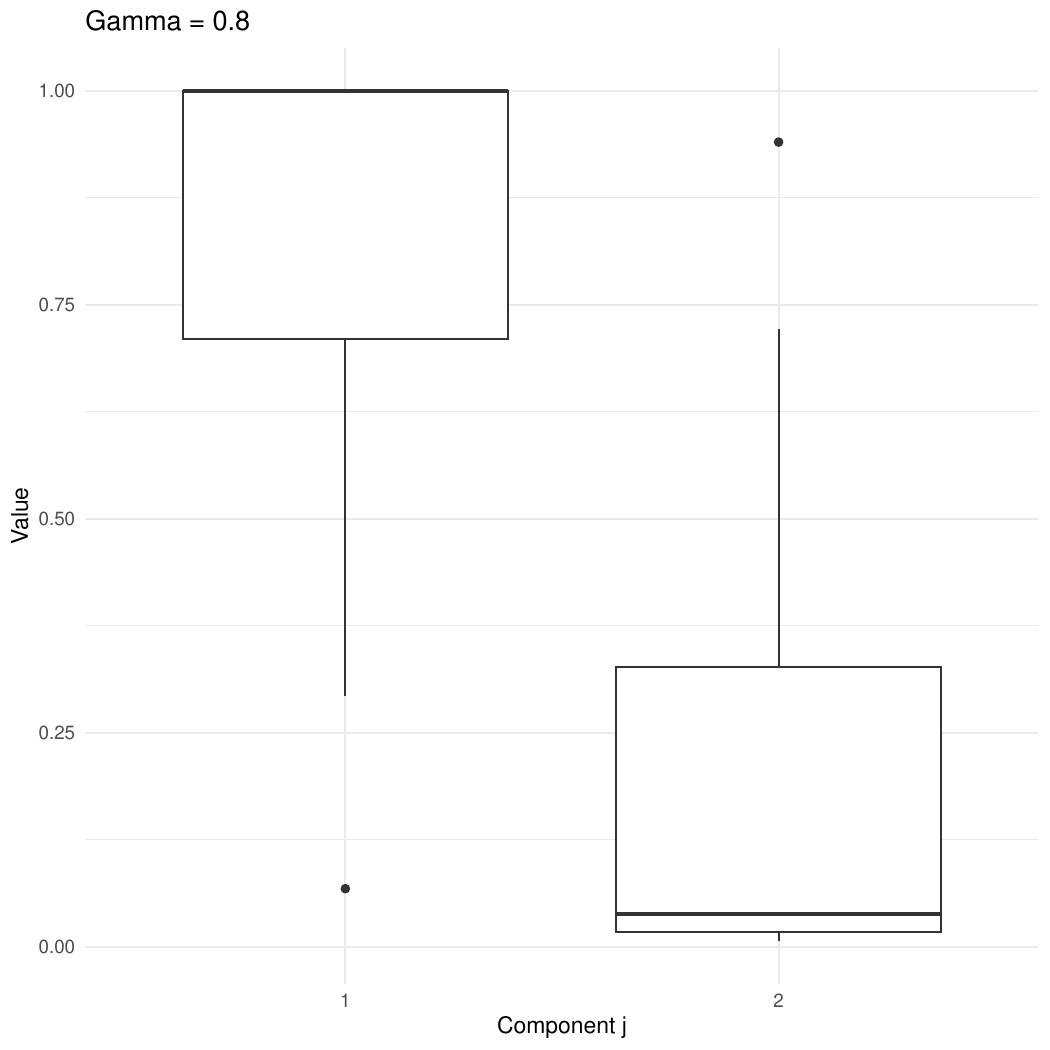}
        \caption{$\gamma = 0.8$}
    \end{subfigure}
    \hfill
    \begin{subfigure}[b]{0.3\linewidth}
        \centering
        \includegraphics[width=\textwidth]{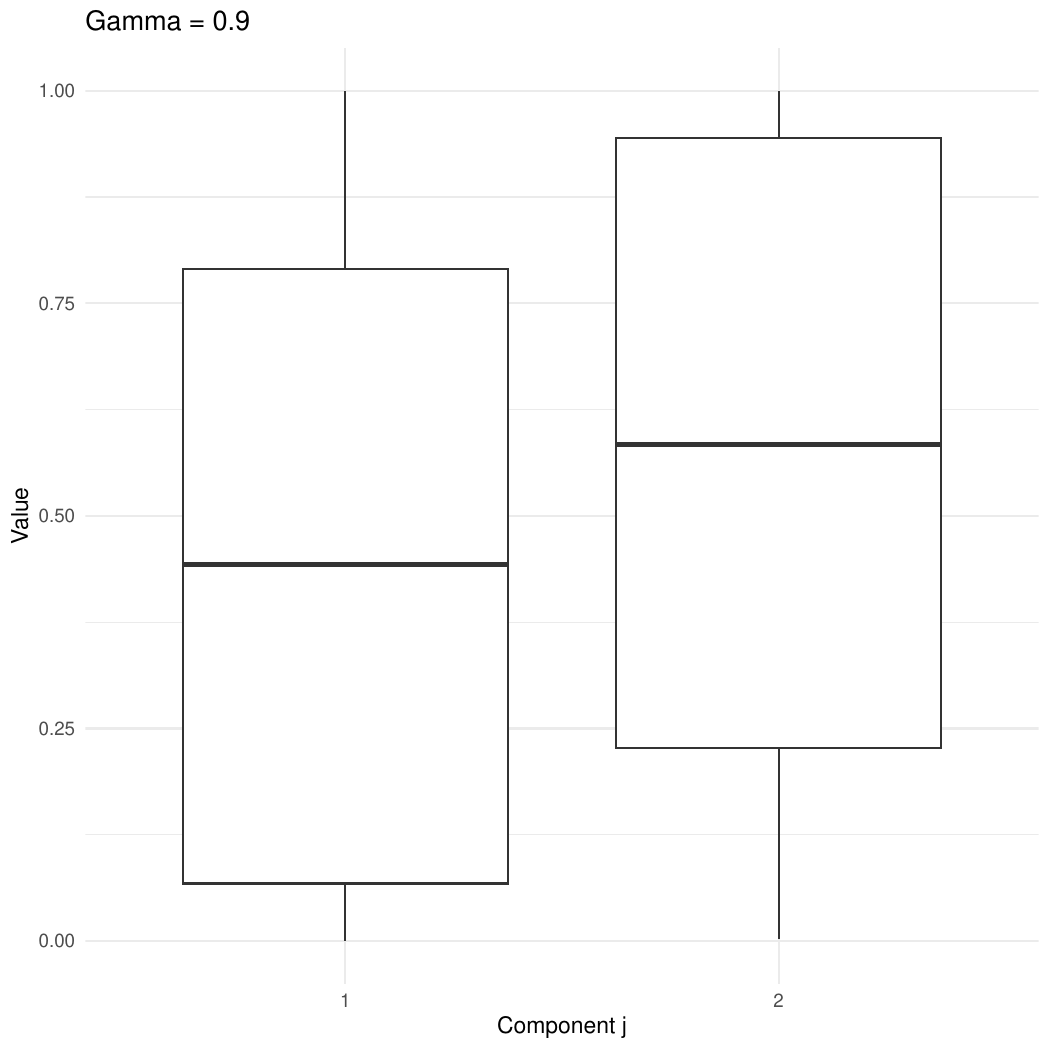}
        \caption{$\gamma = 0.9$}
    \end{subfigure}
    \hfill
    \begin{subfigure}[b]{0.3\linewidth}
        \centering
        \includegraphics[width=\textwidth]{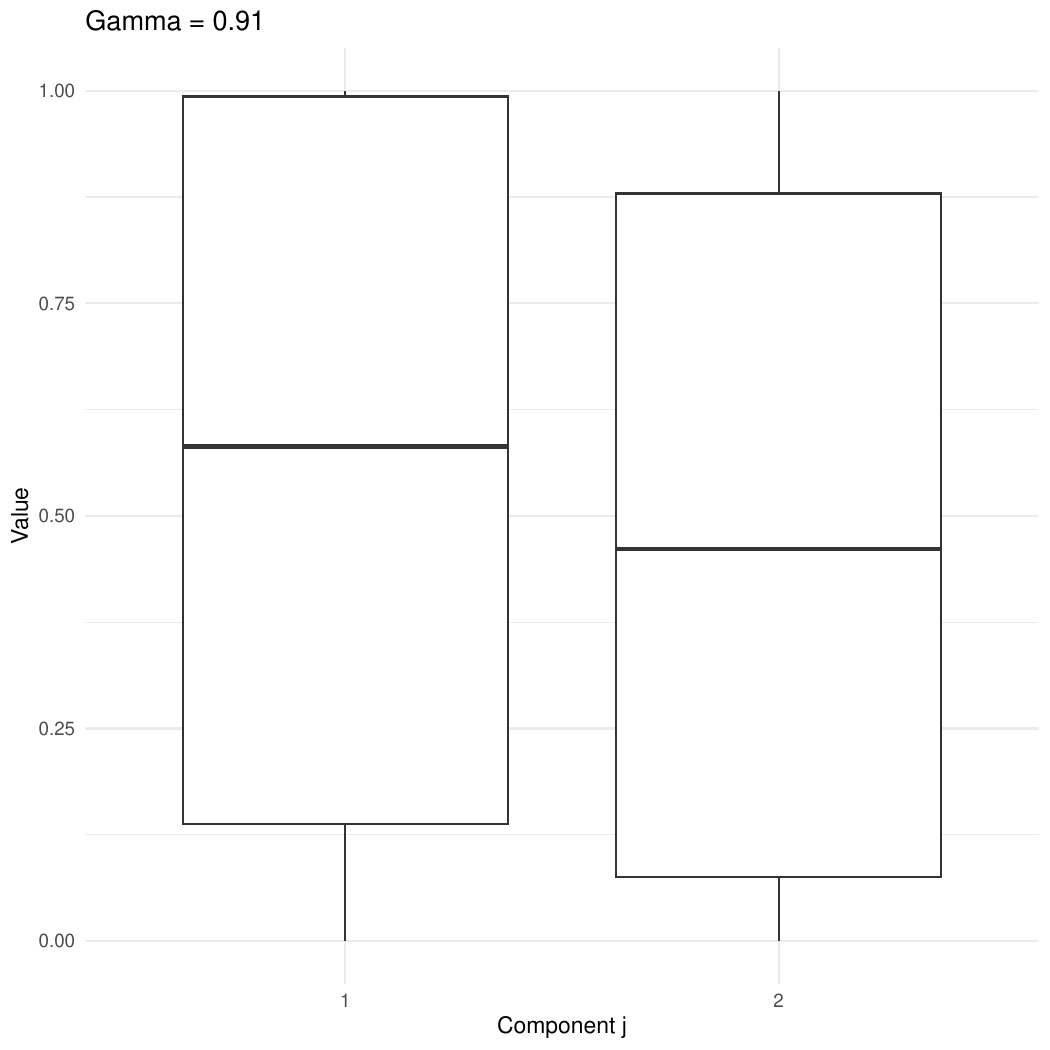}
        \caption{$\gamma = 0.91$}
    \end{subfigure}
    \hfill
    \begin{subfigure}[b]{0.3\linewidth}
        \centering
        \includegraphics[width=\textwidth]{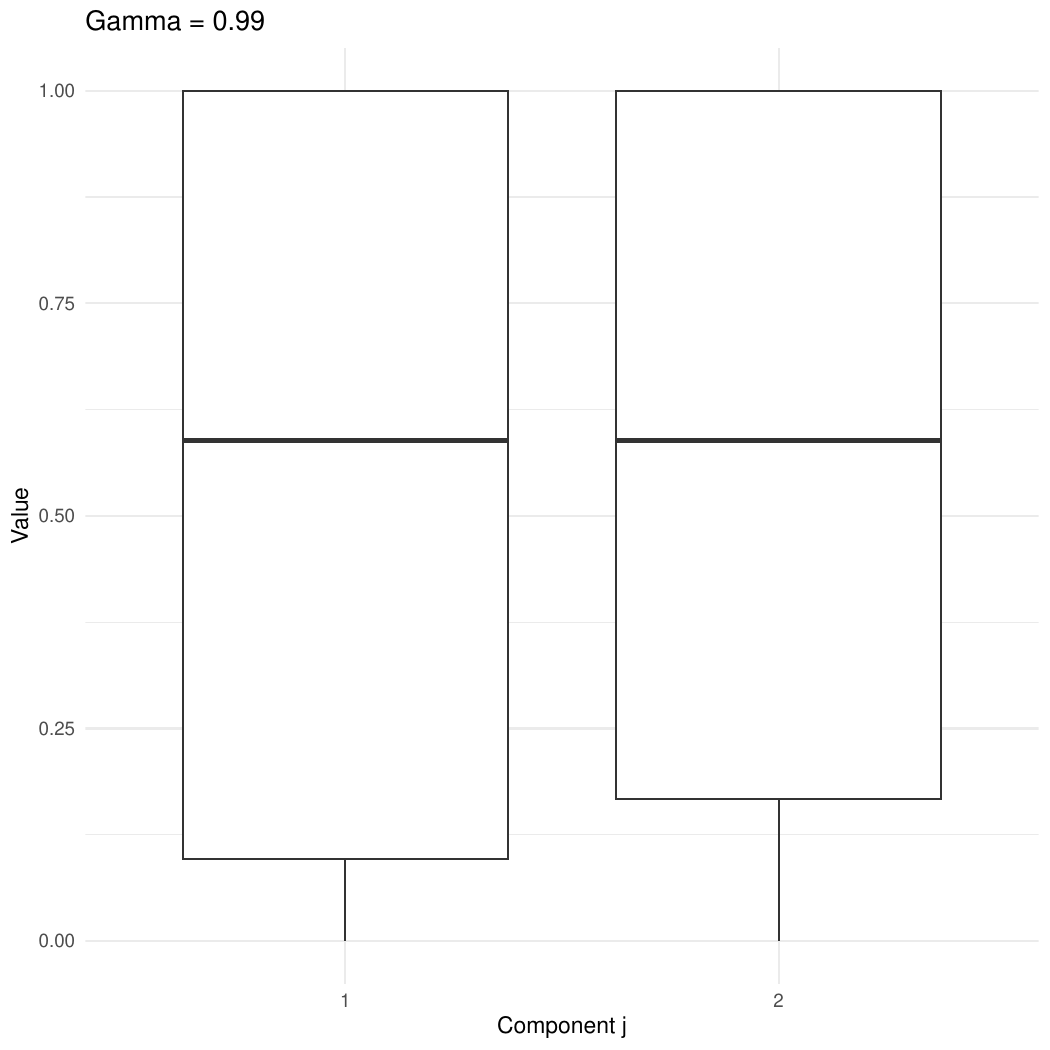}
        \caption{$\gamma = 0.99$}
    \end{subfigure}
    \caption{Poisson Model. Posterior inclusion probabilities $\Pr(z_j = 1)$ for the two covariates across different values of $\gamma$. Each boxplot summarizes the results over multiple random seeds. As correlation increases, the distinction between informative ($X_1$) and redundant ($X_2$) covariates becomes less clear.}
    \label{fig:correlation_boxplots_inclusionprob_poi}
\end{figure}
To capture how overall selection performance varies with the degree of correlation, we report the selection accuracy (proportion of correctly recovered entries of $\mathbf z$) across seeds, for different values of $\gamma$ (Figure~\ref{fig:correlation_boxplots_acc_poi}); and the average and median accuracy across repetitions as $\gamma$ increases (Figure~\ref{fig:correlation_acc_vs_gamma_poi}).

\begin{figure}[ht!]
    \centering
    \begin{subfigure}[b]{0.3\linewidth}
        \centering
        \includegraphics[width=\textwidth]{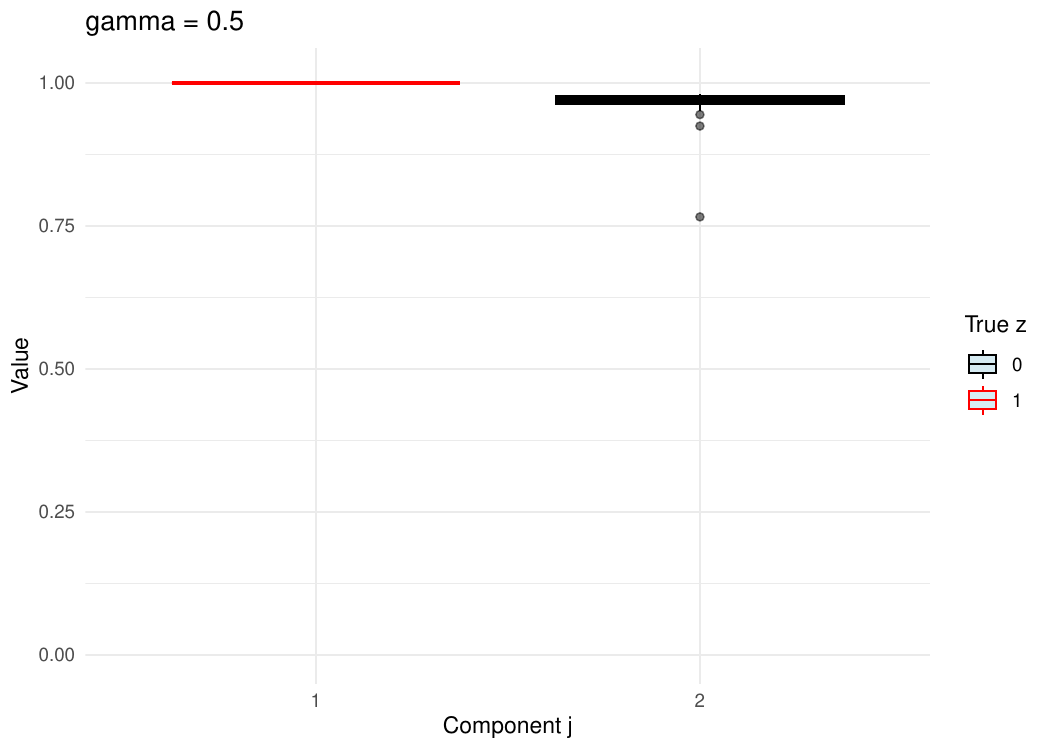}
        \caption{$\gamma = 0.5$}
    \end{subfigure}
    \hfill
    \begin{subfigure}[b]{0.3\linewidth}
        \centering
        \includegraphics[width=\textwidth]{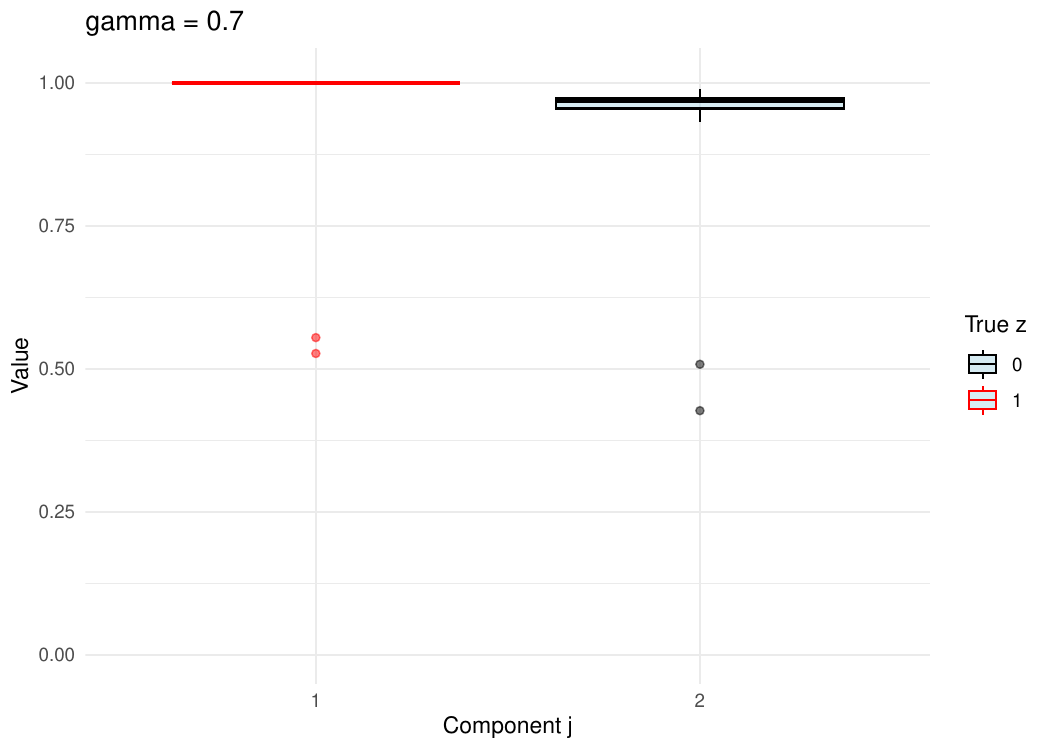}
        \caption{$\gamma = 0.7$}
    \end{subfigure}
    \hfill
    \begin{subfigure}[b]{0.3\linewidth}
        \centering
        \includegraphics[width=\textwidth]{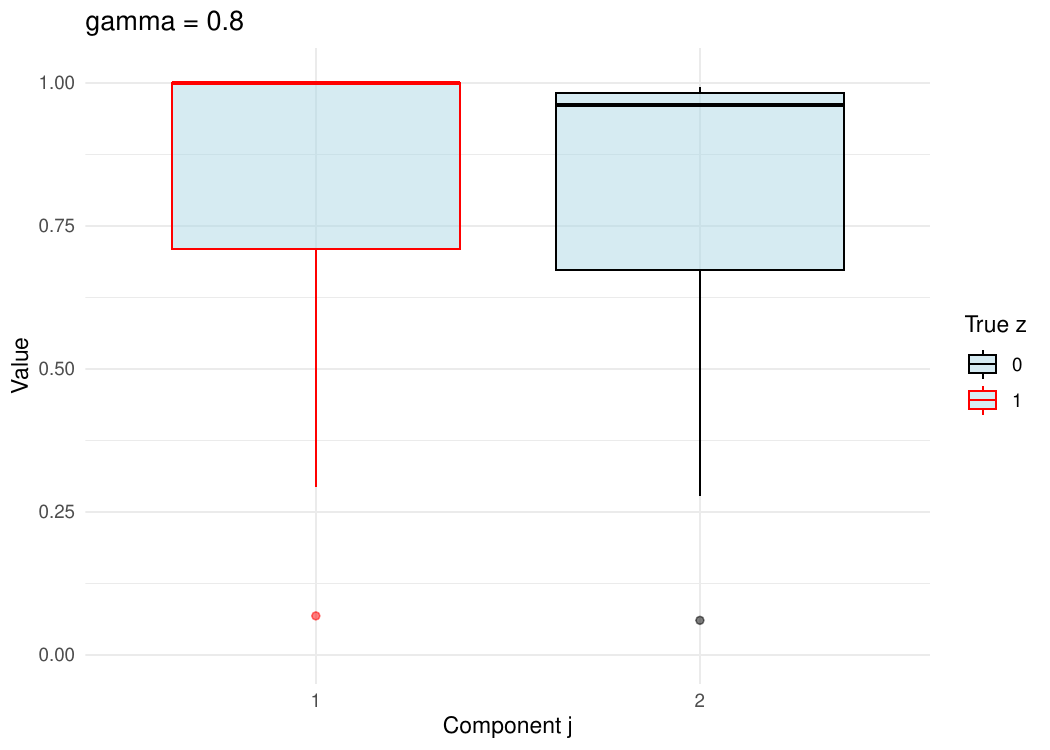}
        \caption{$\gamma = 0.8$}
    \end{subfigure}
    \hfill
    \begin{subfigure}[b]{0.3\linewidth}
        \centering
        \includegraphics[width=\textwidth]{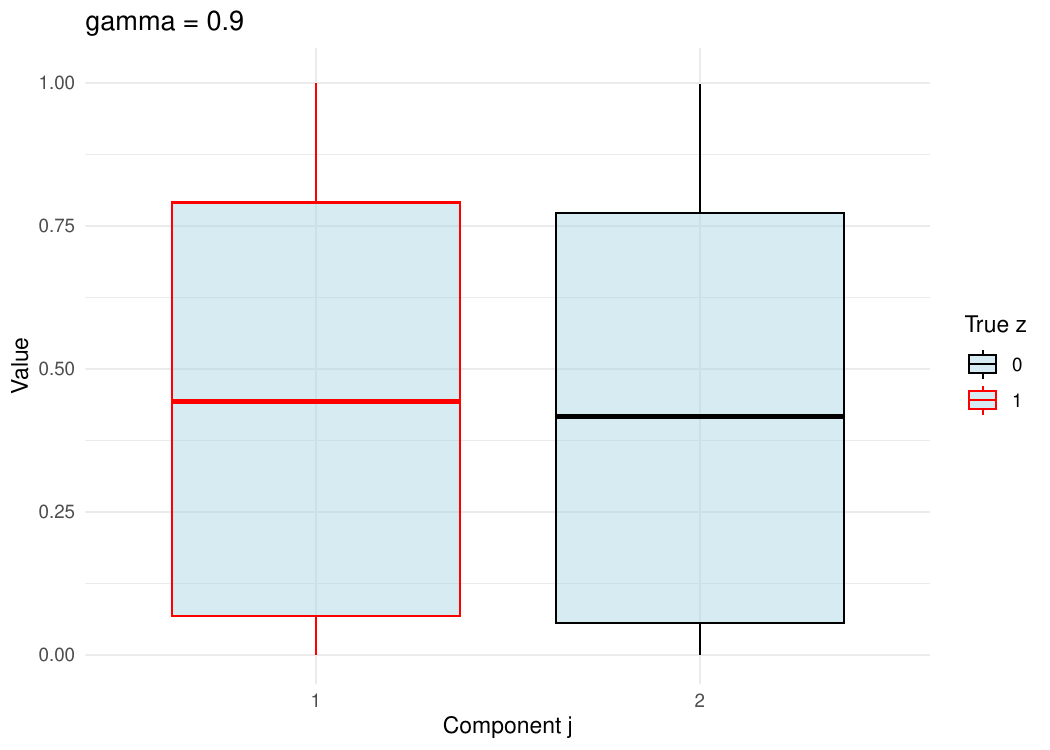}
        \caption{$\gamma = 0.9$}
    \end{subfigure}
    \hfill
    \begin{subfigure}[b]{0.3\linewidth}
        \centering
        \includegraphics[width=\textwidth]{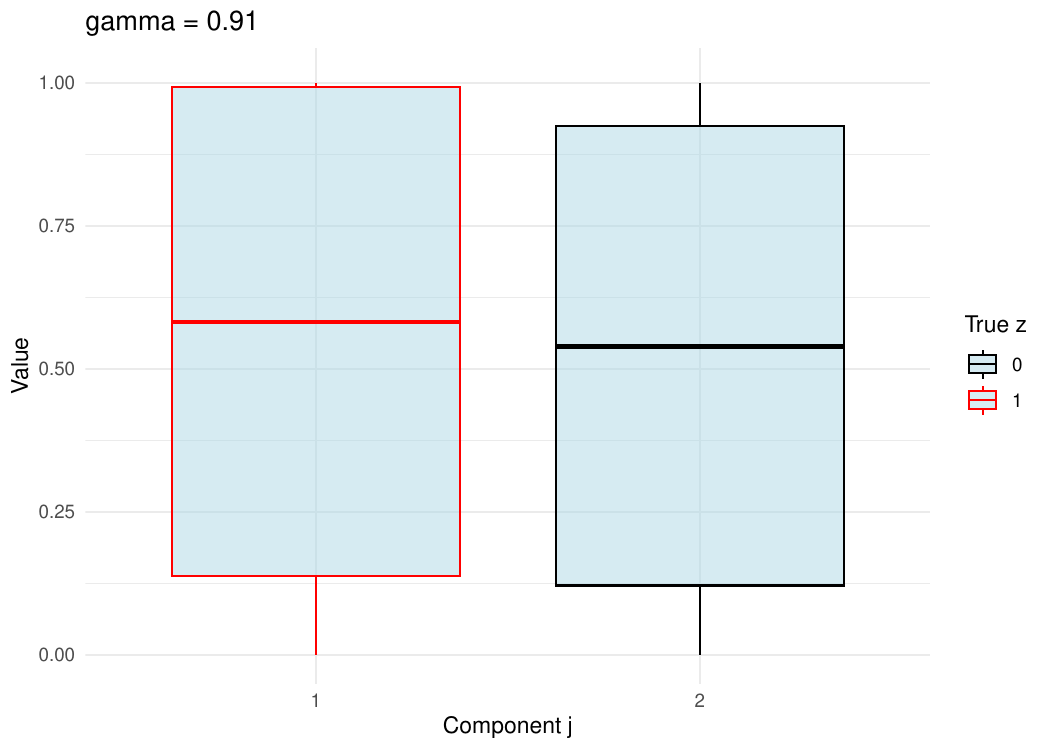}
        \caption{$\gamma = 0.91$}
    \end{subfigure}
    \hfill
    \begin{subfigure}[b]{0.3\linewidth}
        \centering
        \includegraphics[width=\textwidth]{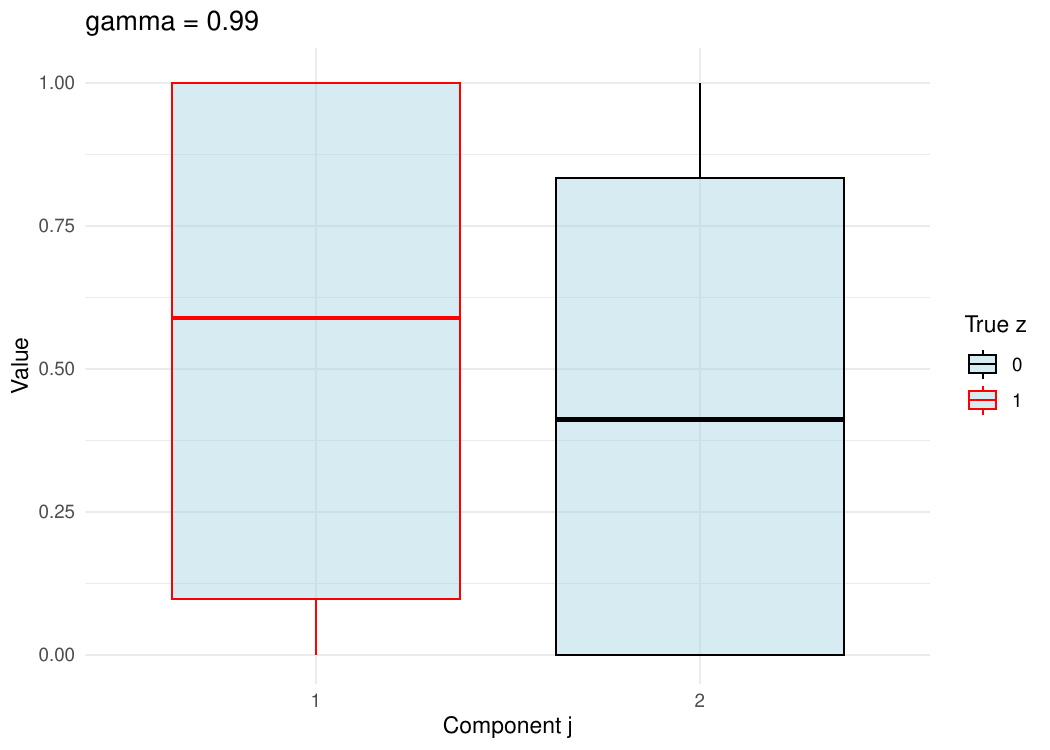}
        \caption{$\gamma = 0.99$}
    \end{subfigure}
    \caption{Poisson Model. Proportion of correctly recovered entries of $\mathbf z$ for different values of $\gamma$. Each boxplot summarizes the results over multiple random seeds. As correlation increases, the distinction between informative ($X_1$) and redundant ($X_2$) covariates becomes less clear.}
    \label{fig:correlation_boxplots_acc_poi}
\end{figure}

\begin{figure}[ht!]
    \centering
    \begin{subfigure}[b]{0.48\linewidth}
        \centering
        \includegraphics[width=\textwidth]{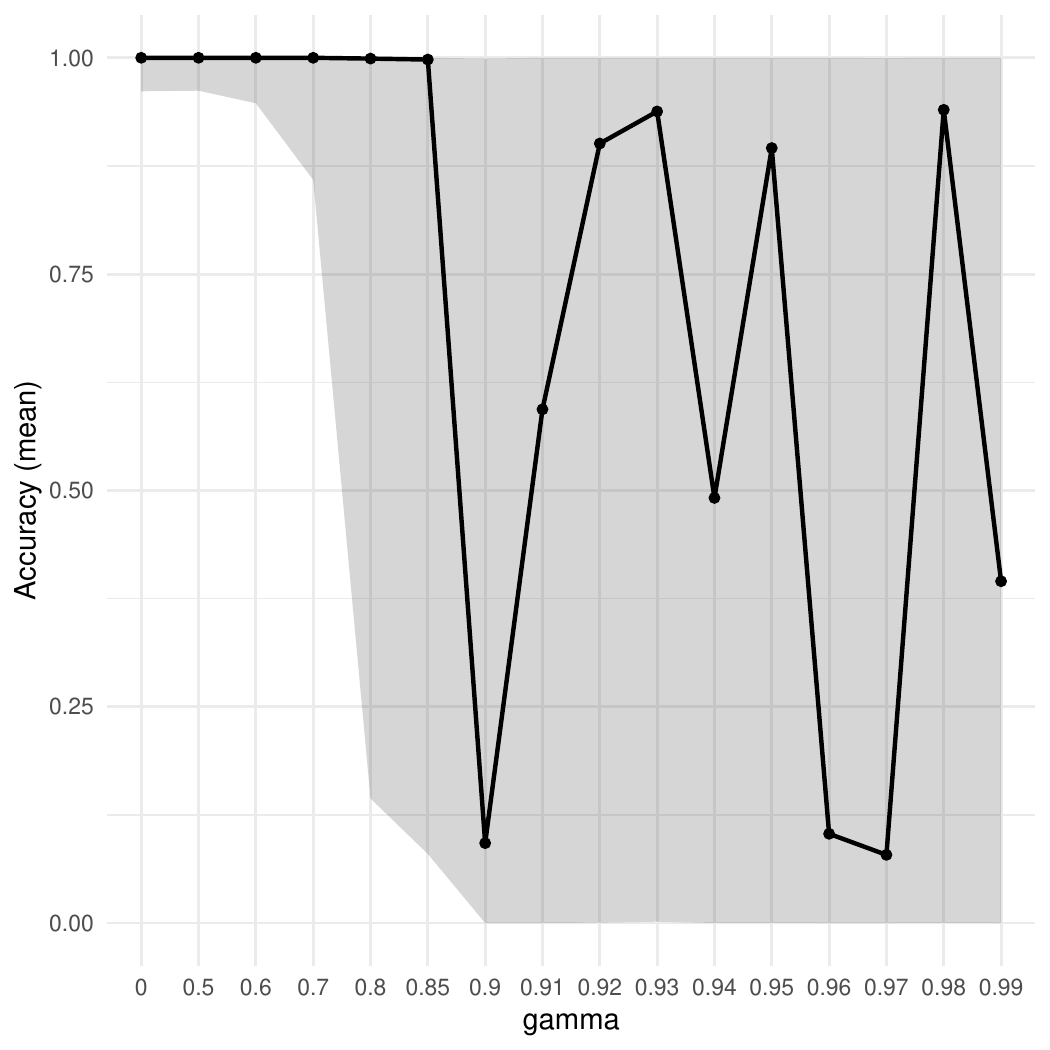}
        \caption{Mean accuracy $\pm$ standard deviation.}
    \end{subfigure}
    \hfill
    \begin{subfigure}[b]{0.48\linewidth}
        \centering
        \includegraphics[width=\textwidth]{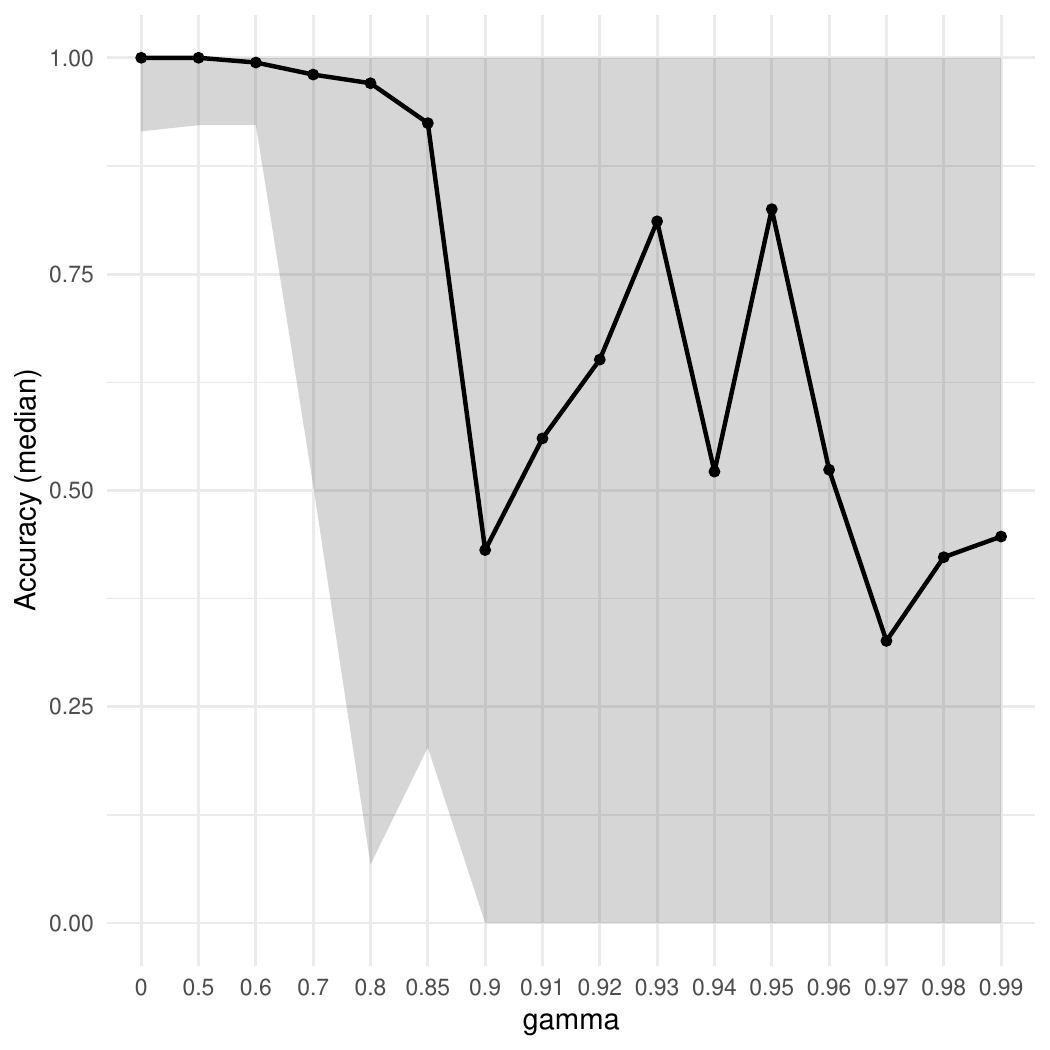}
        \caption{Median accuracy and interquartile range (IQR).}
    \end{subfigure}
    \caption{Poisson Model. Overall selection accuracy (proportion of correctly identified entries of $\mathbf{z}$) as a function of the correlation parameter $\gamma$.
    Each point reports the mean or median accuracy over multiple random seeds, while the shaded area represents the corresponding variability. As $\gamma$ increases, the ability to correctly recover the true model decreases, reflecting the growing difficulty in distinguishing correlated predictors.}
    \label{fig:correlation_acc_vs_gamma_poi}
\end{figure}
In the Poisson case, we observe that both the mean and the median accuracies remain stable until high correlation levels ($\gamma = 0.85$). However, the variability across random seed increases rapidly starting from $\gamma \ge 0.8$, with several runs dropping below $0.5$ accuracy. For higher correlations, results become very unstable ---accuracy values range from $0$ to $1$--- making the outcomes unreliable and preventing meaningful conclusions about variables selection.

\clearpage
\subsection{Logistic Model}
We now analise the Logistic Model, fixing the prior's hyperparameters to $\DparamScalar_0 =  0.05$, $\DparamVector_0 \overset{\iid}{\sim} Ber(0.5)$, and $\alpha=0.2$.
The sample size is $n=500$ and the posterior inference is carried out using Gibbs Sampling with $1000$ iterations and a burn-in of $10\%$, for $12$ multiple random seeds. 

Figure~\ref{fig:correlation_boxplots_inclusionprob_bin} reports boxplots of the posterior inclusion probabilities $\Pr(z_j = 1)$ for the two covariates across different values of $\gamma$.
\begin{figure}[ht!]
    \centering
    \begin{subfigure}[b]{0.3\linewidth}
        \centering
        \includegraphics[width=\textwidth]{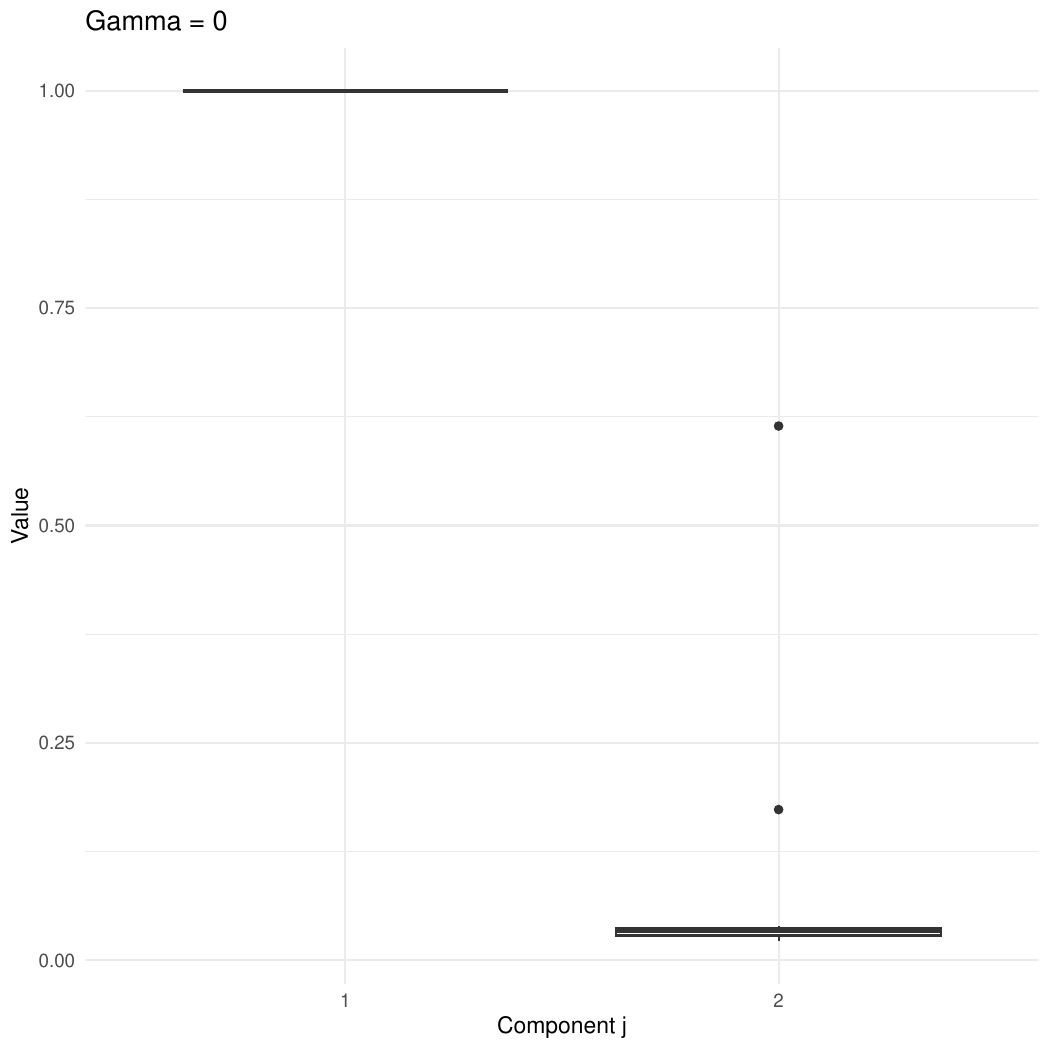}
        \caption{$\gamma = 0$}
    \end{subfigure}
    \hfill
    \begin{subfigure}[b]{0.3\linewidth}
        \centering
        \includegraphics[width=\textwidth]{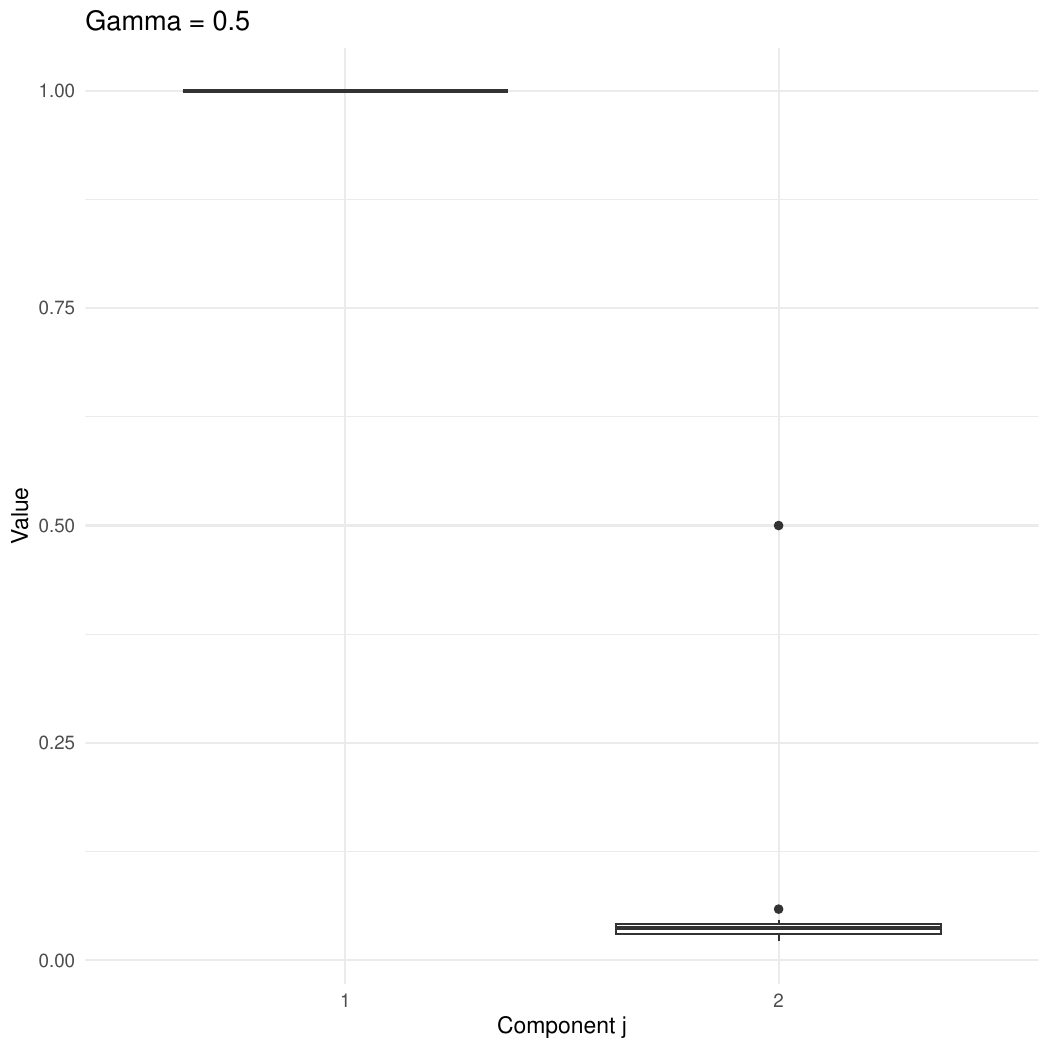}
        \caption{$\gamma = 0.5$}
    \end{subfigure}
    \hfill
    \begin{subfigure}[b]{0.3\linewidth}
        \centering
        \includegraphics[width=\textwidth]{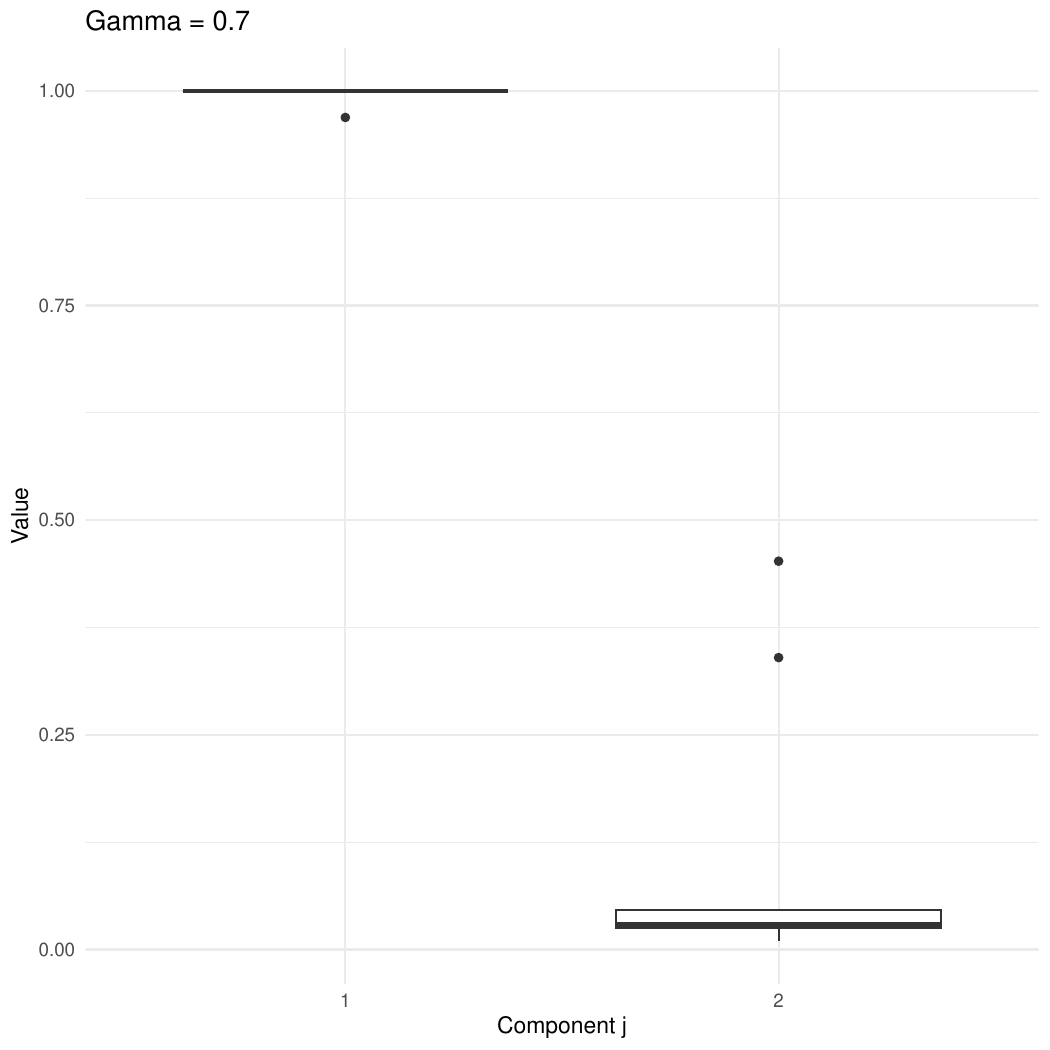}
        \caption{$\gamma = 0.7$}
    \end{subfigure}
    \hfill
    \begin{subfigure}[b]{0.3\linewidth}
        \centering
        \includegraphics[width=\textwidth]{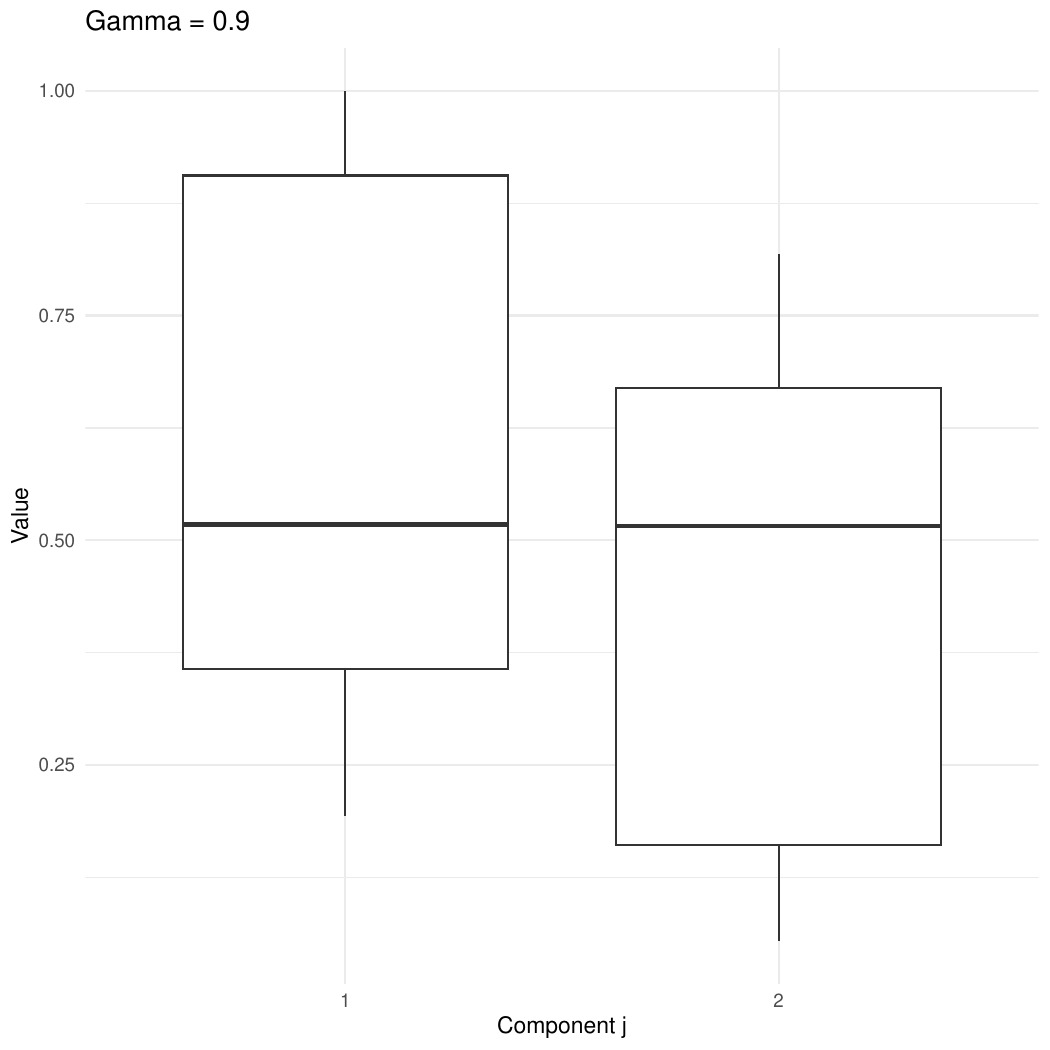}
        \caption{$\gamma = 0.9$}
    \end{subfigure}
    \hfill
    \begin{subfigure}[b]{0.3\linewidth}
        \centering
        \includegraphics[width=\textwidth]{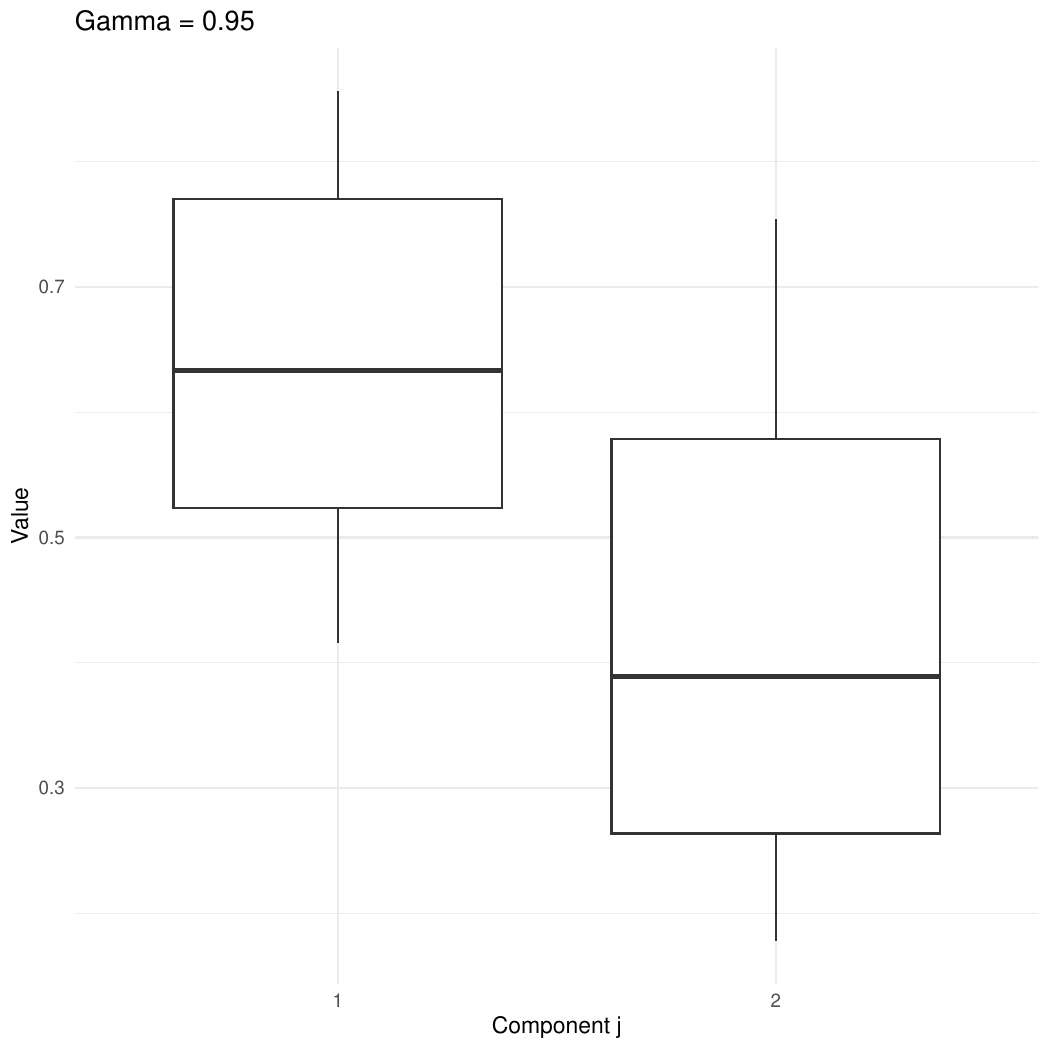}
        \caption{$\gamma = 0.95$}
    \end{subfigure}
    \hfill
    \begin{subfigure}[b]{0.3\linewidth}
        \centering
        \includegraphics[width=\textwidth]{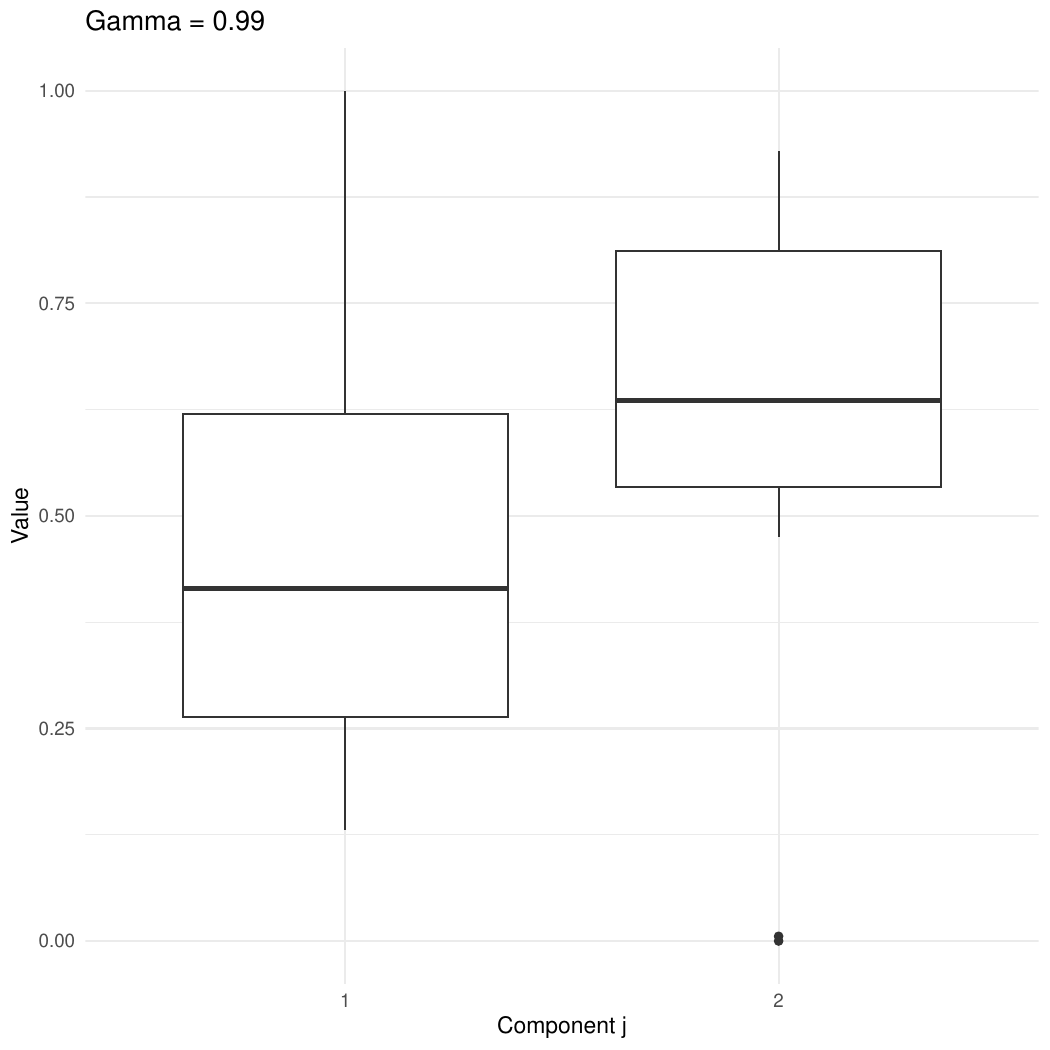}
        \caption{$\gamma = 0.99$}
    \end{subfigure}
    \caption{Logistic Model. Posterior inclusion probabilities $\Pr(z_j = 1)$ for the two covariates across different values of $\gamma$. Each boxplot summarizes the results over multiple random seeds. As correlation increases, the distinction between informative ($X_1$) and redundant ($X_2$) covariates becomes less clear.}
    \label{fig:correlation_boxplots_inclusionprob_bin}
\end{figure}

Figure~\ref{fig:correlation_boxplots_acc_bin} reports the selection accuracy (proportion of correctly recovered entries of $\mathbf z$) across seeds, for different values of $\gamma$; while Figure~\ref{fig:correlation_acc_vs_gamma_bin} reports the average and median accuracy across repetitions as $\gamma$ increases.

\begin{figure}[ht!]
    \centering
    \begin{subfigure}[b]{0.3\linewidth}
        \centering
        \includegraphics[width=\textwidth]{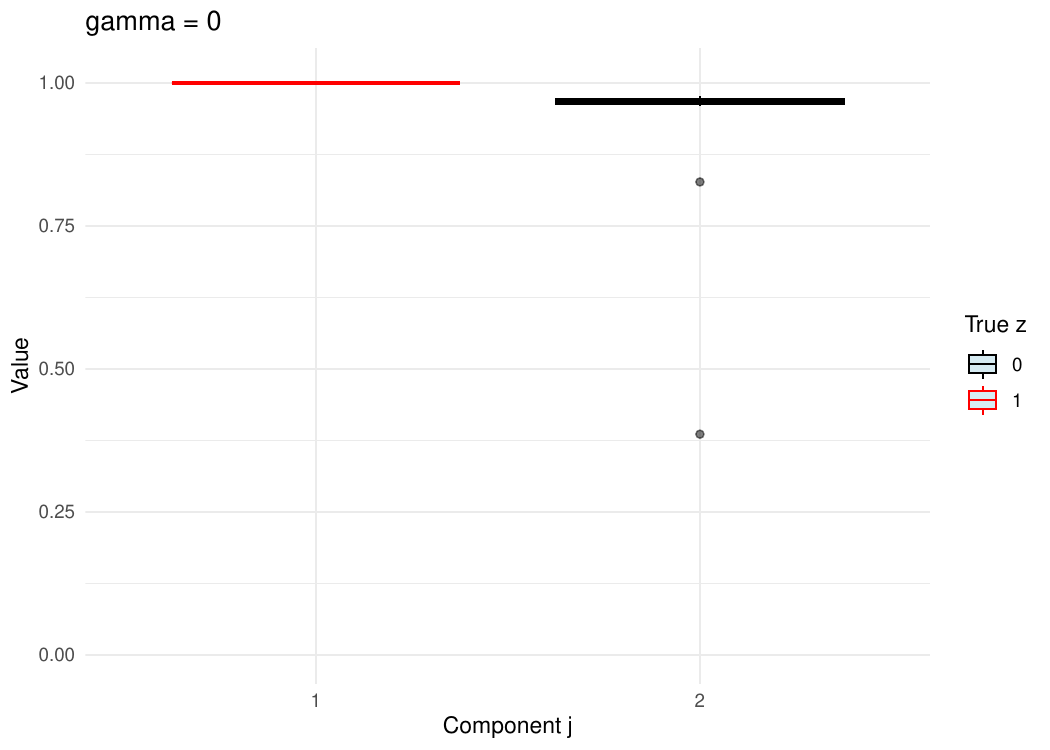}
        \caption{$\gamma = 0$}
    \end{subfigure}
    \hfill
    \begin{subfigure}[b]{0.3\linewidth}
        \centering
        \includegraphics[width=\textwidth]{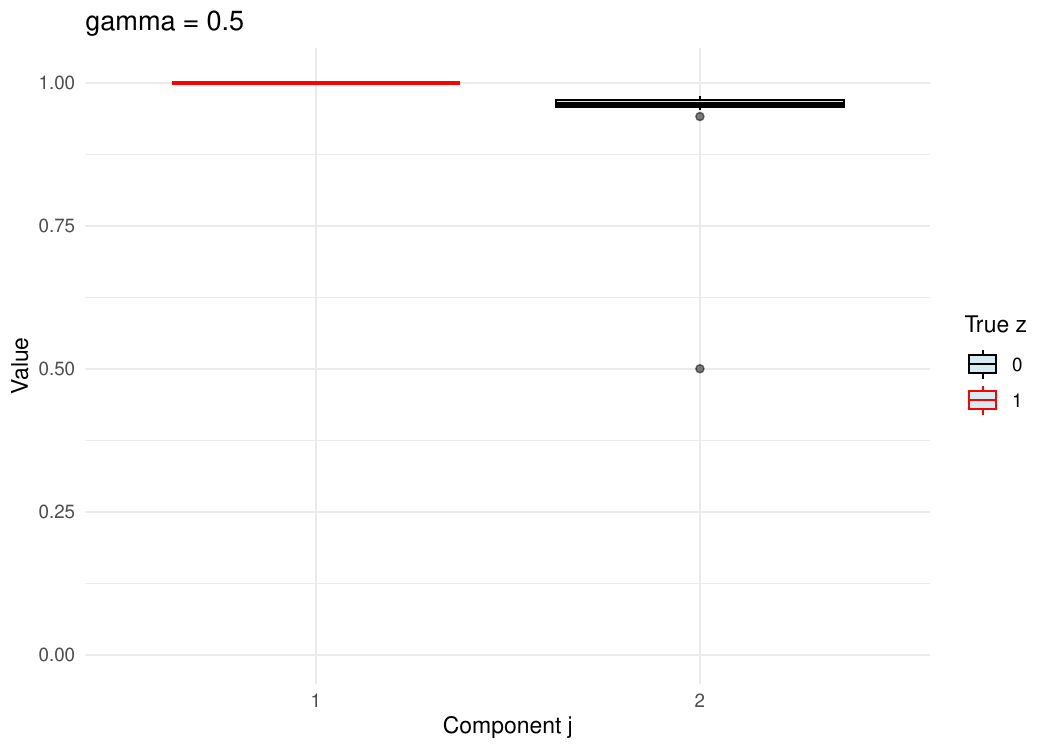}
        \caption{$\gamma = 0.5$}
    \end{subfigure}
    \hfill
    \begin{subfigure}[b]{0.3\linewidth}
        \centering
        \includegraphics[width=\textwidth]{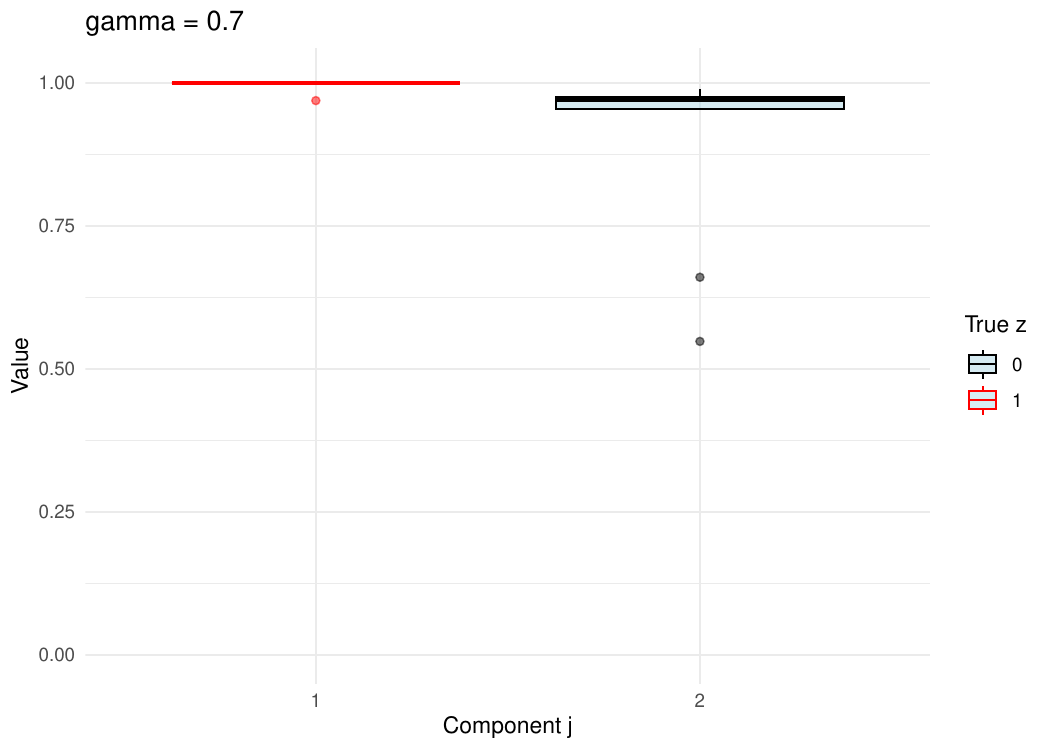}
        \caption{$\gamma = 0.7$}
    \end{subfigure}
    \hfill
    \begin{subfigure}[b]{0.3\linewidth}
        \centering
        \includegraphics[width=\textwidth]{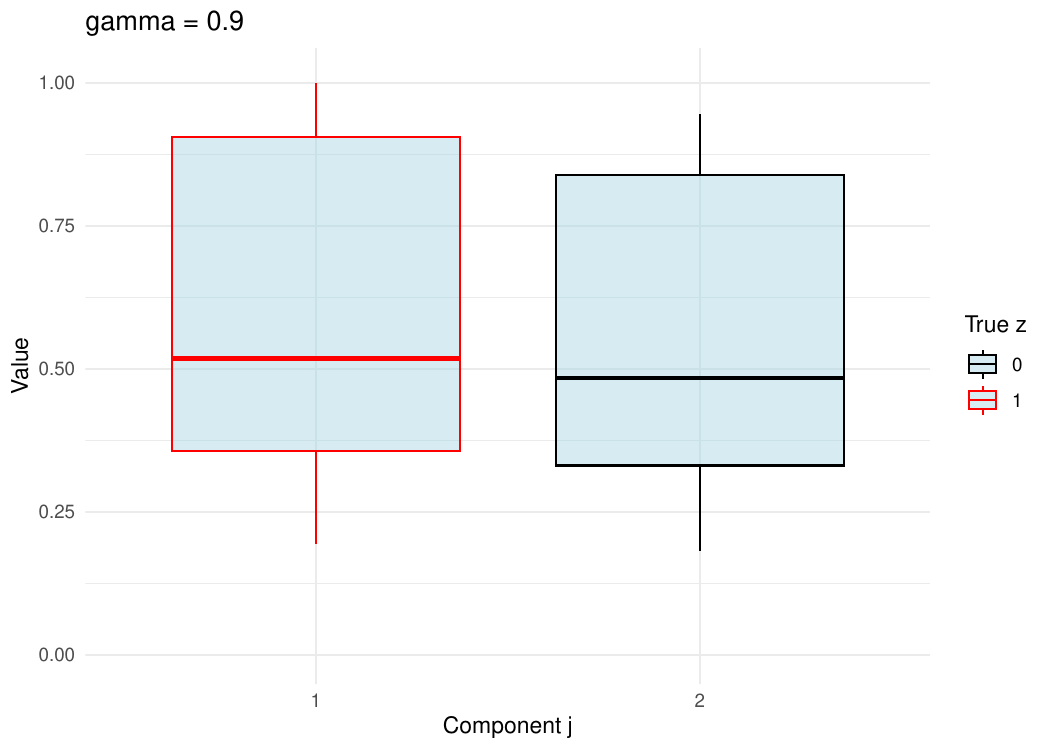}
        \caption{$\gamma = 0.9$}
    \end{subfigure}
    \hfill
    \begin{subfigure}[b]{0.3\linewidth}
        \centering
        \includegraphics[width=\textwidth]{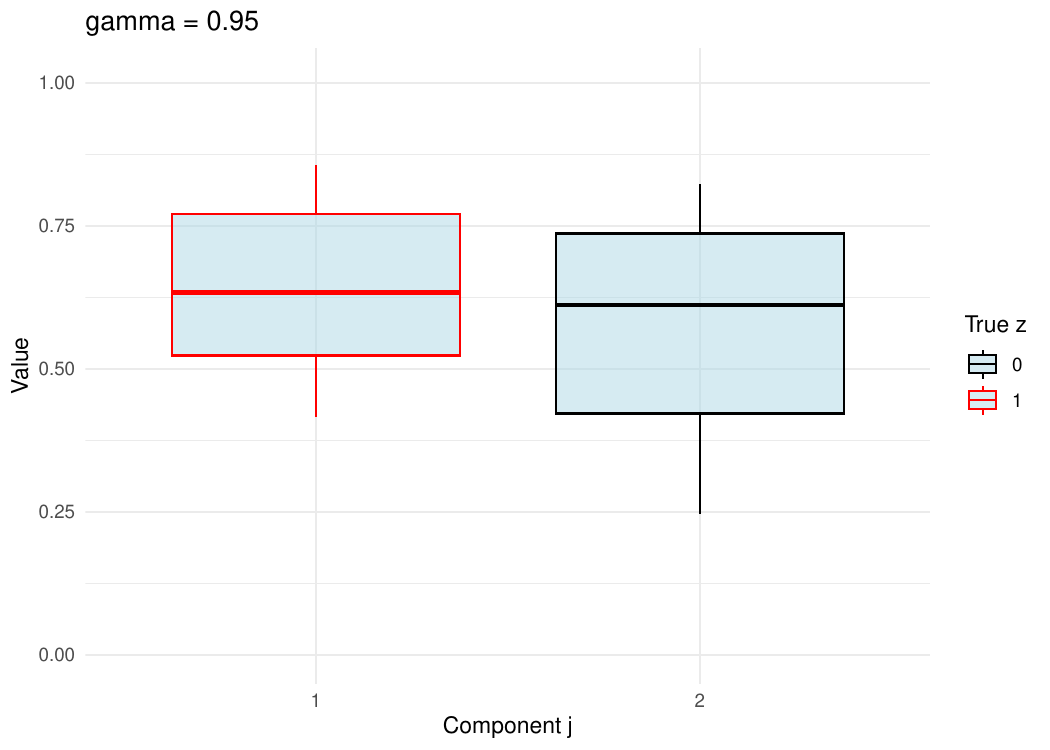}
        \caption{$\gamma = 0.95$}
    \end{subfigure}
    \hfill
    \begin{subfigure}[b]{0.3\linewidth}
        \centering
        \includegraphics[width=\textwidth]{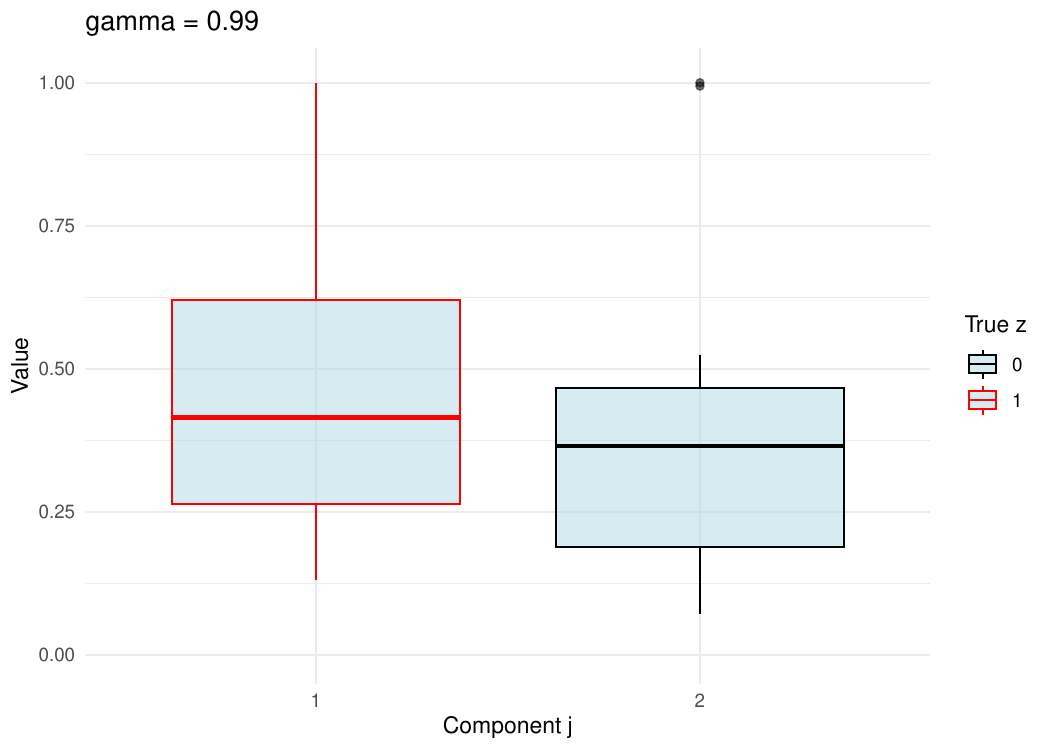}
        \caption{$\gamma = 0.99$}
    \end{subfigure}
    \caption{Logistic Model. Proportion of correctly recovered entries of $\mathbf z$ for different values of $\gamma$. Each boxplot summarizes the results over multiple random seeds. As correlation increases, the distinction between informative ($X_1$) and redundant ($X_2$) covariates becomes less clear.}
    \label{fig:correlation_boxplots_acc_bin}
\end{figure}
\begin{figure}[ht]
    \centering
    \begin{subfigure}[b]{0.48\linewidth}
        \centering
        \includegraphics[width=\textwidth]{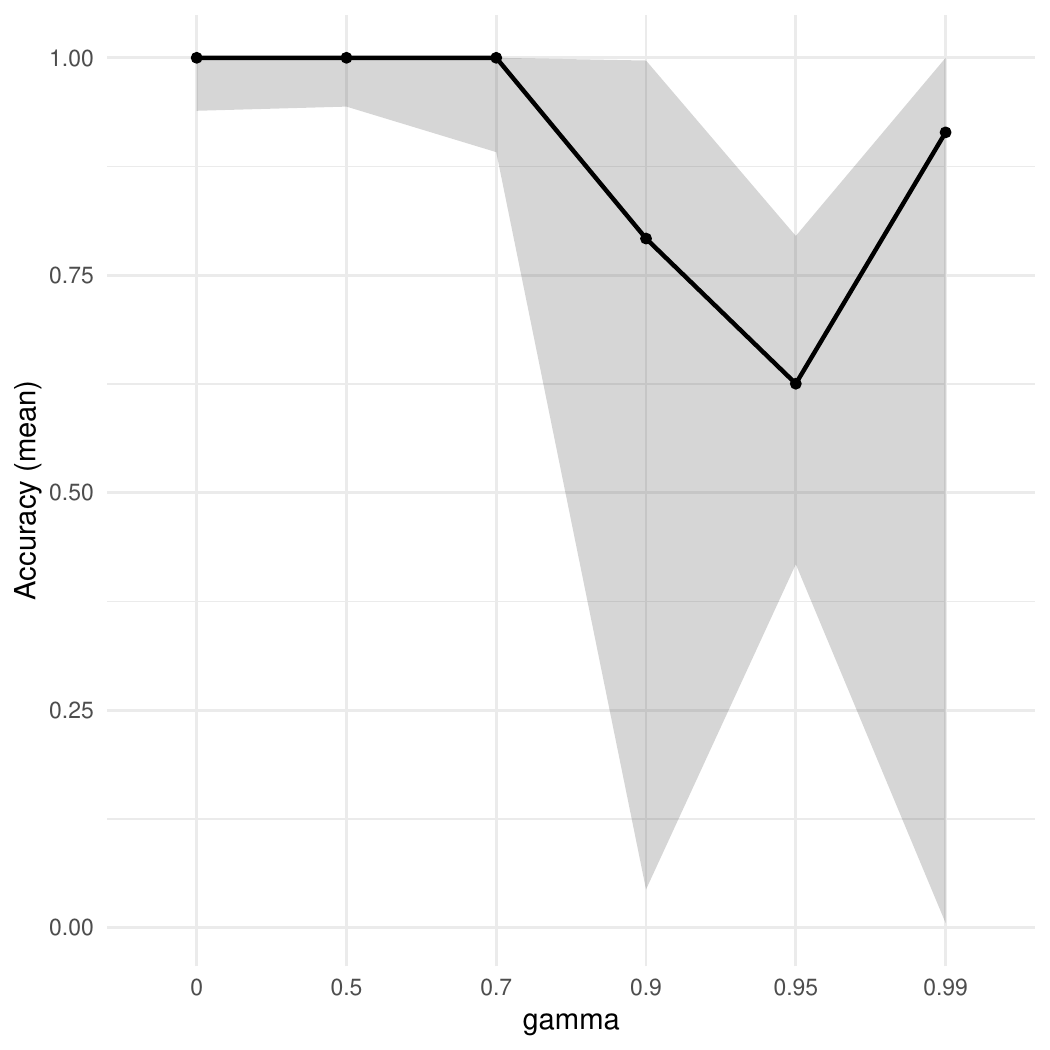}
        \caption{Mean accuracy $\pm$ standard deviation.}
    \end{subfigure}
    \hfill
    \begin{subfigure}[b]{0.48\linewidth}
        \centering
        \includegraphics[width=\textwidth]{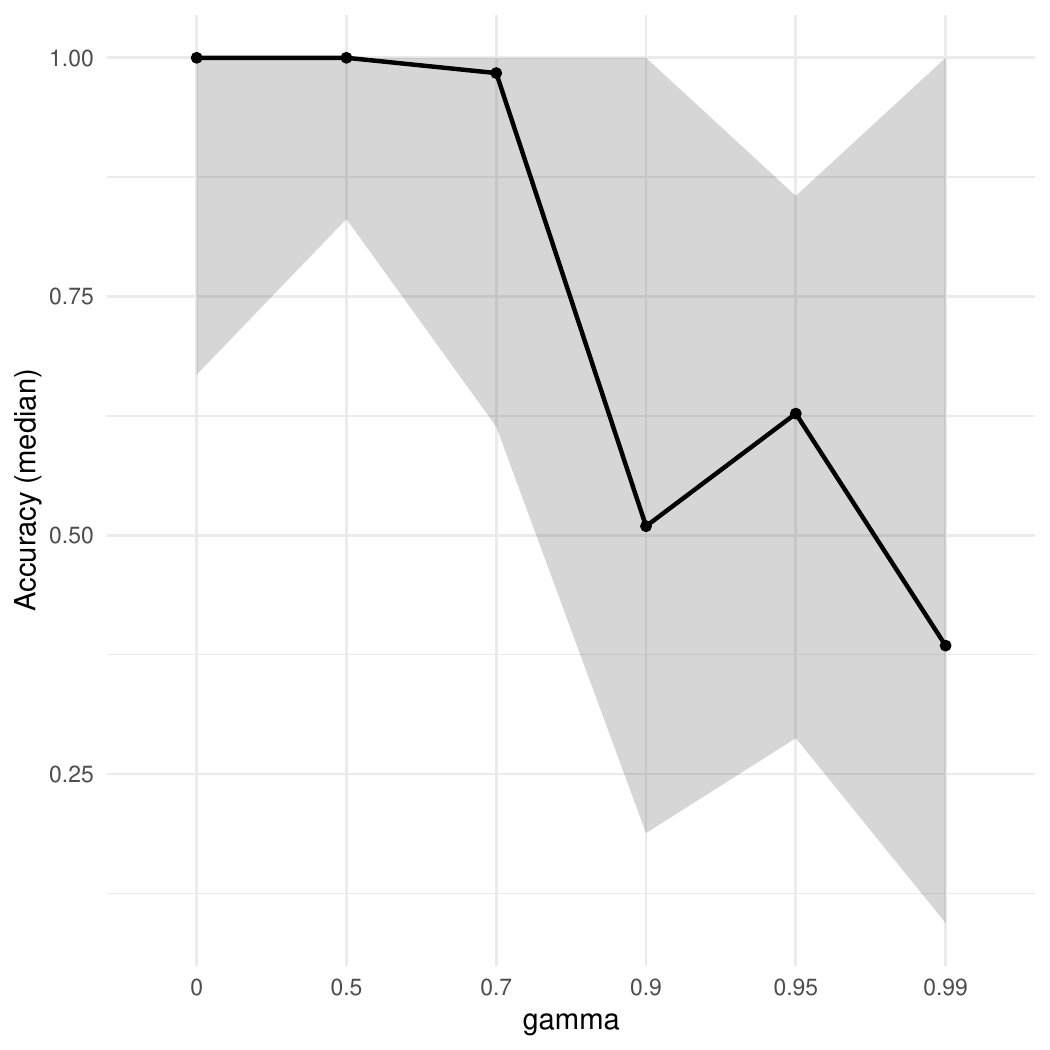}
        \caption{Median accuracy and interquartile range (IQR).}
    \end{subfigure}
    \caption{Logistic Model. Overall selection accuracy (proportion of correctly identified entries of $\mathbf{z}$) as a function of the correlation parameter $\gamma$.
    Each point reports the mean or median accuracy over multiple random seeds, while the shaded area represents the corresponding variability. As $\gamma$ increases, the ability to correctly recover the true model decreases, reflecting the growing difficulty in distinguishing correlated predictors.}
    \label{fig:correlation_acc_vs_gamma_bin}
\end{figure}

\clearpage


\end{document}